\documentclass[article]{siamart190516}
\usepackage[utf8]{inputenc}

\usepackage{todonotes}
\usepackage{enumitem}
\usepackage{amsfonts}
\usepackage{amsmath}
\usepackage{algorithm}
\usepackage{algpseudocode}
\usepackage{hyperref}
\usepackage{multirow}
\usepackage{multicol}
\usepackage{subcaption}
\usepackage[capitalize]{cleveref}
\usepackage{epsfig}
\usepackage{amssymb}
\usepackage{graphicx}
\usepackage{grffile}
\usepackage{appendix}
\usepackage{placeins}

\newcommand{\D}{\displaystyle}
\newcommand{\DF}[2]{\frac{\D#1}{\D#2}}

\newcommand{\ba}[1]{\begin{array}{#1}}

\newcommand{\ea}{\end{array}}
\newcommand{\br}{\textbf{r}}
\newcommand{\Hrm}{{\rm H}}

\newcommand{\beq}[1]{\begin{equation}\label{#1}}
\newcommand{\eeq}{\end{equation}}

\usepackage{color}

\newcommand{\mat}[1]{{#1}}

\newcommand\fpath{./}
\newcommand\tfpath{./}

\allowdisplaybreaks

\newcommand{\ignore}[1]{}
\newcommand{\ylrev}[1]{{\color{black}{#1}}}
\newcommand{\ylrevnew}[1]{{\color{black}{#1}}}

\usepackage{tikz}
\tikzstyle{cnode} = [draw, circle,scale=.6]
\tikzstyle{level 1} = [level distance=.35\textwidth, sibling distance=1\textwidth]
\tikzstyle{level 2} = [level distance=.35\textwidth, sibling distance=.45\textwidth]
\tikzstyle{level 3} = [level distance=.3\textwidth, sibling distance=.25\textwidth]
\usetikzlibrary{decorations.pathreplacing,patterns}

\title{A Fast Butterfly-compressed Hadamard-Babich Integrator for High-Frequency Helmholtz Equations in Inhomogeneous Media with Arbitrary Sources}

\author{Yang Liu\thanks{\ylrev{Applied Mathematics and Computational Research Division}, Lawrence Berkeley National Laboratory, Berkeley, CA, USA.\newline Email: 
		\texttt{liuyangzhuan@lbl.gov}} \and Jian Song \thanks{Department of Mathematics, Michigan State University, East Lansing, MI 48824, USA. Email: 
		\texttt{\{songji12,jqian\}@msu.edu}} \and Robert Burridge \thanks{Department of Mathematics and Statistics, University of New Mexico, Albuquerque, NM 87131, USA. Email: {\tt burridge137@gmail.com}} \and Jianliang Qian\footnotemark[2]}

\begin{document}
	
	\maketitle
	
	\begin{abstract}
		We present a butterfly-compressed representation of the Hadamard-Babich (HB) ansatz for the Green's function of the high-frequency Helmholtz equation in smooth inhomogeneous media. For a computational domain discretized with $N_v$ \ylrev{discretization cells}, the proposed algorithm first solves and tabulates the phase and HB coefficients via eikonal and transport equations with \ylrev{observation points and} point sources located at the Chebyshev nodes \ylrev{using a set of much coarser computation grids}, and then butterfly compresses the resulting HB interactions from all $N_v$ \ylrev{cell centers} to each other. The overall CPU time and memory requirement scale as $O(N_v\log^2N_v)$ \ylrev{for any bounded 2D domains with arbitrary excitation sources. A direct extension of this scheme to bounded 3D domains yields an $O(N_v^{4/3})$ CPU complexity, which can be further reduced to quasi-linear complexities with proposed remedies.} The scheme can also efficiently handle scattering problems involving inclusions in inhomogeneous media. Although the current construction of our HB integrator does not accommodate caustics, the resulting HB integrator itself can be applied to certain sources, such as concave-shaped sources, to produce caustic effects. Compared to finite-difference frequency-domain (FDFD) methods, the proposed HB integrator is free of numerical dispersion and requires fewer discretization points per wavelength. As a result, it can solve wave-propagation problems well beyond the capability of existing solvers. Remarkably, the proposed scheme can accurately model wave propagation in 2D domains with 640 wavelengths per direction and in 3D domains with 54 wavelengths per direction \ylrev{on a state-the-art supercomputer at Lawrence Berkeley National Laboratory}.        
	\end{abstract}
	
	

\section{Introduction}
We are interested in finding the solution to the high-frequency Helmholtz equation with variable refractive index $n(\mathbf{r})$ subject to an arbitrary source:  
\begin{equation}
\label{eq:helm}[ \Delta + \omega^2 n^2(\mathbf{r}) ] u = -s(\mathbf{r})  \quad \mbox{in}\quad R^d, 
\end{equation}
where
\begin{equation}
\Delta = \sum_{i=1}^{d} \DF{\partial^2}{\partial x^2_i}, \quad \mathbf{r} = [x_1, x_2, \ldots, x_d]^T,
\end{equation}
$n(\mathbf{r})$ is the index of refraction (or the slowness function), $s(\mathbf{r})$ is a generic source function with compact support in a bounded domain $V \subset R^d$, $\omega$ is a large angular frequency, $d$ is the dimension, and the Sommerfeld radiation condition is imposed at infinity. When the source is a point source, $s(\mathbf{r})=\delta(\mathbf{r}, \mathbf{r}_0)$, the point-source solution of \eqref{eq:helm} is the Green's function $G(\mathbf{r},\mathbf{r}_0)$ with source location $\mathbf{r}_0$. Assume that we are given a volumetric discretization of the computational domain $V$ with $N_v$ degrees-of-freedom (DOFs) so that the Shannon sampling principle, $\omega h = O(1)$, is obeyed, where $h$ is the volumetric mesh size, implying that $N_v= O(\omega^d)$. Accordingly, numerical discretization of the Helmholtz equation \eqref{eq:helm} by a variety of methods, such as finite-difference, finite-element, integral-equation, or hybrid asymptotic  finite-element, gives rise to an  $N_v\times N_v$ linear system. Ideally, we desire a numerical scheme to have two properties: having almost linear complexity, up to poly-logarithmic factors, in both CPU time and memory storage units (in solving the linear system), and having asymptotically uniform accuracy with respect to $1/\omega$ (at least) as $\omega\rightarrow\infty$ while respecting the Shannon  sampling principle. However, so far, no method is available enjoying the two properties simultaneously in the literature. Observing that the Helmholtz solution can be written as  
\begin{eqnarray}
\label{v2vintro}u(\mathbf{r})=\iint_V G(\mathbf{r},\mathbf{r}_0)s(\mathbf{r}_0)d\mathbf{r}_0
\end{eqnarray} 
if the Green's function $G(\mathbf{r},\mathbf{r}_0)$ is known, we propose to first use the Hadamard-Babich (HB) high-frequency asymptotic ansatz to compute the Green's function, and then use the fast butterfly algorithm to compress the HB integrator, and finally apply the compressed integrator to the source function to obtain the desired Helmholtz solution. As we will see, the resulting new numerical scheme enjoys the two desired properties simultaneously.

{\bf Why to use the HB ansatz.} High-frequency asymptotics typically assumes an expansion series for the Green's function in terms of the phase (or traveltime) and amplitude functions, which satisfy the eikonal and transport equations, respectively. The Eulerian asymptotics solves these equations with partial differential equation (PDE) solvers and utilizes the resulting asymptotic ingredients to construct the Green's function for each point source. However, the usual geometrical-optics ansatz \cite{avikel63} does not yield uniform accuracy near the source as $\omega\rightarrow\infty$ and poses difficulties when initializing the amplitudes. Recently, the HB ansatz \cite{bab65} based Eulerian asymptotics has been developed in \cite{qiayualiuluobur16} which yields a uniform asymptotic solution in the region of space containing a point source but no other caustics. The eikonal and transport equations for the HB coefficients are solved with high-order Lax-Friedrichs weighted non-oscillatory (WENO) schemes which are initialized near the source point with high-order Taylor expansions \cite{qiayualiuluobur16}. The resulting HB integrator is a highly accurate approximation of the Green's function for media with a smooth and analytic refractive index $n(\mathbf{r})$ that does not introduce caustics. That said, for a computational domain discretized with $N_v$ \ylrev{cells}, the Eulerian asymptotics requires solving the eikonal and transport equations $N_v$ times for an arbitrary source function $s(\mathbf{r})$ that can be nonzero across the entire computational domain. Because such an $N_v\times N_v$ discretized HB integrator is still prohibitively expensive to compute, a fast compressed representation is called for. Therefore, we will develop low-rank representations using the Chebyshev interpolation for these HB ingredients. 

{\bf Why to use the butterfly algorithm.} We consider an algebraic compression tool called butterfly \cite{Eric_1994_butterfly,liu2021butterfly,li2015butterfly,Yingzhou_2017_IBF,Pang2020IDBF}, a multilevel numerical linear algebra algorithm well-suited for representing highly oscillatory operators such as Fourier transforms and integral operators \cite{Candes_butterfly_2009,Lexing_SFT_2009,Haizhao_2018_Phase}, special function transforms \cite{Tygert_2010_spherical,bremer2020SHT,Oneil_2010_specialfunction}, free-space \cite{michielssen_multilevel_1996,Eric_1994_butterfly,Yang_2020_BFprecondition}, numerical \cite{liu2021sparse}, and inverse \cite{Han_2013_butterflyLU,Liu_2017_HODBF,Han_2017_butterflyLUPEC,Han_2018_butterflyLUDielectric} Green's functions for Helmholtz and Maxwell's equations. As the 
HB integrator consists of non-oscillatory HB coefficient functions and oscillatory Hankel functions defined via non-oscillatory phase functions, we show that the discretized HB integrator is butterfly compressible. The proposed scheme first constructs low-rank representations for the phase function and HB coefficients via solving the eikonal and transport equations with a set of coarse grids for a constant number of point sources located at the Chebyshev interpolation nodes. Next, it leverages butterfly algorithms and its hierarchical extension, the hierarchical off-diagonal butterfly (HODBF) matrix \cite{Liu_2017_HODBF}, to compress the HB integrator for \ylrev{cell} sizes proportional to the angular frequency, via sampling the phase and HB coefficients in a manageable way. Once compressed, the HB integrator can be applied to any source function as a simple matrix-vector multiplication. This framework is also extended to handle a computational domain with sound-hard inclusion, where an additional surface integral equation using the HB integrator is solved. We analyze our proposed algorithm to validate that the CPU time and memory requirement for most involved discretized integrators scale at most as $O(N_v\log^2N_v)$. Moreover, this scheme obeys the Shannon sampling principle, is free of dispersion errors due to the asymptotic nature of the method, and requires much smaller numbers of \ylrev{cell} points per wavelength than finite-difference solvers, and it further has been distributed-memory parallelized. As a result, it can solve wave propagation problems well beyond the capability of existing FDFD solvers. Remarkably, the proposed scheme can accurately model wave propagation in 2D domains with 640 wavelengths per direction, and 3D domains with 54 wavelengths per direction \ylrev{on a state-the-art supercomputer at Lawrence Berkeley National Lab}.      

\subsection{Related works} 
To put our work into perspective, let us first point out that in the high-frequency regime, the notion of convergence is different from the standard numerical analysis. Because of pollution errors (numerical dispersions) \cite{babsau00}, numerical errors of standard methods for the Helmholtz equation \eqref{eq:helm} do not decay as $\omega \rightarrow \infty$ if $\omega h$ is fixed, namely, the Shannon sampling principle is respected.  Consequently, in the high-frequency asymptotic regime, namely, $\omega\rightarrow \infty$, we seek numerical methods which both converge  asymptotically with respect to $\omega$ and obey the Shannon sampling principle. 

There are three popular classes of numerical methods for solving the variable coefficient Helmholtz (or Maxwell's) equations: the differential-equation method such as finite-difference \cite{chen2013optimal,operto20073d} or finite-element \cite{Babuska:A_Generalized_Finite_Element_Method_for_solving_the_Helmholtz_equation_in_two_dimensions_with_minimal_pollution, Monk_Wang:A_least-squares_method_for_the_Helmholtz_equation} method, the volume-integral equation (VIE) method \cite{schaubert1984tetrahedral,YURKIN2007558,Eikrem2020IterativeLS}, and the hybrid asymptotic finite-element based methods \cite{Giladi_Keller:A_hybrid_numerical_asymptotic_method_for_scattering_problems,DBLP:journals/jcphy/NguyenPRC15,Howarth2014:New_generation_finite_element_methods_for_forward_seismic_modelling,fanqiazepzha17,fanqiazepzha18,lamqia19}. We consider the following three aspects. 

{\bf Accuracy.} The differential-equation method, for instance the finite-difference frequency-domain (FDFD) method, leverages absorbing boundary conditions  and finite-difference stencils to form a sparse $N_v\times N_v$ linear system, whose inverse gives numerical Green's functions. Given a uniform accuracy requirement for all frequencies $\omega$, pollution errors demand that differential-equation methods oversample the numerical solution to mitigate the pollution effects, leading to large-scale systems with sub-optimal DOFs to solve when $\omega\rightarrow \infty$, and the resulting CPU time and memory storage units are sub-optimal with respect to the Shannon principle. On the other hand, the VIE method \cite{Sadeed2022VIE} leverages volumetric equivalent sources and the exact, free-space Green's function to form a dense $N_v\times N_v$ linear system. The solution of the linear system yields the equivalent source densities. Although VIE is almost free of numerical dispersion, the system solve is even more expensive than FDFD. 
The hybrid asymptotic finite-element methods \cite{Giladi_Keller:A_hybrid_numerical_asymptotic_method_for_scattering_problems,DBLP:journals/jcphy/NguyenPRC15,Howarth2014:New_generation_finite_element_methods_for_forward_seismic_modelling,fanqiazepzha17,fanqiazepzha18,lamqia19} incorporate phase or phase gradient information into the formulation, and the hybrid methods in \cite{fanqiazepzha17,fanqiazepzha18} demonstrate almost linear complexity in CPU time and converge asymptotically as $1/\sqrt{\omega}$ when $\omega\rightarrow \infty$ for 2-D problems. We refer to \cite{fanqiazepzha18} for references on various approaches to eliminate or mitigate pollution effects.  Our proposed approach is different from the above three classes, and it is based on the uniform asymptotic HB ansatz and enjoys the following unique feature: given a uniform accuracy requirement for all frequencies $\omega$, the accuracy behaves as $1/\omega$ (at least) asymptotically and shows no apparent dispersion errors when $\omega\rightarrow \infty$, as long as $\omega h$ is fixed. 

{\bf Efficiency.} Given CPU time and storage requirement scalable with respect to $\omega$ by fixing $\omega h$ to be a constant, we would like to solve the linear system $N_v\times N_v$ in $O(N_v)$ time and $O(N_v)$ memory storage, up to poly-logarithmic factors.  
As stated in \cite{fanqiazepzha18}, standard sparse linear algebra algorithms based on nested dissection \cite{GeorgeNested_dissection} and multi-frontal methods \cite{Duff_Reid:The_Multifrontal_Solution_of_Indefinite_Sparse_Symmetric_Linear} have a suboptimal complexity, and they are prohibitively expensive memory-wise in dimension greater than two \cite{Davis:UMFPACK,li_demmel03:SuperLU_DIST,Amestoy_Weisbecker:compressed_MUMPS,Wang_Li_Sia_Situ_Hoop:Efficient_Scalable_Algorithms_for_Solving_Dense_Linear_Systems_with_Hierarchically_Semiseparable_Structures,Bebendorf:2008,Hackbusch:Hierarchical_matrices}.  As a result, quasi-linear-cost   preconditioners are developed, such as \cite{EngquistYing:Sweeping_H,EngquistYing:Sweeping_PML,CStolk_rapidily_converging_domain_decomposition,Chen_Xiang:a_source_transfer_ddm_for_helmholtz_equations_in_unbounded_domain,GeuzaineVion:double_sweep,ZepedaDemanet:the_method_of_polarized_traces}, among many others.  In a most recent work \cite{liu2021sparse}, a sparse approximate multifrontal factorization with butterfly compression for high-frequency wave equations has been developed, and complexity analysis and numerical experiments demonstrate that it enjoys $O(N_v \log^2N_v)$ computation and $O(N_v)$ memory complexity when applied to an $N_v\times N_v$ sparse system arising from 3-D high-frequency Helmholtz problems. However, most of these methods use low-order discretizations so that they require oversampling to produce accurate solutions, thus resulting in suboptimal complexities with respect to the frequency $\omega$. Our proposed approach is also based on butterfly compression, but it enjoys the following unique feature: 
 the overall CPU time and memory requirement scale as $O(N_v\log^2N_v)$ when $\omega\rightarrow \infty$ as long as $\omega h$ is fixed (no oversampling). 

{\bf Accuracy and efficiency.} As stated in \cite{fanqiazepzha18}, only a few references deal with both accuracy and efficiency simultaneously. \cite{Taus_Demanet_Zepeda:HDG_Helmholtz} develops a hybridizable discontinuous Galerkin method coupled with the method of polarized traces, \cite{ZepedaZhao:Fast_Lippmann_Schwinger_solver,YingLiu:Sparsify_and_sweep} deal with an integral version of the Helmholtz equation with inclusions by coupling with sparsification and a fast preconditioner, and \cite{fanqiazepzha17,fanqiazepzha18,lamqia19} develop adaptive discretizations for Helmholtz equations by learning the dominant wave directions. In our work, we also consider an integral version of the Helmholtz equation for curved inclusions. Our proposed approach for both cases of inclusion and no-inclusion enjoys both accuracy and efficiency simultaneously in the sense that the overall CPU time and memory requirement scale as $O(N_v\log^2N_v)$ and the accuracy behaves as $1/\omega$ (at least) asymptotically when 
$\omega\rightarrow \infty$, as long as $\omega h$ is fixed (no oversampling). 

Finally, we remark that a version of the butterfly algorithm \cite{Candes_butterfly_2009} has been implemented in fast Huygens sweeping methods for computing high-frequency Green's functions of point-source Helmholtz equations based on traditional geometrical optics \cite{luoqiabur14a} and the HB ansatz \cite{luqiabur16}, respectively. Our butterfly compressed HB integrator proposed here is different from those in \cite{luoqiabur14a,luqiabur16,qialuyualuobur16} in that we are treating arbitrary sources rather than a single point source and we are using the hierarchical off-diagonal butterfly (HODBF) algorithm \cite{Liu_2017_HODBF}, which is more efficient than the classical butterfly algorithm \cite{Candes_butterfly_2009} or hierarchical matrices \cite{Han_2017_butterflyLUPEC} to compress the high-frequency interaction matrix. We also mention in passing that the necessity of using butterfly compression for fast computing high-frequency waves has been shown in \cite{engzha18}.

\subsection{Contents}
In Section \ref{sec:babich}, we introduce the HB ansatz for the point-source problem of Helmholtz equations. Direct computation of the HB integrator for arbitrary sources is detailed in Section \ref{sec:babich_anysource}. Fast computation of the HB integrator for arbitrary sources is described in Section \ref{sec:fastbabich_anysource}. In Section \ref{sec:example}, we show numerical results to demonstrate the performance of the proposed scheme.

\section{Hadamard-Babich Ansatz for Point Source}\label{sec:babich}
To solve equation \eqref{eq:helm} asymptotically with a point source when $\omega\rightarrow\infty$, Babich \cite{bab65} used Hadamard's method to obtain the following Hankel-based ansatz  so as to expand the solution $u(\mathbf{r})=G(\mathbf{r},\mathbf{r}_0)$, 
\begin{equation}
\label{hb} g_{hb}(\mathbf{r},\mathbf{r}_0) = \sum_{s=0}^{\infty} v_s(\mathbf{r},\mathbf{r}_0) f_{s - (d-2)/2}(\omega,\tau),
\end{equation}
where
\begin{equation}
\label{termf} f_{q}(\omega,\tau) = \mathrm{i} \DF{\sqrt{\pi}}{2} e^{i q \pi} \left( \DF{2 \tau}{\omega} \right)^q \Hrm^{(1)}_{q}(\omega\tau).
\end{equation}
Here $\Hrm^{(1)}_{q}$ is the $q$-th Hankel function of the first kind, and the phase $\tau$, more precisely its square $\tau^2$, satisfies the eikonal-squared equation \eqref{eikonal}, 
\begin{equation}
\label{eikonal}
|\nabla\tau^2|^2 = 4\tau^2n^2,\quad \tau^2(\mathbf{r},\mathbf{r}_0)|_{\mathbf{r}=\mathbf{r}_0} =0,  
\end{equation}
which is equivalent to the usual form of the eikonal equation,  
\begin{equation}
\label{eikonalnormal}
|\nabla\tau| = n, \quad \tau(\mathbf{r},\mathbf{r}_0)|_{\mathbf{r}=\mathbf{r}_0} =0.   
\end{equation}
We remark that two forms of eikonal equations as \eqref{eikonal} and \eqref{eikonalnormal} are needed because $\tau$ itself as a distance function near the source is not differentiable at the source but $\tau^2$ is, as long as $n$ is locally smooth \cite{avikel63,symqia03slowm}.  

We call the above ansatz the Hadamard-Babich ansatz and the coefficients $v_s$ the HB coefficients which are different from amplitude functions in the classical geometrical optics. 

The HB coefficients $v_{s+1}$ in expansion (\ref{hb}) satisfy the recurrent system
\begin{equation}
\label{equVs} 4 \tau n^2 \DF{\partial v_{s+1}}{\partial \tau} + v_{s+1} \left[ \Delta \tau^2 + 2 n^2 (2 s + 2 -d) \right] = \Delta v_s, \quad s = -1, 0, 1, \ldots,
\end{equation}
and $v_{-1} \equiv 0$, where the differentiation $\frac{\partial}{\partial\tau}$ is performed along the ray departing from $\mathbf{r}_0$. Assuming $v_s(\mathbf{r},\mathbf{r}_0)$ to be continuous in the neighborhood of $\mathbf{r}=\mathbf{r}_0$, we get the initial conditions for $v_0$ at $\mathbf{r} = \mathbf{r}_0$,
\begin{eqnarray}
\label{initialv0} v_0(\mathbf{r},\mathbf{r}_0)|_{\mathbf{r}=\mathbf{r}_0} &=& \DF{n_0^{d-2}}{2 \pi^{(d-1)/2}}, \quad n_0 = n(\mathbf{r}_0).
\end{eqnarray}
If $n(\mathbf{r})$ is smooth, then $\tau$ is smooth in the neighborhood of $\mathbf{r}_0$ except at the source point itself, but $\tau^2$ is smooth in the source neighborhood, including the source itself. If  $n(\mathbf{r})$ is analytic, it can be shown \cite{bab65,luqiabur16} that the function $v_0(\mathbf{r},\mathbf{r}_0)$ will also be analytic in $\mathbf{r}$ when $\mathbf{r}$ is in the neighborhood of $\mathbf{r}_0$; furthermore, $v_{s+1}(\mathbf{r},\mathbf{r}_0)$ are determined in terms of $v_0$ and $\tau$ so that $v_{s+1}$ are analytic when $v_0$ are analytic for $s=0, 1, \ldots$. 

\subsection{Assumption and essential estimates}
In the following, we {\bf assume} that these functions $\tau^2$ and $v_s$ ($s=0,1,2, \cdots$) are analytic in the computational domain $V$ for all source points $\mathbf{r}_0\in V$, so that {\it the HB ansatz is valid for all point sources} $\mathbf{r}_0\in V$. This means that all these HB ingredients are smooth single-valued functions in $V$ for any point source in $V$. Although this is a strong assumption, it will serve as a good starting point for many applications.  

By the essence of asymptotics \cite{lax57}, the difference between the true Green's function of \eqref{eq:helm} and the HB ansatz \eqref{hb} can be written as 
\begin{eqnarray}
G(\mathbf{r},\mathbf{r}_0) - g_{hb}(\mathbf{r}, \mathbf{r}_0)  = O({1}/{\omega^\infty}),
\end{eqnarray}
where the ``error'' term on the right-hand side means that the difference can be made arbitrarily smooth for all $\mathbf{r}$, as long as the HB ingredients $\tau^2$ and $v_s$ ($s=0,1, \cdots$) are analytic. 

Moreover, when $|\tau|\leq A\equiv{\rm constant}$, the Hankel-based terms $f_q$ with $q=N-(d-2)/2$ have the following asymptotic forms for large $\omega$ \cite{bab65}, 
\begin{equation}\label{eq:Babich_f_asymptotics}
f_q(\omega, \tau) = \left \{
\begin{array}{ll}
O \left( \left( \frac{\tau}{\omega} \right)^q \left( \omega \tau \right)^{-1/2} \right) = O(\omega^{-q-1/2} \tau^{q-1/2}) , & \quad \mbox{if }  \omega \tau \geq A_1 \equiv \mbox{constant}, \\

O \left(\ln (\omega \tau) +1 \right), &\quad  \mbox{if }\omega \tau \leq A_1 \mbox{ and } q = 0, \\

O \left( \left( \frac{\tau}{\omega} \right)^q \left( \omega \tau \right)^{-q} \right) = O(\omega^{-2q}) , & \quad \mbox{if } \omega \tau \leq A_1 \mbox{ and } q \geq 1,
\end{array}  \right .
\end{equation}
where $N=0,1,2,\ldots$, and both $A$ and $A_1$ are arbitrarily fixed positive constants. 

In this article, we truncate the HB ansatz \eqref{hb} to obtain a $(N+1)$-term expansion, 
\begin{equation}
\label{babich} g(\mathbf{r},\mathbf{r}_0) = \sum_{s=0}^{N} v_s(\mathbf{r},\mathbf{r}_0) f_{s - (d-2)/2}(\omega,\tau),
\end{equation}
and we can estimate the truncation error asymptotically in $\omega$ by using the asymptotic forms of $f_q$ in \eqref{eq:Babich_f_asymptotics}, 
\begin{equation}\label{eq:inferror}
\|g_{hb}(\cdot, \mathbf{r}_0)-g(\cdot,\mathbf{r}_0)\|_{L^{\infty}(V)}\leq O\left(\left({1}/{\omega}\right)^{(N+1-\frac{d-3}{2})}\right); 
\end{equation}
see a similar analysis in \cite{fanqiazepzha18}. 

Now we estimate the difference of the following solution formulas,
\begin{eqnarray}
u_{true}(\mathbf{r})&=&\iint_V G(\mathbf{r},\mathbf{r}_0)s(\mathbf{r}_0)d\mathbf{r}_0  \quad \mbox{for}\quad \mathbf{r}\in V,\\
u_{\rm hb}(\mathbf{r})&=&\iint_V g(\mathbf{r},\mathbf{r}_0)s(\mathbf{r}_0)d\mathbf{r}_0   \quad \mbox{for}\quad \mathbf{r}\in V,
\end{eqnarray}
where we have assumed that the source $s$ is compactly supported and an appropriate absorbing boundary condition has been used to truncate the entire space $R^d$ to the computational domain $V$. Accordingly, we have 
\begin{eqnarray} 
\|u_{true}-u_{\rm hb}\|_{L^\infty(V)} &\leq& \iint_V \|G(\cdot,\mathbf{r}_0)-g(\cdot, \mathbf{r}_0)\|_{L^{\infty}(V)} |s(\mathbf{r}_0)|d\mathbf{r}_0 \nonumber\\
&\leq& \iint_V \|G(\cdot,\mathbf{r}_0)-g_{hb}(\cdot, \mathbf{r}_0)\|_{L^{\infty}(V)} |s(\mathbf{r}_0)|d\mathbf{r}_0 \nonumber \\
& & \quad + \iint_V \|g_{hb}(\cdot,\mathbf{r}_0)-g(\cdot, \mathbf{r}_0)\|_{L^{\infty}(V)}|s(\mathbf{r}_0)|d\mathbf{r}_0 \nonumber \\
&\leq& O\left(\left({1}/{\omega}\right)^{(N+1-\frac{d-3}{2})}\right). 
\end{eqnarray}
This estimate is the foundation for our new numerical schemes, and further numerical analysis is an ongoing work. 

In what follows, we drop the subscript of $u$ and by default refer to it as the HB ansatz solution, and we will only consider two cases: $N=0$ and $N=1$, yielding the first- and second-order asymptotic expansion, respectively.   

\subsection{Approximations of eikonals and HB coefficients near the source}

Evaluating \eqref{babich} requires computation of the squared-phase function $\tau^2$ and the HB coefficient functions $v_0$ and $v_1$. In a neighborhood of the point source where the function $\tau$ is single-valued, $\tau^2$, $v_0$, and $v_1$ can be computed by solving the eikonal equation \eqref{eikonalnormal} and transport equations \eqref{equVs} with initial conditions $v_0$ in \eqref{initialv0} and $v_1(\mathbf{r},\mathbf{r}_0)|_{\mathbf{r}=\mathbf{r}_0}$ to be determined, and the numerical details have been given in \cite{qiayualiuluobur16}. 

To expedite our presentation, we summarize some numerical aspects in the following subsections. One essential difficulty in computing the eikonal and HB coefficients to high-order accuracy is how to initialize these quantities near the point source for numerical PDE solvers, such as Lax-Friedrichs WENO schemes \cite{kaooshqia04,zha05,zhazhaqia06}. Because initial conditions for the eikonal and transport equations are only specified at the source point and high-order schemes need accurate initial values within a small neighborhood of the source to start with, our analyticity assumptions allow us to extract high-order approximations of eikonals and HB coefficients near the source. To do that, we may carry out local Taylor expansions of these functions and further insert these relations into related PDEs so that we may obtain some recursive relations to compute these functions locally. Since such recursive relations actually provide (crude) approximations to the HB ingredients which in turn may be used to construct the Green's function (in a very crude manner) in a large neighborhood, we illustrate how to obtain such recursive relations in the following.

\subsubsection{High-order factorization of eikonals}
We have the following expansion near the source point $\mathbf{r}_0$ (dropped in the expressions below) for squared $\tau$ and $n$:
\begin{equation}
\label{expansionT} \tau^2({\bf{r}}) \approx \sum_{k=0}^{P_T} T_k({\bf{r}}), \quad n^2({\bf{r}}) \approx \sum_{k=0}^{P_S} S_k({\bf{r}}),
\end{equation}
where $T_k({\bf{r}})$ and $S_k({\bf{r}})$ are homogeneous polynomials of degree $k$ in ${\bf{r}}$, and $P_T$ and $P_S$ denote the truncation order of $\tau^2$ and $n^2$, respectively. Upon substituting \eqref{expansionT} into \eqref{eikonal}, we can determine $T_k$ term-by-term by
\begin{eqnarray}
&&T_0 = 0, \quad T_1 = 0,\quad  T_2({\bf{r}}) = S_0 {\bf{r}}^2, \\
\label{recursiveT} && (P-1) S_0 T_P = \sum_{k=1}^{P-2} S_k T_{P-k} - \DF{1}{4} \sum_{k=2}^{P-2} \nabla T_{k+1} \cdot \nabla T_{P-k+1}, ~~P\geq 3. 
\end{eqnarray}

Since we will solve for $\tau$ in the eikonal equation \eqref{eikonalnormal} rather than $\tau^2$, we will take the square root of 
\eqref{expansionT} to obtain an approximate $\tau$, which will serve as a high-order approximation of the exact $\tau$ in numerically solving the eikonal equation \eqref{eikonalnormal}. 

\subsubsection{High-order factorization of coefficients $v_s$}
Although $v_s$ are assumed to be analytic near the source, we still need to obtain high-order approximations to $v_s$ within a small neighborhood of the source so that high-order numerical schemes can be initialized near the source. Therefore, we will expand the coefficients $v_s$ as homogeneous polynomials of degree $k$ in $\br$ as well.

The coefficient $v_0$ can be expanded as
\begin{equation}
\label{expansionV0} v_0 = \sum_{k=0}^{P_B} B_k({\bf r}),
\end{equation}
where $B_k({\bf{r}})$ are homogeneous polynomials of degree $k$ in ${\bf{r}}$. Upon substituting \eqref{expansionV0} into \eqref{equVs}, we can determine $B_k$ term-by-term by
\begin{eqnarray}
B_0 &=& \DF{n_0^{d-2}}{2 \pi^{(d-1)/2}},\quad B_1 = \frac{1}{2S_0} \left(-\frac{1}{2}B_0 \Delta T_3 + d B_0S_1\right), \\
\label{recursiveB} 2 P S_0 B_P &=& - \sum_{k=1}^{P-1} \nabla B_k \cdot \nabla T_{P+2-k} - \DF{1}{2} \sum_{k=0}^{P-1} B_k \Delta T_{P+2-k} \\ 
\nonumber          & &\quad + d \sum_{k=0}^{P-1} B_k S_{P-k}, ~ P\geq 2. 
\end{eqnarray}

Similarly, we expand $v_1$ as 
\begin{equation}
\label{expansionV1} v_1 = \sum_{k=0}^{P_C}C_k({\bf r}), 
\end{equation}
and obtain $C_0$ and $C_1$, 
\begin{eqnarray}
2S_0 C_0 &=& \frac{1}{2}\Delta{B}_2. \label{C0}, \\
4S_0 C_1 &=& \frac{1}{2}\Delta{B}_3 -\frac{1}{2}\Delta{T}_3\;C_0 - (2-d)S_1\;C_0, \label{C1} 
\end{eqnarray}
and $C_P$ for $P\geq2$ by
\begin{eqnarray}
2 (P+1) S_0 C_P &=& \frac{1}{2}\Delta{B}_{p+2}-\sum_{k=1}^{P-1} \nabla C_k \cdot \nabla T_{P+2-k} - \DF{1}{2} \sum_{k=0}^{P-1} C_k \Delta T_{P+2-k} \nonumber \\
&&\quad\quad - (2-d)\sum_{k=0}^{P-1} C_k S_{P-k}. \label{recursiveC}
\end{eqnarray}

To ensure the same order of accuracy for solving \eqref{eikonal} and \eqref{equVs}, the truncation orders are chosen as $P_B=P_S=P_T-2=P_C+2$.  If the analytic function $n$ varies very slowly in a certain neighborhood of the point source that we are interested in so that the resulting Taylor expansions for $\tau^2$, $v_0$, and $v_1$ are sufficiently accurate, then we may use these Taylor expansions to compute the HB ingredients so as to construct the Green's function in this neighborhood. However, in most of situations, we are interested in wave propagation in large domains, and we will use these local Taylor expansions to initialize high-order Lax-Friedrichs WENO sweeping schemes; the related details have been given in \cite{qiayualiuluobur16} which is briefly summarized here.   

\subsection{Efficient algorithms for computing eikonals and HB coefficients}
According to numerical analysis for PDEs, to obtain $v_1$ with first-order accuracy, we need a third-order accurate approximation of $v_0$ and hence a fifth-order accurate approximation of $\tau$. The Lax-Friedrichs WENO schemes as illustrated in \cite{qiayualiuluobur16} can be employed to solve the eikonal equation \eqref{eikonal}. To resolve the singularity at the source, we use the factorization approach \cite{pic97,zharechov05,fomluozha09,luoqiabur14} so that $\tau$ can be factored as 
\begin{equation}
\tau = \tilde{\tau}\bar\tau.
\end{equation}
Here $\tilde{\tau}$ is pre-determined analytically to capture the source singularity, and for instance, we can choose $\tilde{\tau}(\bold r, \bold r_0)$ as $n(\bold r_0)|\bold r-\bold r_0|$ or the square root of the expansion \eqref{expansionT}. Hence $\bar\tau$ is the new unknown which is smooth at the source and satisfies the factored eikonal equation
\begin{equation}\label{faceik}
|\tilde{\tau}\nabla\bar\tau+\bar\tau\nabla\tilde{\tau}| = n.
\end{equation}

Accordingly, high-order Lax-Friedrichs WENO schemes \cite{kaooshqia04,zhazhaqia06,xiozhazhashu10,qiayualiuluobur16} can be applied to solve \eqref{faceik}. In order to obtain a $P_T$-th order accurate approximation of $\tau$ on a mesh of size $h_0$, $\bar\tau$ needs to be initialized in a neighborhood of size $2(P_T-1)h_0$ centered at the source, and these initial values will be fixed during the iterations. In the computation, we will take $P_T=6$ to obtain $\tau$. 

After obtaining a high-order approximation for $\tau$, we apply Lax-Friedrichs WENO schemes to solve transport equations \eqref{equVs} as illustrated in \cite{qiayualiuluobur16}. When $s=-1$, $v_0$ will be initialized as \eqref{expansionV0} in a neighborhood of size $2(P_B-1)h_0$ centered at the source and these values will be fixed during iterations; at other points, high-order Lax-Friedrichs WENO-based schemes are used to update $v_0$. Similarly, when $s=0$, $v_1$ is then initialized as \eqref{expansionV1} in a neighborhood of size $2(P_C-1)h_0$ centered at the source and these values will be fixed during iterations; the values at other points for $v_1$ will be updated using the high-order Lax-Friedrichs WENO-based schemes.

To analyze the complexity of these numerical schemes for computing these HB ingredients, we assume that the $d$-dimensional computational domain is partitioned into a finite-difference mesh of $N_W^d$ points with $N_W$ sampling points in each direction. 
Since these high-order Lax-Friedrichs WENO-based schemes are iterative by nature, we may assume that  these high-order schemes have a superlinear complexity as analyzed \cite{luqiabur16}, and hence the computational cost for these HB ingredients is $O(N_W^d\log N_W)$ in general. 

Nevertheless, we emphasize that since these HB ingredients are independent of the angular frequency $\omega$, the above complexity is for computing the HB ingredients only and consequently is not directly related to computing wave fields, and we can use very coarse meshes to compute these quantities.  Moreover, once they are computed, these ingredients can be compressed into low-rank representations and reused for different frequencies as shown in Section \ref{sec:lowrankpa}. 

\section{Direct Computation of Hadamard-Babich Integrator for Arbitrary Source}\label{sec:babich_anysource}
Considering equation \eqref{eq:helm} in a spatial domain $V=[0,1]^d$, we can apply the HB ansatz to compute the wave field in $V$ due to an arbitrary source function $s(\mathbf{r})$. We consider two situations: domains without inclusion and domains with inclusion. 

\subsection{Domain with no inclusion} When no inclusion is present, the field $u(\mathbf{r})$ can be expressed as 
\begin{eqnarray}
\label{v2v}u(\mathbf{r})=\iint_V g(\mathbf{r},\mathbf{r}_0)s(\mathbf{r}_0)d\mathbf{r}_0,
\end{eqnarray} 
where the Green's function $g(\mathbf{r},\mathbf{r}_0)$ is given in equation (\ref{babich}). We call this the Hadamard-Babich integrator. 

To numerically compute (\ref{v2v}), the domain is discretized into $N_v$ regular \ylrev{cells with cell} size $h$, where $h$ is typically a constant fraction of the $\omega$-dependent wavelength. The source function can be discretized with local volume basis functions $s(\mathbf{r}) = \sum_j s(\mathbf{r}_j)\ylrev{h^d} b^v_j(\mathbf{r})$, where the basis function $b^v_j(\mathbf{r})$ is nonzero only inside the source cell $c_j$ with center $\mathbf{r}_j$:
\begin{eqnarray} 
b^v_j(\mathbf{r})=\begin{cases} \frac{1}{h^d} &\mbox{if } \mathbf{r} \in c_j, \\
0 & \mbox{otherwise}. \end{cases} 
\end{eqnarray}
\ylrev{Note that for a point source located inside cell $c_i$, the source function is approximated with $s(\mathbf{r}_j)={\delta_{ij}}/{h^d}$}. The field at the center of each cell $u(\mathbf{r}_i)$ can be computed as  
\begin{eqnarray}
U = K^{v2v}I, \label{eq:m_v2v}
\end{eqnarray}  
where $I$ denotes a vector of length $N_v$ that collects $s(\mathbf{r}_j)\ylrev{h^d}$, $U$ denotes a vector that collects $u(\mathbf{r}_i)$, and $K^{v2v}_{ij}= g(\mathbf{r}_i,\mathbf{r}_j)$ for $i\neq j$. The self term $K^{v2v}_{ii}$ can be computed analytically by integrating the free space Green's function over the source cell, where the index of refraction $n$ is taken to be constant; see \cref{sec:self} for more details. 

The naive computation of all non-diagonal terms of $K^{v2v}$ requires solving the eikonal equation \eqref{eikonalnormal} and transport equations \eqref{equVs} for $N_v$ times to tabulate $\tau^2(\mathbf{r}_i,\mathbf{r}_j)$, $v_0(\mathbf{r}_i,\mathbf{r}_j)$ and $v_1(\mathbf{r}_i,\mathbf{r}_j)$, requiring at least $O(N_v^2)$ CPU time and memory. Moreover, the assembly and application of $K^{v2v}$ also require $O(N_v^2)$ time and memory. 

We remark in passing that when the medium is homogeneous, the HB integrator \eqref{v2v} is exact since the HB ansatz \eqref{babich} yields the exact Green's function in this case; when the medium is inhomogeneous and smooth without inducing caustics,  the HB integrator \eqref{v2v} is accurate in the asymptotic order $O(1/\omega^{\infty})$ without truncating the ansatz \eqref{babich}, and it is accurate in the asymptotic order $O({1}/{\omega}^{(2-\frac{d-3}{2})})$ when truncating the ansatz \eqref{babich} up to the first two terms. Detailed numerical analysis on this is an ongoing work. 
 
\subsection{Domain with sound hard inclusion} For simplicity, this subsection only considers 2D domains with curve inclusion;  however the proposed scheme can be trivially extended to 3D domains with surface inclusion. Considering a sound hard curve denoted by $C$, the source function will generate an incident field $u^{\rm inc}(\mathbf{r})$ that induces an equivalent source $p(\mathbf{r})$ on $C$, which in turn generates a scattered field $u^{\rm sca}(\mathbf{r})$. More specifically, we have 
\begin{eqnarray}
u^{\rm inc}(\mathbf{r}):=\iint_{V}g(\mathbf{r},\mathbf{r}_0)s(\mathbf{r}_0)d\mathbf{r}_0, ~~\mathbf{r}\in C,\\
u^{\rm sca}(\mathbf{r}):=\int_{C}g(\mathbf{r},\mathbf{r}_0)p(\mathbf{r}_0)d\mathbf{r}_0, ~~\mathbf{r}\in C ~\mathrm{or}~ V.\label{eq:s2v}
\end{eqnarray} 

Here the equivalent source $p$ can be solved by the following integral equation:
\begin{eqnarray}
u^{\rm inc}(\mathbf{r})=-u^{\rm sca}(\mathbf{r}), ~~\mathbf{r}\in C. \label{eq:sie}
\end{eqnarray} 

To numerically solve (\ref{eq:sie}) and compute the total field $u=u^{\rm inc}+u^{\rm sca}$ in $V$, the domain $V$ is discretized with $N_v$ \ylrev{cells with cell size} $h$. After discretization of the inclusion into $N_s$ line segments with length $w_i$ for $i=1,\ldots,N_s$ and enforcement of (\ref{eq:sie}) at segment centers, we solve a linear system 
\begin{eqnarray} 
K^{s2s}P=U^{\rm inc}\label{eq:m_sie}
\end{eqnarray} 
with 
\begin{eqnarray}
K^{s2s}_{ij} = \begin{cases} \frac{-w_i}{4}(1+\mathrm{i}\frac{2}{\pi}(\log(\frac{\gamma \omega n(\mathbf{r}_i^s) w_i}{4}-1)))&\mbox{if } i=j, \\
-g(\mathbf{r}_i^s,\mathbf{r}_j^s)w_j  & \mbox{otherwise}, \end{cases}\label{eq:s2s} 
\end{eqnarray}
where $\gamma$ is the Euler constant, $U^{\rm inc}_i= u^{\rm inc}(\mathbf{r}_i^s)$ and $P_i=p(\mathbf{r}_i^s)$ with $\mathbf{r}_i^s$ denoting segment centers. The right-hand-side (RHS) in (\ref{eq:m_sie}) is computed by 
\begin{eqnarray}
U^{\rm inc}=K^{v2s}I, \label{eq:m_v2s}
\end{eqnarray} 
where $I$ is the same as that in (\ref{eq:m_v2v}), and the discretized volume-to-surface operator is $K^{v2s}_{ij}= g(\mathbf{r}_i^s,\mathbf{r}_j)$. Here the source function does not overlap with the curve. 

Once the equivalent source $p(\mathbf{r})$ is obtained, the scattered field at any point $\mathbf{r}$ of the computational domain is computed using (\ref{eq:s2v}),  
\begin{eqnarray}
U^{\rm sca}=K^{s2v}P. \label{eq:m_s2v}
\end{eqnarray} 
Here $U^{\rm sca}_i = u^{\rm sca}(\mathbf{r}_i)$ with $\mathbf{r}_i$ being the center of cell $i$, and $K^{s2v}_{ij}=g(\mathbf{r}_i,\mathbf{r}_j^s)w_j$ which can be directly calculated from the transpose of $K^{v2s}$. Combining (\ref{eq:m_v2v}), (\ref{eq:m_sie}), (\ref{eq:m_v2s}) and (\ref{eq:m_s2v}), the total fields at the cell centers can be expressed as 
\begin{eqnarray}
U=(K^{v2v}+K^{s2v}\big(K^{s2s}\big)^{-1}K^{v2s})I. \label{eq:m_v2v_inclu}
\end{eqnarray} 

As a typical curve in 2D requires $N_s=N_v^{1/2}$ discretization \ylrev{segments}, the naive computation of $K^{v2s}$ and $K^{s2s}$ requires solving the eikonal equation (\ref{eikonalnormal}) and transport equations (\ref{equVs}) for $N_s$ times. In addition, the computation of $K^{v2s}$, $K^{s2s}$, and $\big(K^{s2s}\big)^{-1}$ requires $O(N_v^{3/2})$ time and memory. Overall, the computation of (\ref{eq:m_v2v_inclu}) is still $O(N_v^{2})$ dominated by the computation of $K^{v2v}$.

\section{Fast Computation of Hadamard-Babich Integrator for Arbitrary Source}\label{sec:fastbabich_anysource}
Here we propose a quasi-linear complexity algorithm for the computation of (\ref{eq:m_v2v}) and (\ref{eq:m_v2v_inclu}). The proposed algorithm leverages the low-rank representation of the squared phase function $\tau^2$ and HB coefficient functions $v_0$ and $v_1$ to avoid solving eikonal and transport equations for all point sources. Once these low-rank representations are obtained, the discretized volume-to-volume, volume-to-surface, and surface-to-surface operators are compressed using butterfly algorithms and their hierarchical extensions. 

\subsection{Low-rank Representation of the Phase and HB Coefficients\label{sec:lowrankpa}}
Since we have assumed that $\tau^2$, $v_0$ and $v_1$ are analytic, they permit low-rank representations as shown in \cite{luqiabur16,luqiabur16b,luqiabur18,qiasonlubur21}. Letting $f=\tau^2$, $v_0$ and $v_1$, respectively, we consider the following analytical low-rank representation using the Chebyshev interpolation, 
\begin{eqnarray}
f(\mathbf{r},\mathbf{r}_0) \approx \sum_{i=1}^{N_I}\sum_{j=1}^{N_I}T_i(\mathbf{r})f(\mathbf{r}^c_i,\mathbf{r}^c_j)T_j(\mathbf{r}_0). \label{eq:chebinterp} 
\end{eqnarray} 
Let $n_I$ denote order of the Chebyshev interpolation and $\mathbf{r}^c_i$ for $i=1,\ldots,N_I$ with $N_I=n_I^d$ be the Chebyshev nodes \ylrev{(i.e., $d$-dimensional Chebyshev sampling of the domain $V=[0,1]^d$)}. Accordingly, we define the Lagrange interpolants $T_i(\mathbf{r})$, where 
 \begin{eqnarray}
T_i(\mathbf{r}) = \prod_{s=1}^{d}l_i(r_s)=\prod_{s=1}^{d}\prod_{\substack{1\leq k\leq n_I \\ k\neq i}}\frac{r_s-r_{sk}^c}{r_{si}^c-r_{sk}^c}, 
 \end{eqnarray} 
$\mathbf{r}=[r_1,\ldots,r_d]$, and $\mathbf{r}^c_i=[r_{1i}^c,\ldots,r_{di}^c]$.

In (\ref{eq:chebinterp}), the set of function samples $\{f(\mathbf{r}^c_i,\mathbf{r}^c_j)\}$ of cardinality $N_I^2$ requires solving the eikonal and transport equations via the Lax-Friedrichs WENO schemes with $N_I$ point sources located at $\mathbf{r}^c_j$, as explained in Section \ref{sec:babich}. More specifically, for each Chebyshev node $\mathbf{r}^c_j$, we create a grid with mesh size $h_0$ and grid points $\mathbf{r}^0_i$ for $i=1,\ldots,1/h_0^d$ that are aligned with the point source at $\mathbf{r}^c_j$. \ylrev{In other words, the grid covers the computation domain $V$ and its ghost regions, and has one grid point collocated with $\mathbf{r}^c_j$.} This gives rise to solutions $f(\mathbf{r}^0_i,\mathbf{r}^c_j)$. For each $\mathbf{r}_i^c$, we compute $f(\mathbf{r}^c_i,\mathbf{r}^c_j)$ with $i\neq j$ from a local cubic interpolation using data points $f(\mathbf{r}^0_i,\mathbf{r}^c_j)$. See \cref{fig:table-compute} for a 2D example with two point sources (in red dots). 

Once the set $\{f(\mathbf{r}^c_i,\mathbf{r}^c_j)\}$ is obtained, the computation of $f(\mathbf{r},\mathbf{r}_0)$ for any point pair $(\mathbf{r},\mathbf{r}_0)$ requires $O(N_I^2)=O(n_I^{2d})$ time. We can assume $N_I$ to be constant as typically $n_I<15$. In addition, we can leverage a blocked version of (\ref{eq:chebinterp}) to further improve its computational efficiency. Consider an $m\times n$ block $\mat{F}$ with $F_{ij}=f(\mathbf{r}_i,\mathbf{r}_j)$ for arbitrary lists of $n$ source points $\mathbf{r}_j$ and $m$ observation points $\mathbf{r}_i$. The block $\mat{F}$ can be computed as
\begin{eqnarray}
\mat{F}\approx\mat{T}^o\mat{F}^c\mat{T}^s.\label{eq:chebinterpblock} 
\end{eqnarray}  
Here, $\mat{T}^o_{ij}=T_j(\mathbf{r}_i)$,  $\mat{T}^s_{ij}=T_i(\mathbf{r}_j)$, and $\mat{F}^c_{i,j}=f(\mathbf{r}^c_i,\mathbf{r}^c_j)$. By using this blocked form, repetitive computation of the interpolants $\mat{T}^o$ and $\mat{T}^s$ is avoided, and high-performance BLAS libraries can be used. As a result, this requires $O(\min(m,n)N_I^2+mnN_I)$ time using (\ref{eq:chebinterpblock}) as opposed to $O(mnN_I^2)$ using (\ref{eq:chebinterp}). 

One may attempt to compute $F$ for all entries of $K^{v2v}$ (and similarly for $K^{s2v}$, $K^{v2s}$ and $K^{s2s}$), but this leads to $O(N_v^2N_I)$ computational time. As we will see next, we propose the butterfly algorithm for constructing a compressed representation of $K^{v2v}$ and the other discretized operators, requiring only a total of $mn=O(N_v\log^2N_v)$ entries in (\ref{eq:chebinterpblock}).   

\begin{figure}
	\centering
	\vspace{-7.5pt}
	\begin{subfigure}[t]{0.9\textwidth}
		\centering
		\includegraphics[width=\linewidth]{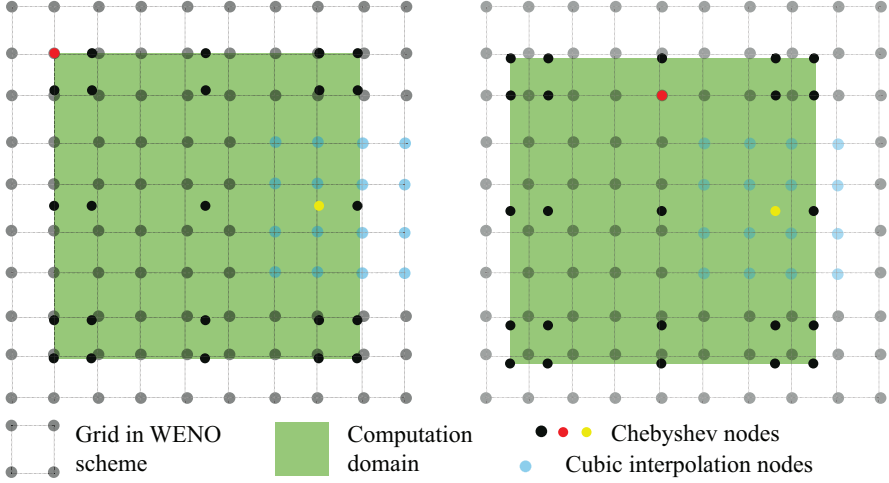}
	\end{subfigure}
	\vspace{-5pt}
	\caption{The computation of $f(\mathbf{r}_i^c,\mathbf{r}_j^c)$ in a 2D domain (shown as the green area) with $f=\tau^2,v_0,v_1$. Here $r_i^c$, $i=1,\ldots,N_I$ denote the Chebyshev nodes with $N_I=5^2$ (the black dots). Two instances of $\mathbf{r}_j^c$ are shown: for each instance, a fast sweeping method with a grid size $h_0$, grid points $\mathbf{r}^0_i$ and a point source collocated at one Chebyshev node (the red dot) is used to compute $f(\mathbf{r}^0_i, \mathbf{r}_j^c)$. For each $\mathbf{r}_i^c$, $i\neq j$ (the yellow dots), $f(\mathbf{r}_i^c,\mathbf{r}_j^c)$ is computed via cubic interpolation from the grid points $\mathbf{r}^0_i$ (the blue dots). }
	\label{fig:table-compute}
\end{figure}

	\subsection{Butterfly Representation of the Discretized Integral Operators\label{sec:BF}}
	The butterfly algorithm exploits the fact that judiciously selected submatrices of the discretized operators $K^{v2v}$, $K^{s2v}$, $K^{v2s}$ and $K^{s2s}$ are low-rank compressible, despite of the full rankness of these operators. 
	
	The algorithm first recursively subdivides the geometry point sets associated with the rows and columns of these operators into two subsets of approximately equal sizes, using such as k-dimensional (K-D) tree clustering algorithms, until the subsets contain a predefined number of points. For the $N_v$ \ylrev{cell} centroids in the computational domain, the procedure generates a complete binary tree $\mathcal{T}_{H_v}$ of $L_v$ levels with root level $0$ and leaf level $L_v$. Each node $\tau$ at level $l$ is an index set $\tau \subset \{1,\ldots,N_v\}$. Similarly for the $N_s$ segment centers for the inclusion, the procedure generates a complete binary tree $\mathcal{T}_{H_s}$ of $L_s$ levels. In both trees, a non-leaf node $\tau$ at level $l$ has two children $\tau_1$ and $\tau_2$, where $\tau=\tau_1\cup\tau_2$ and $\tau_1 \cap \tau_2 = \emptyset$. For a non-root node $\tau$, its parent is denoted $p_\tau$. 
	
The butterfly representation of a matrix requires binary trees $\mathcal{T}_o$ and $\mathcal{T}_s$ with $L$ levels for the row and column indices, respectively, which are defined for the integral operators as follows:
\begin{itemize}
		\item $K^{v2v}$: For any two siblings $\tau_1$ and $\tau_2$ at level $l$ of $\mathcal{T}_{H_v}$, let $o=\tau_1$ and $s=\tau_2$. $K^{v2v}(o,s)$ is compressed as a butterfly with $L=L_v-l$ levels. Let $\mathcal{T}_o$ and $\mathcal{T}_s$ be the subtrees of $\mathcal{T}_{H_v}$ rooted at $o$ and $s$, respectively. As a result, there are $2^l$ butterfly representations at each level $l=1,\ldots,L_v$. The $2^{L_v}$ blocks $K^{v2v}(\tau,\tau)$ for node $\tau$ at level $l=L_v$ are kept as dense blocks. Such a representation is called the hierarchically off-diagonal butterfly (HODBF) representation \cite{Liu_2017_HODBF}.
		\item $K^{s2s}$: For any two siblings $\tau_1$ and $\tau_2$ at level $l$ of $\mathcal{T}_{H_s}$, let $o=\tau_1$ and $s=\tau_2$. $K^{s2s}(o,s)$ is compressed as a butterfly with $L=L_s-l$ levels. Let $\mathcal{T}_o$ and $\mathcal{T}_s$ be the subtrees of $\mathcal{T}_{H_s}$ rooted at $o$ and $s$, respectively. Similar to $K^{v2v}$, we seek a HODBF representation of $K^{s2s}$.
		\item $K^{s2v}$ (or $K^{v2s}$): Let $o$ and $s$ be the roots of $\mathcal{T}_{H_v}$ and $\mathcal{T}_{H_s}$, respectively. $K^{s2v}(o,s)=K^{s2v}$ is compressed as a \textit{single} butterfly with $L=\min\{L_v,L_s\}$ levels. Let $\mathcal{T}_o=\mathcal{T}_{H_v}$ and $\mathcal{T}_s=\mathcal{T}_{H_s}$. 
		  
	\end{itemize} 
	
	\subsubsection{Butterfly algorithm}		
	
	The $L$-level butterfly representation of the integral operator $K(o,s)\in \mathbb{C}^{m\times n}$ (superscripts of $K$ are dropped) requires the complementary low-rank property: at any level $0\leq l \leq L$, for any node $\tau$ at level $l$ of $\mathcal{T}_{o}$ and any node $\nu$ at level $L-l$ of $\mathcal{T}_{s}$, the subblock $K(\tau,\nu)$ is numerically low-rank with rank $r_{\tau,\nu}$ bounded by a small number $r$ called the butterfly rank. We will comment on the butterfly rank for the three integral operators in subsections \ref{sec:v2v}, \ref{sec:s2v}, and \ref{sec:s2s}, respectively. 
	
	For any subblock $K(\tau,\nu)$, the complementary low-rank property permits a low-rank representation using for instance interpolative decomposition (ID)  as 
	\begin{equation}\label{eqn:ID}
	K(\tau,\nu) \approx K(\tau, \bar{\nu}) {V}_{\tau,\nu},
	\end{equation}
	where the skeleton matrix $K(\tau, \bar{\nu})$ contains $r_{\tau,\nu}$ skeleton columns indexed by $\bar{\nu}$, and the interpolation matrix ${V}_{\tau,\nu}$ has bounded entries. The ID can be computed via for instance rank-revealing QR decomposition with a relative tolerance $tol$. There are several equivalent butterfly representations in literature \cite{liu2021butterfly,Yingzhou_2017_IBF,li2015butterfly,Pang2020IDBF} and here we briefly describe the so-called column-wise butterfly representation \cite{liu2021butterfly}.
	
	At level $l=0$, the interpolation matrices ${V}_{\tau,\nu}$ are explicitly formed. While at level $l>0$, they are represented in a nested fashion. To see this, consider a node pair $(\tau,\nu)$ at level $l>0$ and let $\nu_1,\nu_2$ and $p_\tau$ be the children and parent of $\nu$ and $\tau$, respectively. From (\ref{eqn:ID}), we have  
\begin{align}\label{eqn:transfer_computation}
\mat{K}(\tau,\nu) &=\begin{bmatrix}
\mat{K}(\tau,{\nu_1}) & \mat{K}(\tau,{\nu_2})\end{bmatrix}\nonumber\\
&\approx\begin{bmatrix}
\mat{K}(\tau,\bar{\nu}_1) & \mat{K}(\tau,\bar{\nu}_2)\end{bmatrix}\begin{bmatrix}
\mat{V}_{p_\tau, \nu_1} & \\
& \mat{V}_{p_\tau, \nu_2}
\end{bmatrix}\\
&\approx\mat{K}(\tau,\bar{\nu})\mat{W}_{\tau,\nu}\begin{bmatrix}
\mat{V}_{p_\tau, \nu_1} & \\
& \mat{V}_{p_\tau, \nu_2}
\end{bmatrix}.
\end{align}
Here $\mat{W}_{\tau,\nu}$ and $\bar{\nu}$ are the interpolation matrix and skeleton columns from the ID of $[\mat{K}(\tau,\bar{\nu}_1), \mat{K}(\tau,\bar{\nu}_2)]$, respectively. This allows representing $\mat{V}_{\tau,\nu}$ as 
\begin{equation}\label{eqn:nested_basis}
\mat{V}_{\tau,\nu} =
\mat{W}_{\tau,\nu}\begin{bmatrix}
\mat{V}_{p_\tau, \nu_1} & \\
& \mat{V}_{p_\tau, \nu_2}
\end{bmatrix}.
\end{equation}
We will refer to $\mat{W}_{\tau,\nu}$ as the transfer matrices in the rest of this paper. We note that the computation of interpolation matrices $\mat{V}_{\tau,\nu}$ at level $l=0$ and transfer matrices $\mat{W}_{\tau,\nu}$ at level $0<l<L$ does not require the ID on the full subblocks $\mat{K}(\tau,\nu)$ and $[\mat{K}(\tau,\bar{\nu}_1), \mat{K}(\tau,\bar{\nu}_2)]$, as this immediately leads to an $O(mn)$ compression complexity at level $l=0$. 

Instead, we can select a number of $O(r_{\tau,\nu})$ proxy rows $\hat{\tau}\subset\tau$ to compute $\mat{V}_{\tau,\nu}$ and $\mat{W}_{\tau,\nu}$ via ID as:
 \begin{align}\label{eqn:proxy}
	K(\hat{\tau},\nu) \approx K(\hat{\tau}, \bar{\nu}) {V}_{\tau,\nu},~~l=0,\\
	\begin{bmatrix}
	\mat{K}(\hat{\tau},\bar{\nu}_1) & \mat{K}(\hat{\tau},\bar{\nu}_2)\end{bmatrix}\approx\mat{K}(\hat{\tau},\bar{\nu})\mat{W}_{\tau,\nu},~~l>0. 
 \end{align}
When $l=L$ and $\hat{\tau}=\tau$, no proxy rows are needed. We will discuss the choice of the proxy rows in more details 
in subsections \ref{sec:v2v}, \ref{sec:s2v}, and \ref{sec:s2s}.


With all the interpolation and transfer matrices computed, the butterfly representation of $\mat{\mat{K}(o,s)}$ is:
\begin{align}
\mat{K}(o,s) \approx \mat{K}^{L}\mat{W}^{L}\mat{W}^{L-1}\ldots \mat{W}^1\mat{V}^0.
\label{eqn:butterfly_mat}
\end{align}
Let $\nu_1,\nu_2,\ldots,\nu_{2^{L-l}}$ denote the nodes at level $L-l$ of $\mathcal{T}_s$, and $\tau_1,\tau_2,\ldots,\tau_{2^{l}}$ denote the nodes at level $l$ of $\mathcal{T}_o$. The interpolation factor $\mat{V}^0$, the transfer factors $\mat{W}^l$ for $l=1$, $\ldots$, $L$, and the skeleton factor $\mat{K}^{L}$ are:  
\begin{align}
\mat{V}^0&=\mathrm{diag}(\mat{V}_{\tau,\nu_1},\ldots,\mat{V}_{\tau,\nu_{2^L}})
\label{eqn:butterfly_factors},~~(\tau,\nu_i) \mathrm{~at~level~} l=0,\\
\mat{K}^L&=\mathrm{diag}(\mat{K}(\tau_1,\bar{\nu}),\ldots,\mat{K}(\tau_{2^L},\bar{\nu})),~~(\tau_i,\nu) \mathrm{~at~level~} l=L,\\
\mat{W}^{l}&=\mathrm{diag}(\mat{W}{\tau_1},\ldots,\mat{W}{\tau_{2^{l-1}}}), ~~l=1,\ldots,L,\\
\mat{W}_{\tau_i}&=
\begin{bmatrix}
\text{diag}(\mat{W}_{\tau_i^1,\nu_1}, \dots, \mat{W}_{\tau_i^1,\nu_{2^{L-l}}}) \\
\text{diag}(\mat{W}_{\tau_i^2,\nu_1}, \dots, \mat{W}_{\tau_i^2,\nu_{2^{L-l}}})
\end{bmatrix},~~(\tau_i^{\{1,2\}},\nu_i) \mathrm{~at~level~} l, 
\end{align}
where $\tau_i^1$ and $\tau_i^2$ denote the children of $\tau_i$. Note that $V^0$ and $K^L$ contain $2^L$ diagonal blocks each with $O(r_{\tau\nu})$ nonzeros, and $W^l$ contains $2^{L}$ blocks $W_{\tau,\nu}$ each with $O(r_{\tau\nu}^2)$ nonzeros. The construction of these blocks via (\ref{eqn:proxy}) at each level $l$ requires the computation of $O(n)$ submatrices (i.e., the left-hand side (LHS) of (\ref{eqn:proxy})) of sizes $O(r_{\tau,\nu})\times O(r_{\tau,\nu})$. From the discussion of (\ref{eq:chebinterpblock}) in \cref{sec:lowrankpa}, each submatrix can be computed in $O(r_{\tau,\nu}N_I^2+r_{\tau,\nu}^2N_I)$ time, which is time-wise optimal assuming $N_I$ constant. If $\max_{\tau,\nu} r_{\tau\nu}$ is $O(1)$, it is immediately clear that the butterfly representation (IDs and matrix entry computation) requires $O(n\log n)$ memory and CPU time. We will see that this is not the case for any of $K^{v2v}$, $K^{s2v}$ and $K^{s2s}$, but quasi-linear complexities can still be attained for most of these operators.

\begin{figure}
	\centering
	\vspace{-7.5pt}
	\begin{subfigure}[t]{\textwidth}
		\centering
		\includegraphics[width=\linewidth]{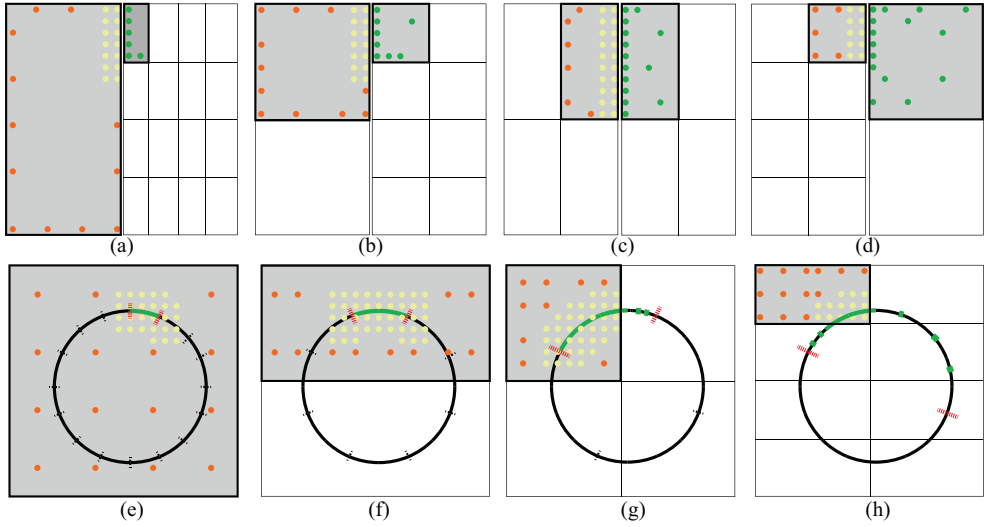}
	\end{subfigure}
	\vspace{-5pt}
	\caption{\ylrev{Top: illustration} of the butterfly compression of $4$ levels for one top-level off-diagonal block of HODBF representation of $K^{v2v}$ in a 2D computational domain. The subdomains at levels $l=0,1,2,3$ ((a)-(d)) are denoted by the vertical and horizontal lines. One subdomain pair $(\tau,\nu)$ at each level with non-constant rank $r_{\tau,\nu}$ is shown in the shaded areas (left: $\tau$, right: $\nu$). The green dots correspond to skeleton columns \ylrev{(cell centroids)} for the ID computation, and the red and yellow dots correspond to uniformly sampled and nearest neighboring proxy rows \ylrev{(cell centroids)} used to compute the ID. \ylrev{Note that the unused cell centroids at each stage are not plotted.} Bottom: illustration of the butterfly compression with $5$ levels for $K^{s2v}$ involving a circular inclusion in a 2D computational domain. The subdomains and arcs at levels $l=0,1,2,3$ ((e)-(h)) are denoted by the vertical/horizontal lines and dashed lines, respectively. One subdomain-arc pair $(\tau,\nu)$ at each level with non-constant rank $r_{\tau,\nu}$ is shown with the shaded areas for the subdomain and curves \ylrev{ending with red dashed lines} for the arc. The green dots correspond to skeleton columns \ylrev{(line segments)} for the ID computation, and the red and yellow dots correspond to uniformly sampled and nearest neighboring proxy rows \ylrev{(cell centroids)} used to compute the ID. }
	\label{fig:BFproxy}
\end{figure}	   

 In what follows, we discuss the rank estimate, proxy row selection, and computational complexity for each of the three discretized integral operators $K^{v2v}$, $K^{s2v}$ and $K^{s2s}$.   
 
\subsubsection{Computation of $K^{v2v}$\label{sec:v2v}} 
Consider the top level off-diagonal block of its HODBF representation, $K^{v2v}(o,s)$ with $o$ and $s$ being the children of the root of $\mathcal{T}_{H_v}$, and sizes $m=n=N_v/2$. This permits a butterfly representation with $L=L_v-1$ levels. Letting  $l_m=L/2$ denote the middle butterfly level, we can show that among the $O(N_v)$ subblocks $\mat{K}(\tau,\nu)$ at each level $l$, there are $O(2^{|l-l_m|/d}N_v^{(d-1)/2d})$ subblocks representing interactions between adjacent or close-by geometry subdomains, where $d=2$ or 3 denotes the problem dimension. Their ranks scale as $r_{\tau,\nu}=O(2^{-|l-l_m|/d}N_v^{(d-1)/2d})$ dominated by the interface DOFs between the two computational subdomains. Each of these non-constant rank subblocks requires $O(r_{\tau,\nu}^2)$ storage and matrix entry computation, and $O(r_{\tau,\nu}^3)$ ID cost. Thus, they require 
$$\sum_lr_{\tau,\nu}^2O(2^{|l-l_m|/d}N_v^{(d-1)/2d})=O(N_v^{3(d-1)/2d})\leq O(N_v)\quad \mathrm{~storage}$$ 
and 
$$\sum_lr_{\tau,\nu}^3O(2^{|l-l_m|/d}N_v^{(d-1)/2d})=O(N_v^{2(d-1)/d})\quad \mathrm{~CPU~time}$$ 
for the interpolation and transfer matrices. Specifically, the CPU time is $O(N_v)$ when $d=2$ and $O(N_v^{4/3})$ when $d=3$. The rest of the subblocks essentially has $r_{\tau,\nu}=O(1)$ and requires $O(N_v\log N_v)$ CPU time and memory based on the DOF analysis in \cite{michielssen_multilevel_1996}. Note that the sub-optimal CPU complexity for $d=3$ can be improved by considering strong-admissible hierarchical matrices \cite{Han_2017_butterflyLUPEC} to keep butterfly rank constant, or considering alternatives to ID, e.g., analytical interpolation schemes such as (\ref{eq:chebinterp}) or the one in \cite{Candes_butterfly_2009}, which brings the $O(r_{\tau,\nu}^3)$ CPU cost per block back to $O(r_{\tau,\nu}^2)$. However, we will not implement the analytical interpolation-based butterfly representation in this article due to the large prefactors in these schemes and will leave that as a future work. 

As an example, \cref{fig:BFproxy} (top) shows a 2D computational domain with a partitioning tree $\mathcal{T}_{H_v}$ with $L_v=5$ levels. $K^{v2v}(o,s)$ with $o$ and $s$ being the children of the root $\mathcal{T}_{H_v}$ (i.e., the left and right half of the domain) is compressed as a $4$-level butterfly. \cref{fig:BFproxy}(a)-(d) shows the subdomain pairs at levels $l=0,1,2,3$,  respectively. For each level, one subdomain pair with non-constant interaction rank is shown in grey. The \ylrev{cell centroids} in green represent the skeleton columns $\bar{\nu}$, which mainly lie on the subdomain interface. 

As mentioned above, the computation of $\bar{\nu}$ is performed via (\ref{eqn:proxy}) with proxy rows. Specifically, consider $(\tau,\nu)$ at level $0<l<L$ (with $l=0$, $L$ being similar). Let $n_i$ denote the nearest neighboring \ylrev{centroids} $i$ (e.g., all \ylrev{centroids} that are within a $10h$ distance of $\mathrm{r}_i$), and let $f_\tau$ denote the $\chi|\bar{\nu}_1\cup\bar{\nu}_2|$ uniformly selected \ylrev{centroids} near the boundary of subdomain $\tau$ with an oversampling parameter $\chi$. We choose the proxy rows as 
\begin{align}
\hat{\tau}=(\cup_{i\in\bar{\nu}_1\cup\bar{\nu}_2}n_i)\cap\tau\cup f_\tau.\label{eq:proxyv2v}
\end{align} 
For example, \cref{fig:BFproxy}(a)-(d) shows the nearest neighboring proxy rows in yellow and the uniform proxy rows in red. 

Because of the above rank estimate and proxy row selection scheme, $O(N_v\log N_v)$ CPU time and memory complexities can be achieved for the top-level off-diagonal block of the HODBF. This leads to an $O(N_v\log^2 N_v)$ complexity for the overall HODBF construction.

\subsubsection{Computation of $K^{s2v}$\label{sec:s2v}}
Consider $K^{s2v}$ between a 2D computational domain with $N_v$ \ylrev{cells} and a curve inclusion with $N_s$ discretization segments with typically $N_s=N_v^{1/2}$. $K^{s2v}$ is compressed as a butterfly with $L=\min\{L_v,L_s\}$ levels. For simplicity we assume that $L_s=L_v$, each leaf node in $\mathcal{T}_s$ contains $O(1)$ columns, and each leaf node in $\mathcal{T}_v$ contains $O(N_v/N_s)=O(N_s)$ rows. Just like the $K^{v2v}$ operator, we need to identify the subblocks with non-constant ranks $r_{\tau,\nu}$. We first identify a level $l_m$ at which the side length of $\tau$ is on the same order as the length of $\nu$. From $N_s^{1/2}2^{(L-l_m)/2}=N_s2^{l_m}$, where the LHS and right-hand-side (RHS) represent lengths of $\tau$ and $\nu$ at level $l_m$, respectively, we have $l_m=O(\frac{2L}{3})$. For each level $l\leq l_m$, we can show that there are $O(2^{-l}N_s)$ subblocks, out of the total of $O(N_s)$ subblocks, representing interactions between a node $\nu$ of $\mathcal{T}_s$ fully contained in or close to some node $\tau$ of $\mathcal{T}_o$. These subblocks have ranks at most $r_{\tau,\nu}=O(2^{l})$ (i.e., the length of $\nu$). These non-constant rank subblocks require $\sum_{l\leq l_m}r_{\tau,\nu}^2O(2^{-l}N_s)=O(N_s^{5/3})=O(N_v^{5/6})$ storage and matrix entry computation, and $\sum_{l\leq l_m}r_{\tau,\nu}^3O(2^{-l}N_s)=O(N_s^{7/3})=O(N_v^{7/6})$ CPU time for the IDs. On the other hand, for each level $l> l_m$, we can show that there are $O(2^{l/2})$ subblocks, out of the total of $O(N_s)$ subblocks, representing interactions between a node $\tau$ of $\mathcal{T}_o$ intersecting with some node $\nu$ of $\mathcal{T}_s$. These subblocks have ranks at most $r_{\tau,\nu}=O(2^{l/2}N_s)$ (i.e., the side length of $\tau$). These non-constant rank subblocks again require $O(N_v^{5/6})$ storage and matrix entry computation, and $O(N_v^{7/6})$ CPU time for the IDs. Just like $K^{v2v}$ in 3D, we can use analytical interpolation schemes to reduce $O(N_v^{7/6})$ to $O(N_v^{5/6})$. The rest of the subblocks essentially have $r_{\tau,\nu}=O(1)$ based on the DOF analysis in \cite{michielssen_multilevel_1996}, and their CPU and memory complexities are dominated by $l=L$, which scale as $O(N_v)$. 

As an example, \cref{fig:BFproxy} (bottom) shows the compression of $K^{s2v}$ representing interaction between a 2D computational domain and a circle inclusion, with $L=L_s=5$ levels. \cref{fig:BFproxy}(e)-(h) shows the subdomain-arc pairs at levels $l=0,1,2,3$ respectively. For each level, one subdomain-arc pair with non-constant interaction rank is shown in grey for the subdomain and red for the arc. The discretization \ylrev{line segments} in green on the curve represent the skeleton columns $\bar{\nu}$, which mainly lie on the intersection of the arc with the subdomain. Note that for $l=0,1$, the arc is fully contained in the subdomain, and the green points overlap with the red arc; for $l=2,3$, the subdomain intersects with the arc. Just like the $K^{v2v}$ operator, proxy rows in the subdomains are chosen by (\ref{eq:proxyv2v}), except that $f_\tau$ is a set of $\chi|\bar{\nu}_1\cup\bar{\nu}_2|$ uniformly distributed \ylrev{centroids} in the subdomains. In \cref{fig:BFproxy}(e)-(h), the nearest neighboring and uniform proxy rows (\ylrev{centroids}) are shown in yellow and red, respectively. 

\subsubsection{Computation of $K^{s2s}$ and its inverse\label{sec:s2s}}
Just like $K^{v2v}$, we seek a HODBF representation of $K^{s2s}$ for $d=2$. Considering a top-level off-diagonal block $K^{s2s}(o,s)$, it has been shown in \cite{Yang_2020_BFprecondition} that its butterfly rank scales as $O(\log N_s)$ and its CPU time and memory scales as $O(N_s\log N_s)$. As a result, the HODBF representation of $K^{s2s}$ requires $O(N_s\log^2 N_s)=O(N_v^{1/2}\log^2N_v)$ CPU time and storage units.  

Once constructed, the inverse of the HODBF compressed $K^{s2s}$ can be computed using algorithms described in \cite{Liu_2017_HODBF,liu2021sparse} leveraging sketching-based butterfly construction algorithms \cite{liu2021butterfly}. For $K^{s2s}$, the HODBF inversion requires $O(N_s^{3/2}\log N_s)=O(N_v^{3/4}\log N_v)$ based on the analysis in \cite{Liu_2017_HODBF}.

We summarize the algorithmic complexities in subsections \ref{sec:v2v}, \ref{sec:s2v}, and \ref{sec:s2s}. The computation of $K^{v2v}$ requires $O(N_v\log^2N_v)$ memory for $d=2,3$, $O(N_v\log^2N_v)$ time for $d=2$, and $O(N_v^{4/3})$ time for $d=3$ (the latter becoming quasi-linear if analytical interpolation rather than ID-based butterfly algorithms are used or strong-admissible hierarchical matrices are used). The computation of $K^{s2v}$ for $d=2$ requires $O(N_v)$ memory and $O(N_v^{7/6})$ CPU time ($O(N_v)$ attainable if the analytical interpolation is used). The computation and inversion of $K^{s2s}$ for $d=2$ require $O(N_v^{1/2}\log^2N_v)$ memory and $O(N_v^{3/4}\log N_v)$ CPU time. Therefore, the overall complexities of the proposed butterfly compressed HB integrator scale as $O(N_v\log^2N_v)$. 

\subsection{Expected convergence behavior}
Although a complete numerical analysis of our new method is an ongoing work, we sketch the expected convergence behavior of the algorithm. 

Our new algorithm has three principal sources of numerical errors: the first source is due to numerical computation of HB ingredients, such as $\tau$, $v_0$ and $v_1$; the second one is due to truncating the infinite asymptotic series to keep only the leading or the first two terms: $N=0$ or $N=1$ in \eqref{babich}; the third one is due to the butterfly compression. Therefore, the overall error of our HB integrator will be roughly controlled by the above three sources of errors.  

To start with, we consider the first source of errors. According to equation \eqref{equVs}, $v_1$ depends on $\Delta v_0$ (the Laplacian of $v_0$) and $\Delta \tau$. To have first-order accurate $v_1$ (which is a minimum requirement for our current setup of the new algorithm), we need at least first-order accurate $\Delta v_0$ and $\Delta \tau$. To have first-order $\Delta v_0$, $v_0$ itself must be computed to third-order accuracy so that it can be numerically differentiated twice to yield first-order accurate Laplacian $\Delta v_0$. To have third-order accurate $v_0$, $\Delta \tau$ (the Laplacian of $\tau$) must be computed to third-order accuracy according to equation \eqref{equVs} again, which in turn implies that $\tau$ itself must be computed to fifth-order accuracy. Consequently, we have chosen to apply the fifth-order Lax-Friedrichs WENO scheme to compute $\tau$, $v_0$ and $v_1$ in our current implementation so that the computed $\tau$, $v_0$ and $v_1$ will have fifth-, third-, and first-order accuracy, respectively. We denote the computational errors of these ingredients as $O(h_0^{\beta})$, where $h_0$ is the mesh size, and $\beta$ is the order of accuracy: $\beta=1,3,5$, respectively. Since these ingredients are independent of the frequency parameter $\omega$, they can be computed just once and re-used for many different frequencies \ylrev{as explained below.} 

Now we consider the second source of errors. According to equation \eqref{babich}, when we keep the leading-order term of the HB ansatz, we are expecting the first-order asymptotic convergence in the form of $O\left(\left({1}/{\omega}\right)^{(N+1-\frac{d-3}{2})}\right)$, where $N=0$; when we keep the first two terms of the HB ansatz, we are expecting the second-order asymptotic convergence in the form of $O\left(\left({1}/{\omega}\right)^{(N+1-\frac{d-3}{2})}\right)$, where $N=1$. 

In addition, we consider the error due to the butterfly compression. Since a detailed analysis of the butterfly compression algorithm \ylrev{with prescribed proxy rows} is beyond the scope of the current work, we assume that the error is simply represented as $O(e_{\rm bf})$ to simplify the matter, where $e_{\rm bf}$ is a small positive constant. \ylrev{That said, one can refer to Section 5 of \cite{liu2021butterfly} for a detailed analysis of SVD-based butterfly algorithms.} 

Finally, \ylrev{when $d=3$, by using some estimates provided in \cref{sec:go2hb}, we can write the overall error of our algorithm roughly as 
\begin{equation}
\label{totalerror}
   E_{\rm total}= O\left(({1}/{\omega})^{(N+1-\frac{d-3}{2})}\right)+O(h_0^3) + O(\omega h_0^5) + O(e_{\rm bf}),
\end{equation}
for $N=0$, and 
\begin{equation}
\label{totalerrorAdded}
   E_{\rm total}= O\left(({1}/{\omega})^{(N+1-\frac{d-3}{2})}\right)+O(h_0^3) + O(\omega h_0^5) + O\left(\frac{h_0}{\omega}\right)+O(e_{\rm bf}),
\end{equation}
for $N=1$. When $d=2$, we may use the analysis in \cite{fanqiazepzha18} for the H-B ansatz to obtain analogous estimates. Note that here we ignore the errors of Chebyshev and cubic interpolations assuming that they do not dominate over those induced by the high-order LxF-WENO schemes.} 

\ylrev{We remark that the term $O(\omega h_0^5)$ in \eqref{totalerror} captures the magnification of numerical phase errors by the frequency $\omega$ in the Hankel-based H-B ansatz, in which the phase function $\tau$ appears together with $\omega$ in the form of ${\omega\tau}$; this implies two things: first, given a set of computed H-B ingredients, we can reuse these ingredients for many different frequencies as long as the sum of the error terms $O(\omega h_0^5)$ and $O(h_0^3)$ is not dominant over the other errors; second, if it happens that the frequency $\omega$ is so large that the error from $O(\omega h_0^5)+O(h_0^3)$ is dominant over other errors for a given set of computed H-B ingredients, then we can always compute more accurate H-B ingredients on finer meshes so that the resulting error from $O(\omega h_0^5)+O(h_0^3)$ is not dominant. Note that such computation tabulates the H-B ingredients on the fixed Chebyshev nodes in the {\it off-line} stage and does not affect the computation time for the wave function in the {\it online} stage. Similar observations apply to the estimate \eqref{totalerrorAdded}.}

Our numerical results demonstrate that the above error estimates are sound. 

\section{Numerical Results}\label{sec:example}
This section provides several numerical examples to demonstrate the accuracy and efficiency of the proposed butterfly-compressed HB integrator when applied to 2D and 3D computational domains with both homogeneous and inhomogeneous media. The low-rank compression of the phase/HB coefficients in \cref{sec:lowrankpa} and the butterfly compression of $K^{v2v}$, $K^{s2v}$, $K^{s2s}$ and $(K^{s2s})^{-1}$ in \cref{sec:BF} have been implemented with distributed-memory parallelism. Most of the experiments are performed on the Haswell nodes of the Cori machine, a Cray XC40, at NERSC in Berkeley, where each of the $2,388$ Haswell nodes has two $16$-core Intel Xeon E5-2698v3 processors and 128GB of 2133MHz DDR4 memory. For most experiments we use 64 Haswell nodes for both the proposed algorithm and the reference FDFD solver. Part of the experiments are performed on development nodes at the High Performance Computing Center (HPCC) of MSU, where each node has two 2.4Ghz 20-core Intel Xeon Gold 6148 CPU and 377GB of RAM.

\subsection{Accuracy comparison with FDFD}
We first compare the performance of the proposed algorithm with that of state-of-the-art FDFD solvers for both 2D and 3D computational domains. 
\subsubsection{2D domains without and with inclusion\label{sec:acc_2d}}
We consider the following examples of homogeneous and inhomogeneous media:
\begin{itemize}
\item \textbf{Constant media}: the computational domain is $V=[0,2]^2$ with $n(\mathbf{r})=n(x,y)=2$. The phase has an analytical form $\tau(\mathbf{r},\mathbf{r}_0)=n(\mathbf{r}_0)|\mathbf{r}-\mathbf{r}_0|$. The HB coefficients have analytical formulas $v_0(\mathbf{r},\mathbf{r}_0)=1/(2\sqrt{\pi})$ and $v_1(\mathbf{r},\mathbf{r}_0)=0$. Therefore, the HB integrator becomes $g(\mathbf{r},\mathbf{r}_0)=\frac{\mathrm{i}}{4}H_0^{(1)}(\omega n(\mathbf{r}_0) |\mathbf{r}-\mathbf{r}_0|)$, i.e., the well-known form of the free-space Green's function. We use high-order Lax-Friedrichs WENO methods with $h_0=0.01$ to solve equations (\ref{eikonalnormal}) and (\ref{equVs}) with point sources, construct their low-rank representation with an order of $n_I=3$ for the Chebyshev interpolation, and compare the results with these exact formulas.

\item \textbf{Constant-gradient media}: the computational domain is $V=[0,1]^2$ with $n(\mathbf{r})\equiv n(x,y)=\frac{1}{0.5-0.25(y-0.5)}$. Note that $1/n$ has a nonzero constant partial derivative in $y$. Let $\mathbf{r}_c=(0.5, 0.5)$. When the point source $\mathbf{r}_0$ is in the interior of $V$, the phase function has an analytical formula $\tau(\mathbf{r},\mathbf{r}_0)=\frac{1}{|\mathbf{G_0}|}\mathrm{arccosh}\Big(1+\frac{1}{2}n(\mathbf{r})n_0|\mathbf{G}_0|^2|\mathbf{r}-\mathbf{r}_0|^2\Big)$, where $\mathbf{G}_0=[0,-0.25]$ and $n_0=\frac{n(\mathbf{r}_c)}{1+n(\mathbf{r}_c)\mathbf{G}_0\cdot(\mathbf{r}_0-\mathbf{r}_c)}$. The HB coefficients have no known analytical expressions. We use the high-order Lax-Friedrichs(LxF)-WENO method with $h_0=0.0025$ to solve equation (\ref{eikonalnormal}) and (\ref{equVs}) with point sources, and construct the their low-rank representation with an order of $n_I=13$ for the Chebyshev interpolation.

\item \textbf{Sinusoidal media}: the computational domain is $V=[0,1]^2$ with $n(\mathbf{r})=n(x,y)=\frac{1}{1+0.2\sin(\pi(x+0.05))\sin(0.5\pi y)}$. In this case, both the phase and the HB coefficients have no known analytical expressions. We use the high-order LxF-WENO method with $h_0=0.01$ to solve \eqref{eikonalnormal} and \eqref{equVs} with point sources, and construct their low-rank representation with an order of $n_I=13$ for the Chebyshev interpolation. Since the exact solutions are not available, we use FDFD solutions as references and compare the corresponding results accordingly.

\item \textbf{Waveguide media}:  the computational domain is $V=[0,1]^2$ with $n(\mathbf{r})=n(x,y)=\frac{1}{1-0.5e^{-2.0(x-0.5)^2}}$. In this case, both the phase and the HB coefficients have no known analytical expressions. We first use the high-order LxF-WENO  method with $h_0=0.01$ to solve \eqref{eikonalnormal} and \eqref{equVs} with point sources, and construct the  low-rank representation with an order of $n_I=13$ for the Chebyshev interpolation. We then compare the results with the FDFD solutions.

\end{itemize}	
	
For the inclusion (if present), we consider a semi-circle of radius $0.5$ and an open square of side length $0.8$ centered 
at the domain center. When computing $K^{v2v}$, $K^{s2v}$, and $K^{s2s}$, let $n_p$ denote the number of points per wavelength (PPW) for the discretization of the computational domain or the curve inclusion. The computational domain is discretized with $h=2\pi/(\omega n_{\max}n_p)$ with $n_p=10$ and $n_{\max}$ being the maximum refractive index over the domain. Similarly, the curve inclusion is discretized with $w_i=2\pi/(\omega n_{\max}n_p)$ with $n_p=100\sim500$ to ensure highly accurate  approximation for $K^{s2s}$. The butterfly and HODBF compression of $K^{v2v}$, $K^{s2v}$ and $K^{s2s}$ are computed with tolerance ${\rm tol}=10^{-8}$ in (\ref{eqn:ID}) and oversampling factor $\chi=20$ in (\ref{eq:proxyv2v}). 

Once the discretized integral operators are computed, we apply (\ref{eq:m_v2v}) and (\ref{eq:m_v2v_inclu}) to the following source functions (i.e., RHSs) centered at $\mathbf{r}_c=(x_{c},y_{c})$ (the domain center): 
\begin{itemize}
\item \textbf{Point source:} $s(\mathbf{r})=1/h^2$ if $\mathbf{r}$ is inside the source cell centered at $\mathbf{r}_c$.  
\item \textbf{Gaussian wavepacket source:} $s(\mathbf{r})={\rm exp}(-|\mathbf{r}-\mathbf{r}_c|^2/(2\sigma^2)){\rm exp}(i\omega_0(\mathbf{r}\cdot\mathbf{d}))t(|\mathbf{r}-\mathbf{r}_c|,w_1,w_2)$ with $\sigma=0.15$, $w_1=0.3$, $w_2=0.1$, $\omega_0=0.9\omega$, and $\mathbf{d}=\frac{1}{\sqrt 2}[1,1]$. Here $t(x,w_1,w_2)$ is the cosine tapering function: 
\begin{eqnarray}
t(x,w_1,w_2) = 0.5(1+\cos((x-w_1)\pi/w_2)) ~\mathrm{if}~ w_1<x<w_2,\label{eq:taper}
\end{eqnarray}
and $t(x,w_1,w_2)=1$ if $x\leq w_1$, and $t(x,w_1,w_2)=0$ if $x\geq w_2$.  	
\item \textbf{Concave kite-shaped source:} Let a kite-shaped curve be $\partial\Omega=\{(x,y): x(t)=b(\cos(t)+0.65\cos(2t)-0.65)+x_{c},\; y(t)=1.5b\sin(t)+y_{c}\mbox{ for } 0\leq t\leq 2\pi\}$ with a scaling factor $b=0.2$. $s(\mathbf{r})=1$ if $\mathbf{r}\in\Omega$ and $s(\mathbf{r})=t(|\mathbf{r}-\mathbf{r}_c|,w_1,w_2)$ if $\mathbf{r}\notin\Omega$ with the tapering function in (\ref{eq:taper}). Here $w_2=0.1$, $w_1=|\mathbf{r}_{\min}|$, and  $\mathbf{r}_{\min}$ is the point on $\partial\Omega$ closest to the line $\mathbf{r}-\mathbf{r}_c$.  

\end{itemize}

Note that although the current form of HB ansatz \eqref{babich} needs further modification to handle media permitting presence of caustics, we can still model caustics induced by the interaction of the RHS (the source) with the medium, as shown in the two examples: the concave kite-shaped source and the open cavity inclusion; both cases will induce caustics in the wave field as we will see. 

As for the reference FDFD solver for computing wavefields $u_{\rm fd}(\mathbf{r})$, we use the 9-point staggered grid scheme in \cite{chen2013optimal}. The computational domain is extended in each direction with a perfectly matched layer (PML) of thickness $8\pi/(\omega n_{\max})$ (i.e., 4 wavelengths). The extended domain is discretized with $h=2\pi/(\omega n_{\max}n_p)$ with PPW $n_p=10\sim50$. The resulting sparse linear system is solved with a multi-frontal sparse direct solver STRUMPACK \cite{ghysels2017robust,liu2021sparse}. When inclusion is present, the rows and columns of the system which represent grid cells overlapping with the inclusion are removed from the system. This can introduce significant staircase approximation errors to the inclusion, unless more sophisticated subgridding techniques are used. Nevertheless, by careful implementation, we still manage to produce good FDFD results in the examples that we are going to show.

For the constant medium, we first consider $\omega=80\pi$, which amounts to 160 wavelengths in each direction. The fields computed by the proposed scheme (PPW=10) and FDFD (PPW=10, 50) and their differences are shown in \cref{fig:ex1_f40}. Note that for the point source without inclusion (Row 1 in \cref{fig:ex1_f40}), the exact solution is also plotted (in dashed green) in the third column. One can clearly see that the solution by the proposed scheme matches well with the exact solution, while it requires PPW=50 or higher for FDFD to achieve a similar order of accuracy. For the kite-shaped source (Row 2 in \cref{fig:ex1_f40}), the concave shape can induce caustics, which are well-captured by the proposed scheme. FDFD matches better with the proposed scheme if PPW=50 other than PPW=10 is used. For the Gaussian wavepacket source with the semi-circle inclusion (Row 3 in \cref{fig:ex1_f40}), FDFD results match poorly with the proposed scheme even using PPW=50, particularly near the inclusion, due to the stair-case error in FDFD (this is the case for all the source functions considered). As a workaround, we consider the open square inclusion (Row 4 in \cref{fig:ex1_f40}) for which FDFD introduces no staircase error. Again, FDFD requires PPW=50 or higher to achieve a similar accuracy as the proposed scheme. Note that the square inclusion permits a hierarchical matrix representation of $K^{s2s}$ as an alternative to HODBF, but this is not considered in this paper.

Based on these experiments, we further consider $\omega=320\pi$, which amounts to 640 wavelengths in each direction. The fields computed by the proposed scheme (PPW=10) and FDFD (PPW=10) and their differences are shown in \cref{fig:ex1_f160}. We note that the FDFD solver with PPW=10 already results in a sparse system of dimension $387,223,684$, and denser discretization causes memory crashes. As a comparison, the proposed scheme results in a dense, butterfly compressed $K^{v2v}$ matrix of dimension $N_v=40,972,801$. One can see from the point source case (Row 1 in \cref{fig:ex1_f160}) that the proposed scheme is still very accurate when changing from $\omega=80\pi$ to $\omega=320\pi$, but FDFD suffers from dispersion errors. From the Gaussian wavepacket source without and with inclusion (Row 2-3 in \cref{fig:ex1_f160}), it is clear that PPW=10 for FDFD does not give satisfactory results.

For the constant-gradient medium, we consider $\omega=50\pi$ and $\omega=100\pi$, which amount to 100 and 200 wavelengths in each direction, respectively. The results are shown in \cref{fig:ex2_f25} and \cref{fig:ex2_f50}. It is not hard to see that for the point source, kite source, and Gaussian wavepacket source without inclusion (Row 1 and 2 in \cref{fig:ex2_f25} and \cref{fig:ex2_f50}), FDFD requires at least PPW=50 to achieve a similar accuracy as the proposed scheme. However for $\omega=100\pi$, FDFD with PPW=20 already leads to a sparse system of dimension $193,710,724$. In contrast, the proposed scheme leads to a dense, butterfly compressed $K^{v2v}$ matrix of dimension $N_v=16,008,001$. Similar conclusions can be drawn for the Gaussian wavepacket source with inclusions (Row 3 and 4 in \cref{fig:ex2_f25} and \cref{fig:ex2_f50}).

To see that our method can be applied to generic inhomogeneous media, we test two more models: the sinusoidal model and the waveguide model. 

For the sinusoidal model, we first consider $\omega=80\pi$, which amounts to 50 wavelengths in each direction, and the results are shown in the first three rows of \cref{fig:ex3_f40}. It is not hard to see that for the point source and Gaussian wavepacket source without inclusion (Row 1 and 2 in \cref{fig:ex3_f40}), the proposed scheme can achieve a similar accuracy as FDFD with PPW=50; similar conclusions can be drawn for the Gaussian wavepacket source with inclusions (Row 3 in \cref{fig:ex3_f40}). In addition, we also show the result in Row 4 of \cref{fig:ex3_f40} when $\omega=40\pi$ for the concave kite-shaped source, and we have chosen this particular frequency so that caustic effects are apparent near the concave region. 

For the waveguide model, we consider $\omega=40\pi$, which amounts to 40 wavelengths in each direction. We compute wavefields for three different types of sources: the point source, the Gaussian wavepacket source, and the concave kite-shaped source, and the results are shown in Row 1 to 3 in \cref{fig:ex4_f20}. The FDTD solutions are computed as references. It can be seen that the proposed scheme can achieve a similar accuracy as FDFD with PPW=50. 

\begin{figure}[!htp]
	\centering
	\vspace{-7.5pt}
	\begin{subfigure}[t]{.29\textwidth}
		\centering
		\includegraphics[width=\linewidth]{\fpath/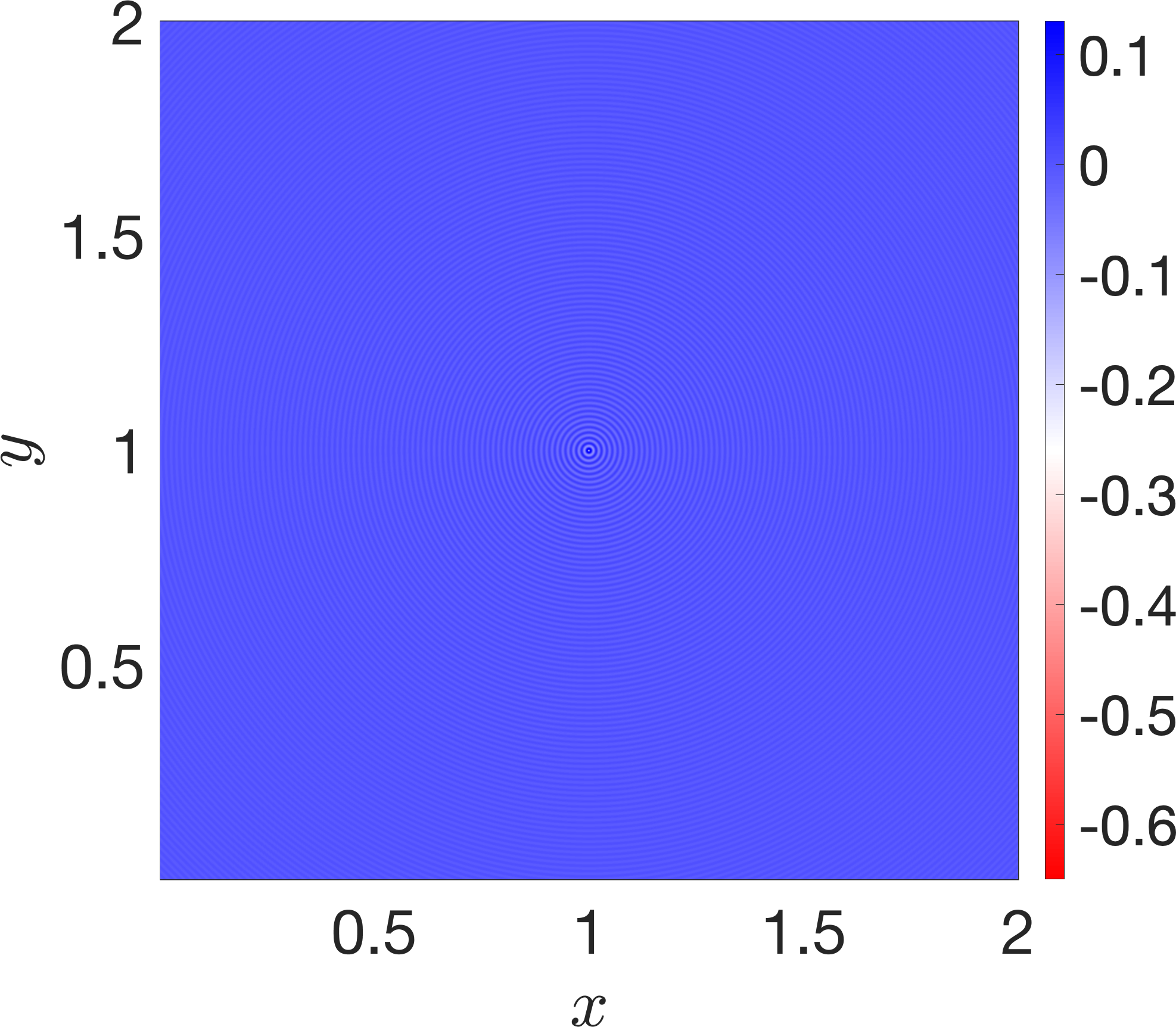}
	\end{subfigure}
	\begin{subfigure}[t]{.29\textwidth}
		\centering
		\includegraphics[width=\linewidth]{\fpath/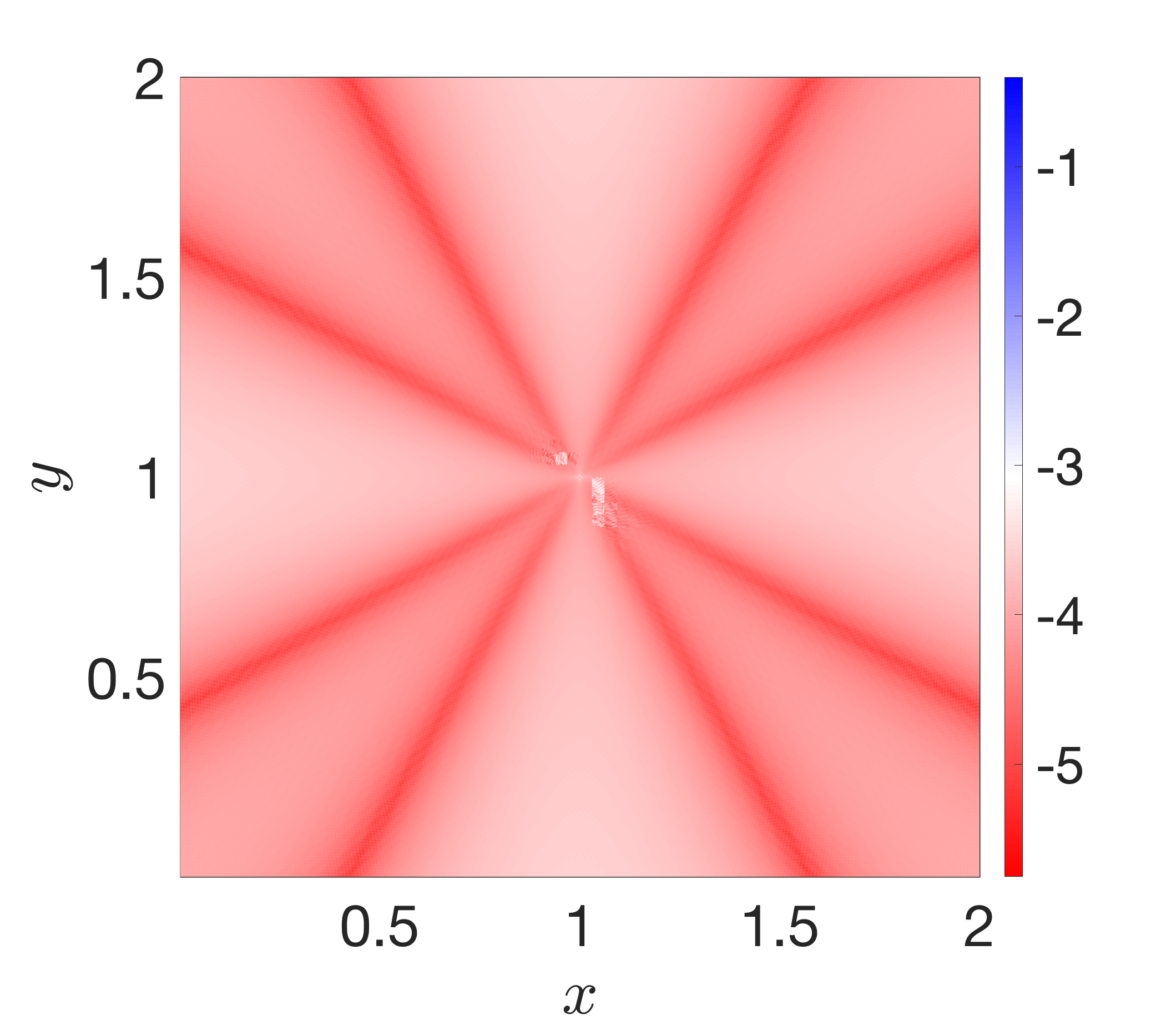}
	\end{subfigure}	
	\begin{subfigure}[t]{.37\textwidth}
	\centering
	\includegraphics[width=\linewidth]{\fpath/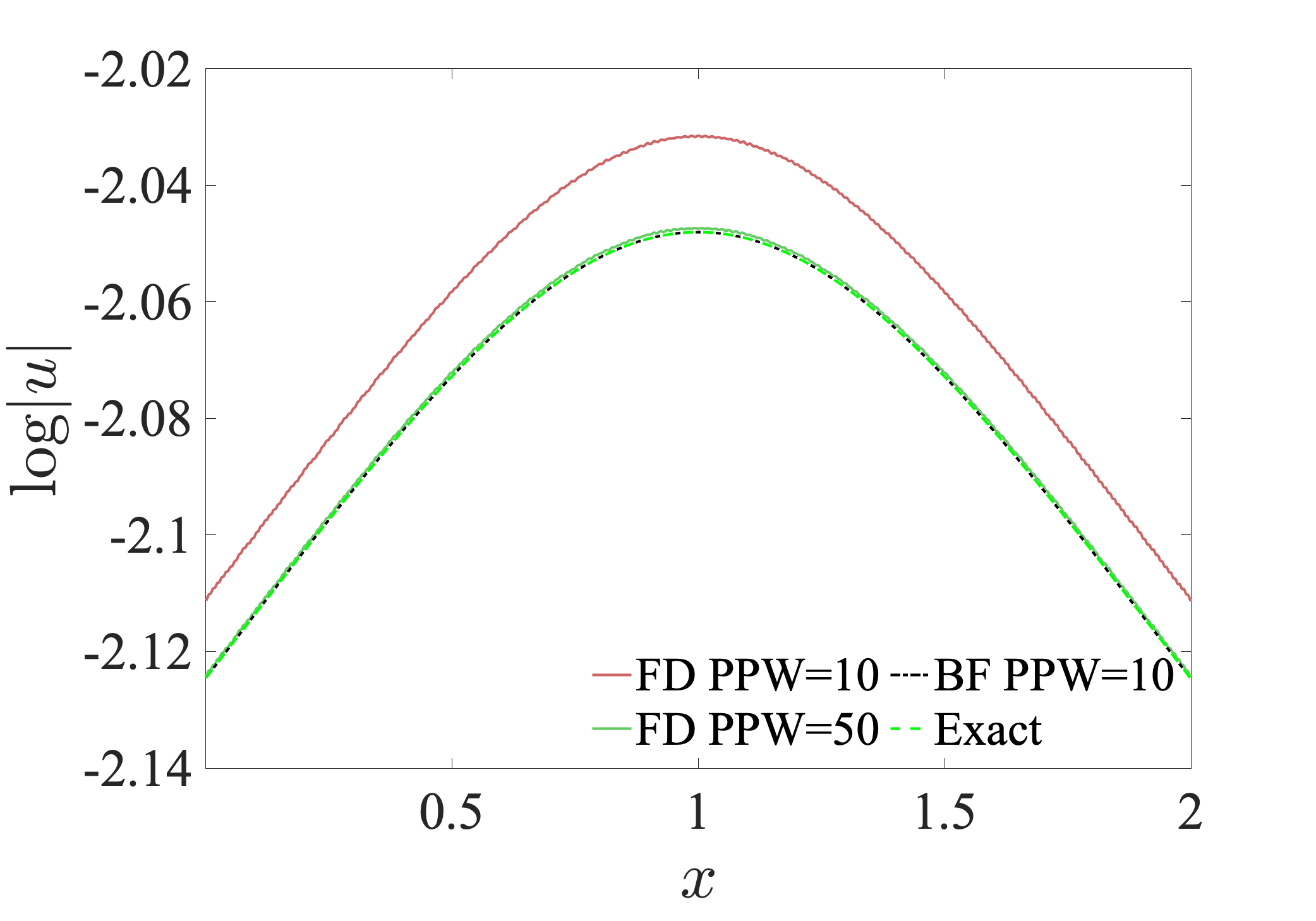}
\end{subfigure}	
	\begin{subfigure}[t]{.29\textwidth}
	\centering
	\includegraphics[width=\linewidth]{\fpath/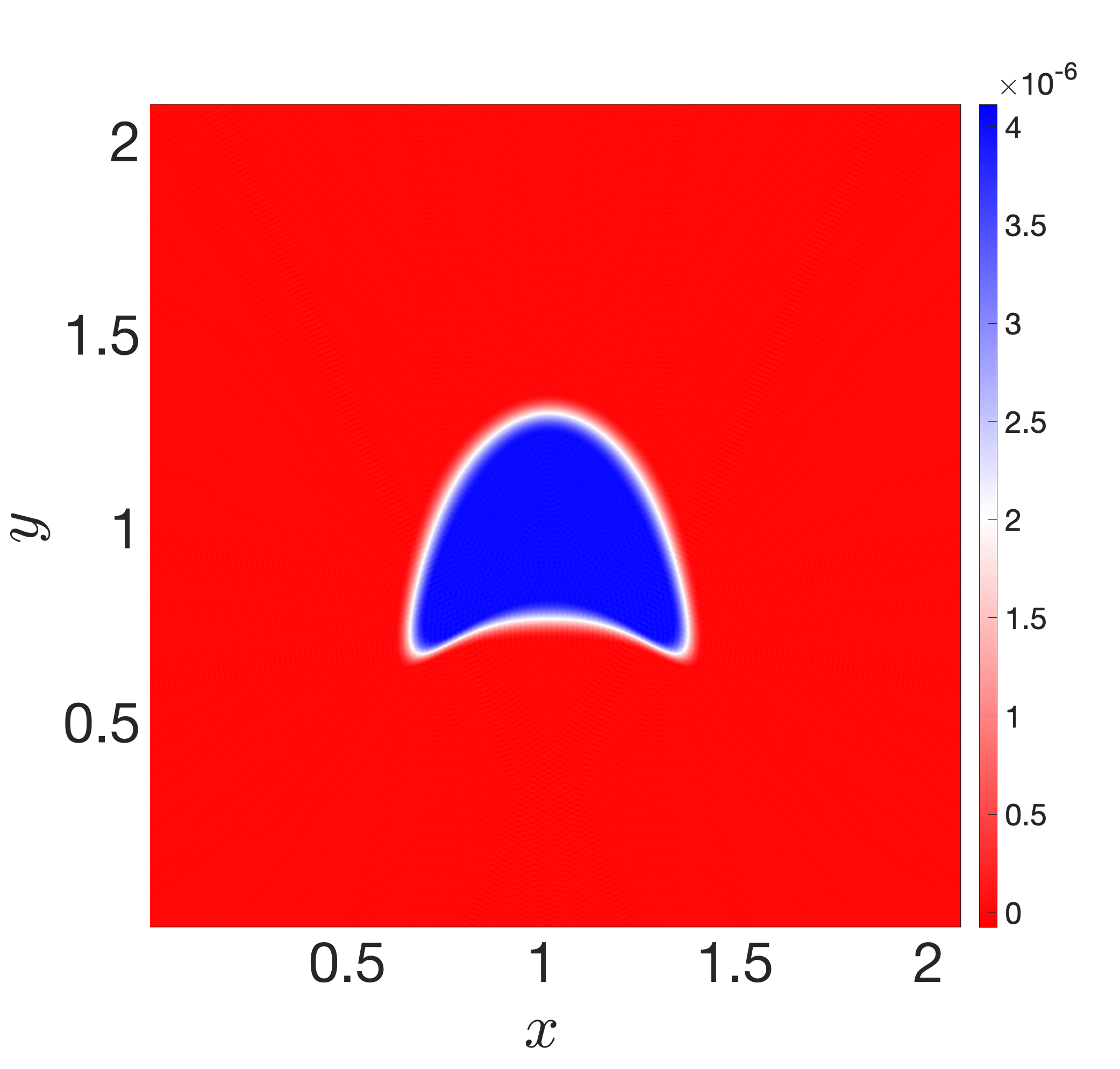}
\end{subfigure}
\begin{subfigure}[t]{.29\textwidth}
	\centering
	\includegraphics[width=\linewidth]{\fpath/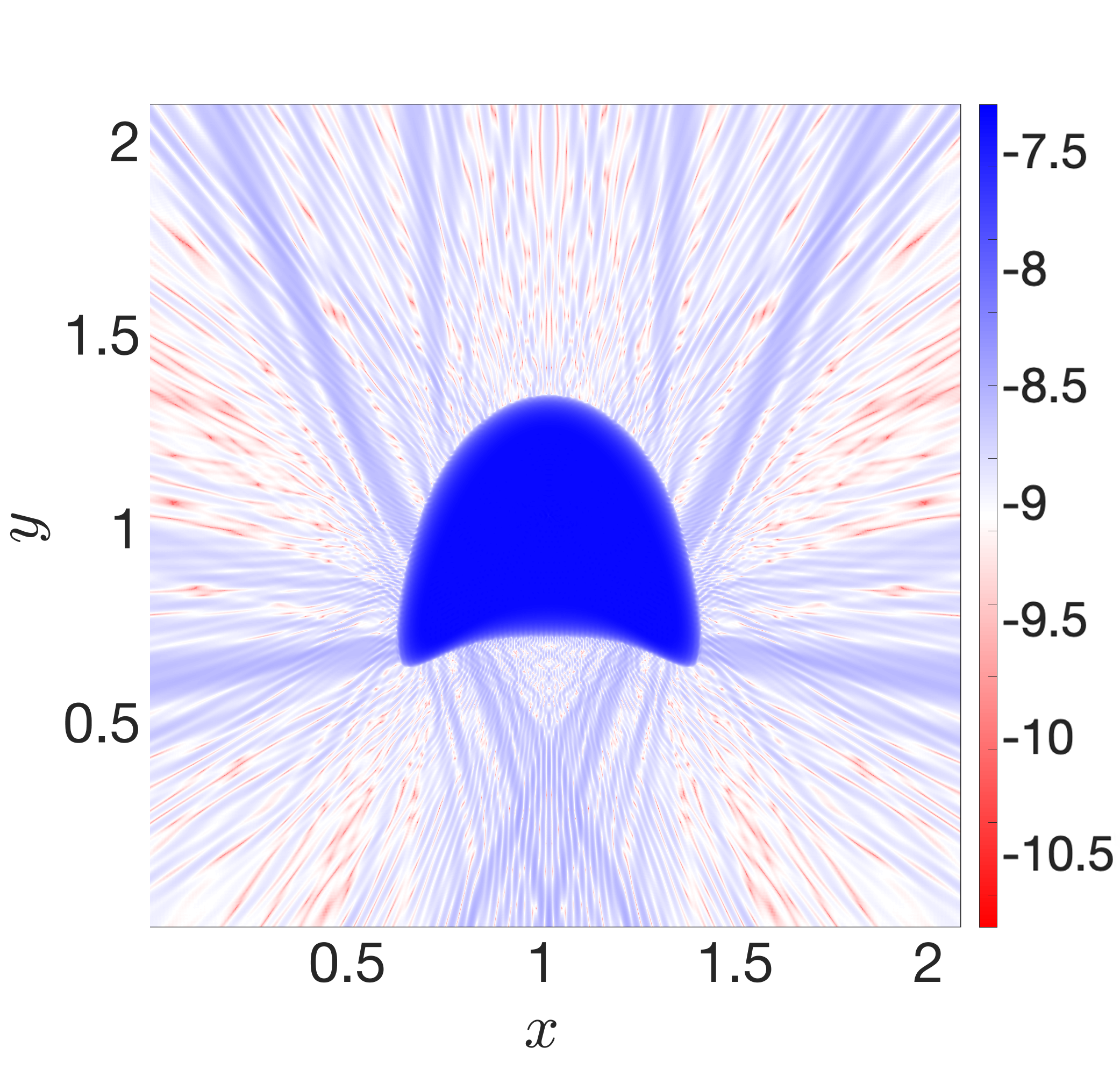}
\end{subfigure}	
\begin{subfigure}[t]{.37\textwidth}
	\centering
	\includegraphics[width=\linewidth]{\fpath/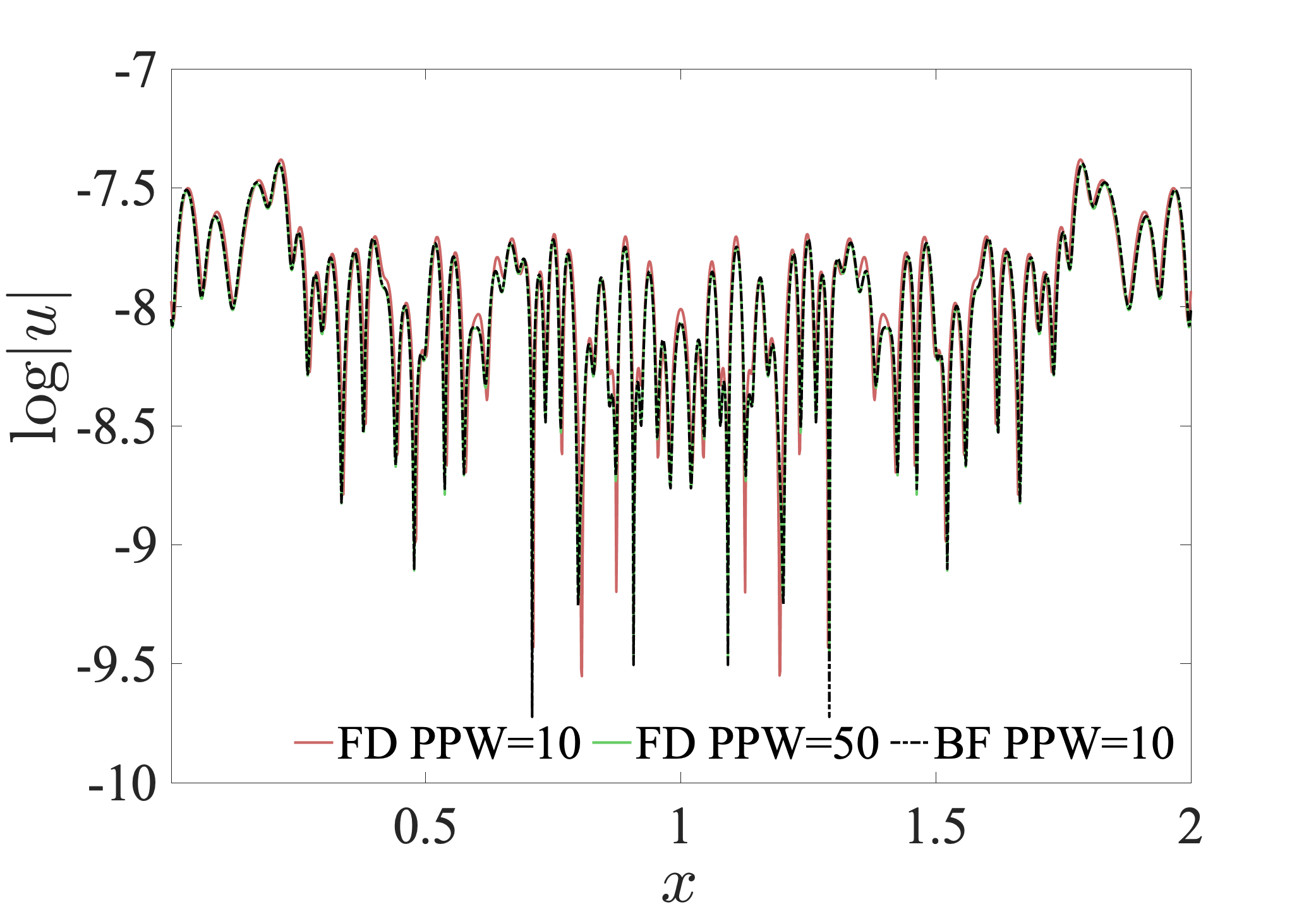}
\end{subfigure}	

\begin{subfigure}[t]{.29\textwidth}
	\centering
	\includegraphics[width=\linewidth]{\fpath/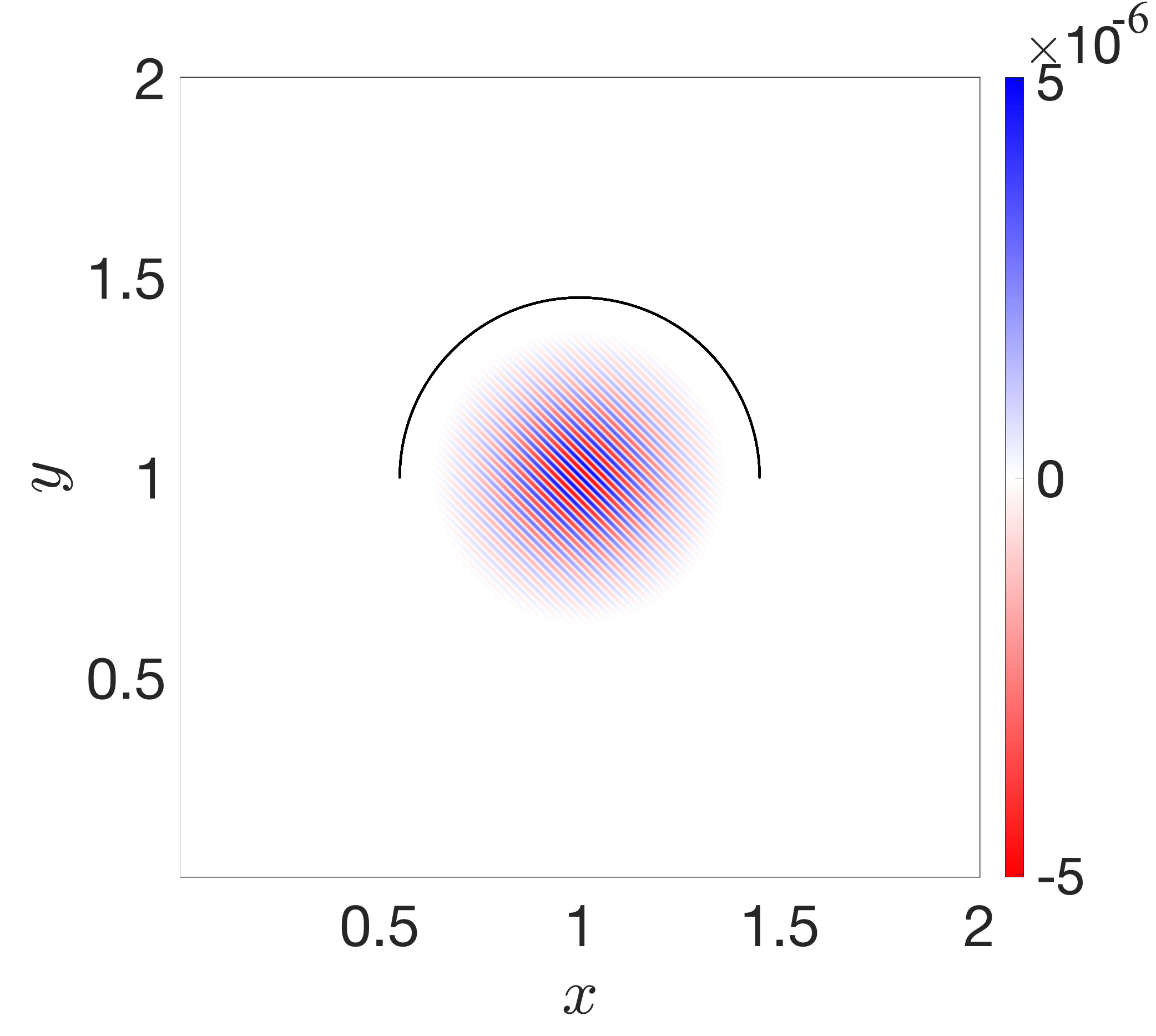}
\end{subfigure}
\begin{subfigure}[t]{.29\textwidth}
	\centering
	\includegraphics[width=\linewidth]{\fpath/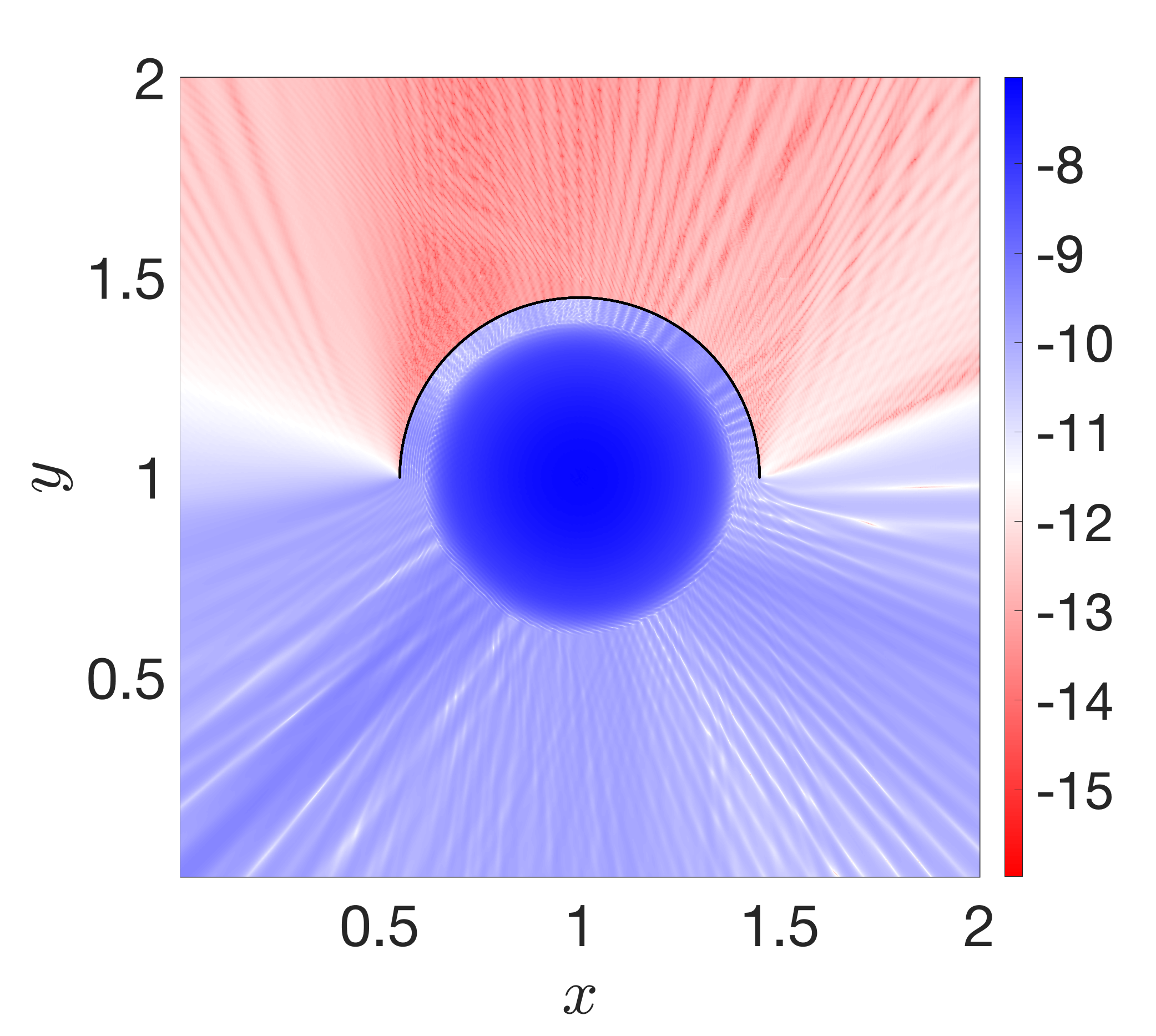}
\end{subfigure}	
\begin{subfigure}[t]{.37\textwidth}
	\centering
	\includegraphics[width=\linewidth]{\fpath/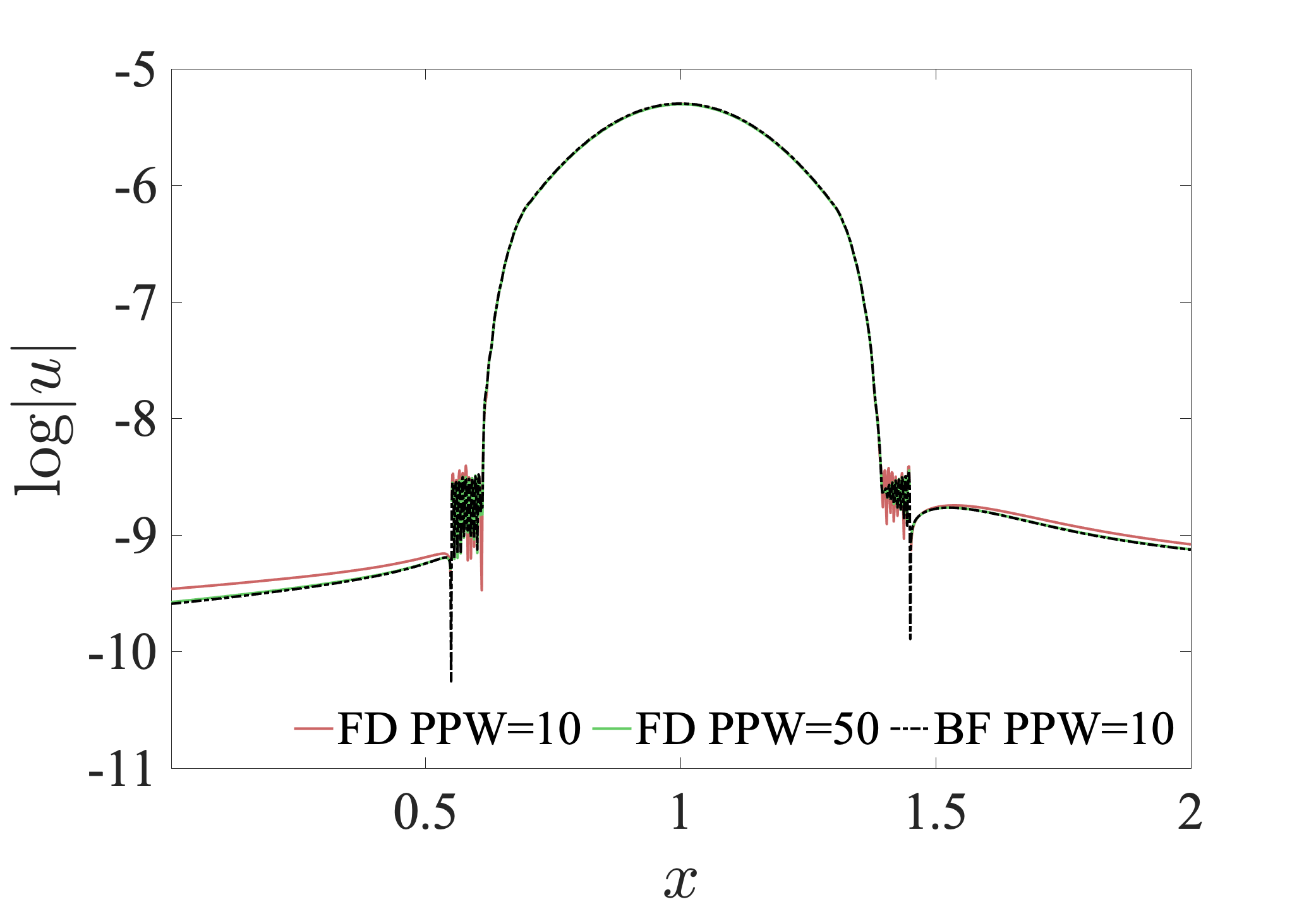}
\end{subfigure}		

\begin{subfigure}[t]{.29\textwidth}
	\centering
	\includegraphics[width=\linewidth]{\fpath/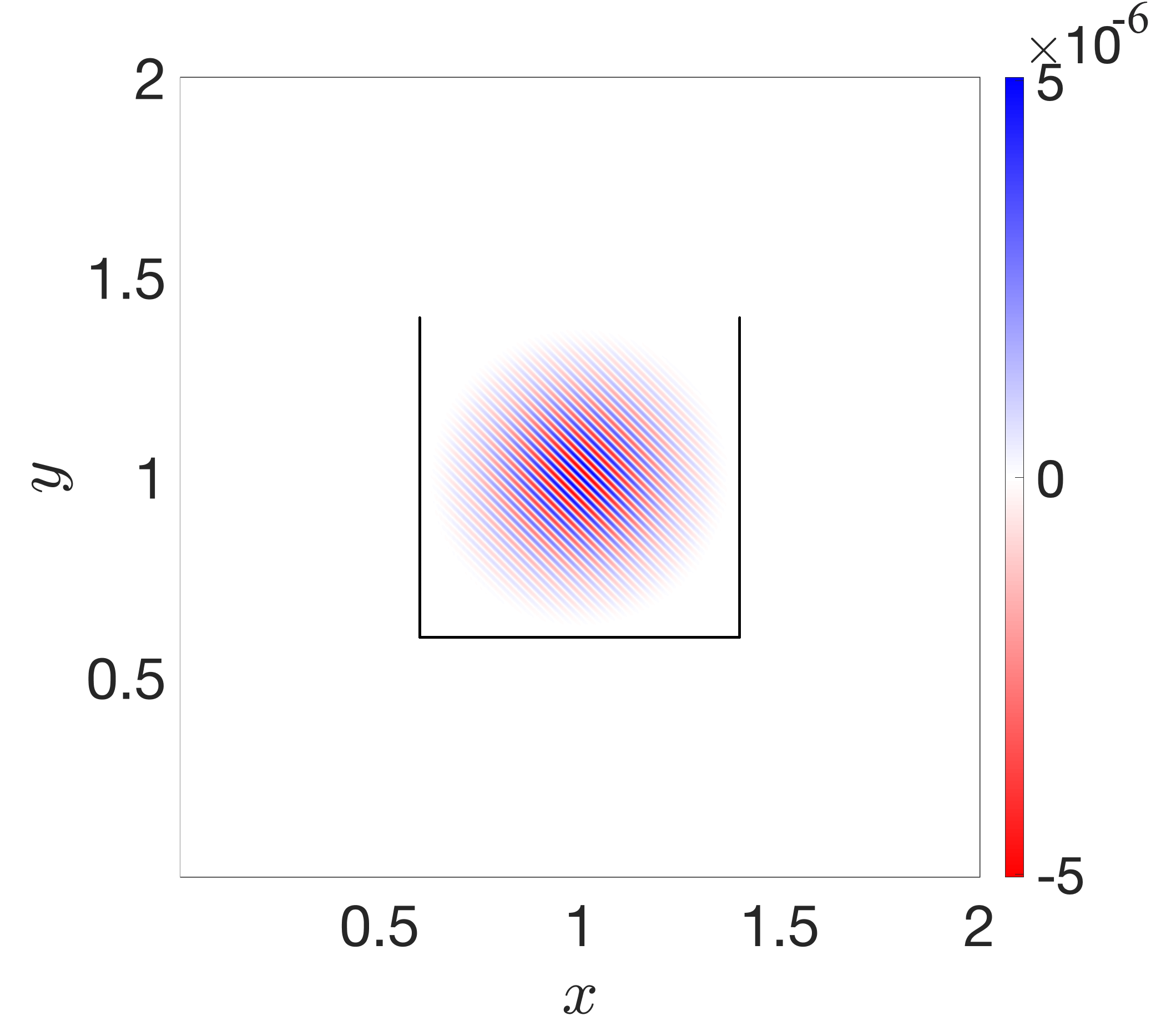}
\end{subfigure}
\begin{subfigure}[t]{.29\textwidth}
	\centering
	\includegraphics[width=\linewidth]{\fpath/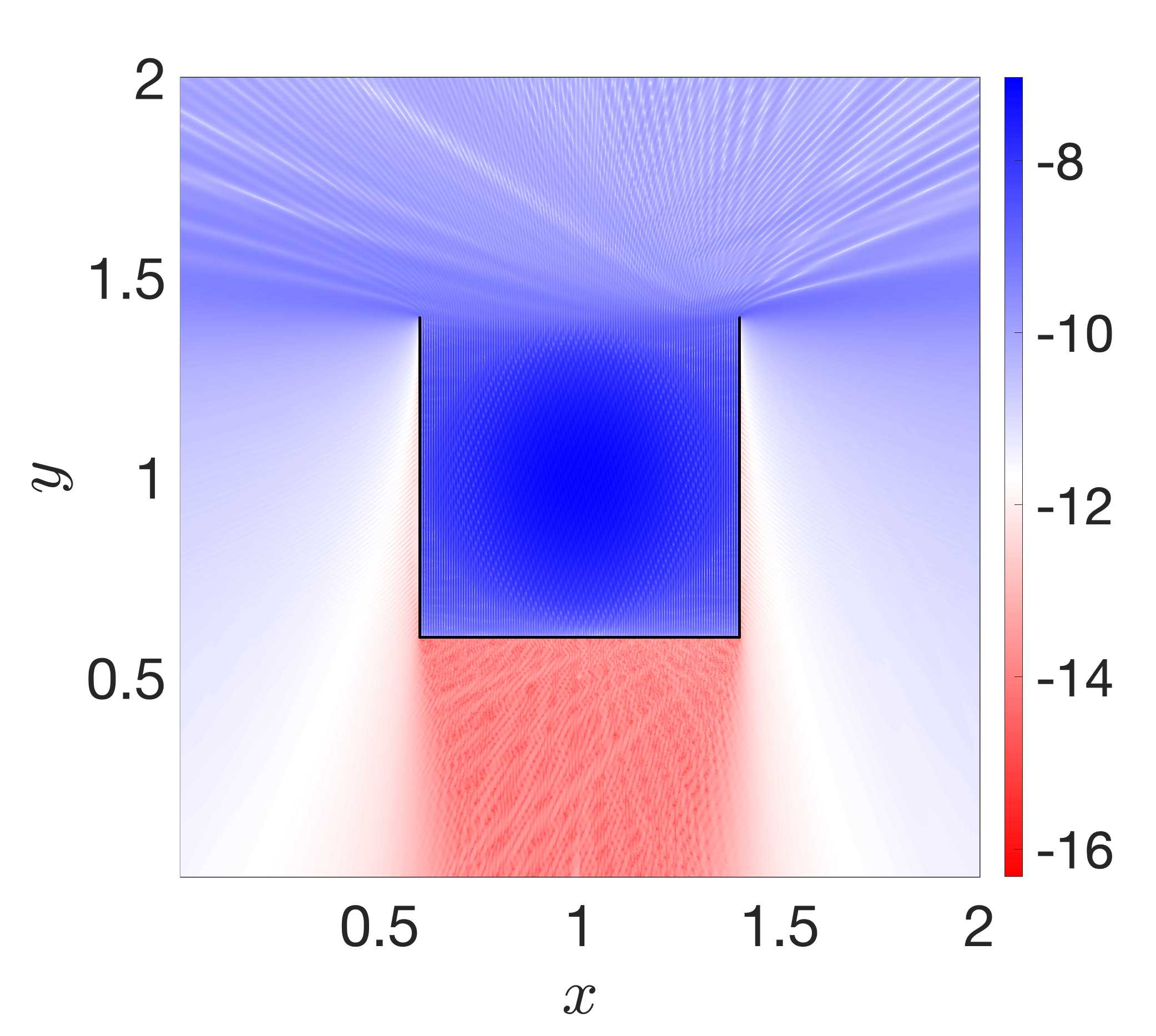}
\end{subfigure}	
\begin{subfigure}[t]{.37\textwidth}
	\centering
	\includegraphics[width=\linewidth]{\fpath/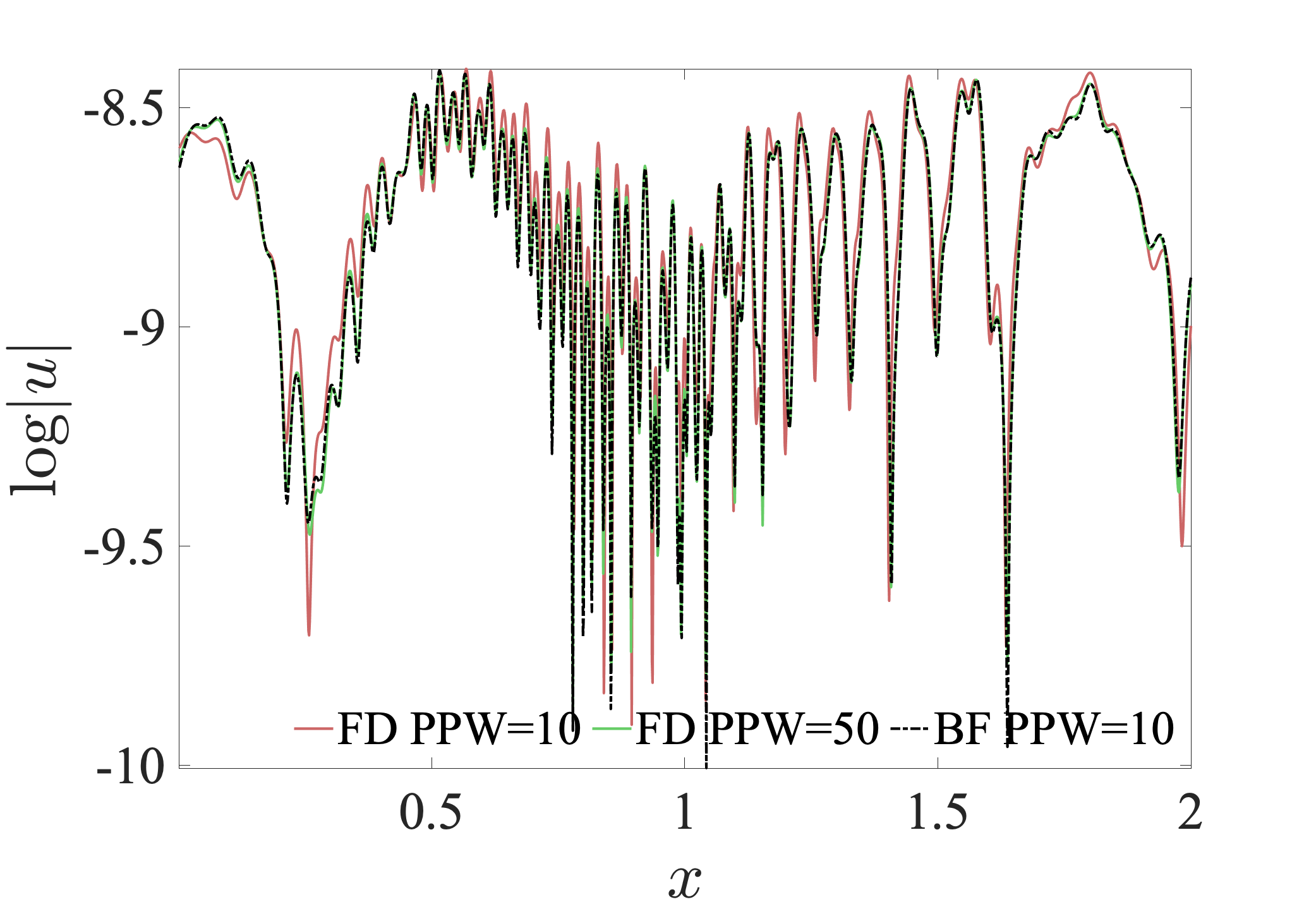}
\end{subfigure}		

	\vspace{-5pt}
	\caption{Constant media. $\omega=80\pi$ (160 wavelengths each direction). Left column: the field $\mathrm{Re}(u_{\rm hb})$ (in linear scale) computed by the proposed scheme. Middle column: difference $|u_{\rm hb}-u_{\rm fd}|$ (in log scale) between the fields computed by the proposed scheme (PPW=10) and FDFD (PPW=50). Right column: the fields $|u_{\rm hb}|,|u_{\rm fd}|,|u_{true}|$ (in log scale) drawn along the line $y=y_{\rm post}$. Row 1: point source with $y_{\rm post}=2-10h$ and $h$ corresponding to PPW=10. Row 2: kite-shaped source with $y_{\rm post}=2-10h$. Row 3: Gaussian packet source with $y_{\rm post}=1-10h$ and a semi-circle inclusion. Row 4: Gaussian packet source with $y_{\rm post}=2-10h$ and an open-square inclusion. }
		\label{fig:ex1_f40}
\end{figure}

\begin{figure}[!htp]
	\centering
	\vspace{-7.5pt}
	\begin{subfigure}[t]{.29\textwidth}
		\centering
		\includegraphics[width=\linewidth]{\fpath/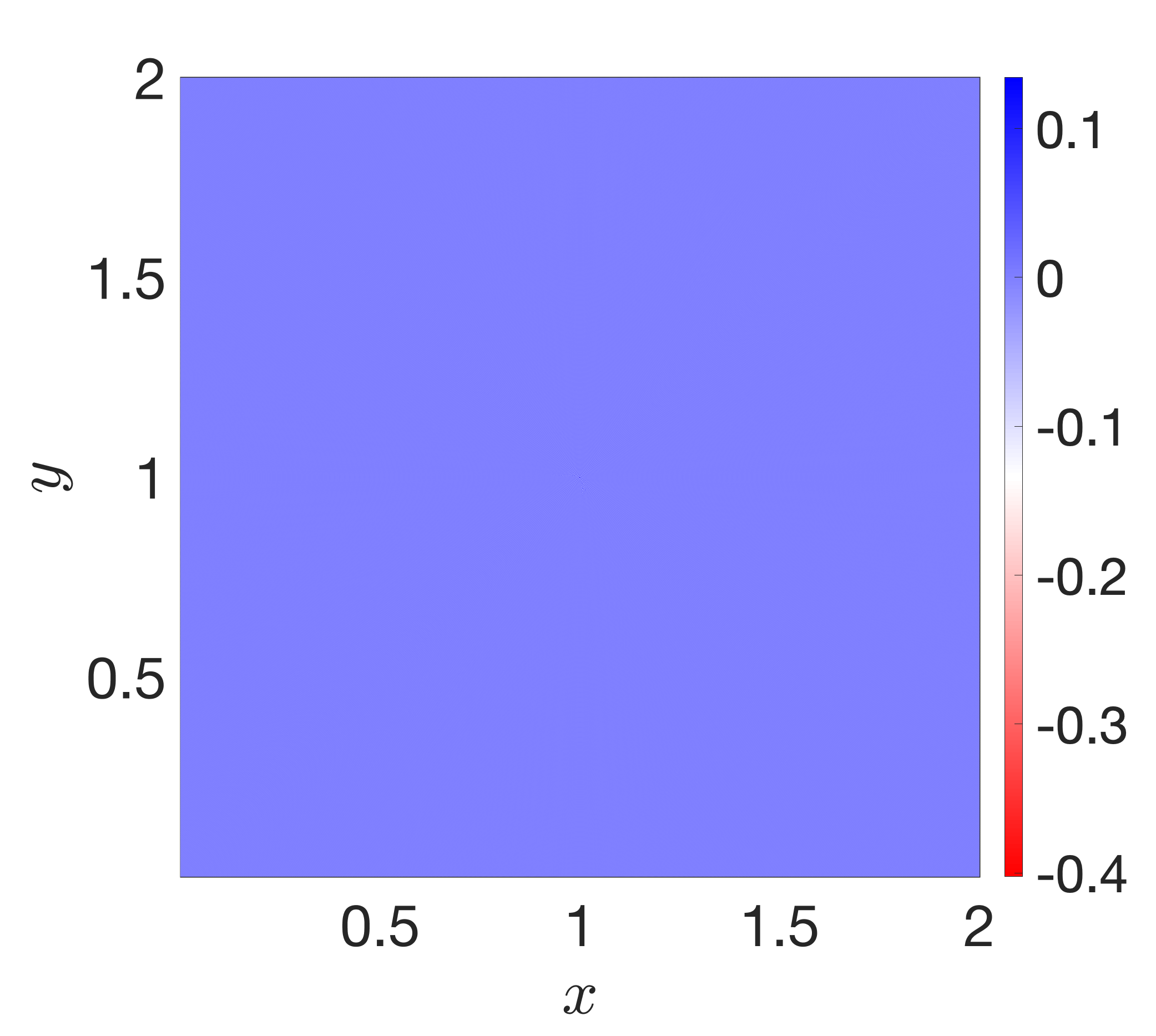}
	\end{subfigure}
	\begin{subfigure}[t]{.29\textwidth}
		\centering
		\includegraphics[width=\linewidth]{\fpath/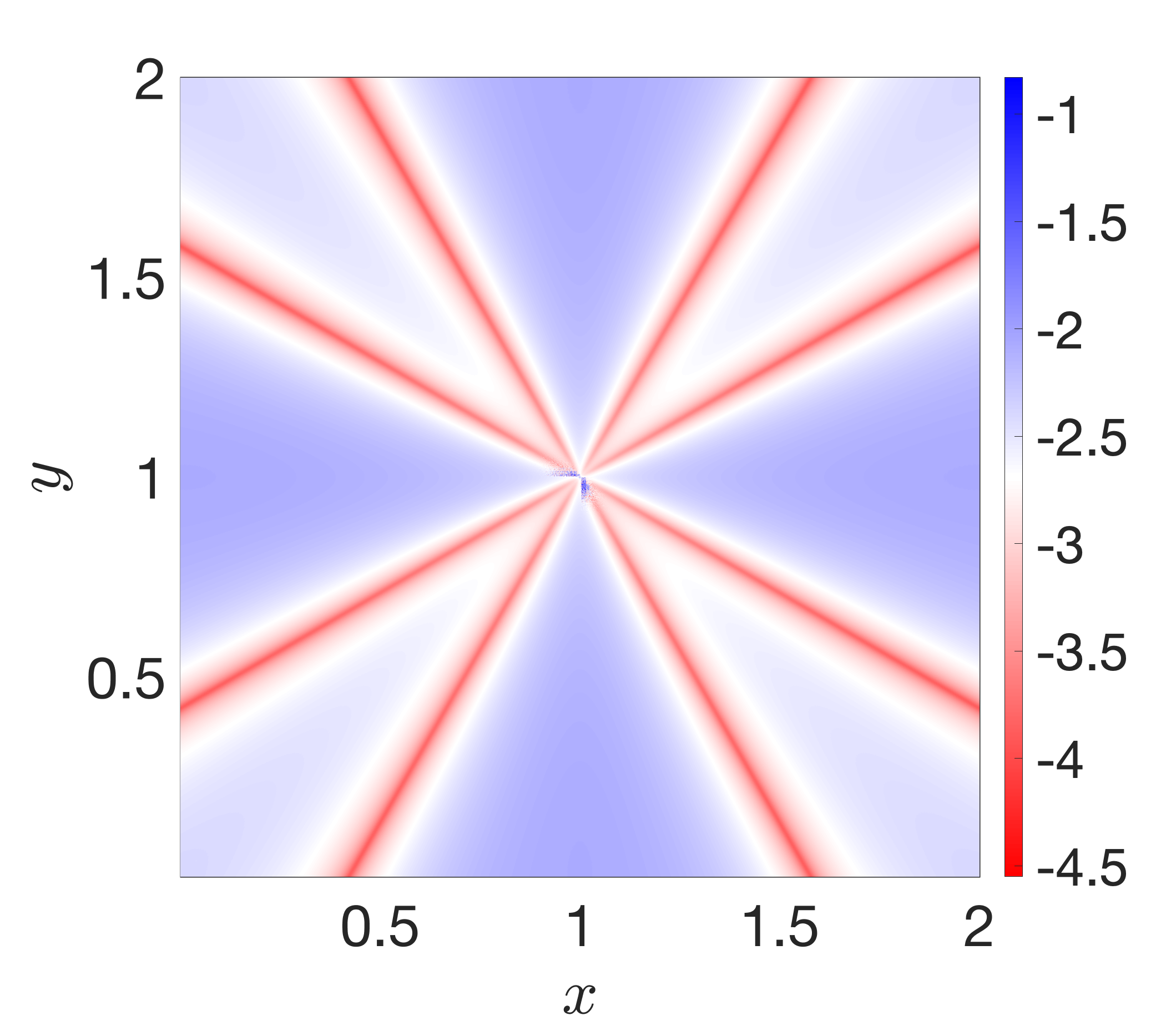}
	\end{subfigure}	
	\begin{subfigure}[t]{.37\textwidth}
		\centering
		\includegraphics[width=\linewidth]{\fpath/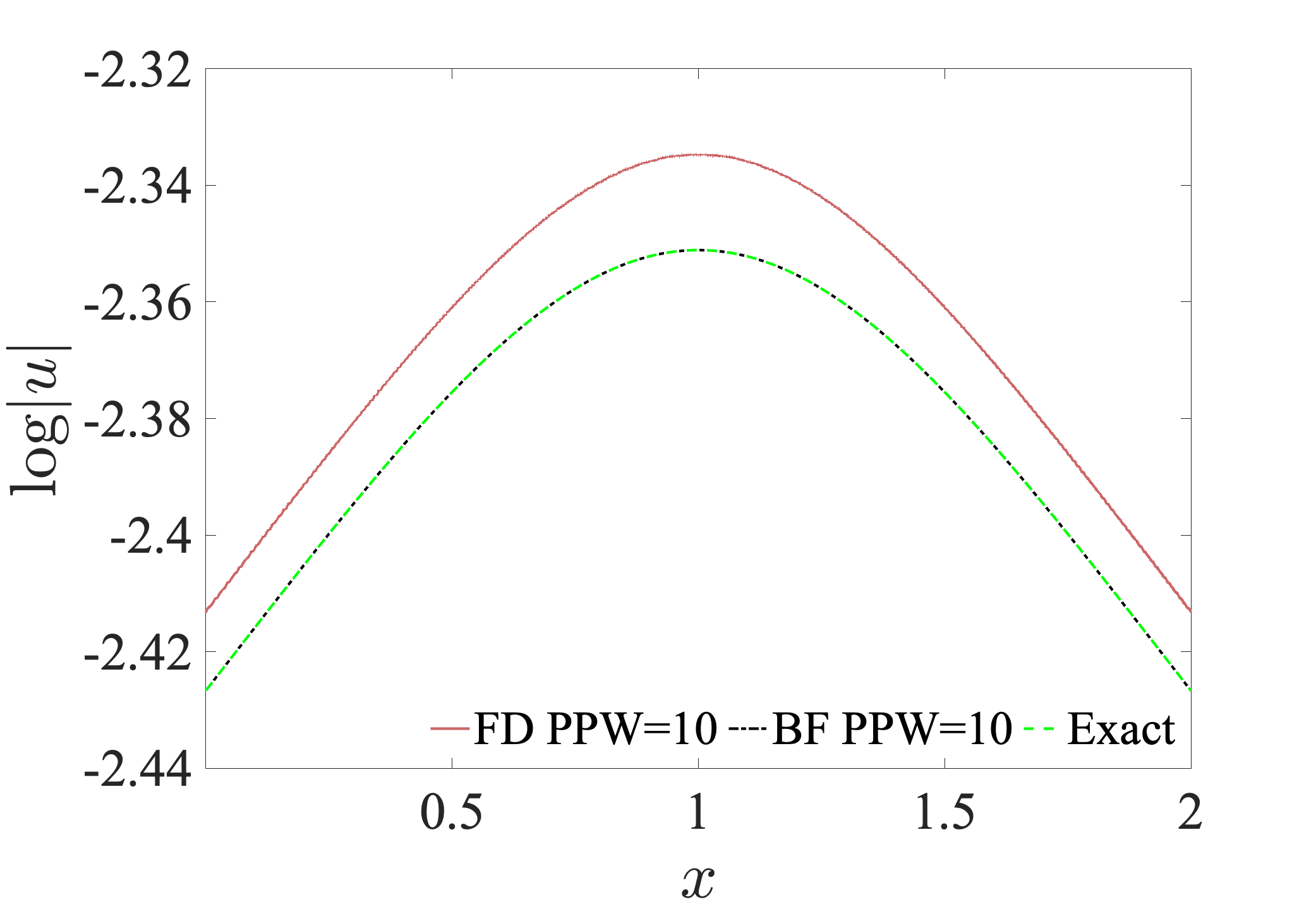}
	\end{subfigure}	
	\begin{subfigure}[t]{.29\textwidth}
	\centering
	\includegraphics[width=\linewidth]{\fpath/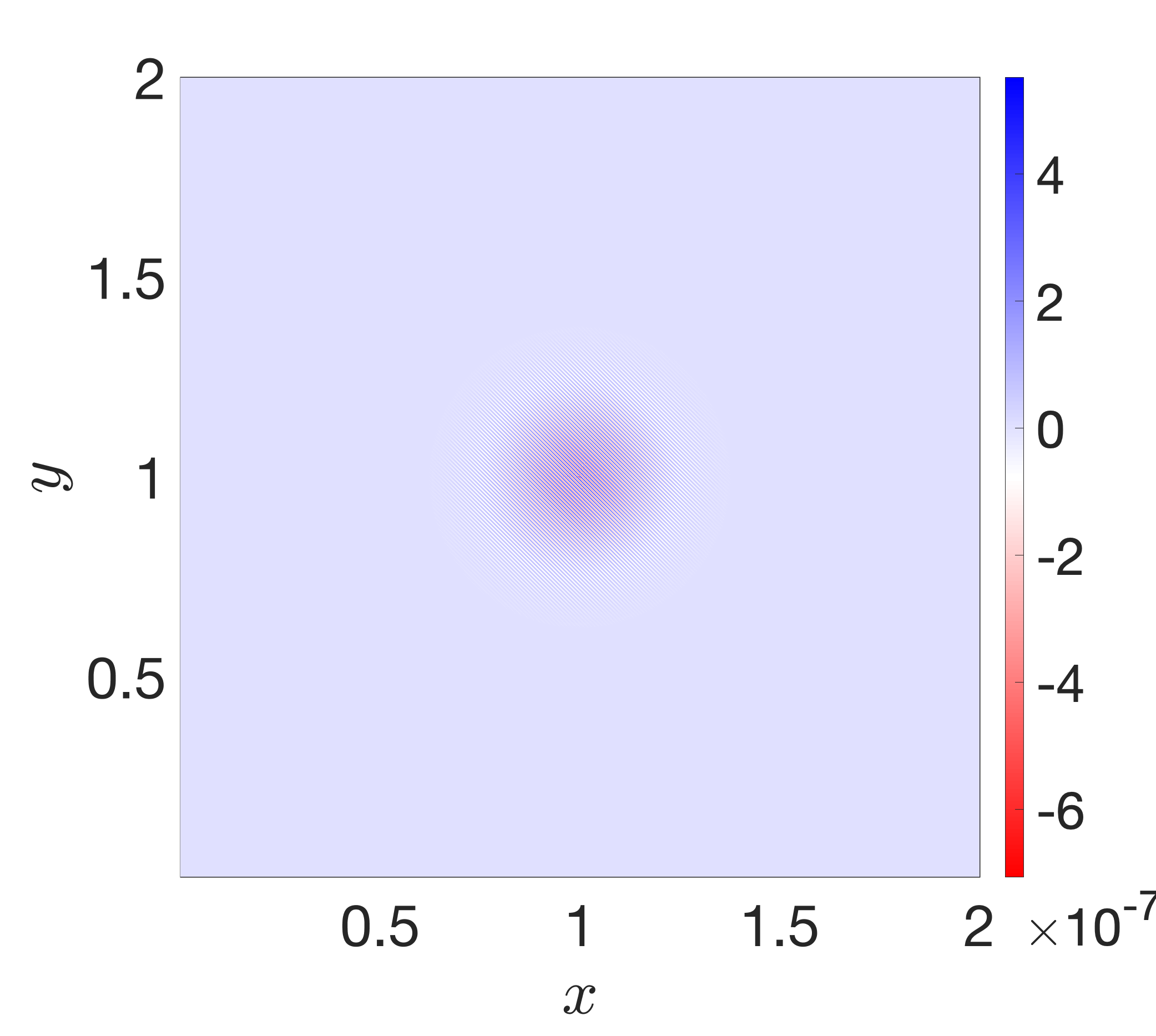}
\end{subfigure}
\begin{subfigure}[t]{.29\textwidth}
	\centering
	\includegraphics[width=\linewidth]{\fpath/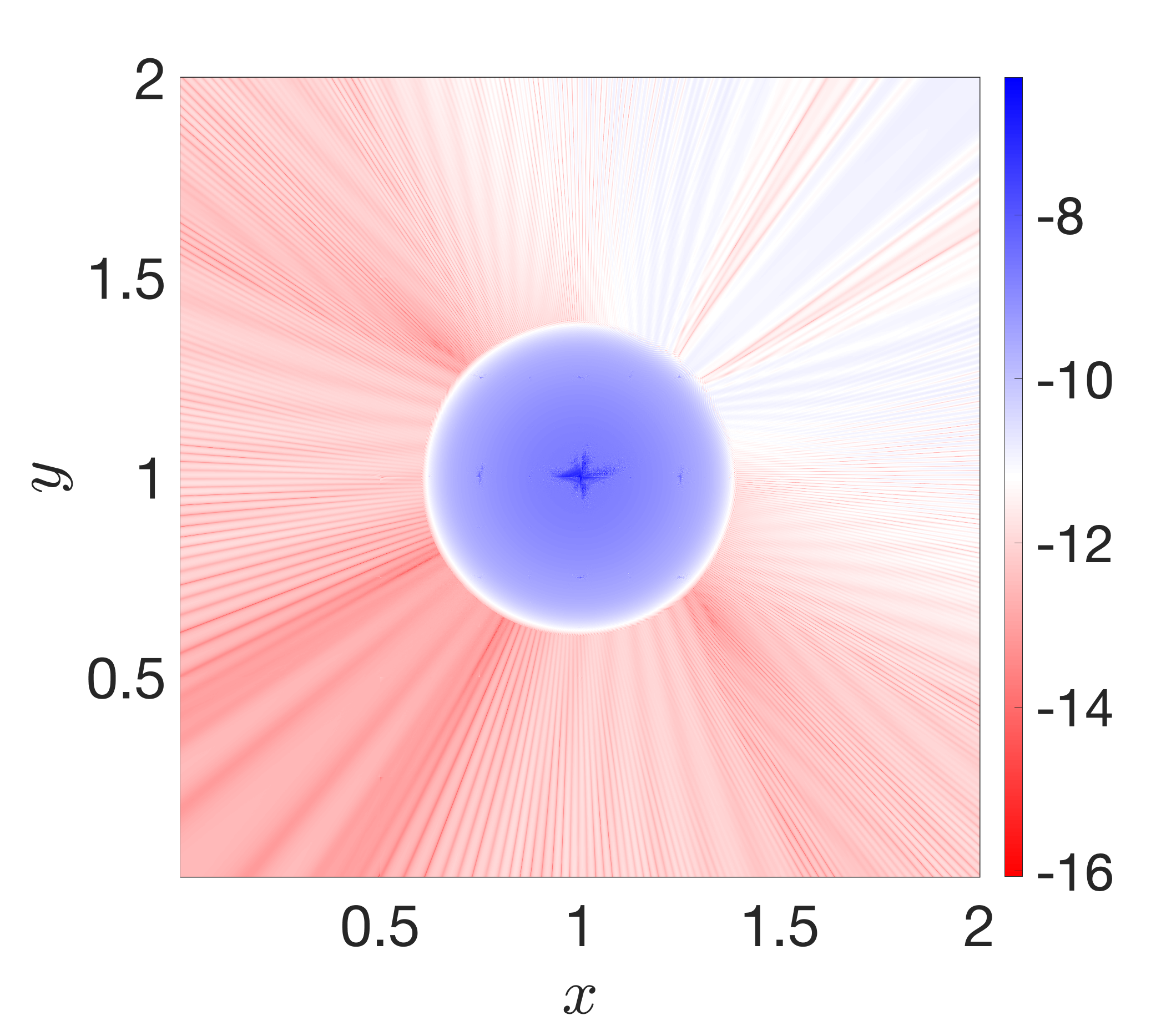}
\end{subfigure}	
\begin{subfigure}[t]{.37\textwidth}
	\centering
	\includegraphics[width=\linewidth]{\fpath/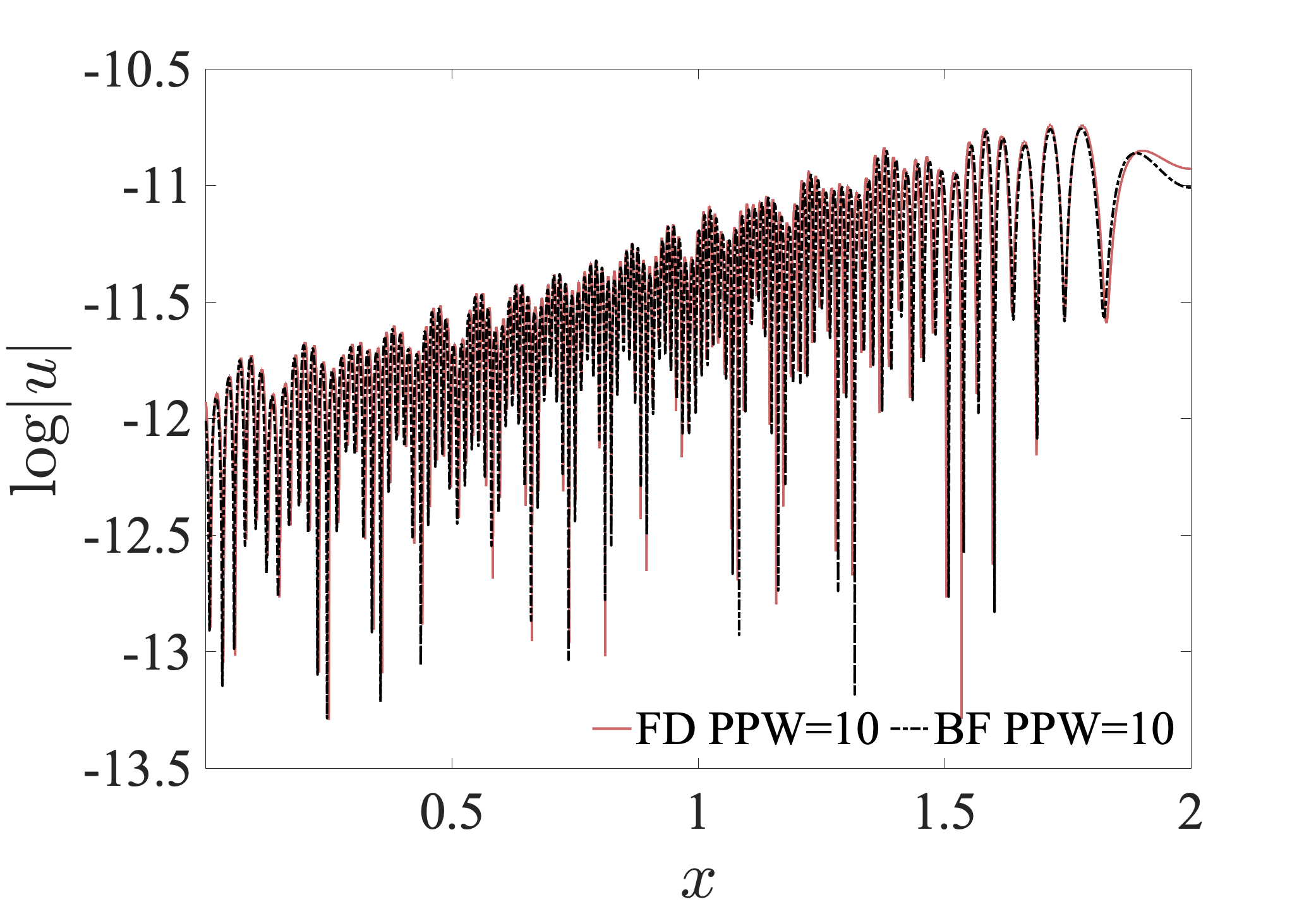}
\end{subfigure}	
	\begin{subfigure}[t]{.29\textwidth}
	\centering
	\includegraphics[width=\linewidth]{\fpath/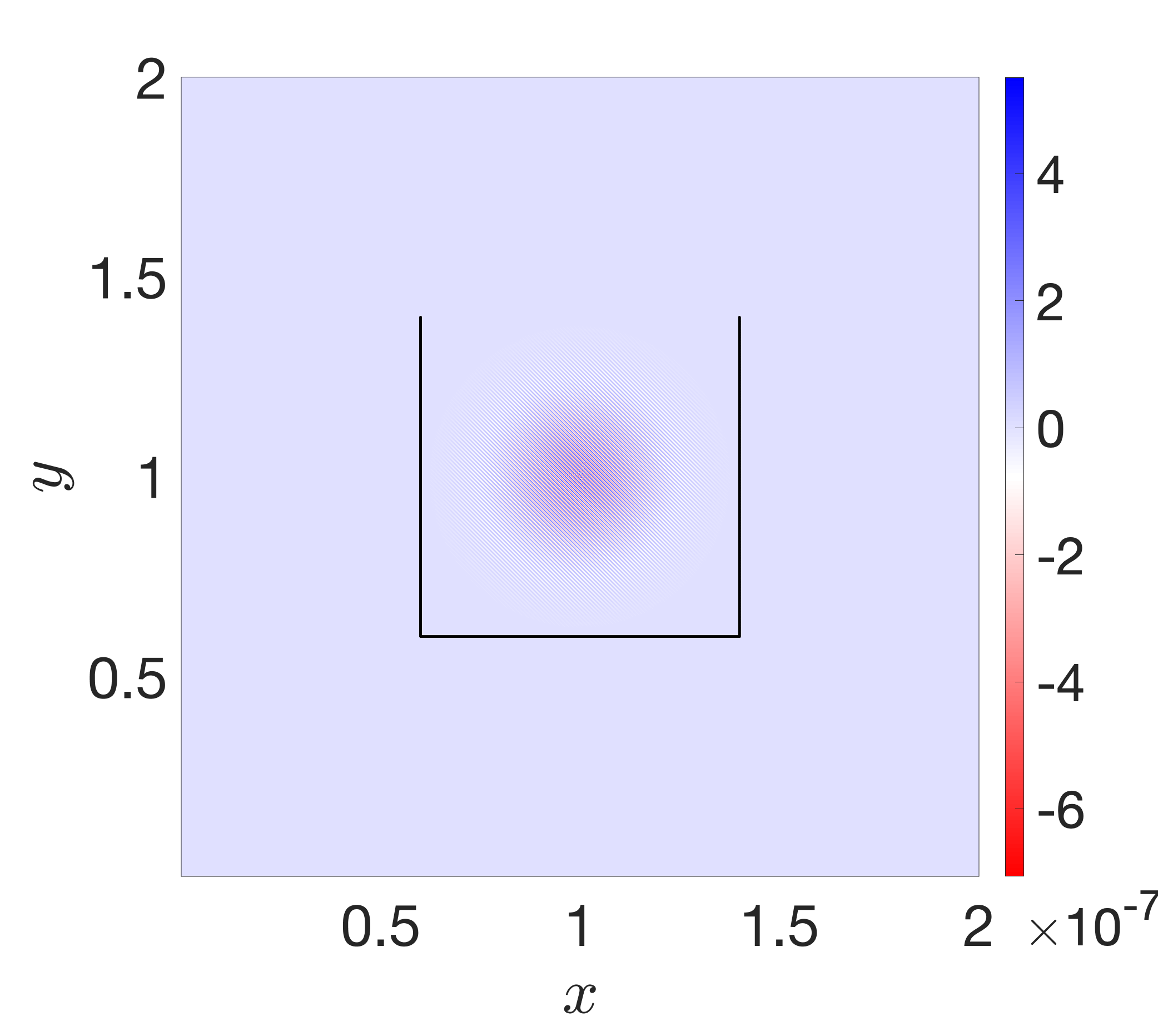}
\end{subfigure}
\begin{subfigure}[t]{.29\textwidth}
	\centering
	\includegraphics[width=\linewidth]{\fpath/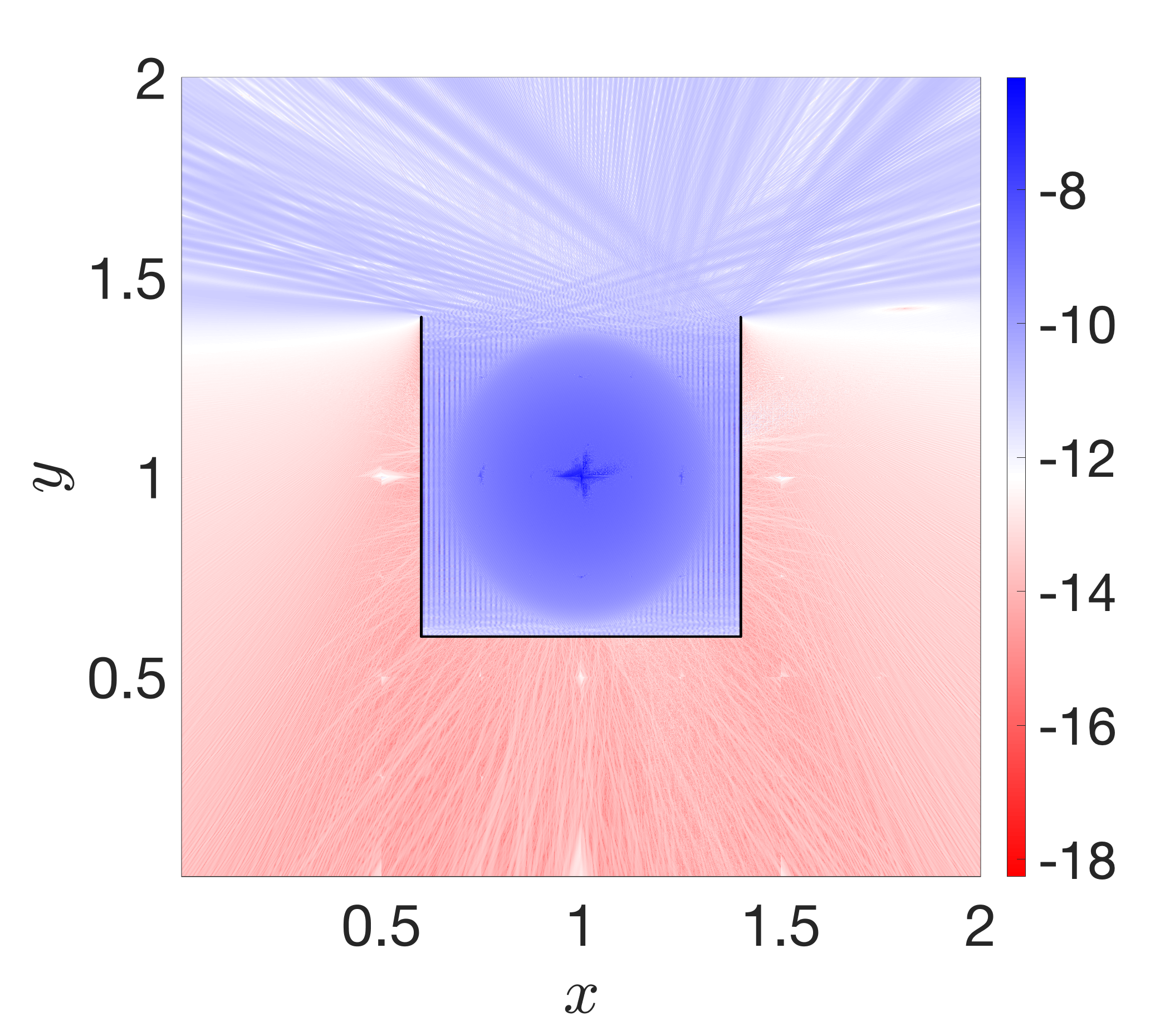}
\end{subfigure}	
\begin{subfigure}[t]{.37\textwidth}
	\centering
	\includegraphics[width=\linewidth]{\fpath/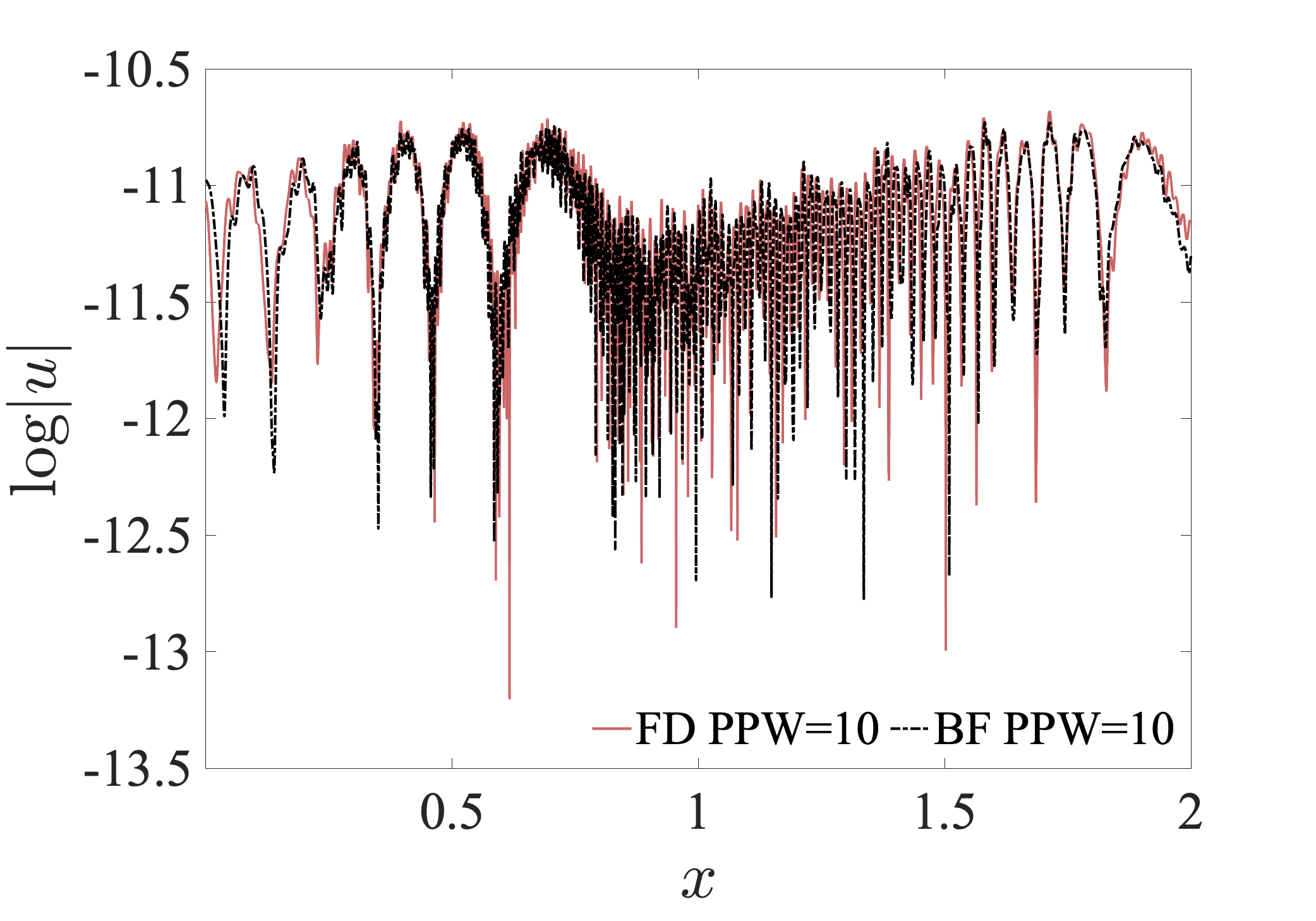}
\end{subfigure}		
	\vspace{-5pt}
	\caption{Constant media. $\omega=320\pi$ (640 wavelengths each direction). Left column: the field $\mathrm{Re}(u_{\rm hb})$ (in linear scale) computed by the proposed scheme. Middle column: difference $|u_{\rm hb}-u_{\rm fd}|$ (in log scale) between the fields computed by the proposed scheme (PPW=10) and FDFD (PPW=10). Right column: the fields $|u_{\rm hb}|,|u_{\rm fd}|,|u_{true}|$ (in log scale) drawn along the line $y=y_{\rm post}$. Row 1: point source with $y_{\rm post}=2-10h$ and $h$ corresponding to PPW=10. Row 2: Gaussian packet source with $y_{\rm post}=2-10h$. Row 3: Gaussian packet source with $y_{\rm post}=2-10h$ and an open-square inclusion. }
	\label{fig:ex1_f160}
\end{figure}

\begin{figure}[!htp]
	\centering
	\vspace{-7.5pt}
	\begin{subfigure}[t]{.29\textwidth}
		\centering
		\includegraphics[width=\linewidth]{\fpath/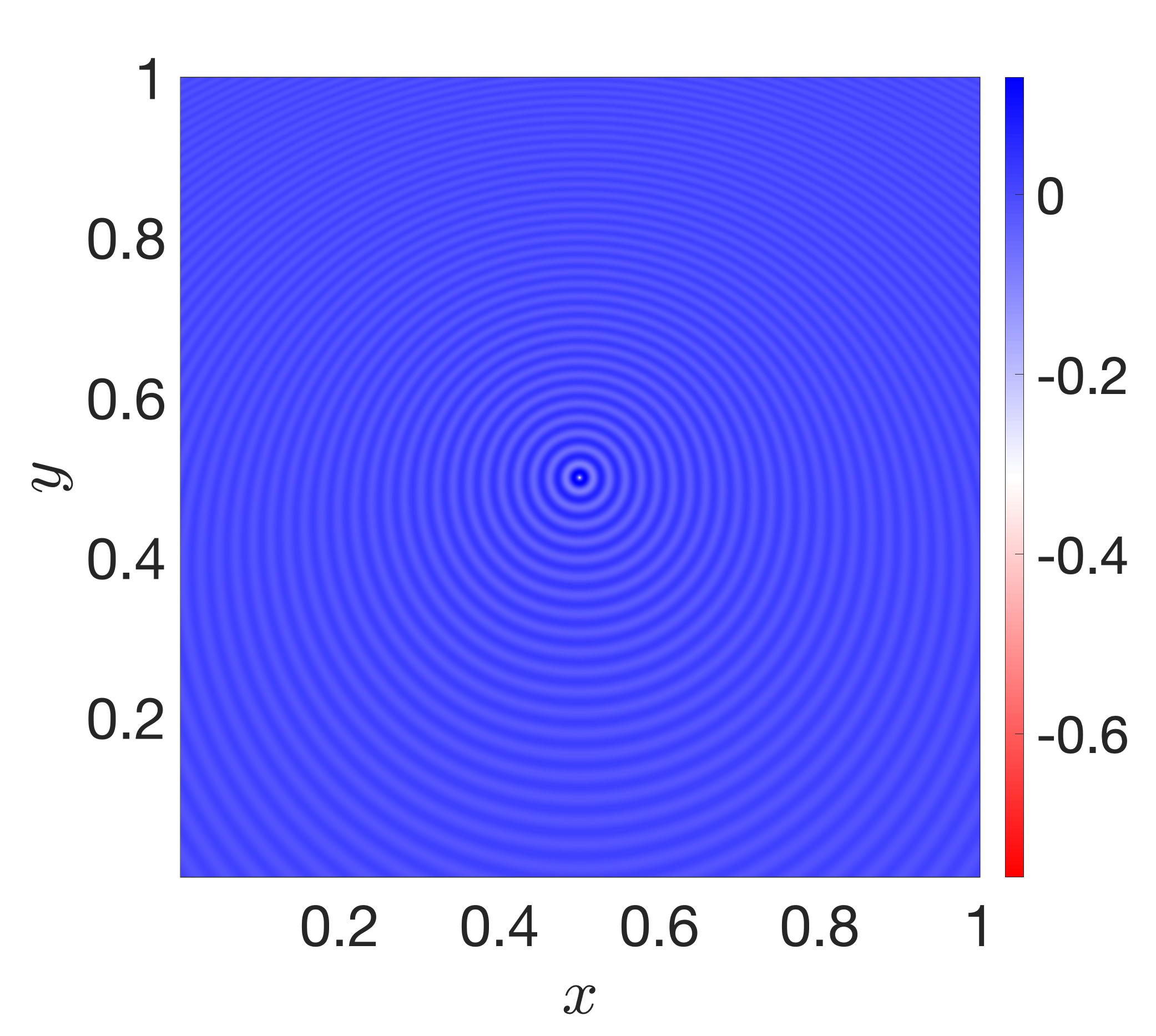}
	\end{subfigure}
	\begin{subfigure}[t]{.29\textwidth}
		\centering
		\includegraphics[width=\linewidth]{\fpath/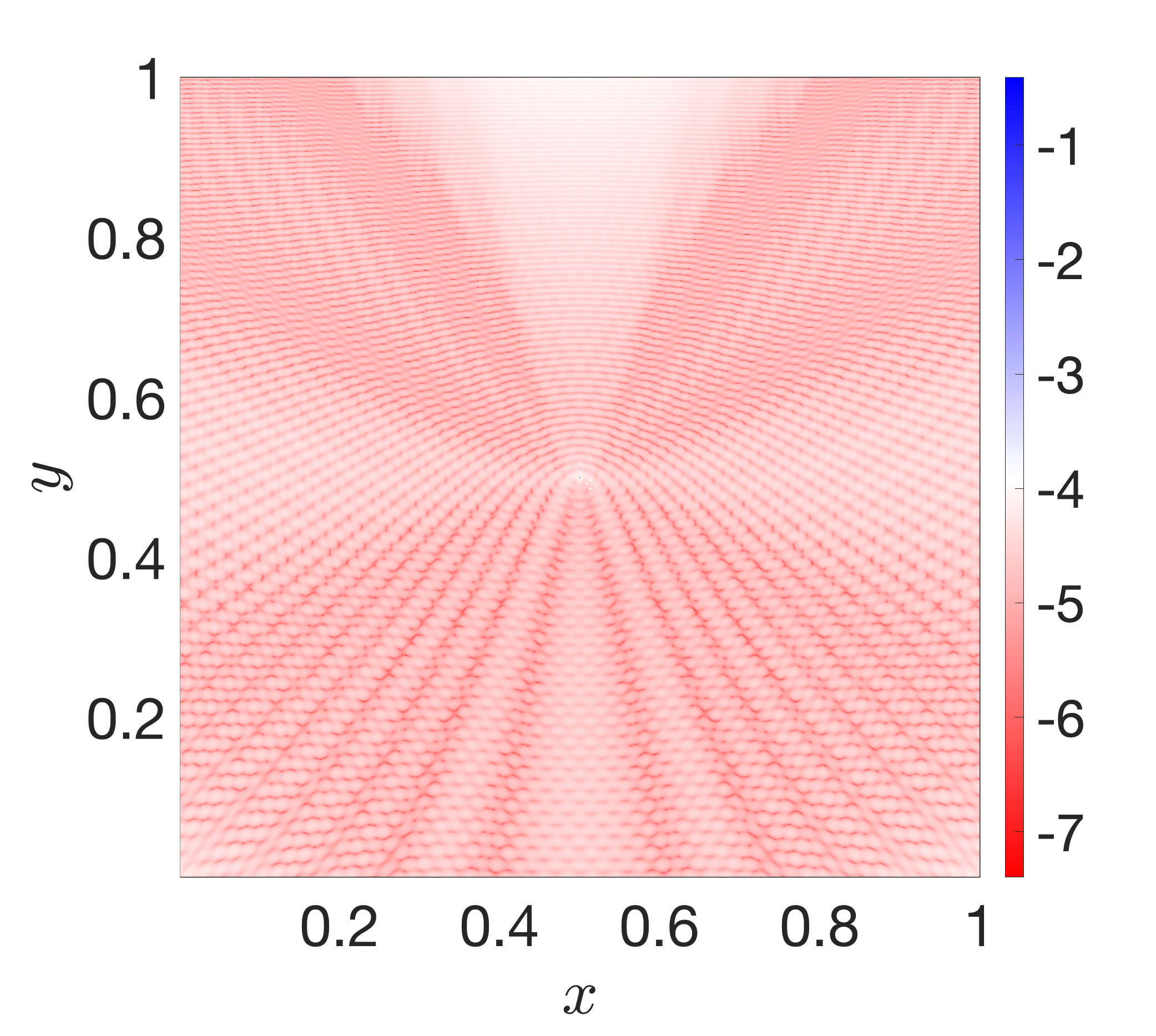}
	\end{subfigure}	
	\begin{subfigure}[t]{.37\textwidth}
		\centering
		\includegraphics[width=\linewidth]{\fpath/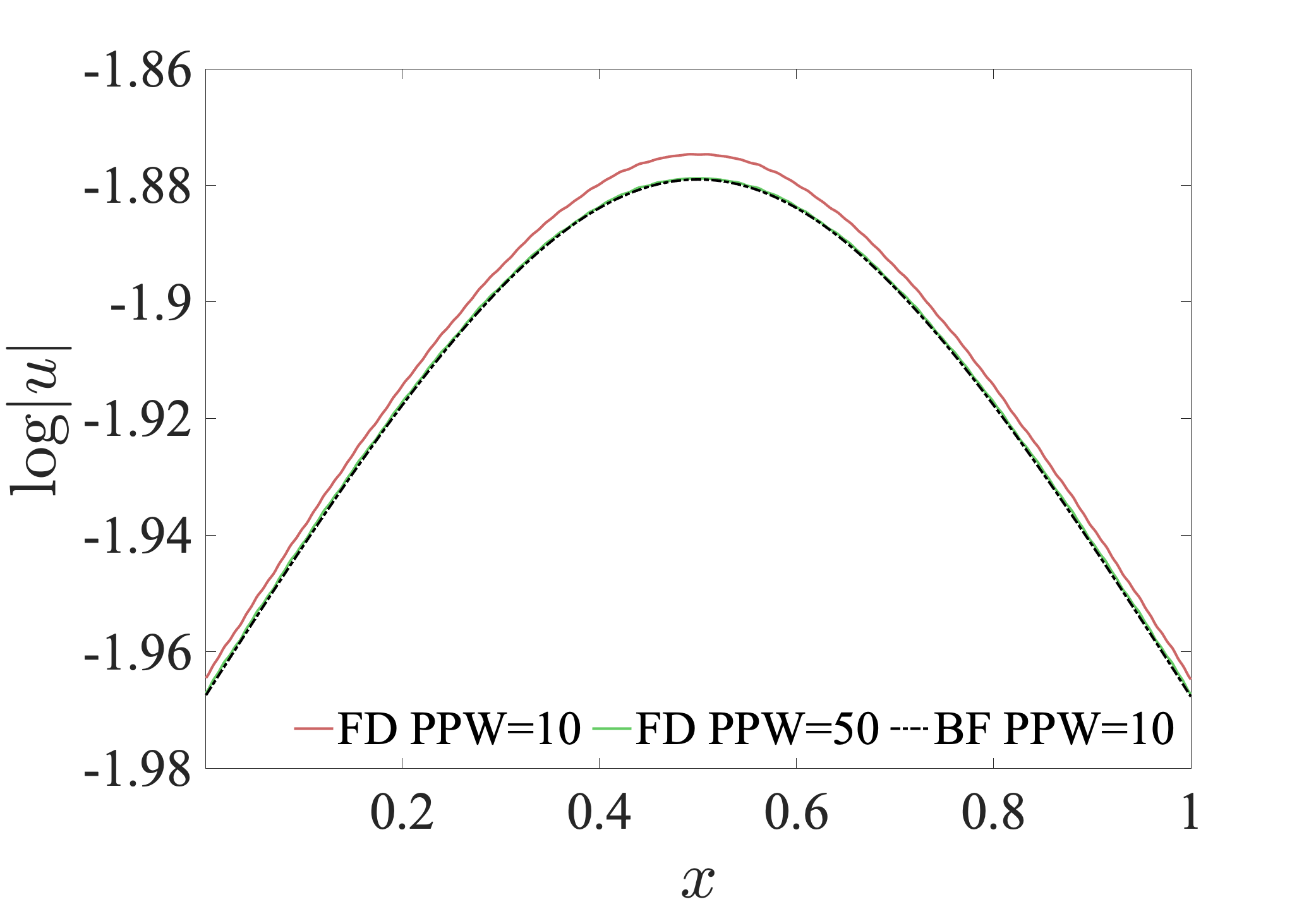}
	\end{subfigure}	
	\begin{subfigure}[t]{.29\textwidth}
	\centering
	\includegraphics[width=\linewidth]{\fpath/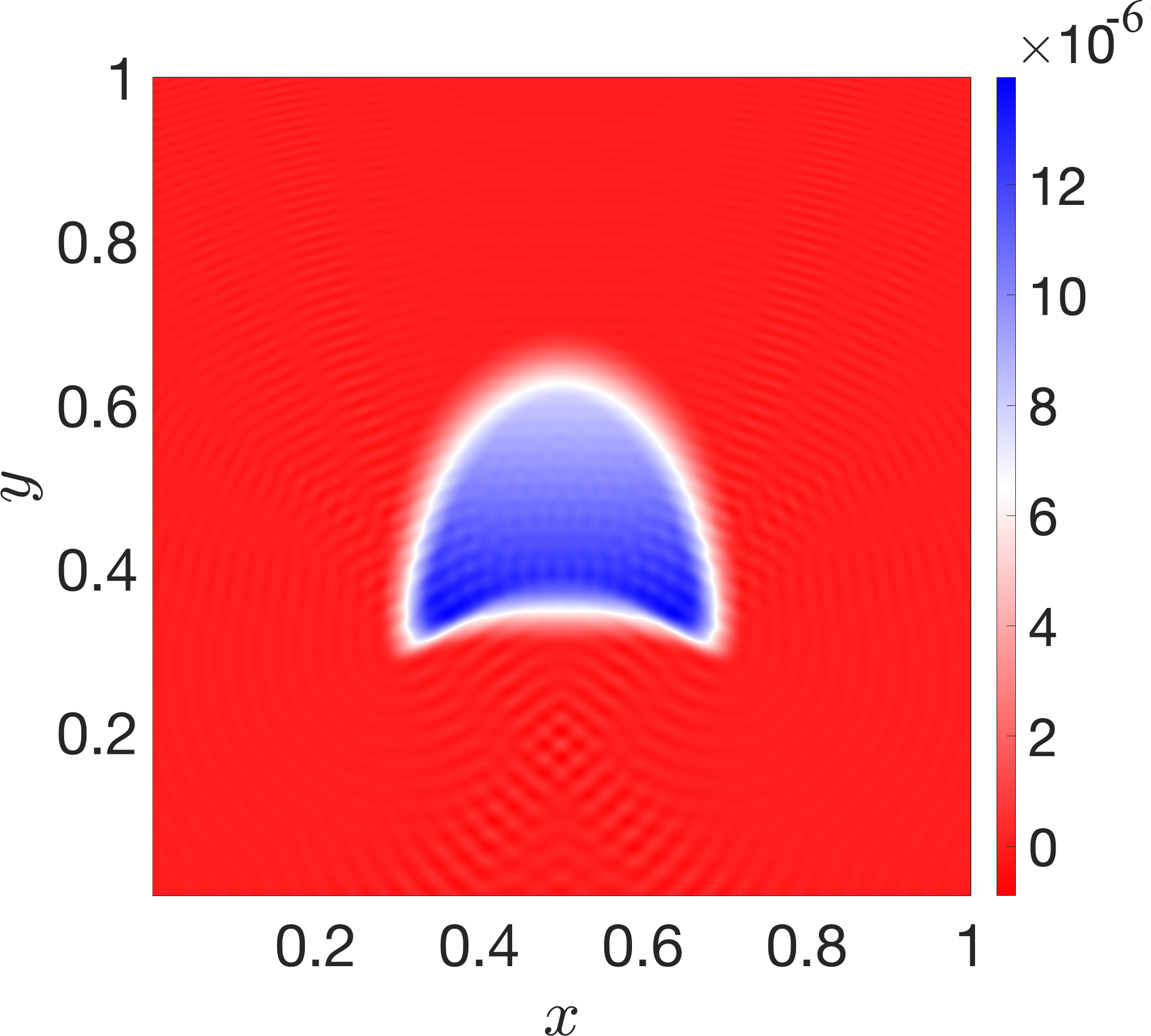}
\end{subfigure}
\begin{subfigure}[t]{.29\textwidth}
	\centering
	\includegraphics[width=\linewidth]{\fpath/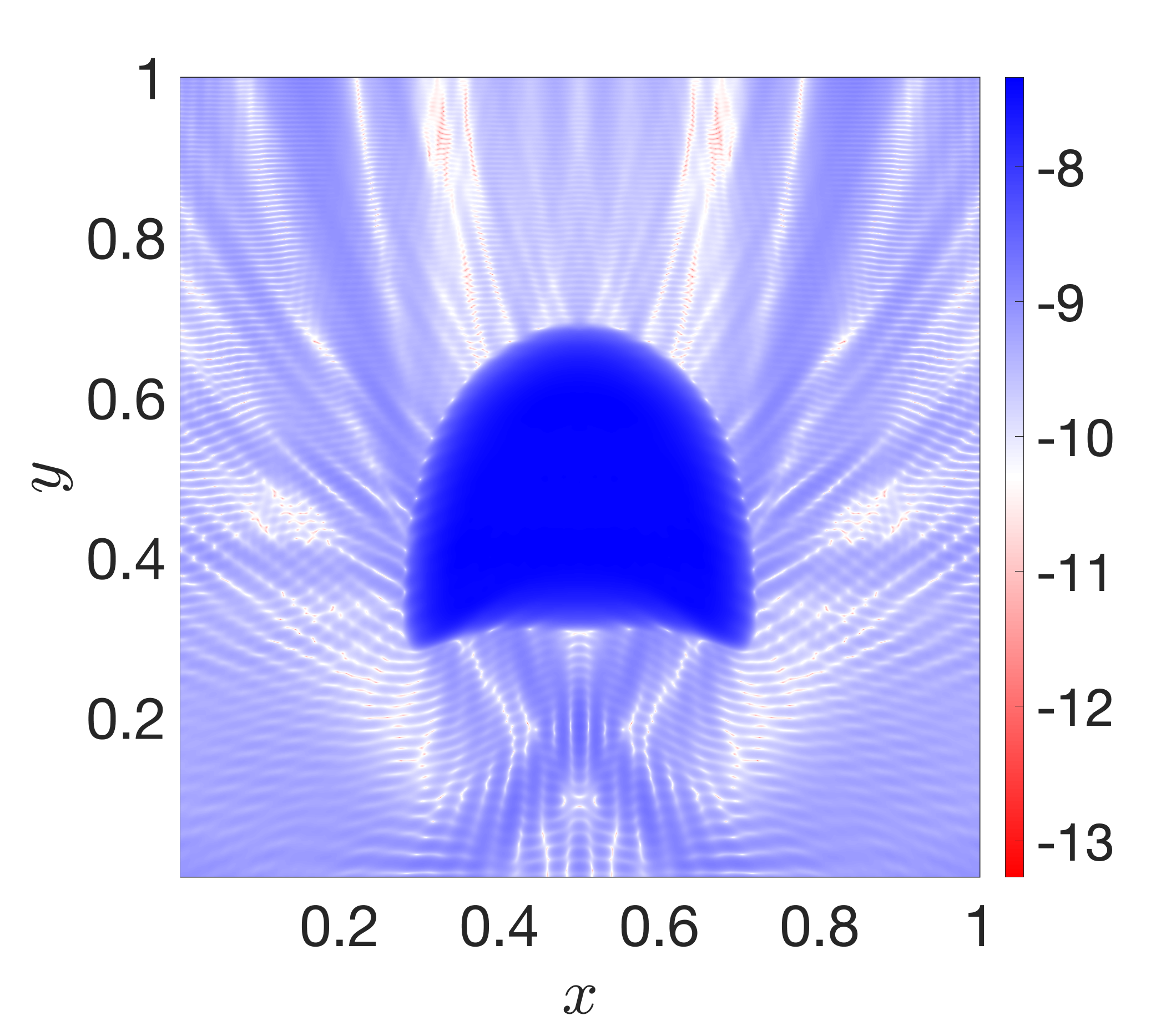}
\end{subfigure}	
\begin{subfigure}[t]{.37\textwidth}
	\centering
	\includegraphics[width=\linewidth]{\fpath/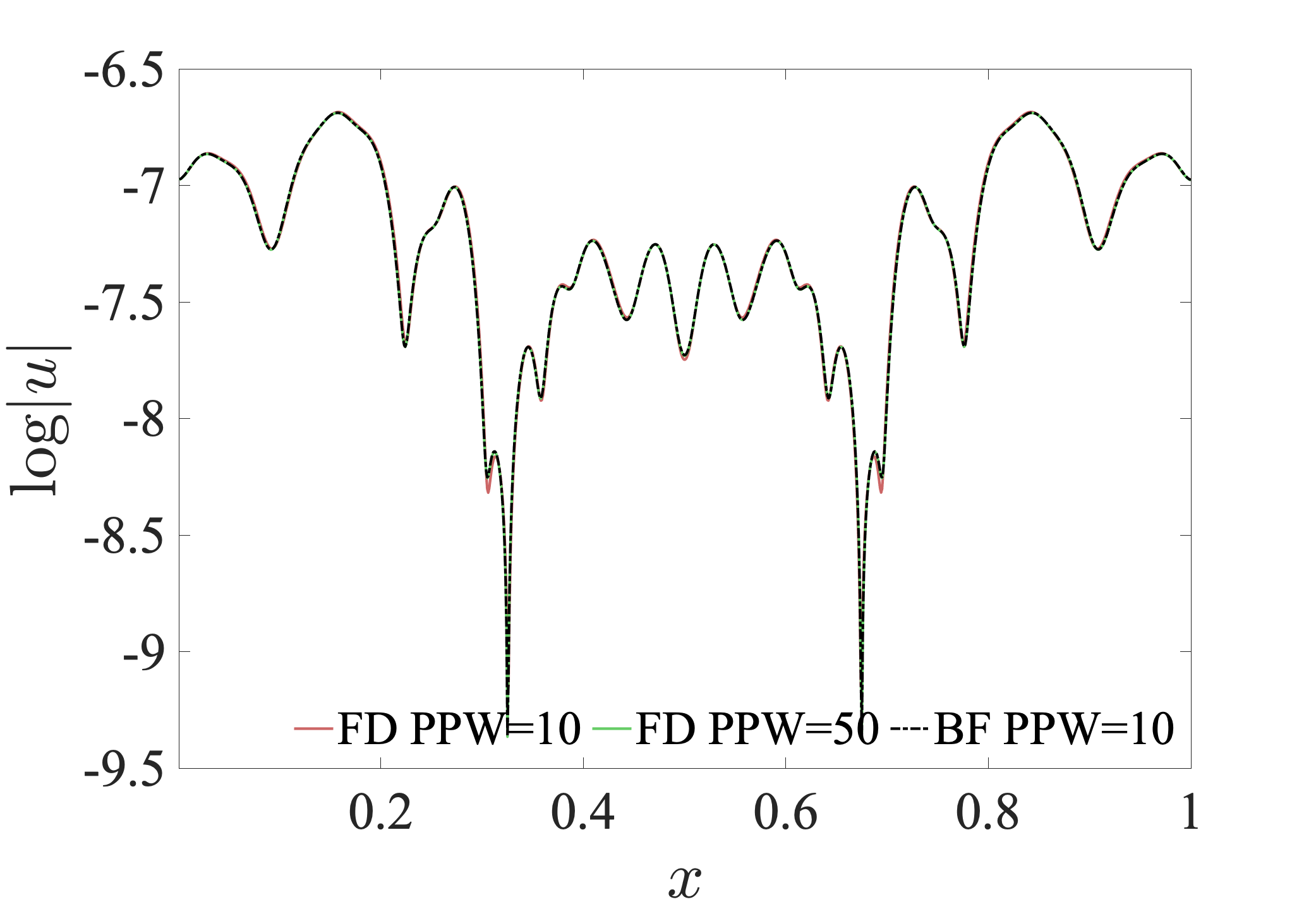}
\end{subfigure}	

\begin{subfigure}[t]{.29\textwidth}
	\centering
	\includegraphics[width=\linewidth]{\fpath/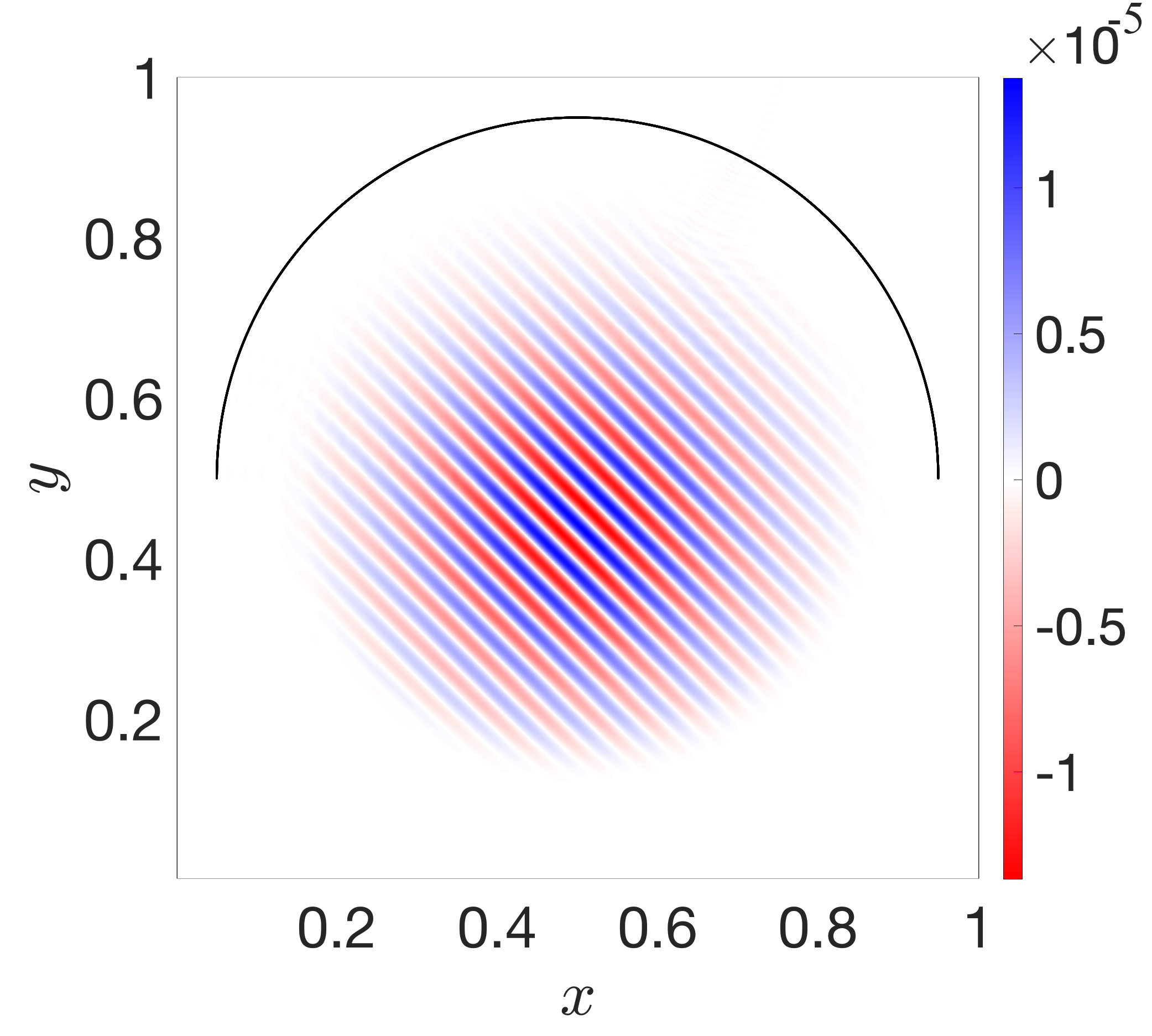}
\end{subfigure}
\begin{subfigure}[t]{.29\textwidth}
	\centering
	\includegraphics[width=\linewidth]{\fpath/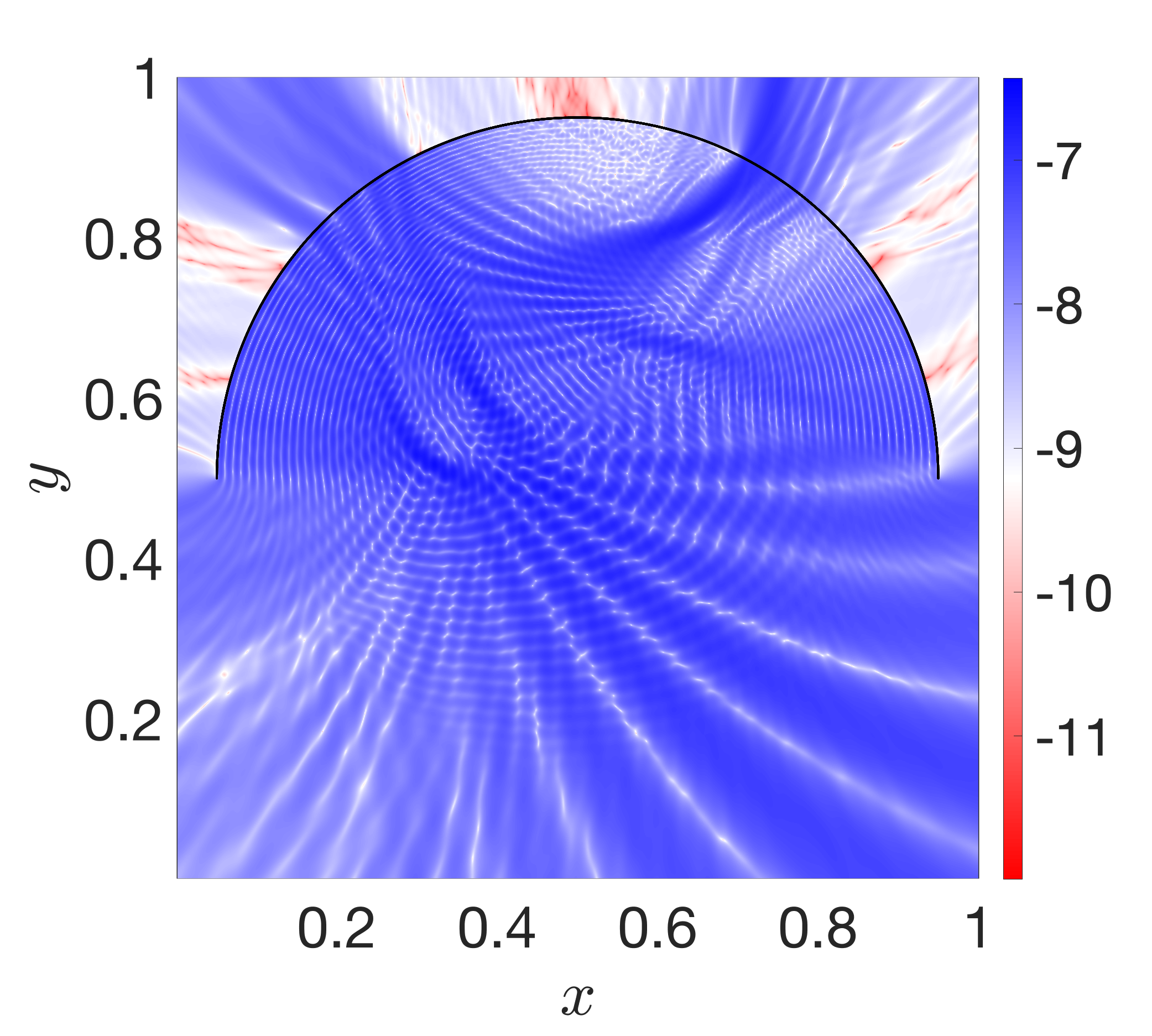}
\end{subfigure}	
\begin{subfigure}[t]{.37\textwidth}
	\centering
	\includegraphics[width=\linewidth]{\fpath/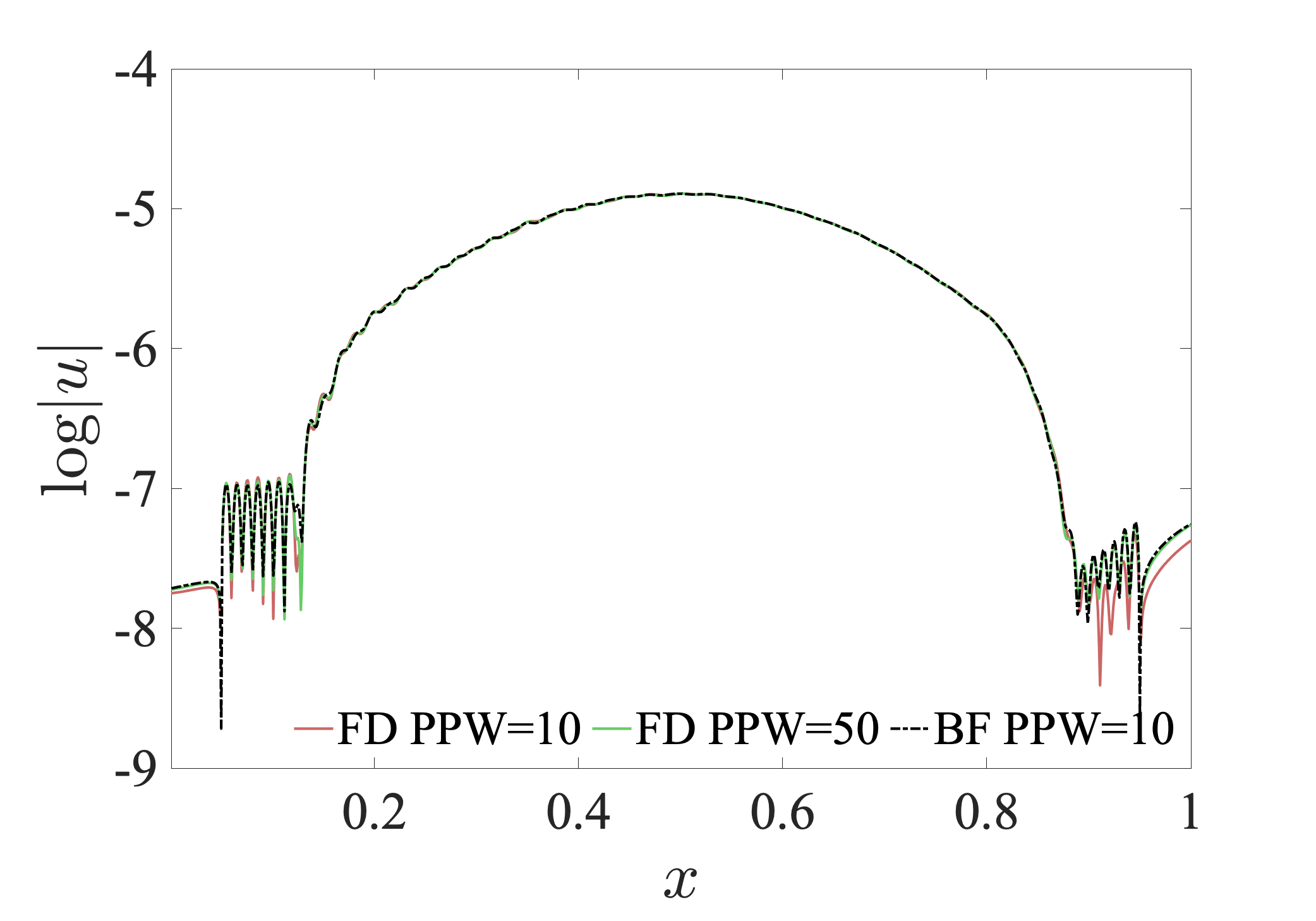}
\end{subfigure}		

\begin{subfigure}[t]{.29\textwidth}
	\centering
	\includegraphics[width=\linewidth]{\fpath/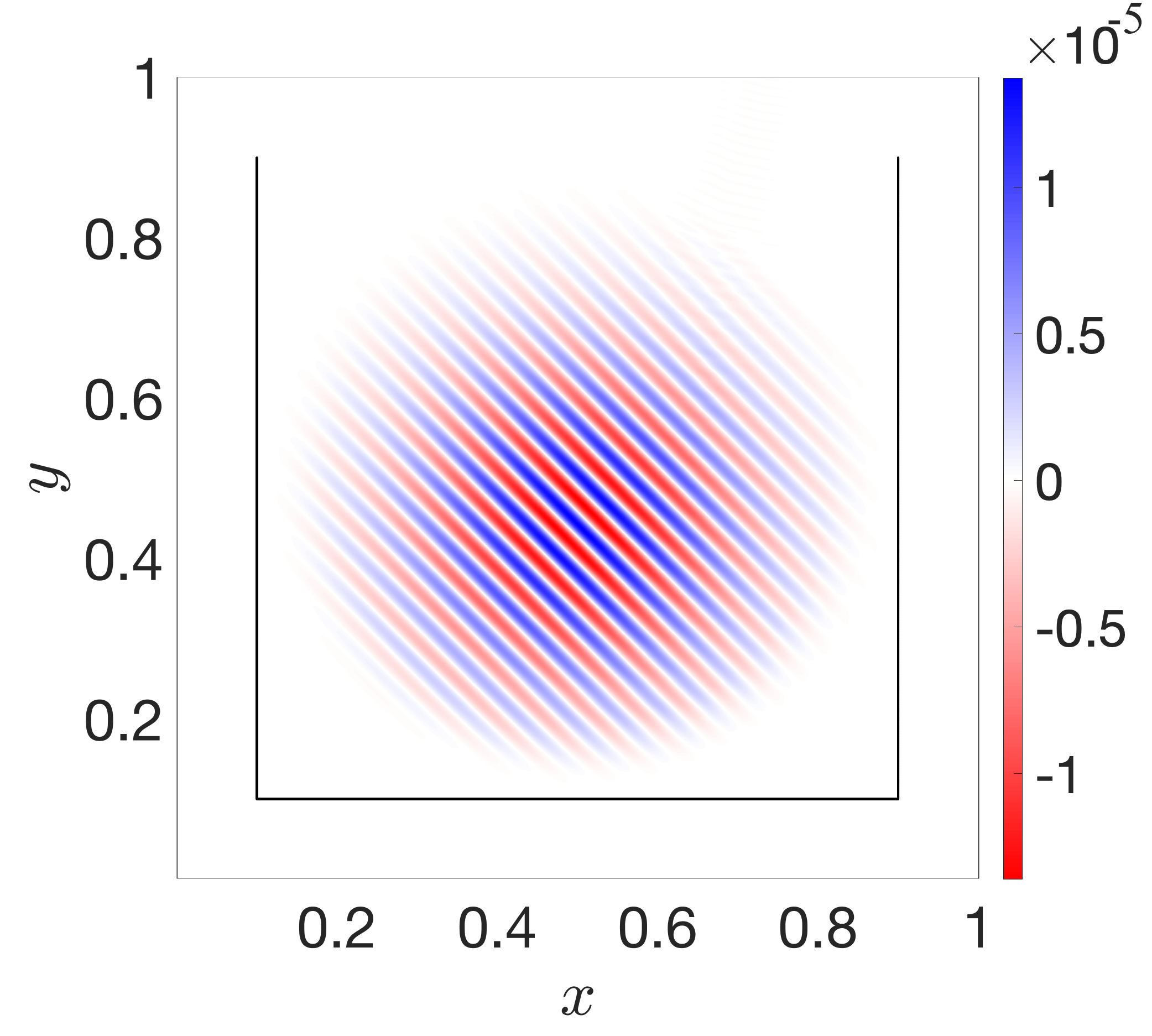}
\end{subfigure}
\begin{subfigure}[t]{.29\textwidth}
	\centering
	\includegraphics[width=\linewidth]{\fpath/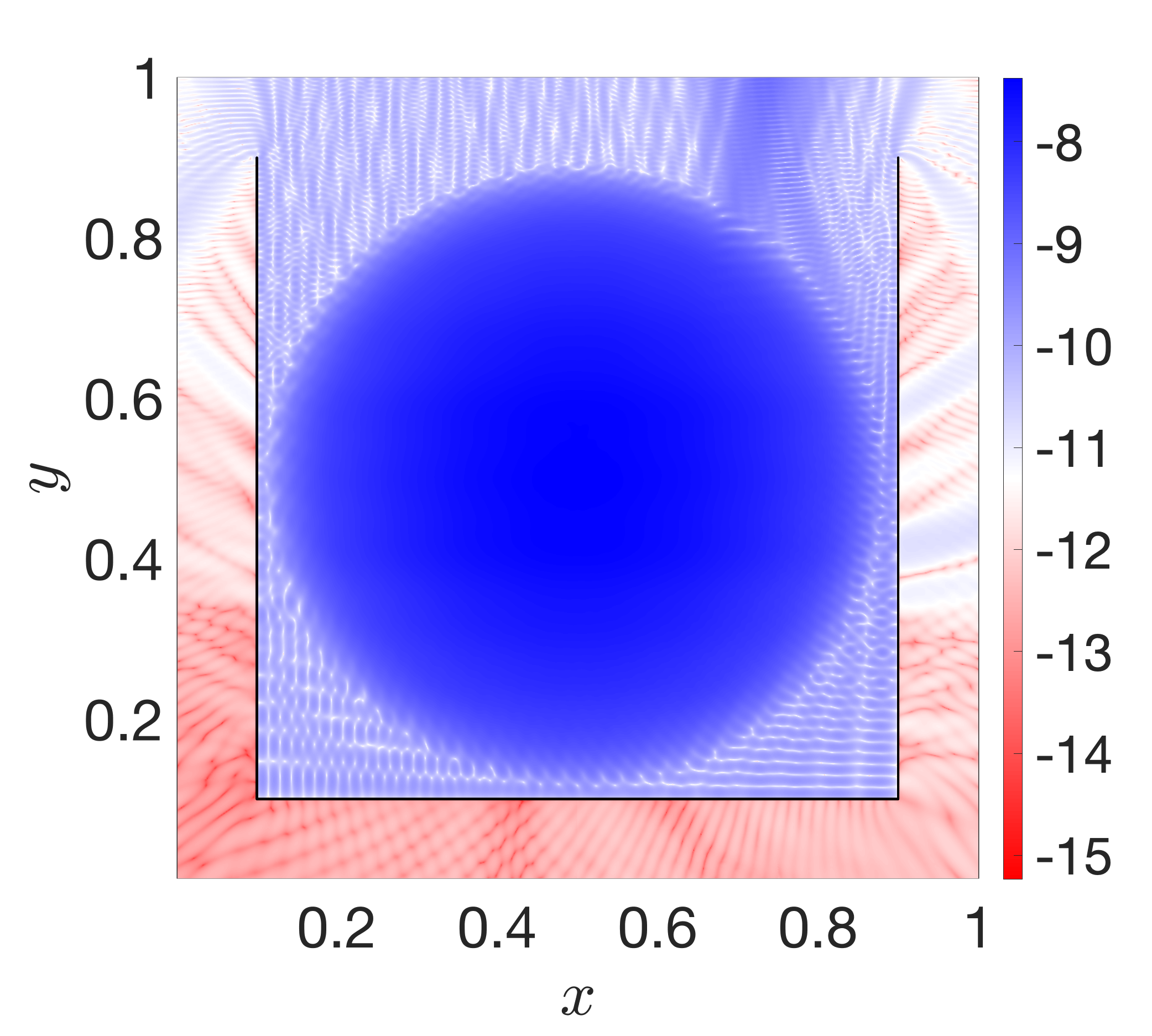}
\end{subfigure}	
\begin{subfigure}[t]{.37\textwidth}
	\centering
	\includegraphics[width=\linewidth]{\fpath/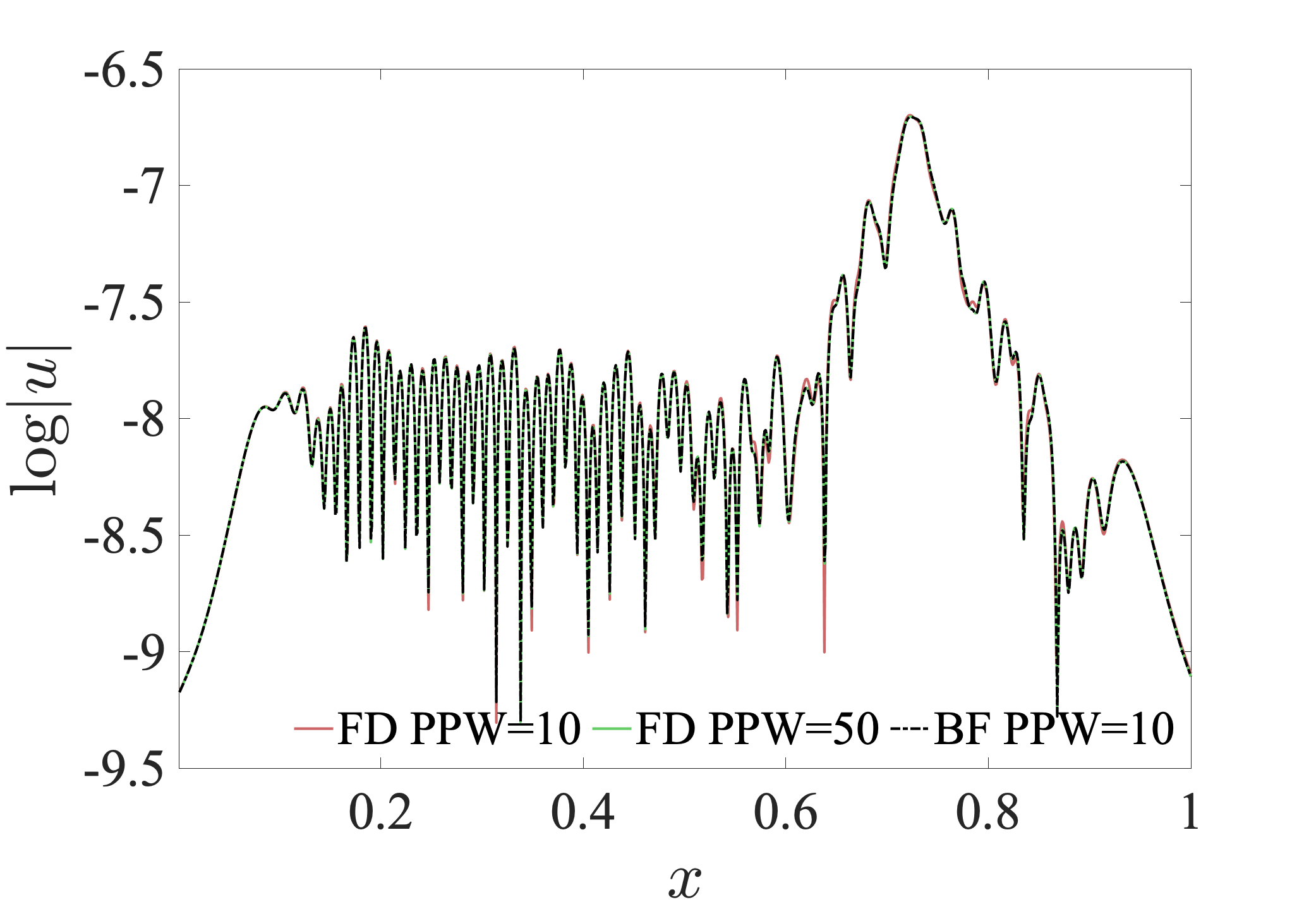}
\end{subfigure}		
	
	\vspace{-5pt}
	\caption{Constant-gradient media. $\omega=50\pi$ (100 wavelengths each direction). Left column: the field $\mathrm{Re}(u_{\rm hb})$ (in linear scale) computed by the proposed scheme. Middle column: difference $|u_{\rm hb}-u_{\rm fd}|$ (in log scale) between the fields computed by the proposed scheme (PPW=10) and FDFD (PPW=50). Right column: the fields $|u_{\rm hb}|,|u_{\rm fd}|$ (in log scale) drawn along the line $y=y_{\rm post}$. Row 1: point source with $y_{\rm post}=1-10h$ and $h$ corresponding to PPW=10. Row 2: kite-shaped source with $y_{\rm post}=1-10h$. Row 3: Gaussian packet source with $y_{\rm post}=0.5-10h$ and a semi-circle inclusion. Row 4: Gaussian packet source with $y_{\rm post}=1-10h$ and an open-square inclusion.}
	\label{fig:ex2_f25}
\end{figure}

\begin{figure}[!htp]
	\centering
	\vspace{-7.5pt}
	\begin{subfigure}[t]{.29\textwidth}
		\centering
		\includegraphics[width=\linewidth]{\fpath/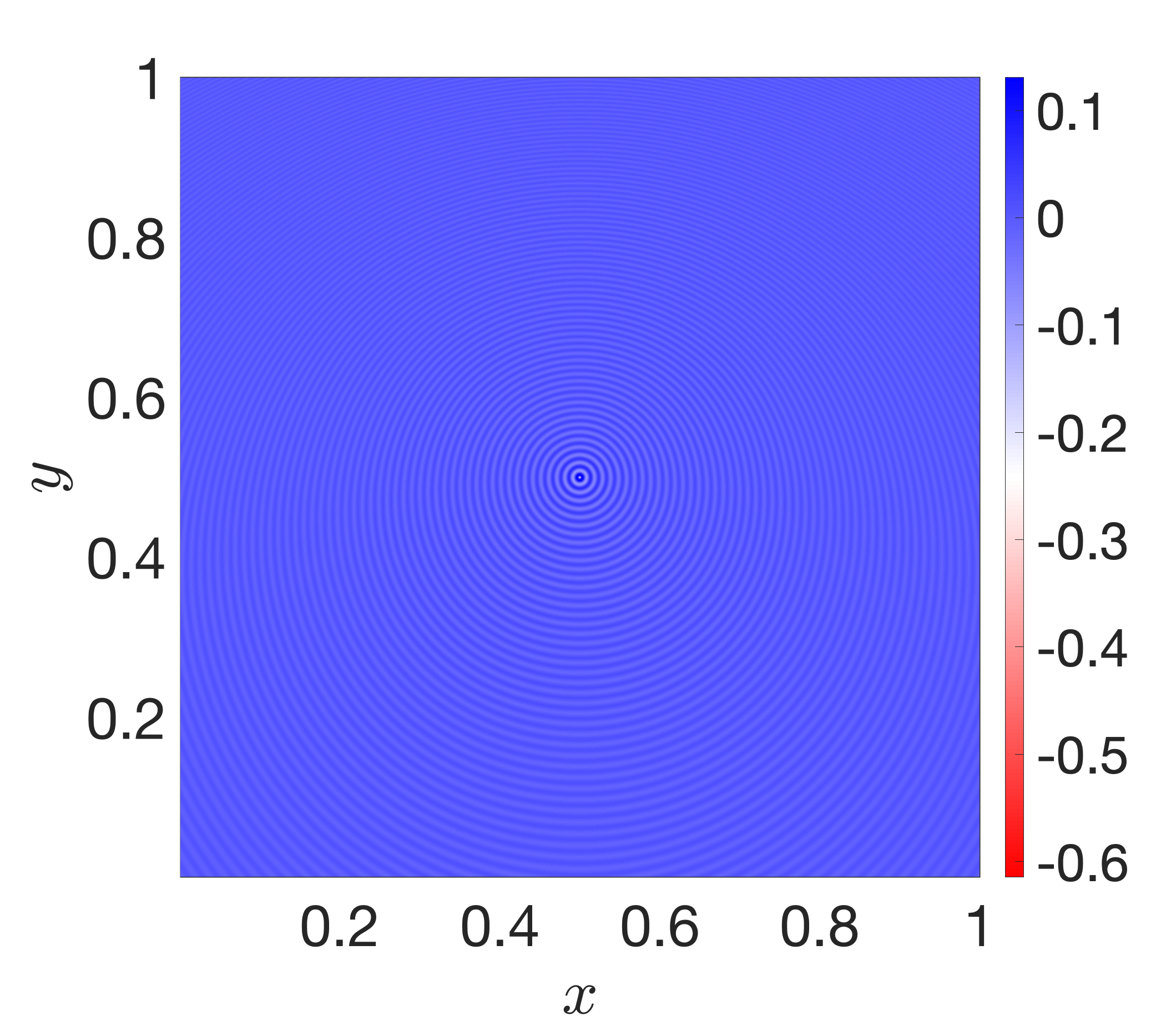}
	\end{subfigure}
	\begin{subfigure}[t]{.29\textwidth}
		\centering
		\includegraphics[width=\linewidth]{\fpath/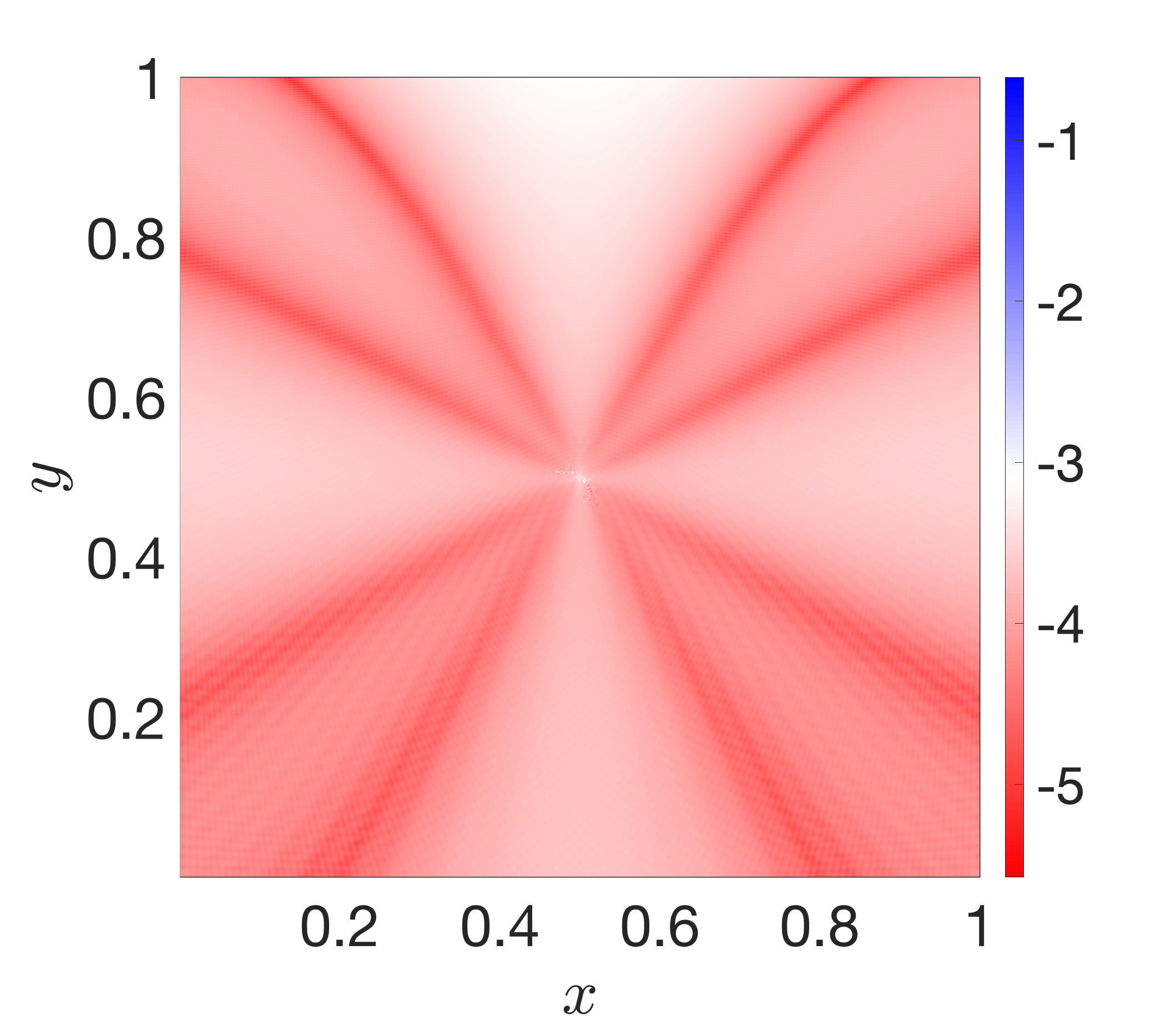}
	\end{subfigure}	
	\begin{subfigure}[t]{.37\textwidth}
		\centering
		\includegraphics[width=\linewidth]{\fpath/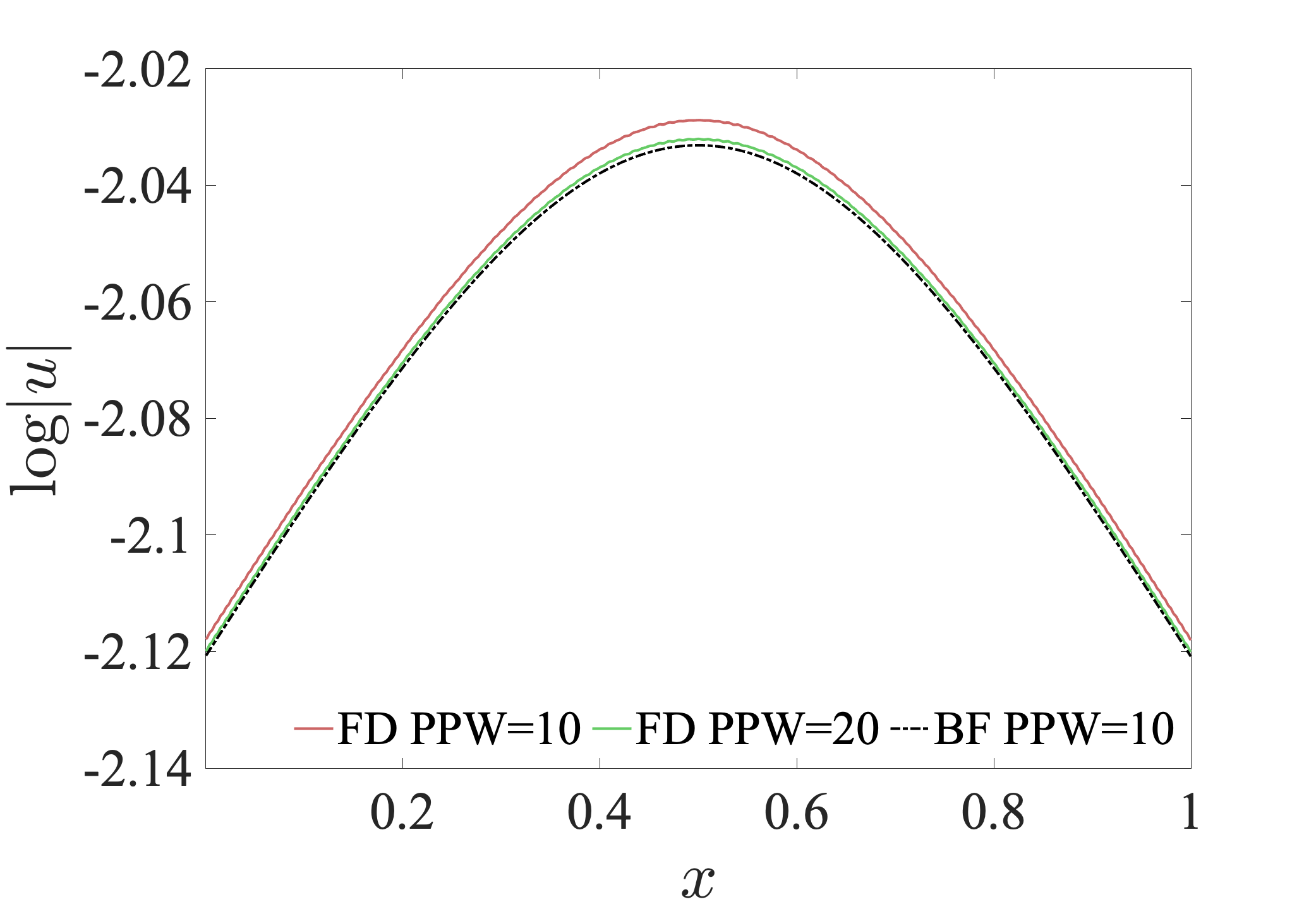}
	\end{subfigure}	
	\begin{subfigure}[t]{.29\textwidth}
		\centering
		\includegraphics[width=\linewidth]{\fpath/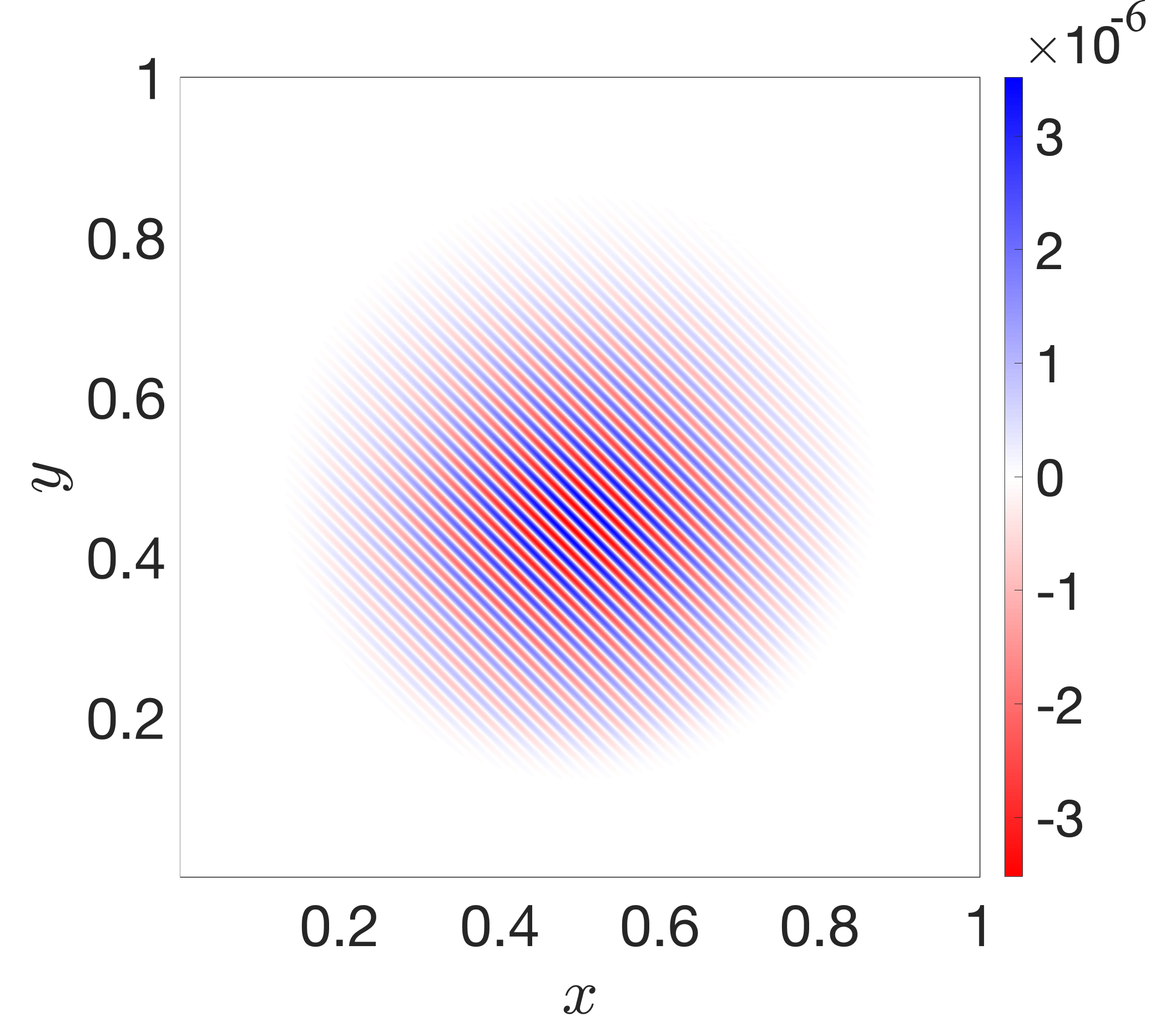}
	\end{subfigure}
	\begin{subfigure}[t]{.29\textwidth}
		\centering
		\includegraphics[width=\linewidth]{\fpath/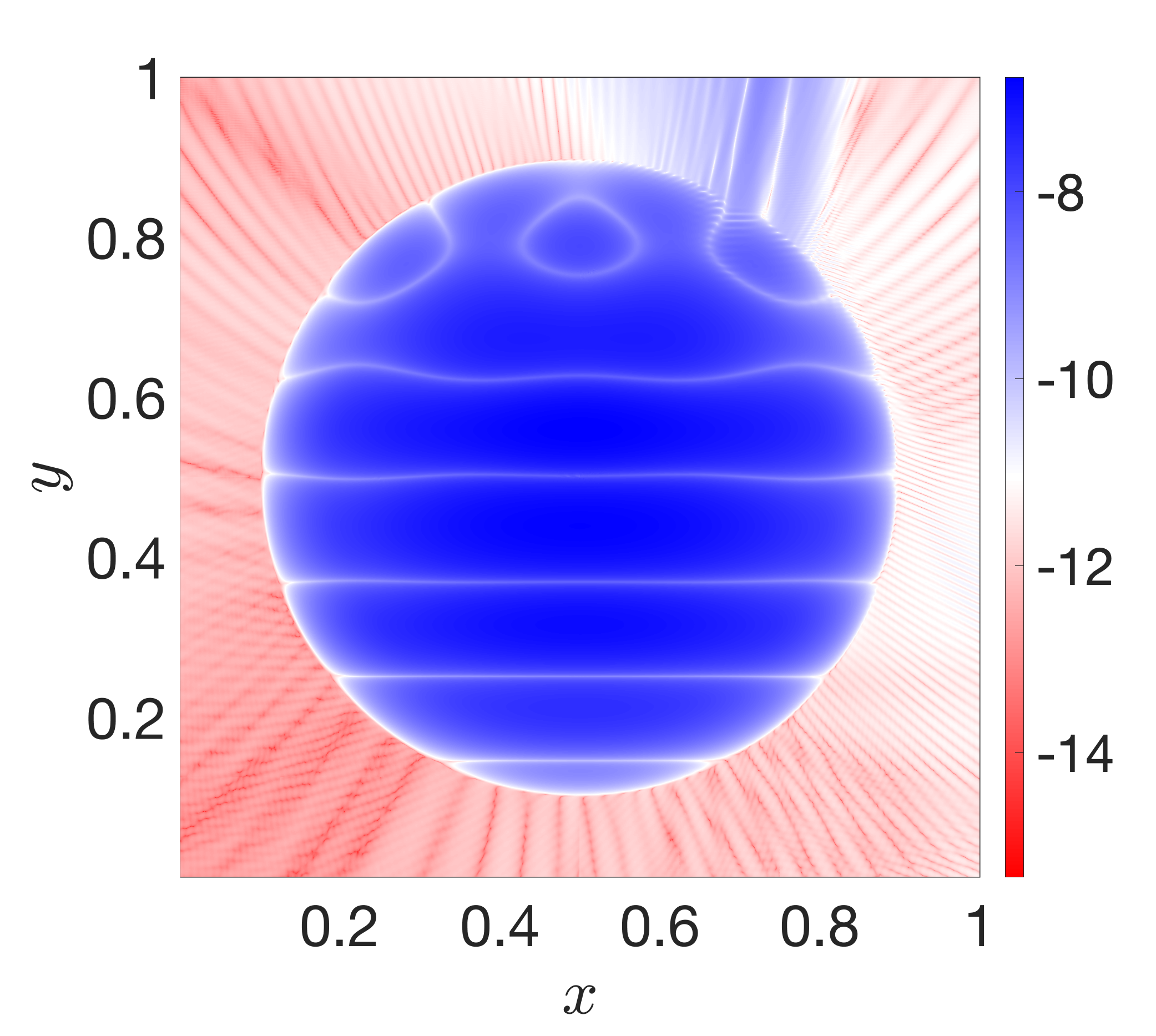}
	\end{subfigure}	
	\begin{subfigure}[t]{.37\textwidth}
		\centering
		\includegraphics[width=\linewidth]{\fpath/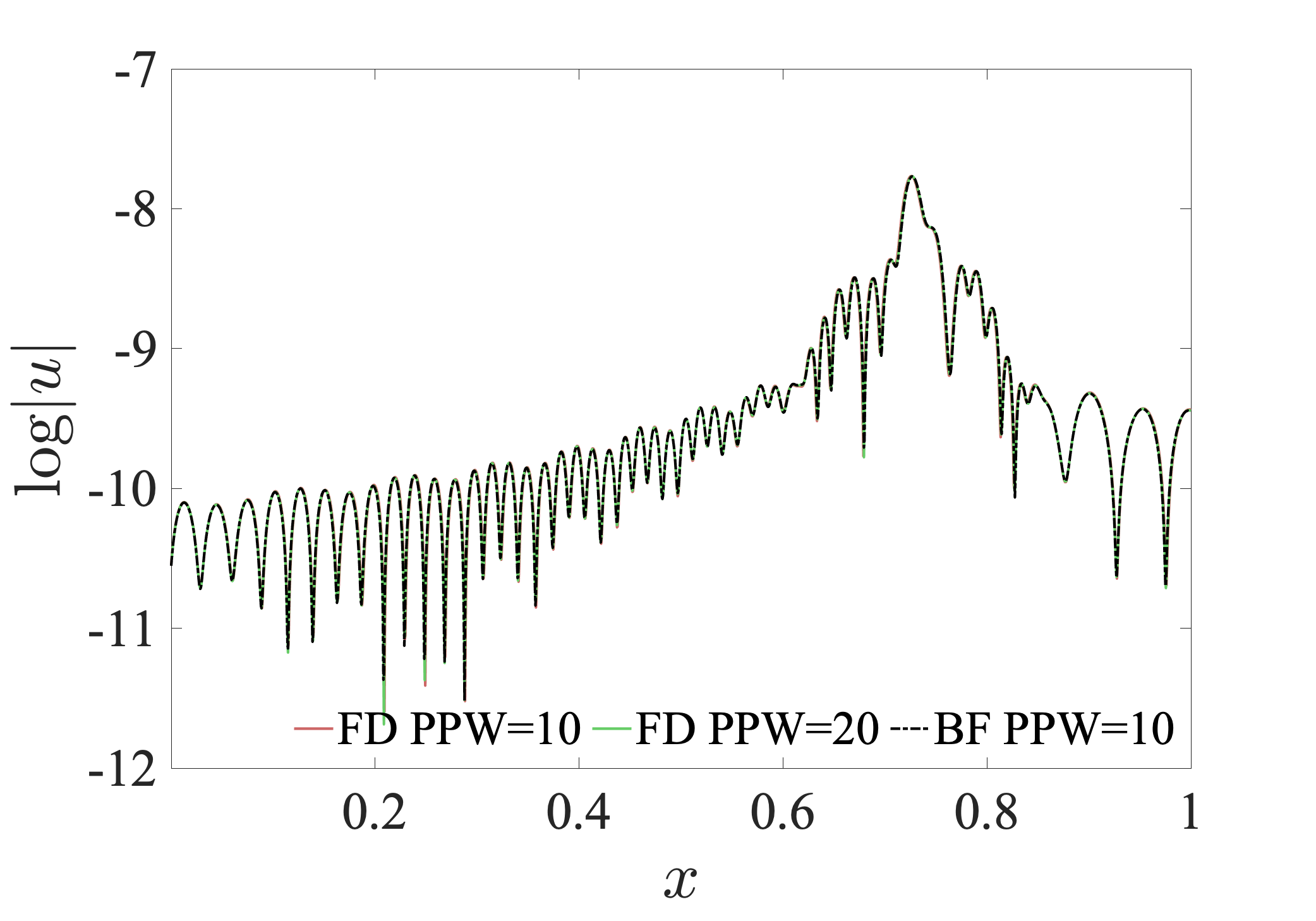}
	\end{subfigure}	
	
	\begin{subfigure}[t]{.29\textwidth}
		\centering
		\includegraphics[width=\linewidth]{\fpath/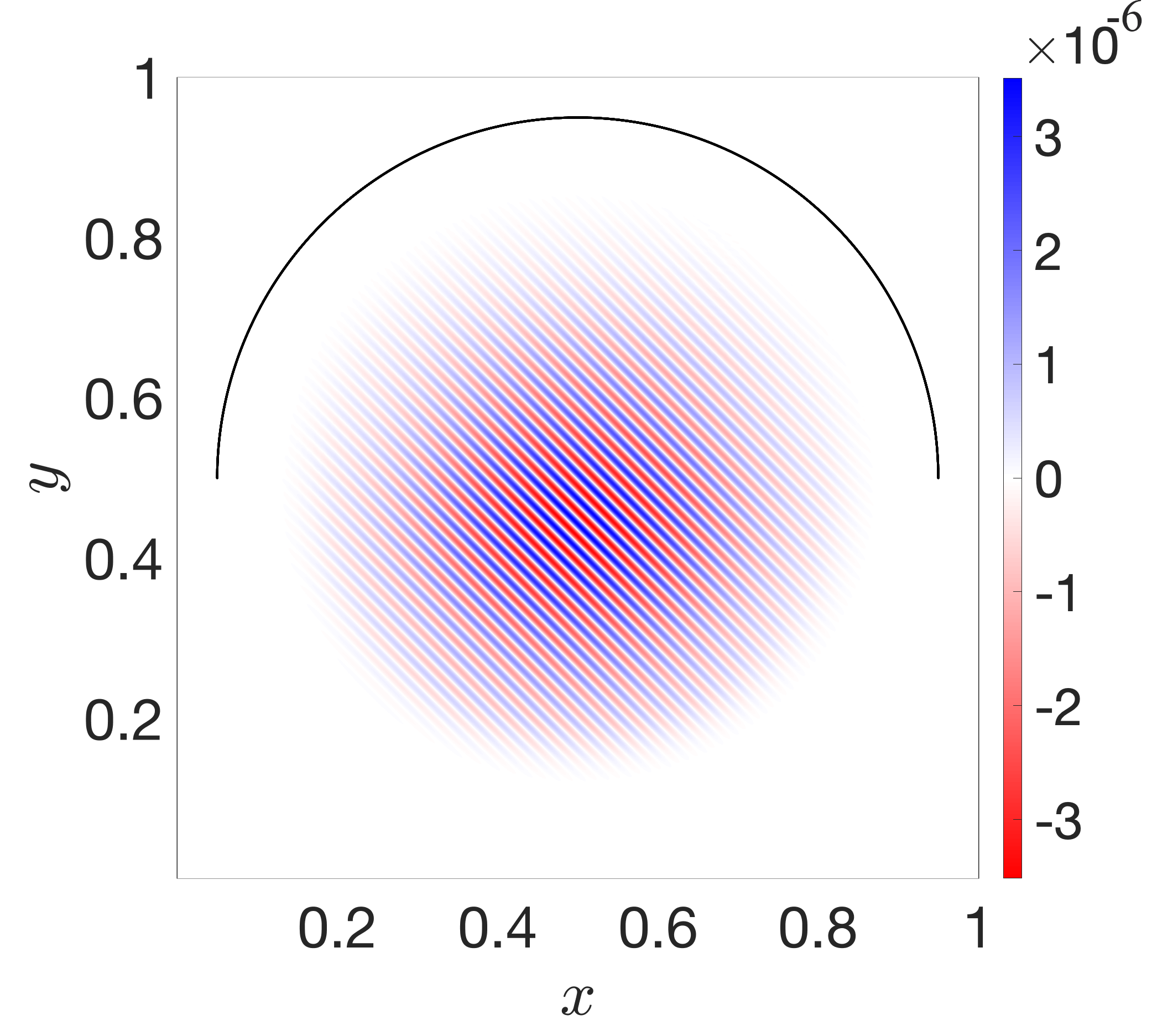}
	\end{subfigure}
	\begin{subfigure}[t]{.29\textwidth}
		\centering
		\includegraphics[width=\linewidth]{\fpath/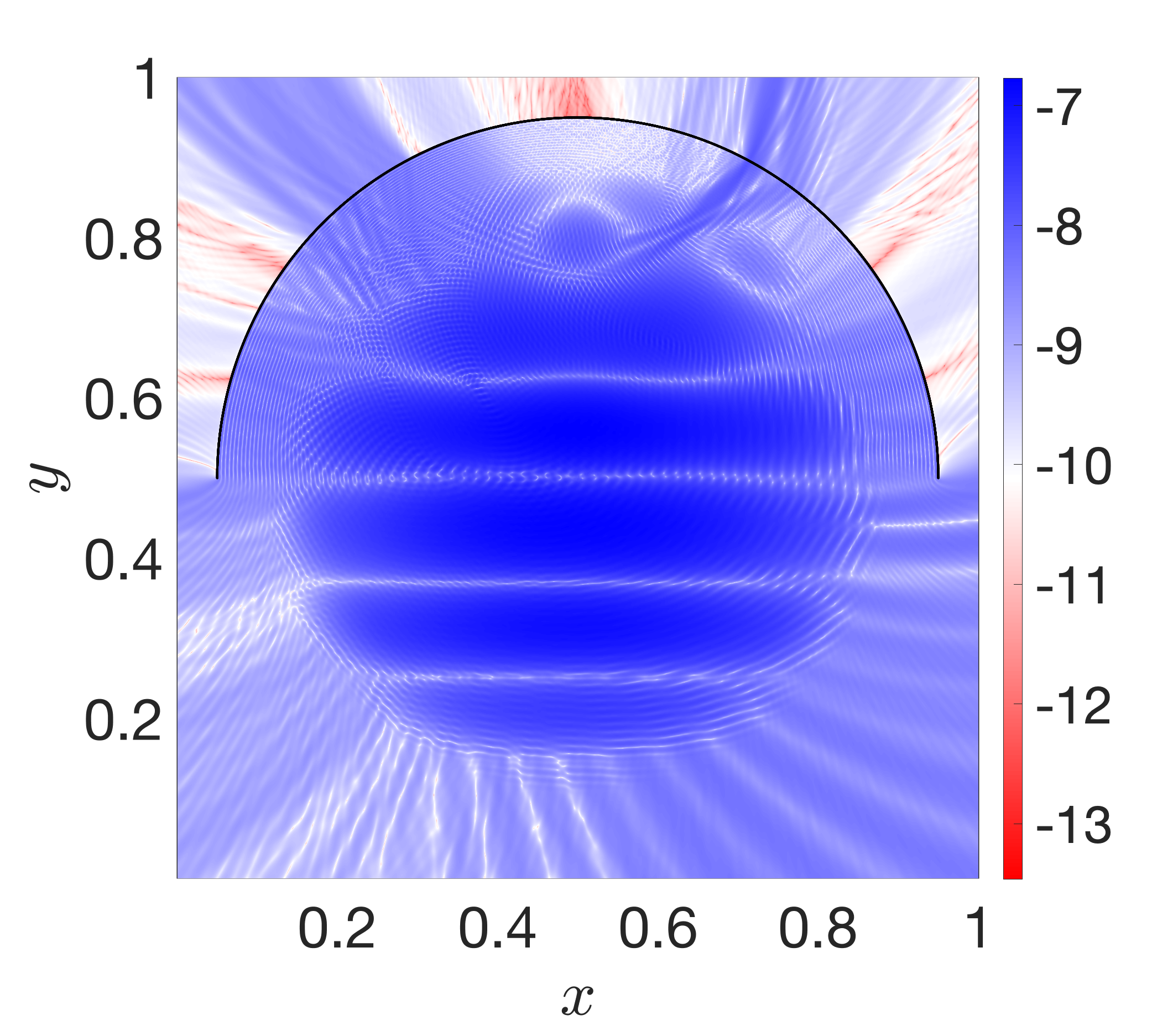}
	\end{subfigure}	
	\begin{subfigure}[t]{.37\textwidth}
		\centering
		\includegraphics[width=\linewidth]{\fpath/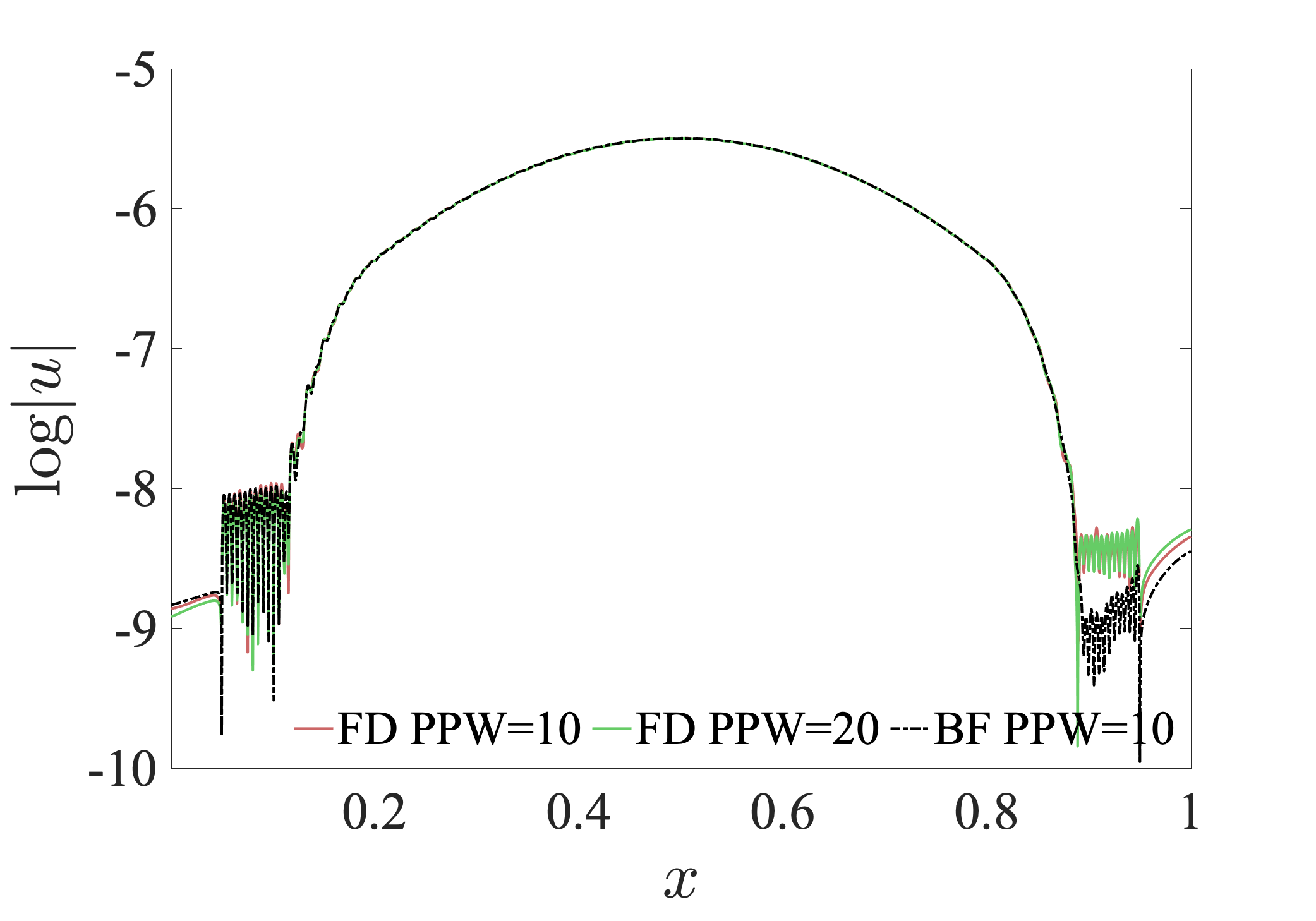}
	\end{subfigure}		
	
	\begin{subfigure}[t]{.29\textwidth}
		\centering
		\includegraphics[width=\linewidth]{\fpath/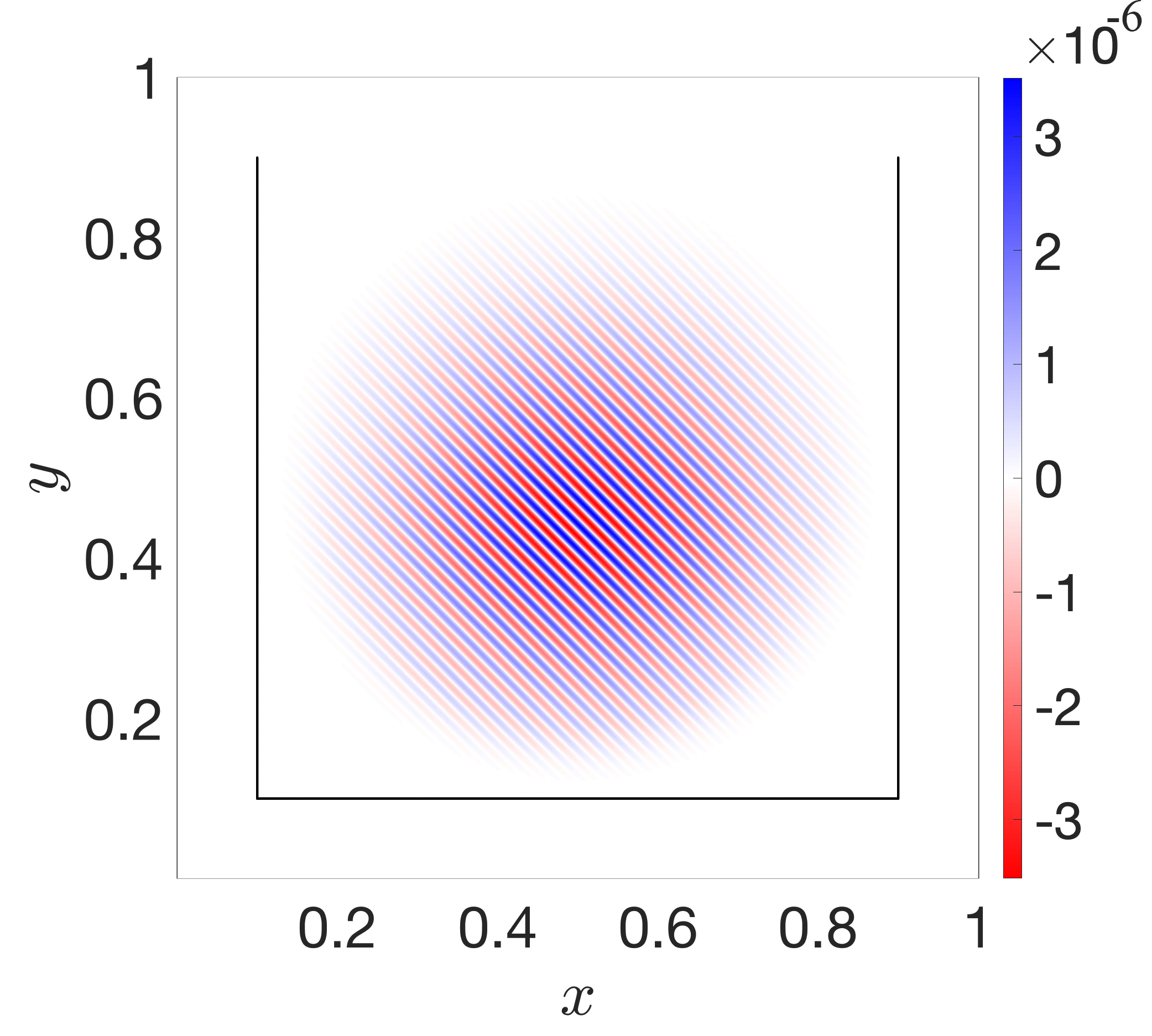}
	\end{subfigure}
	\begin{subfigure}[t]{.29\textwidth}
		\centering
		\includegraphics[width=\linewidth]{\fpath/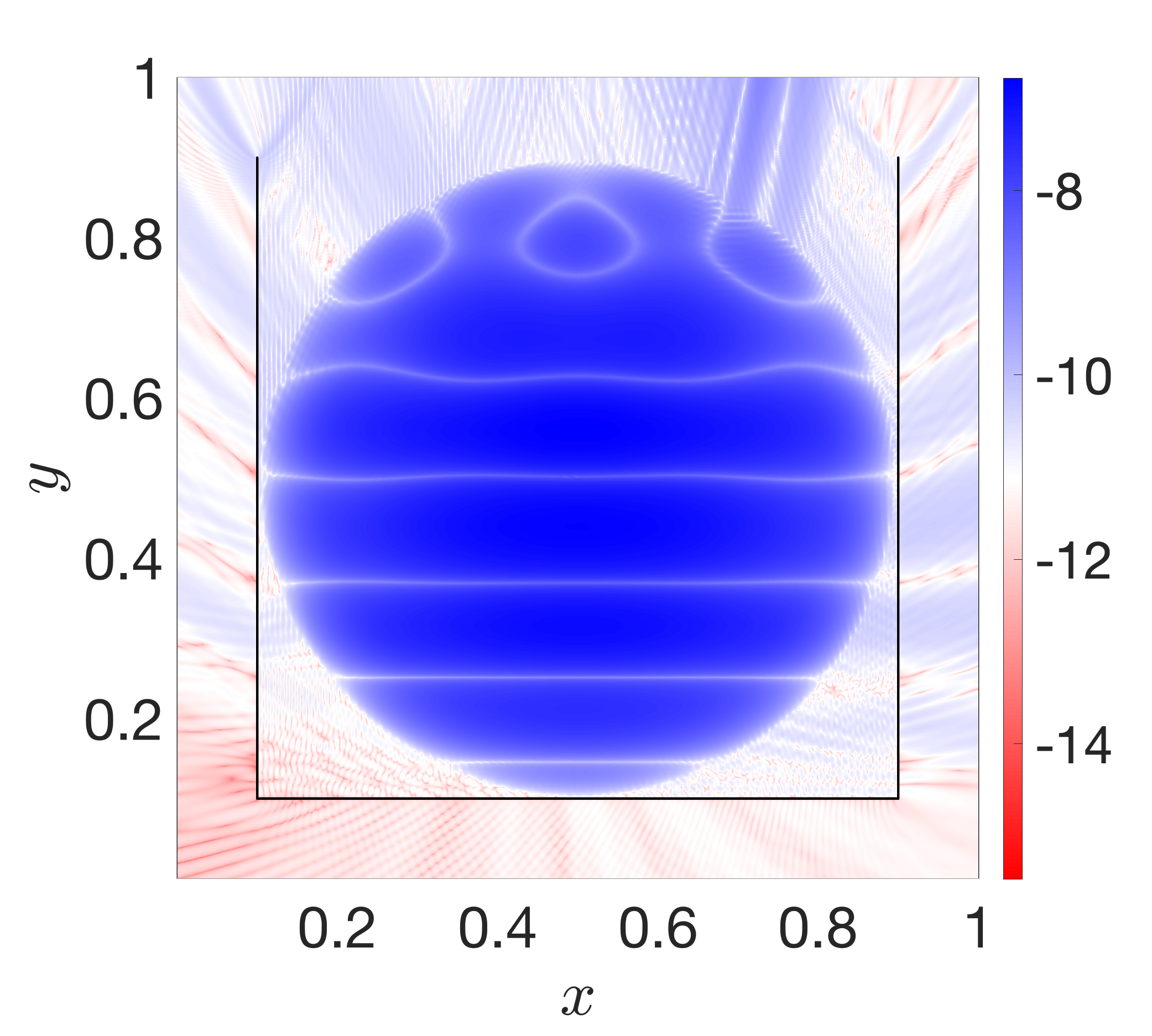}
	\end{subfigure}	
	\begin{subfigure}[t]{.37\textwidth}
		\centering
		\includegraphics[width=\linewidth]{\fpath/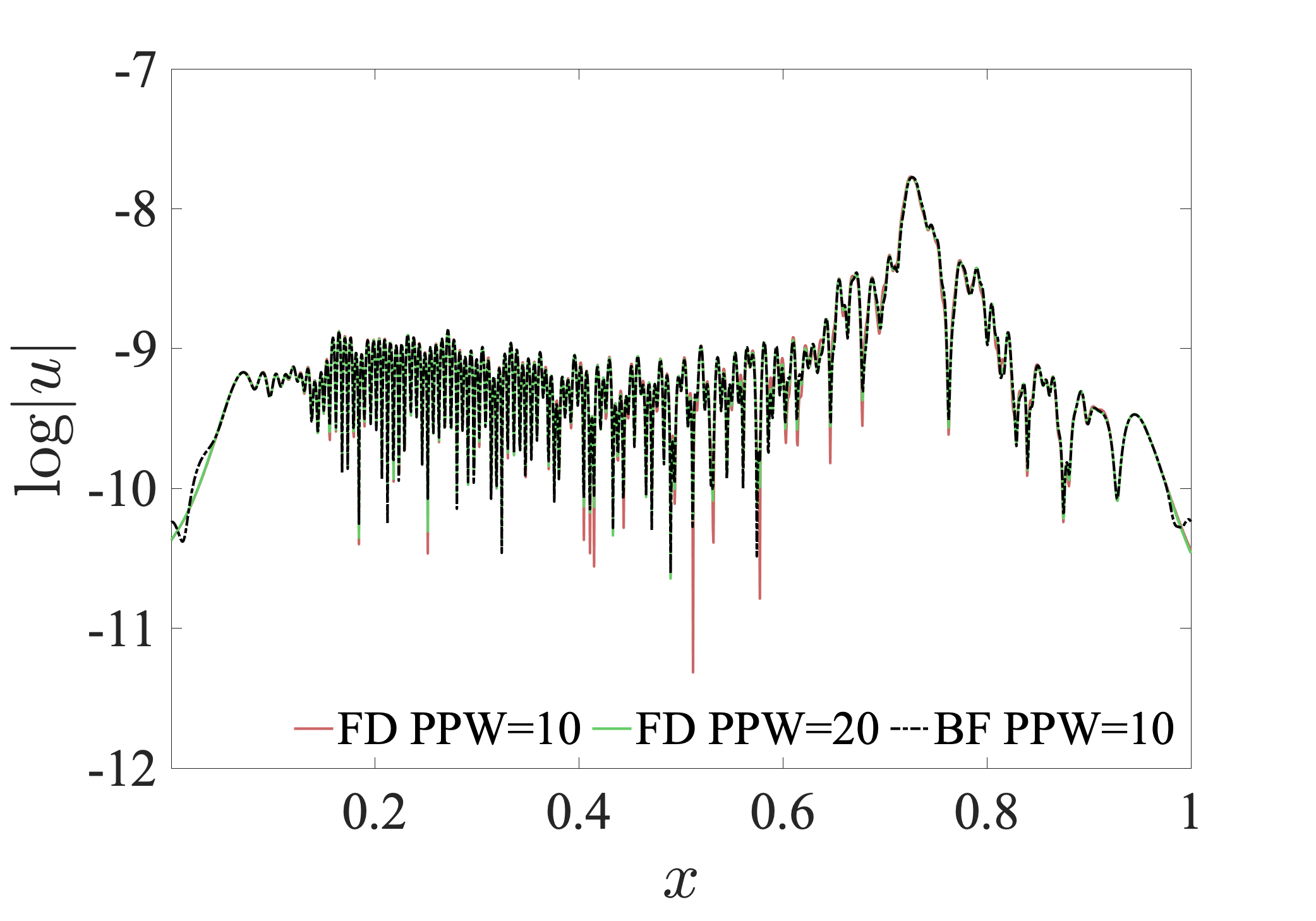}
	\end{subfigure}		
	
	\vspace{-5pt}
	\caption{Constant-gradient media. $\omega=100\pi$ (200 wavelengths each direction). Left column: the field $\mathrm{Re}(u_{\rm hb})$ (in linear scale) computed by the proposed scheme. Middle column: difference $|u_{\rm hb}-u_{\rm fd}|$ (in log scale) between the fields computed by the proposed scheme (PPW=10) and FDFD (PPW=20). Right column: the field $|u_{\rm hb}|,|u_{\rm fd}|$ (in log scale) drawn along the line $y=y_{\rm post}$. Row 1: point source with $y_{\rm post}=1-10h$ and $h$ corresponding to PPW=10. Row 2: Gaussian packet source with $y_{\rm post}=1-10h$. Row 3: Gaussian packet source with $y_{\rm post}=0.5-10h$ and a semi-circle inclusion. Row 4: Gaussian packet source with $y_{\rm post}=1-10h$ and an open-square inclusion.}	
	\label{fig:ex2_f50}
\end{figure}

\begin{figure}[!htp]
	\centering
	\vspace{-7.5pt}
	\begin{subfigure}[t]{.29\textwidth}
		\centering
		\includegraphics[width=\linewidth]{\fpath/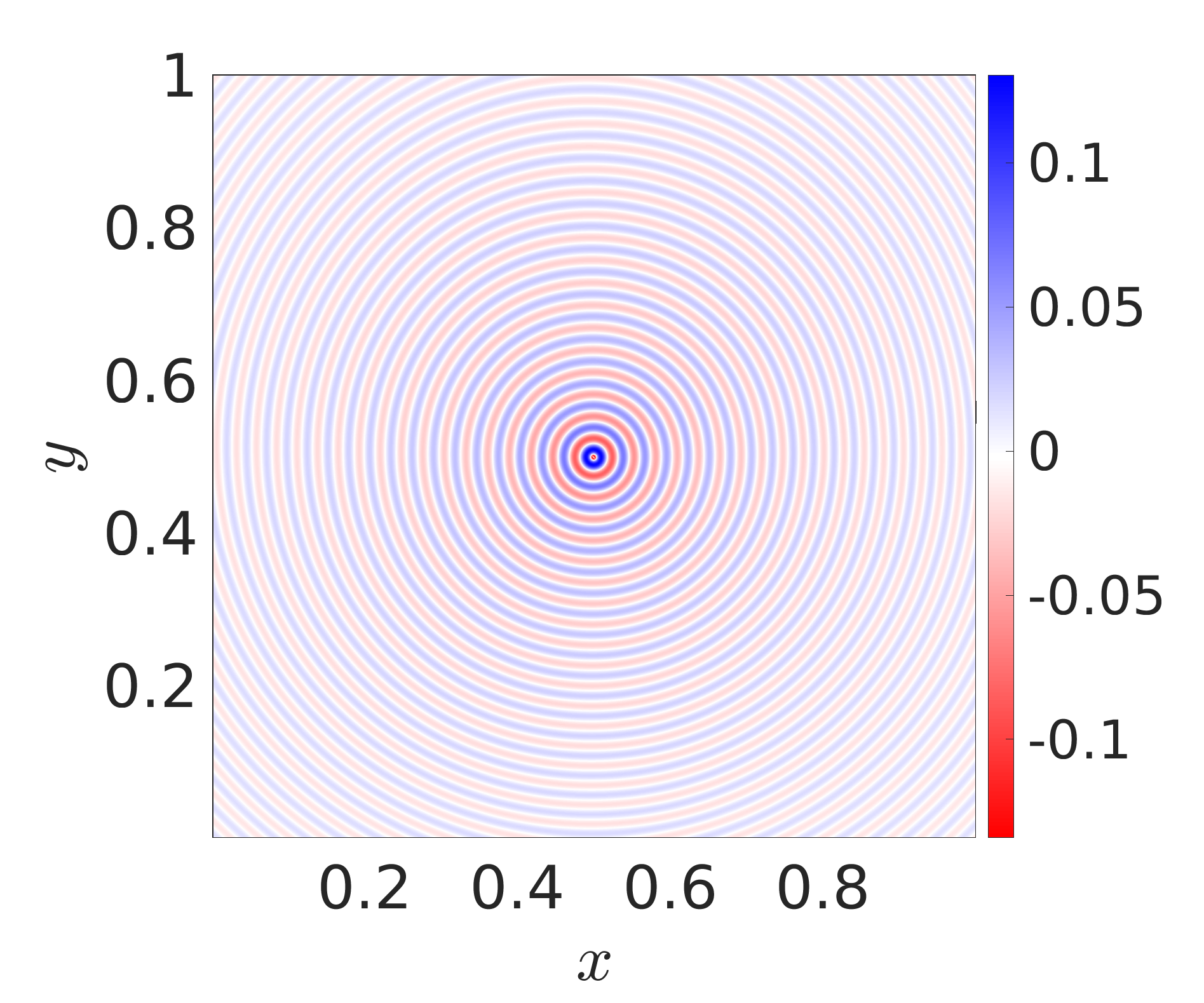}
	\end{subfigure}
	\begin{subfigure}[t]{.29\textwidth}
		\centering
		\includegraphics[width=\linewidth]{\fpath/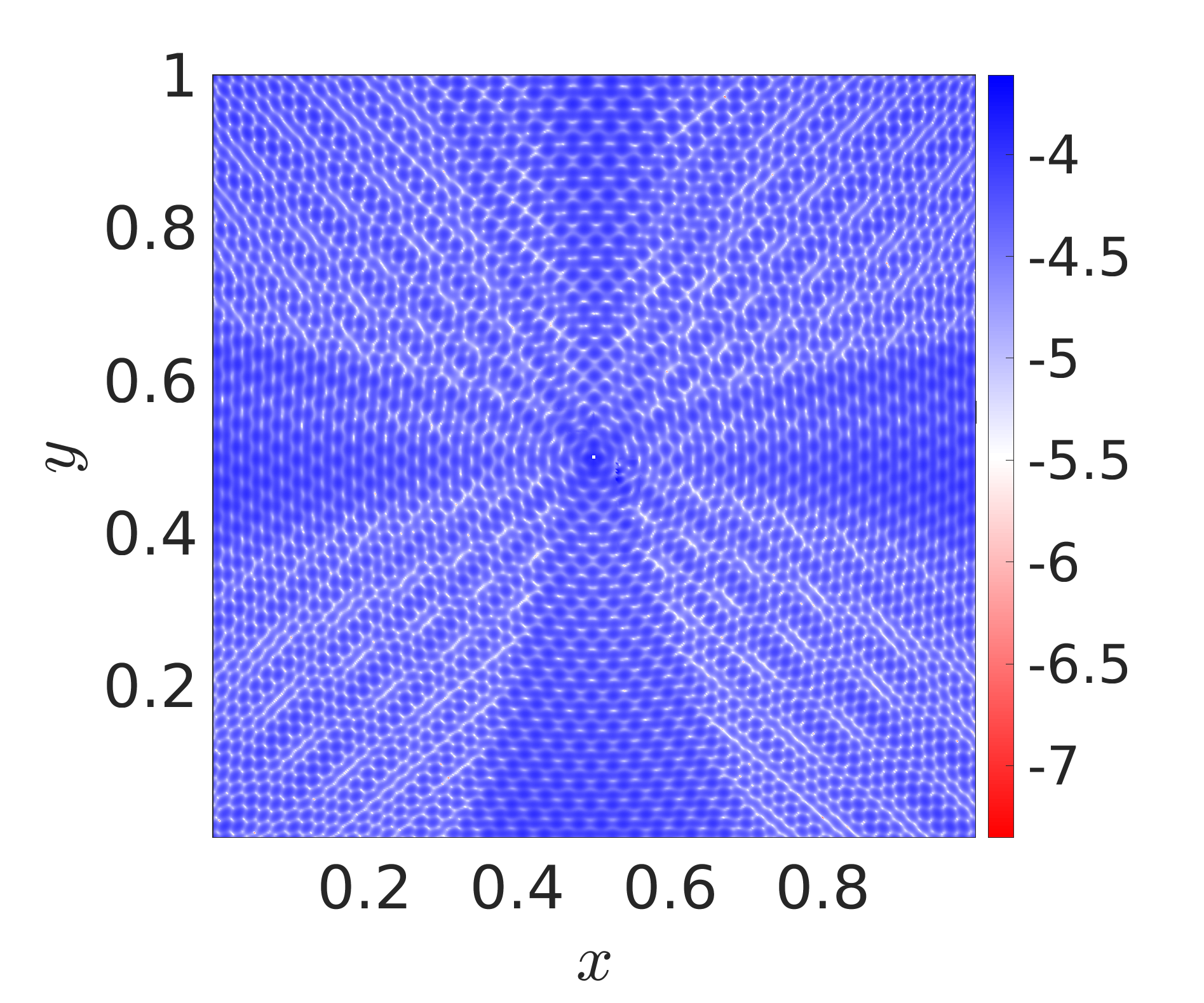}
	\end{subfigure}	
	\begin{subfigure}[t]{.37\textwidth}
		\centering
		\includegraphics[width=\linewidth]{\fpath/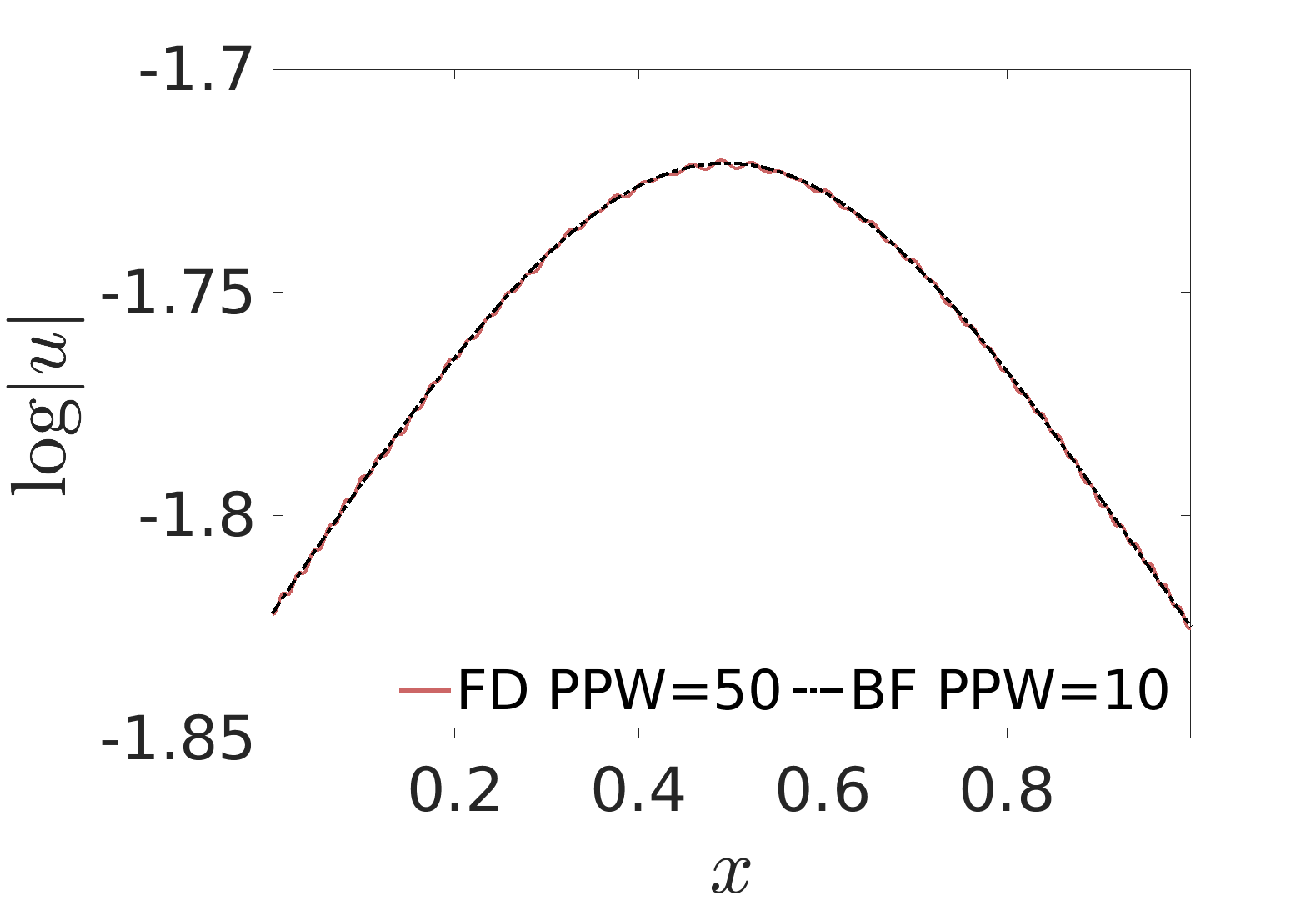}
	\end{subfigure}	
	\begin{subfigure}[t]{.29\textwidth}
		\centering
		\includegraphics[width=\linewidth]{\fpath/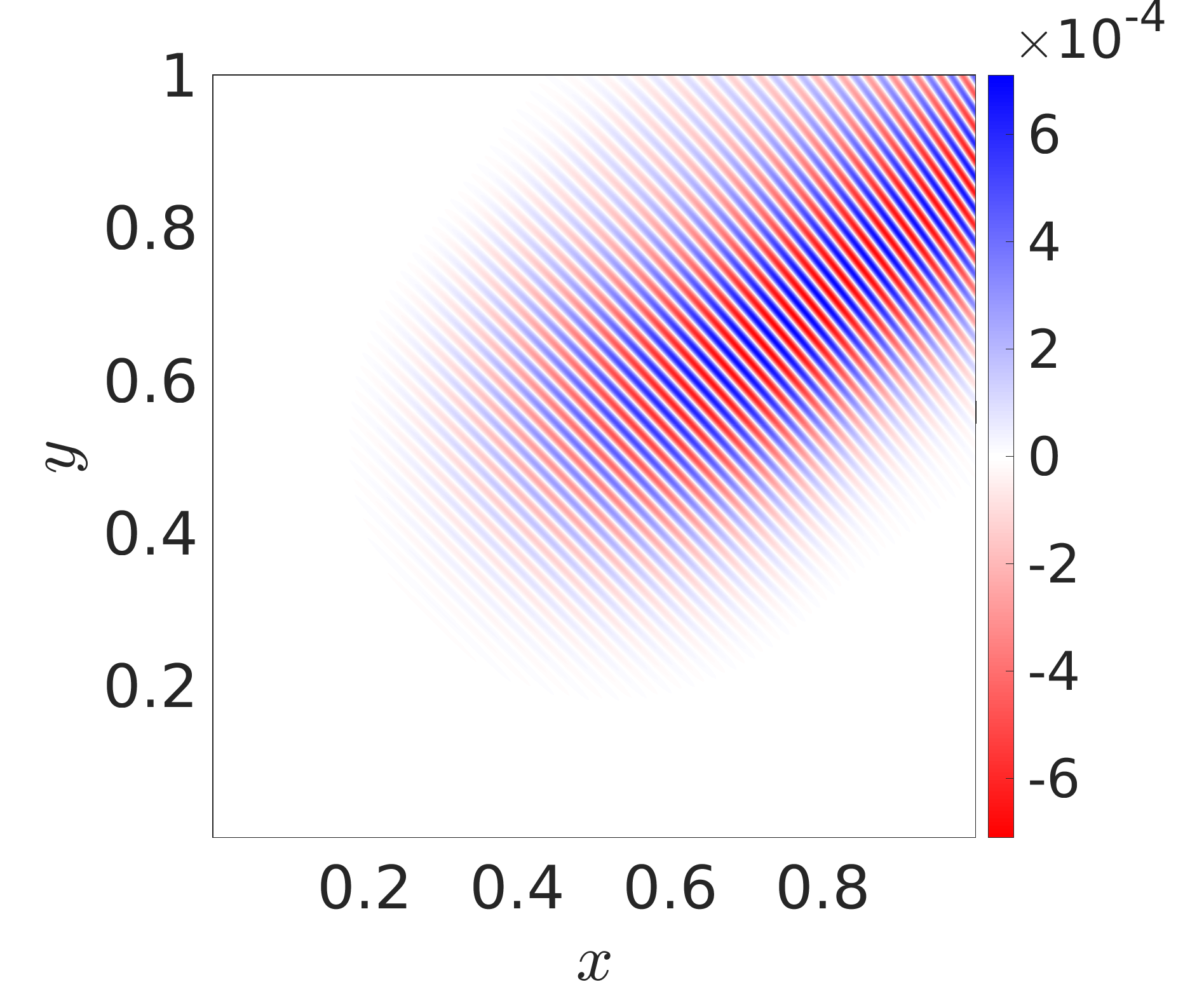}
	\end{subfigure}
	\begin{subfigure}[t]{.29\textwidth}
		\centering
		\includegraphics[width=\linewidth]{\fpath/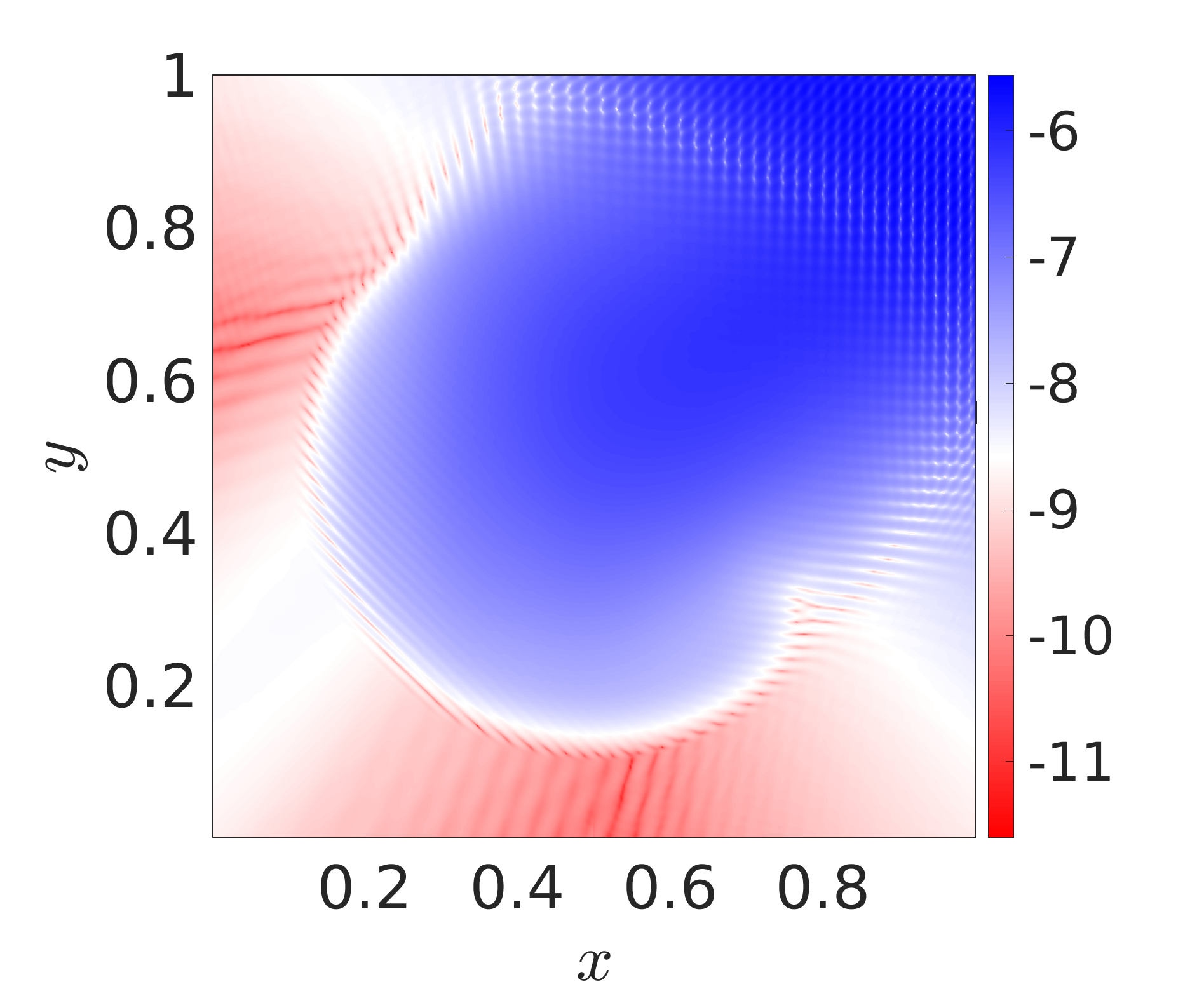}
	\end{subfigure}	
	\begin{subfigure}[t]{.37\textwidth}
		\centering
		\includegraphics[width=\linewidth]{\fpath/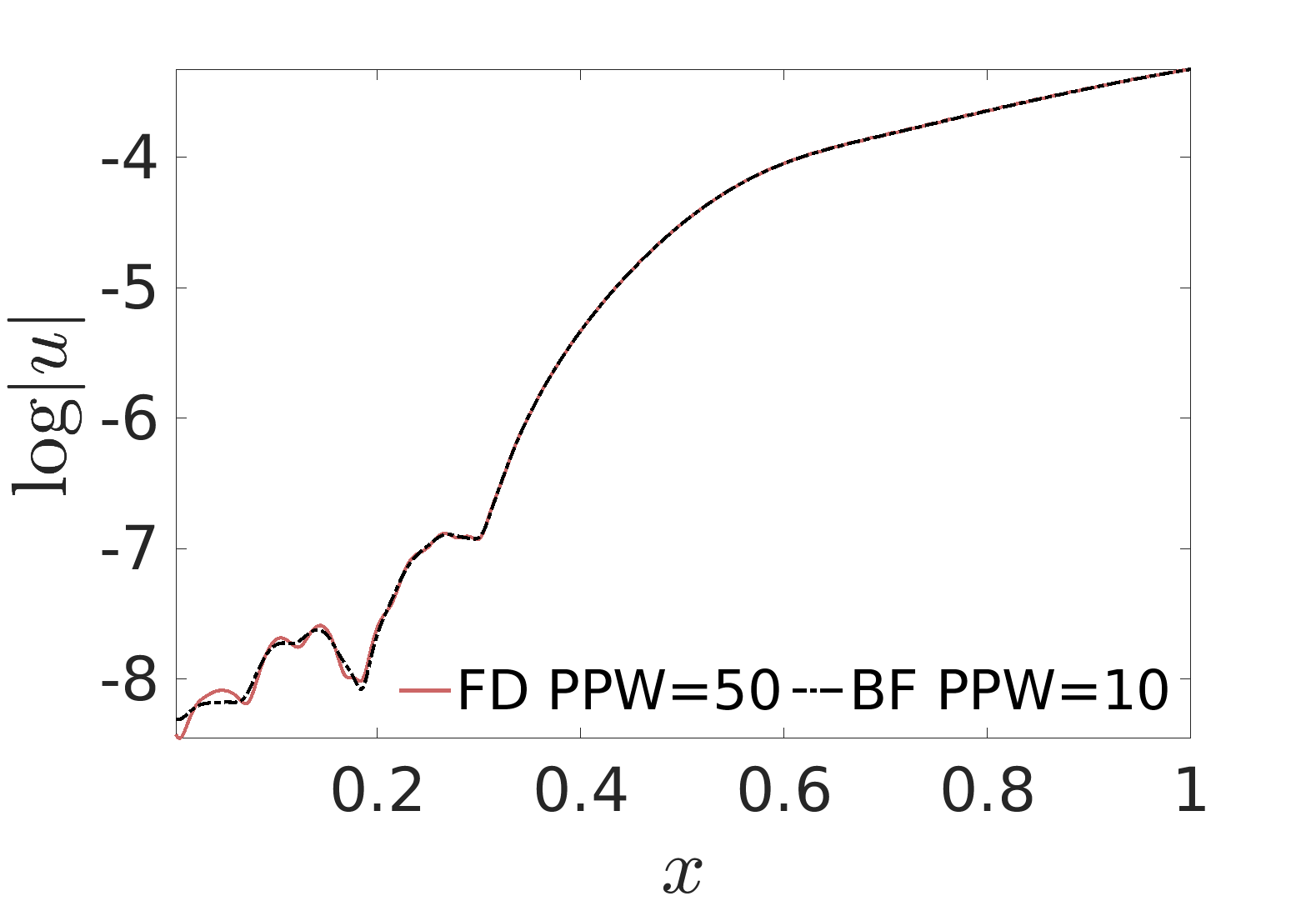}
	\end{subfigure}	
	
	\begin{subfigure}[t]{.29\textwidth}
		\centering
		\includegraphics[width=\linewidth]{\fpath/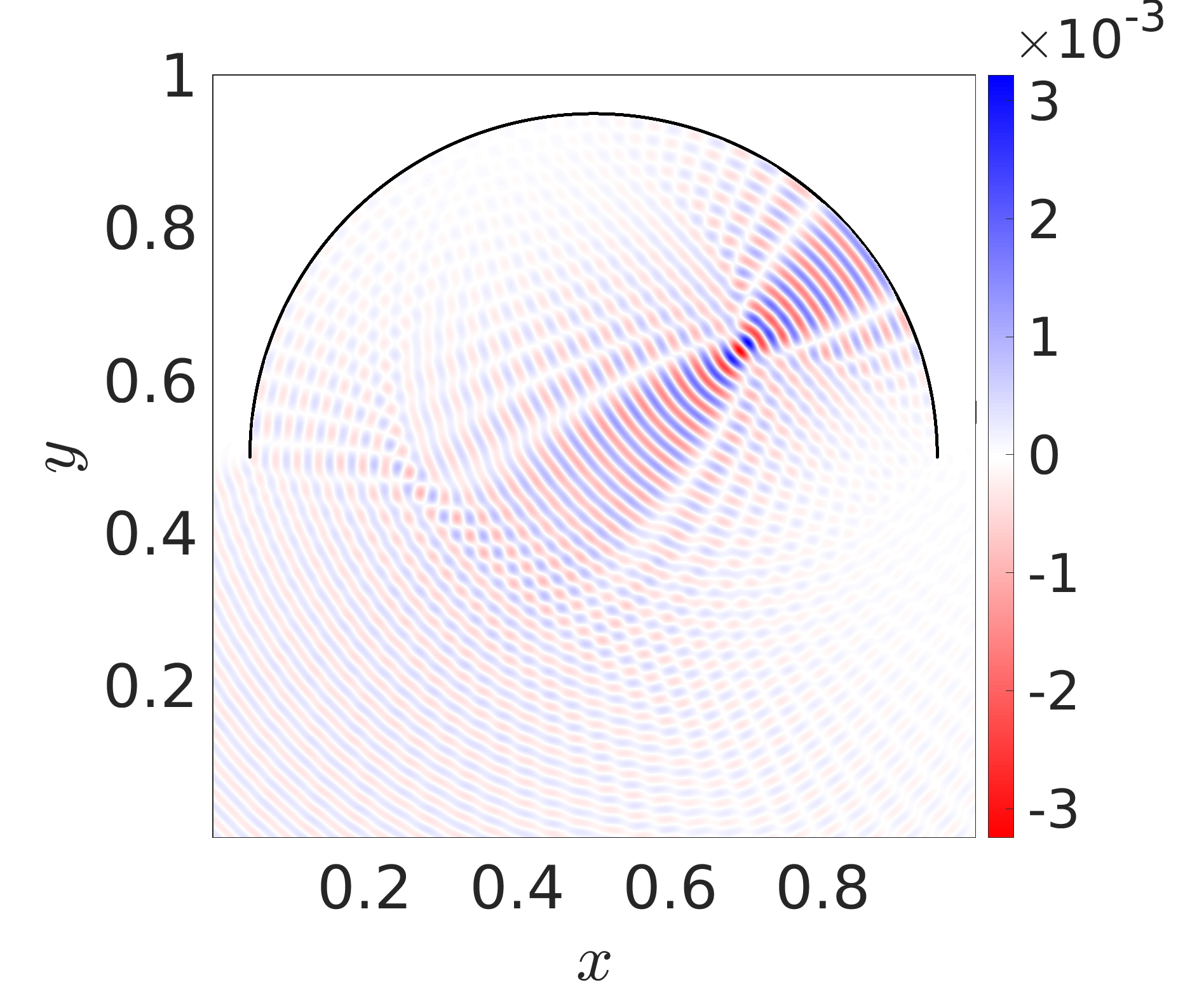}
	\end{subfigure}
	\begin{subfigure}[t]{.29\textwidth}
		\centering
		\includegraphics[width=\linewidth]{\fpath/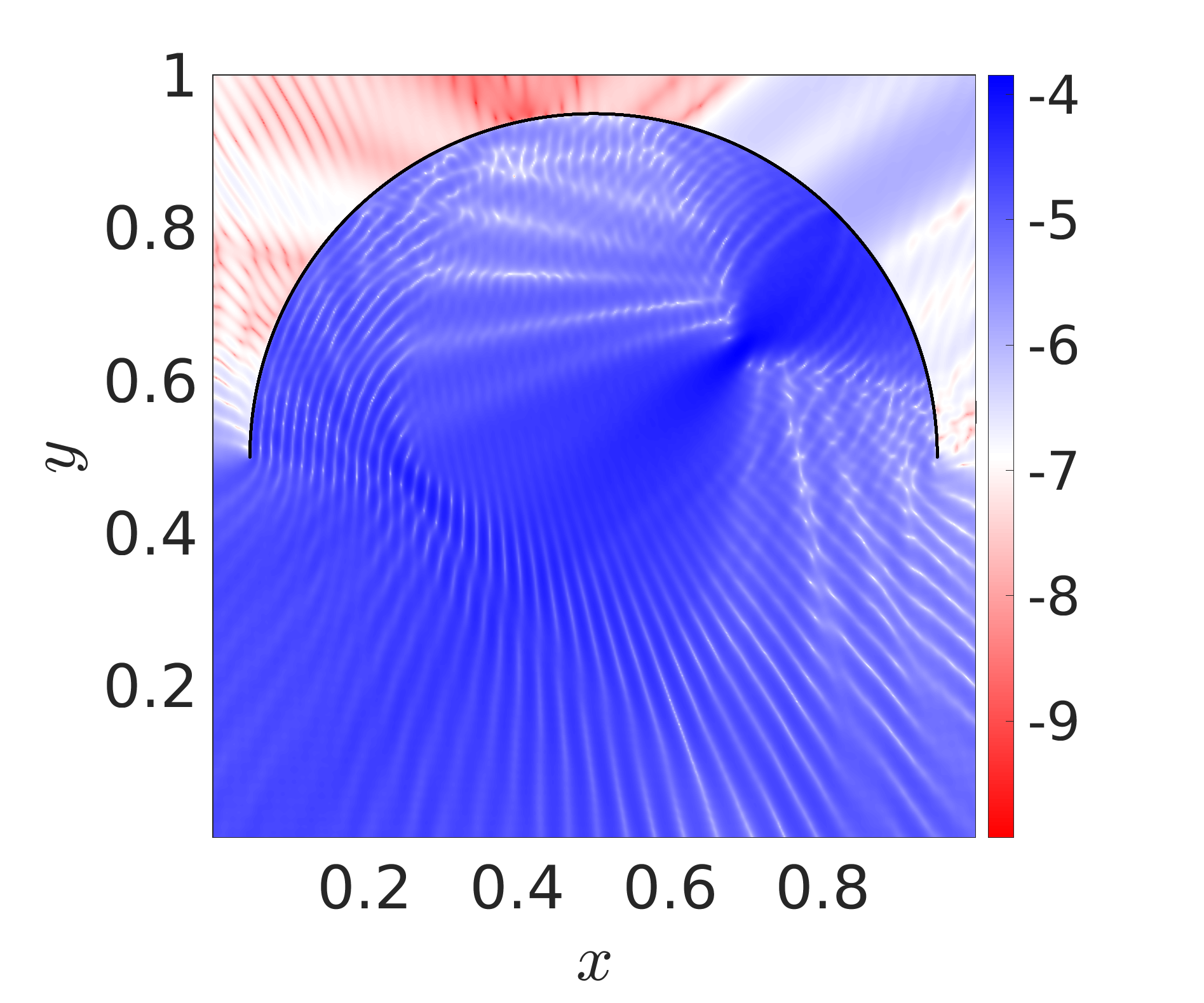}
	\end{subfigure}	
	\begin{subfigure}[t]{.37\textwidth}
		\centering
		\includegraphics[width=\linewidth]{\fpath/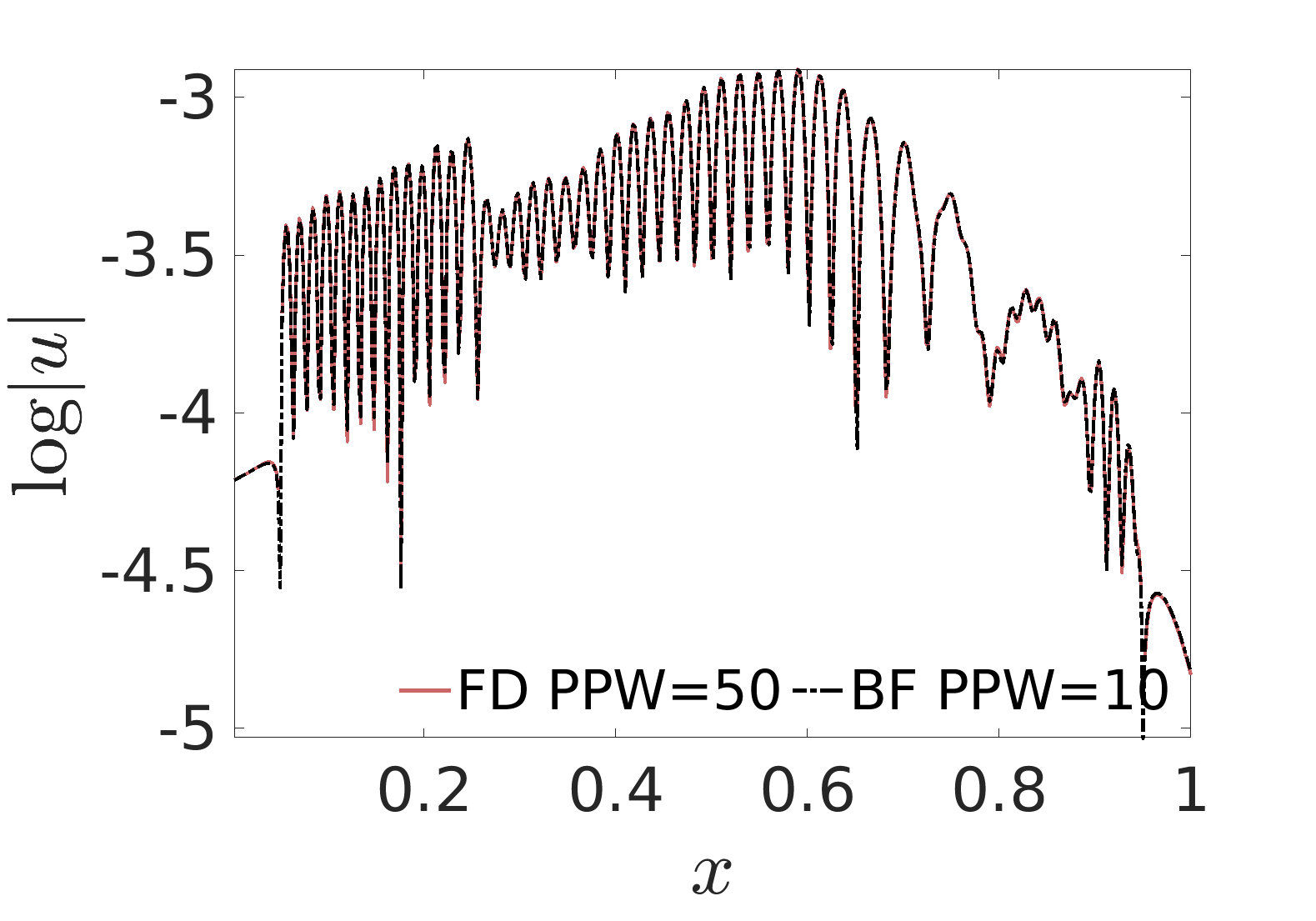}
	\end{subfigure}		
	
		\begin{subfigure}[t]{.29\textwidth}
		\centering
		\includegraphics[width=\linewidth]{\fpath/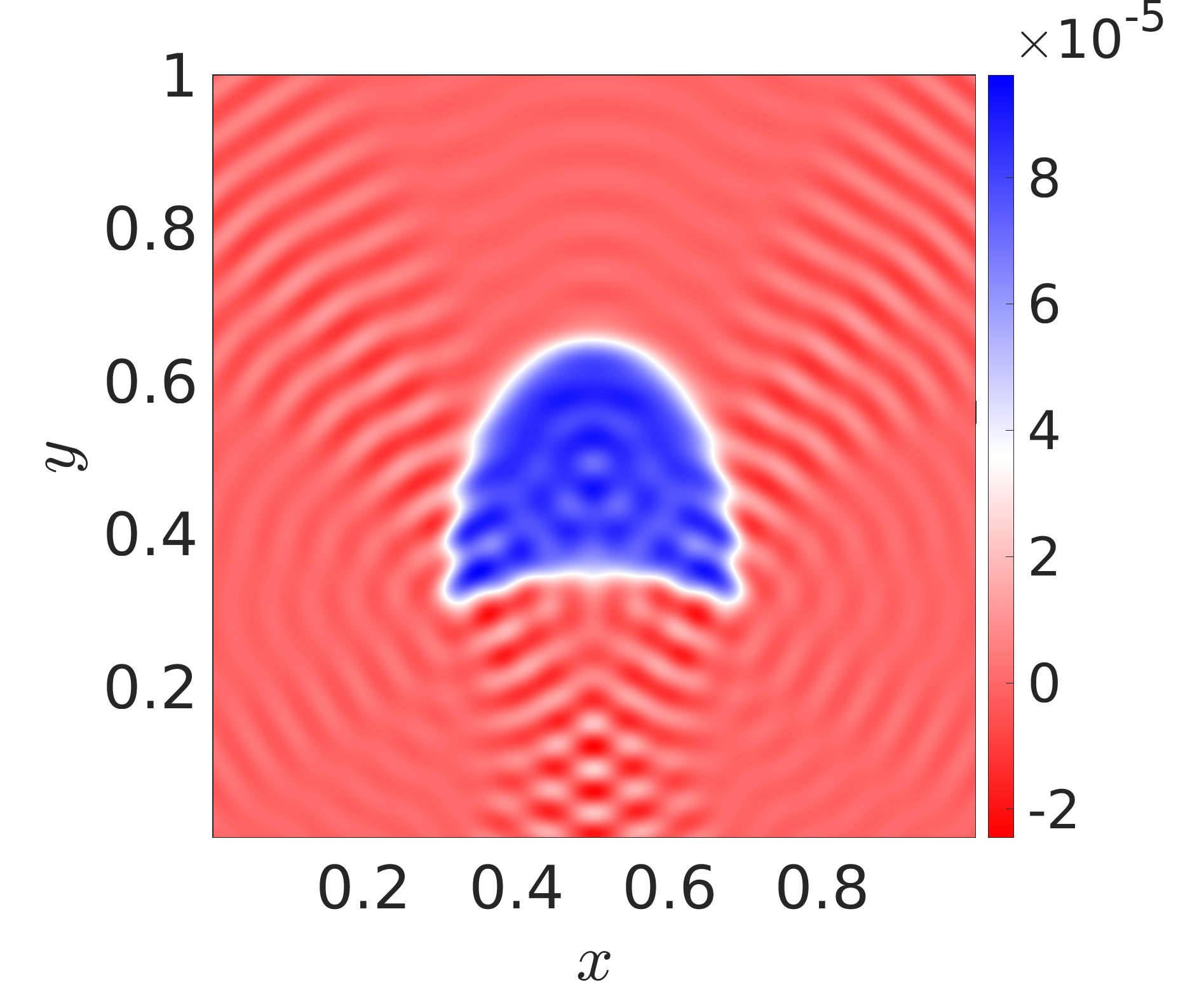}
	\end{subfigure}
	\begin{subfigure}[t]{.29\textwidth}
		\centering
		\includegraphics[width=\linewidth]{\fpath/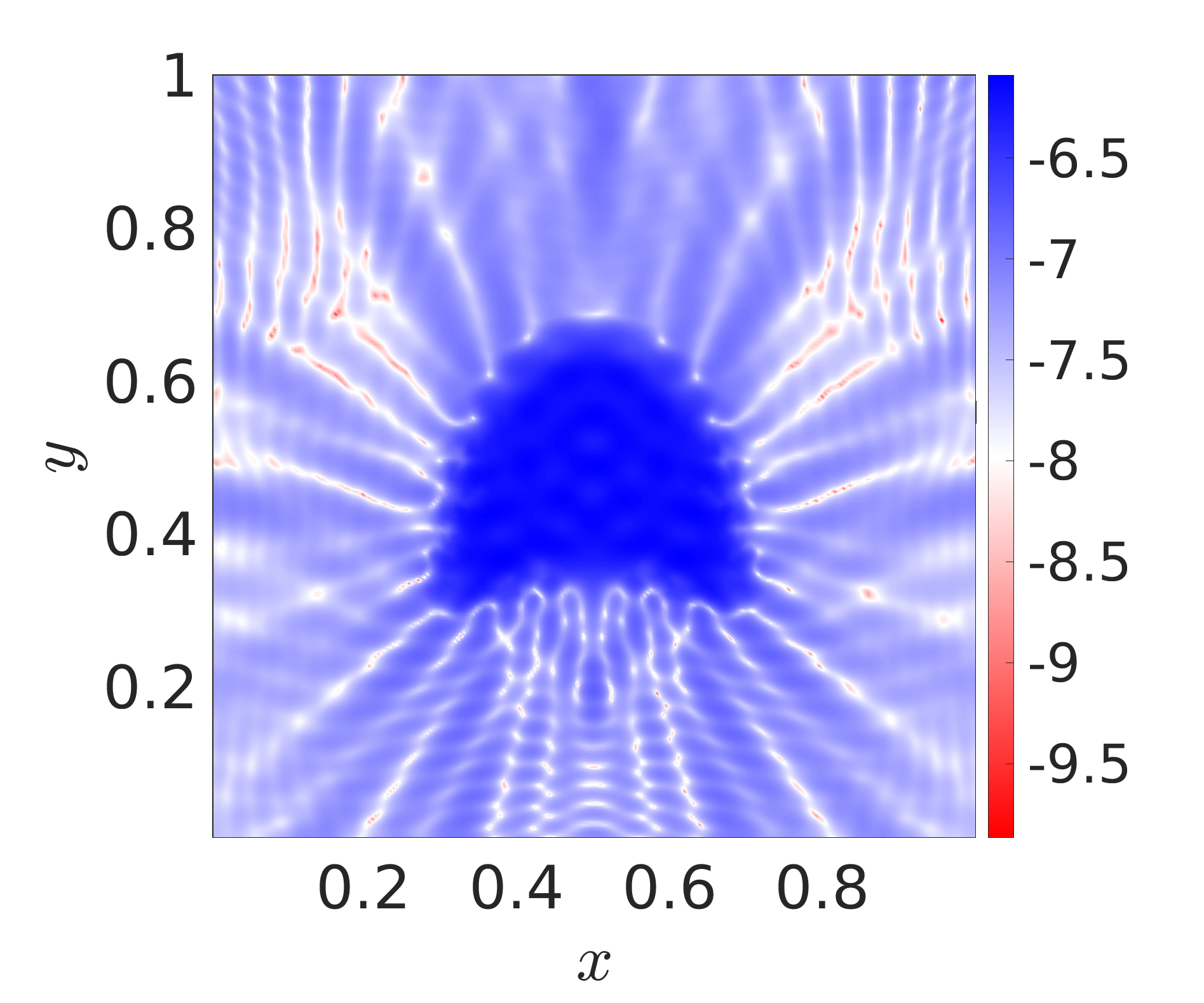}
	\end{subfigure}	
	\begin{subfigure}[t]{.37\textwidth}
		\centering
		\includegraphics[width=\linewidth]{\fpath/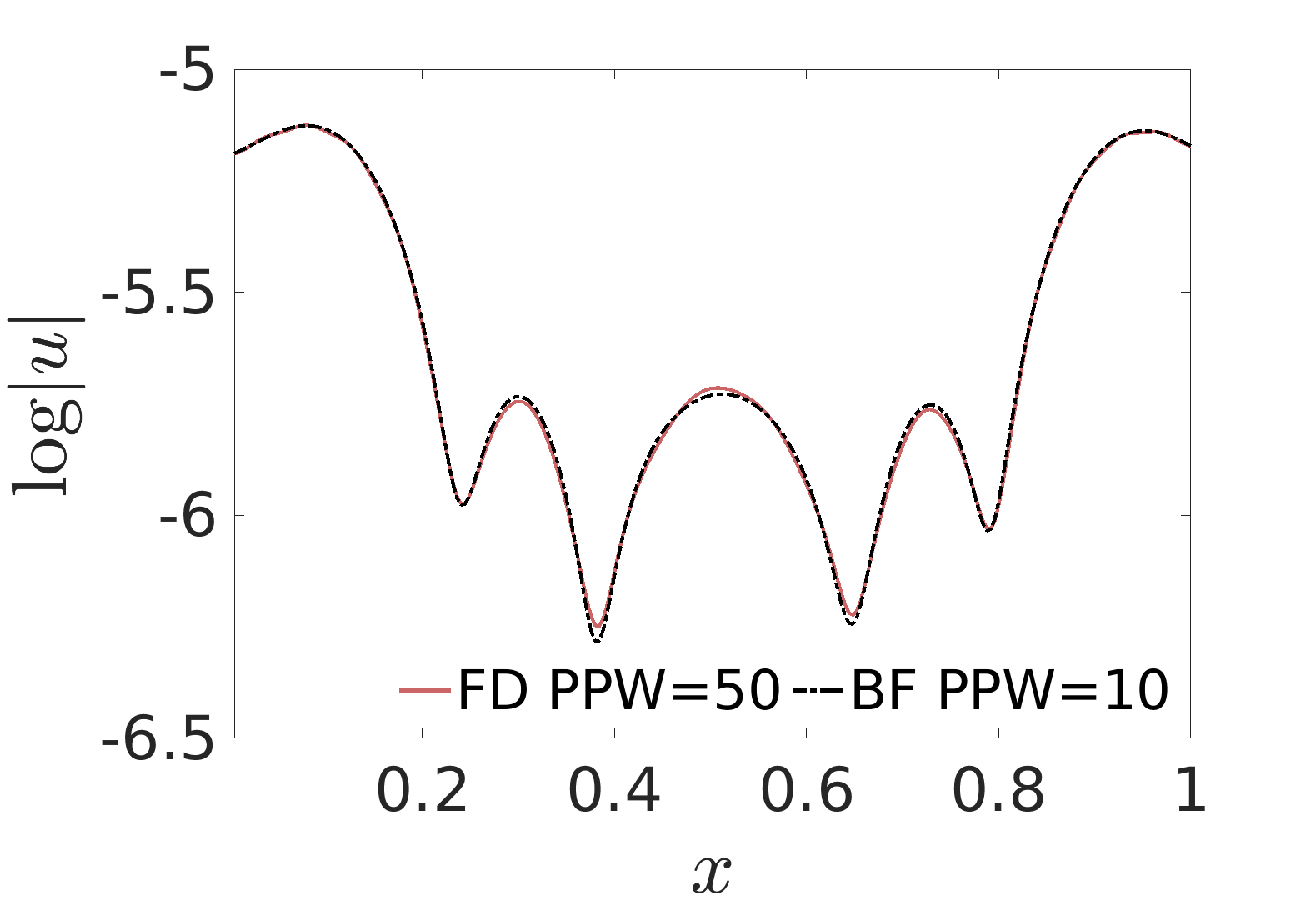}
	\end{subfigure}

	\vspace{-5pt}
	\caption{Sinusoidal media. $\omega=80\pi$  (50 wavelengths each direction) for the first three rows and $\omega=40\pi$ (25 wavelengths each direction) for the last row. Left column: the field $\mathrm{Re}(u_{\rm hb})$ (in the linear scale) computed by the proposed scheme. Middle column: difference $|u_{\rm hb}-u_{\rm fd}|$ (in the log scale) between the fields computed by the proposed scheme (PPW=10) and FDFD (PPW=50). Right column: the fields $|u_{\rm hb}|,|u_{\rm fd}|$ (in the log scale) drawn along the line $y=y_{\rm post}$. Row 1: point source with $y_{\rm post}=1-10h$ and $h$ corresponding to PPW=10. Row 2: Gaussian packet source  with $y_{\rm post}=1-10h$. Row 3: Gaussian packet source with $y_{\rm post}=0.5$ and a semi-circle inclusion. Row 4: concave kite-shaped source with $y_{\rm post}=1-10h$.}	
	\label{fig:ex3_f40}
\end{figure}

\begin{figure}[!htp]
	\centering
	\vspace{-7.5pt}
	\begin{subfigure}[t]{.29\textwidth}
		\centering
		\includegraphics[width=\linewidth]{\fpath/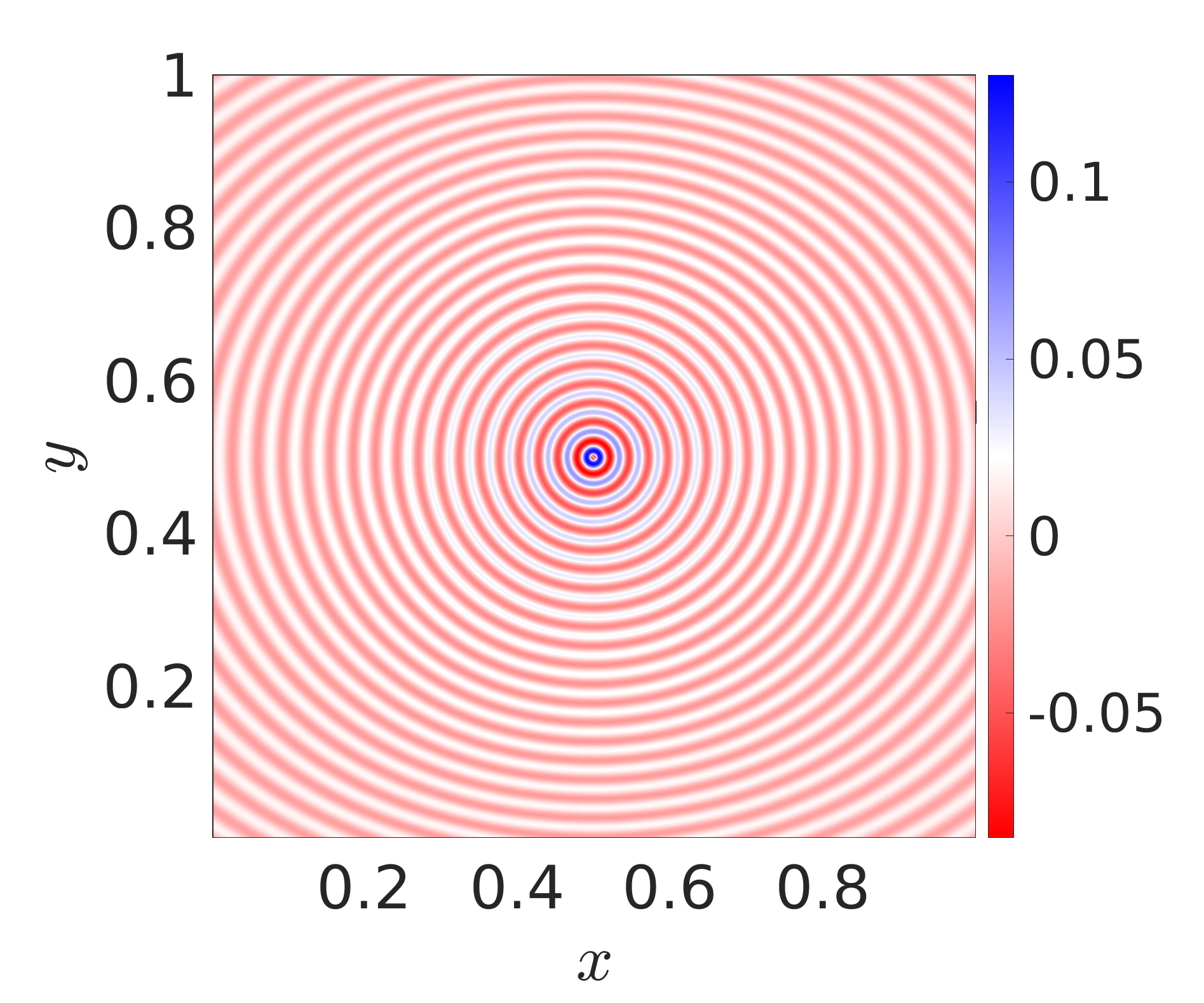}
	\end{subfigure}
	\begin{subfigure}[t]{.29\textwidth}
		\centering
		\includegraphics[width=\linewidth]{\fpath/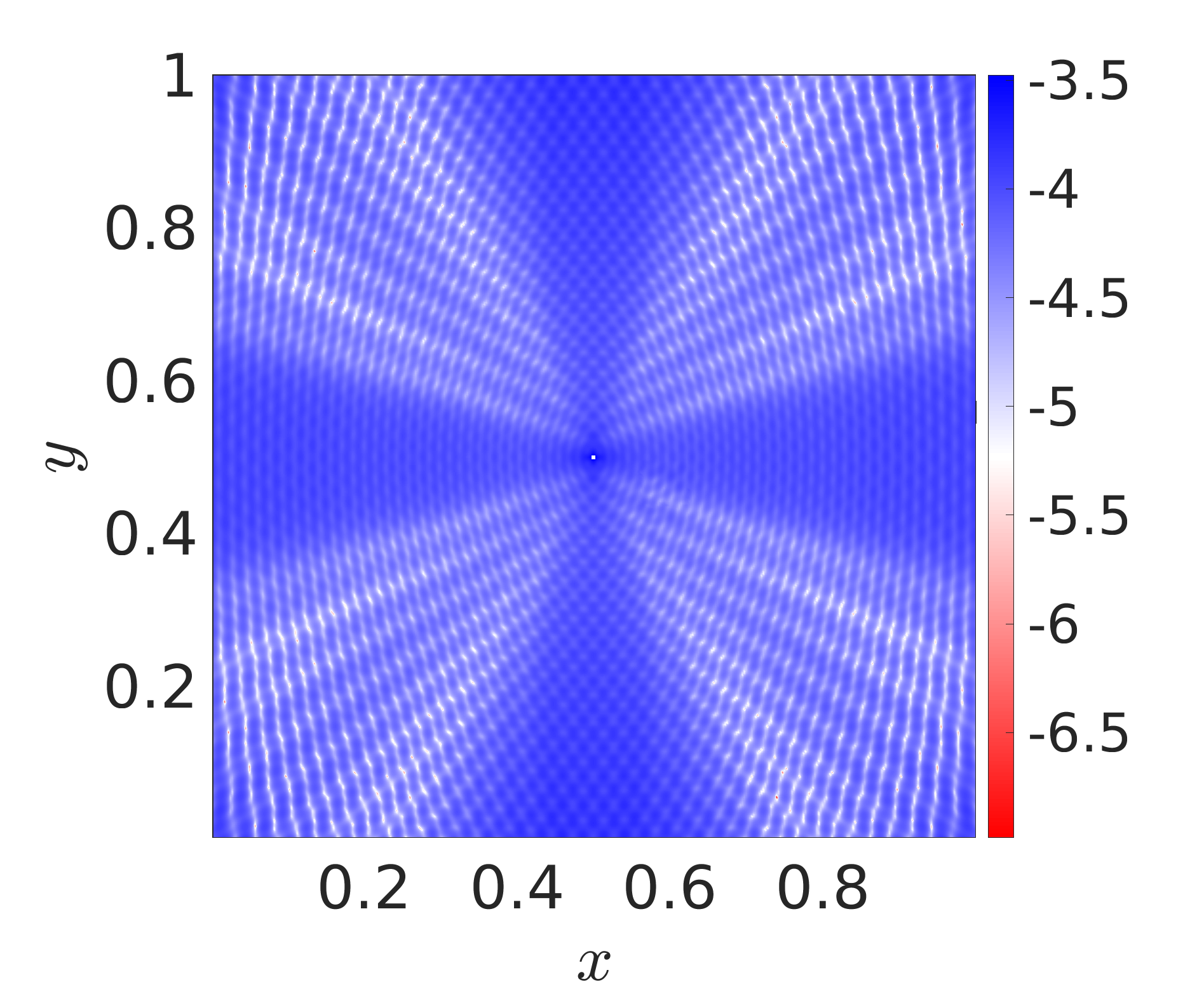}
	\end{subfigure}	
	\begin{subfigure}[t]{.37\textwidth}
		\centering
		\includegraphics[width=\linewidth]{\fpath/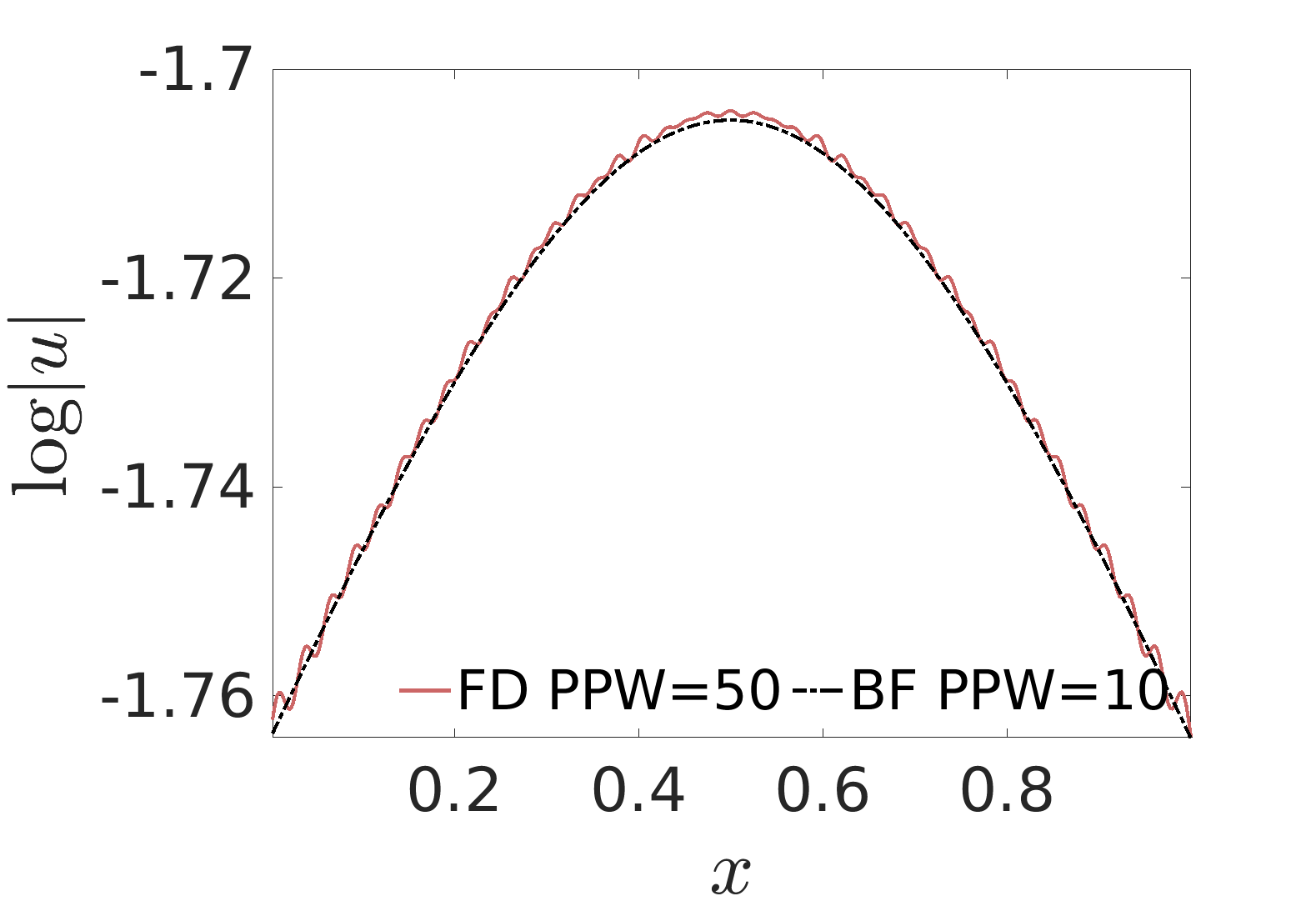}
	\end{subfigure}	
	\begin{subfigure}[t]{.29\textwidth}
		\centering
		\includegraphics[width=\linewidth]{\fpath/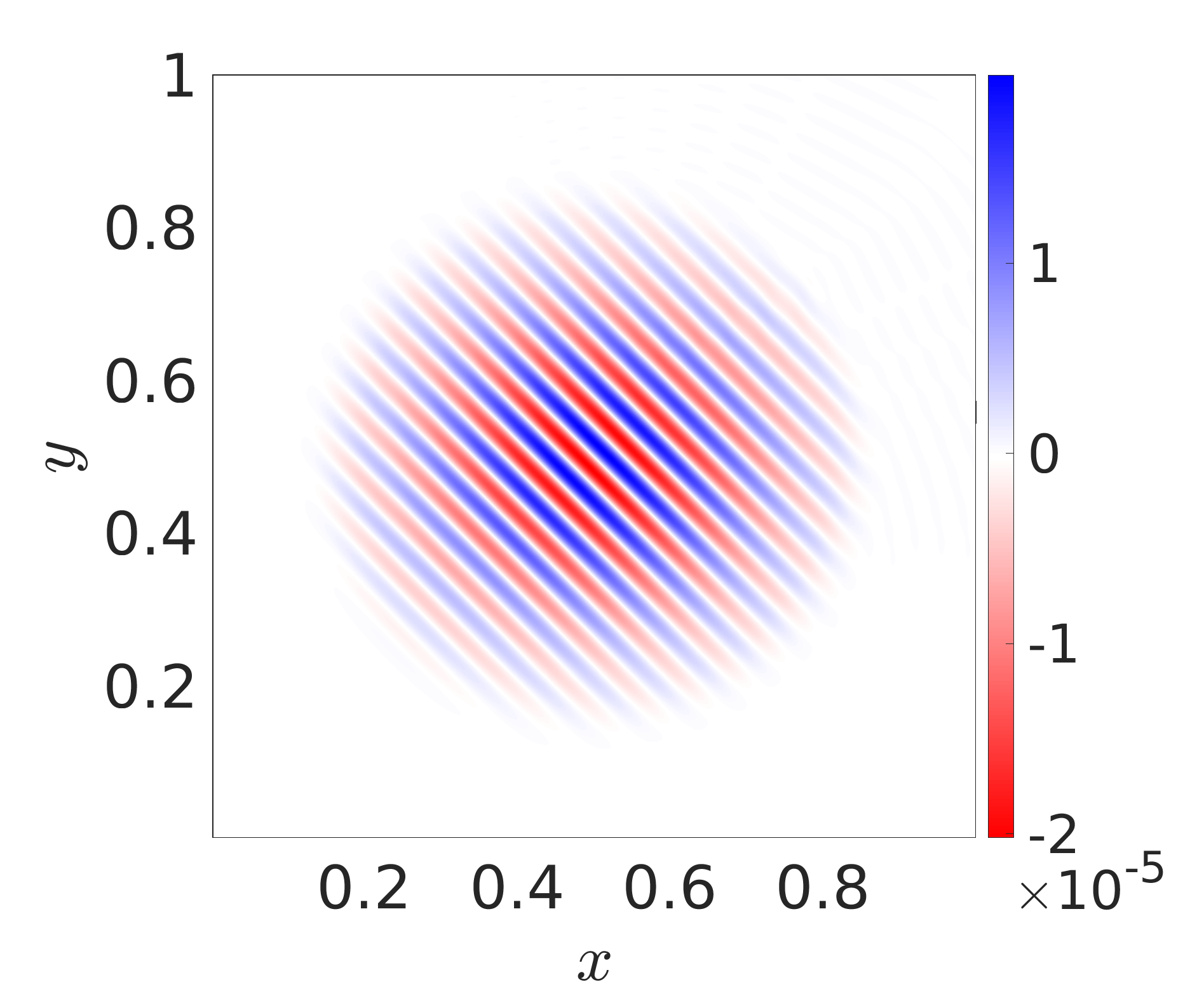}
	\end{subfigure}
	\begin{subfigure}[t]{.29\textwidth}
		\centering
		\includegraphics[width=\linewidth]{\fpath/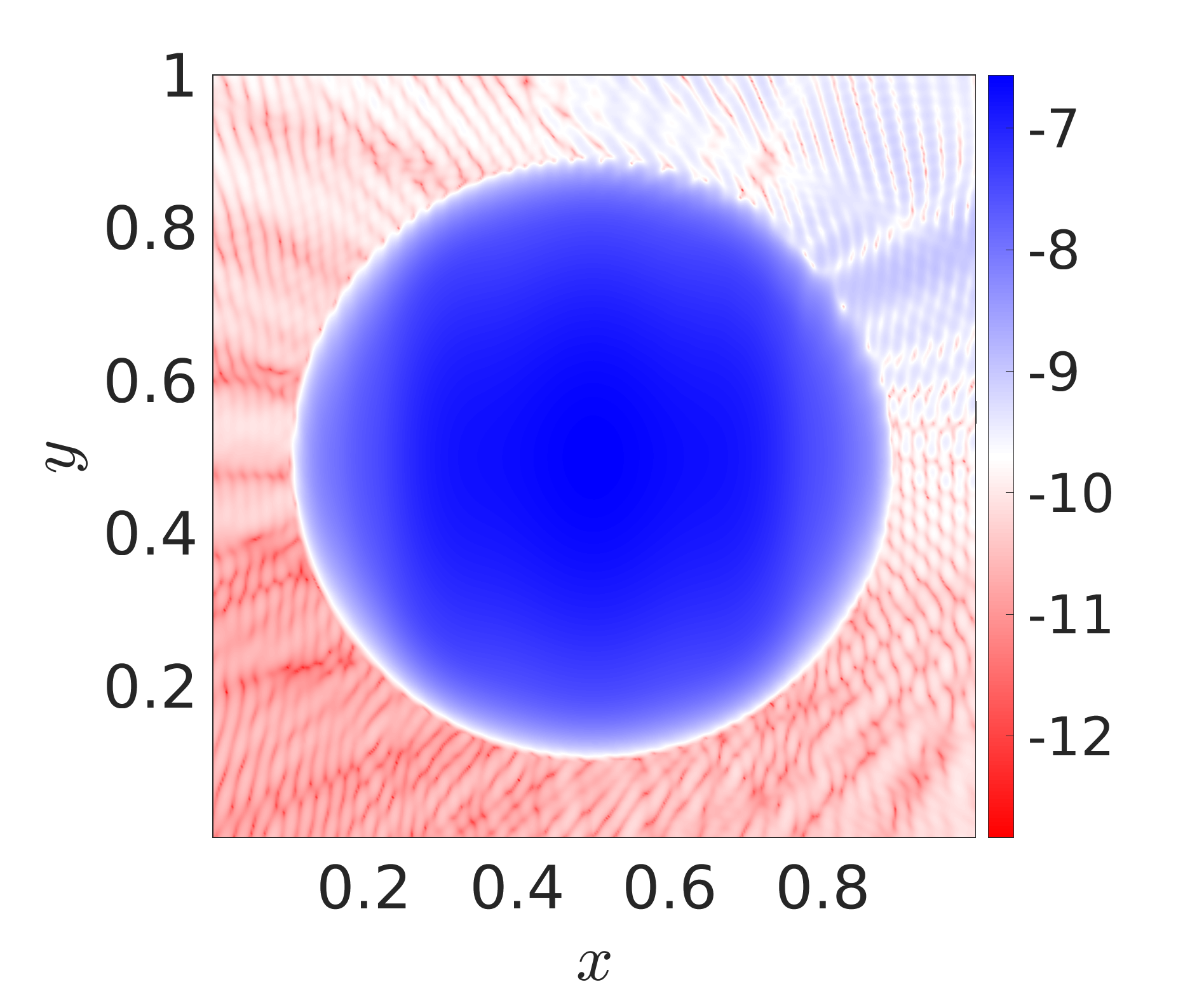}
	\end{subfigure}	
	\begin{subfigure}[t]{.37\textwidth}
		\centering
		\includegraphics[width=\linewidth]{\fpath/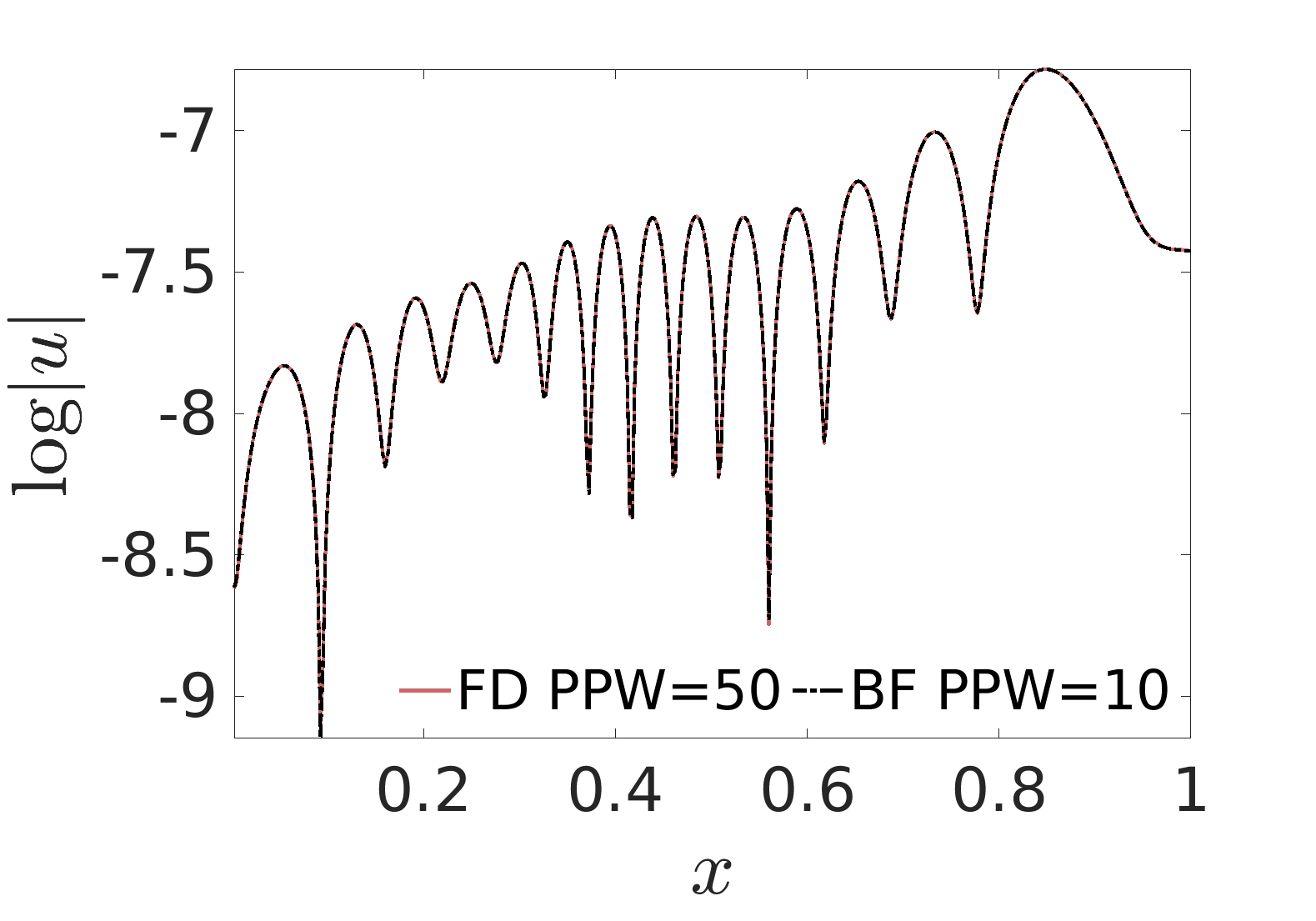}
	\end{subfigure}	
	
	\begin{subfigure}[t]{.29\textwidth}
		\centering
		\includegraphics[width=\linewidth]{\fpath/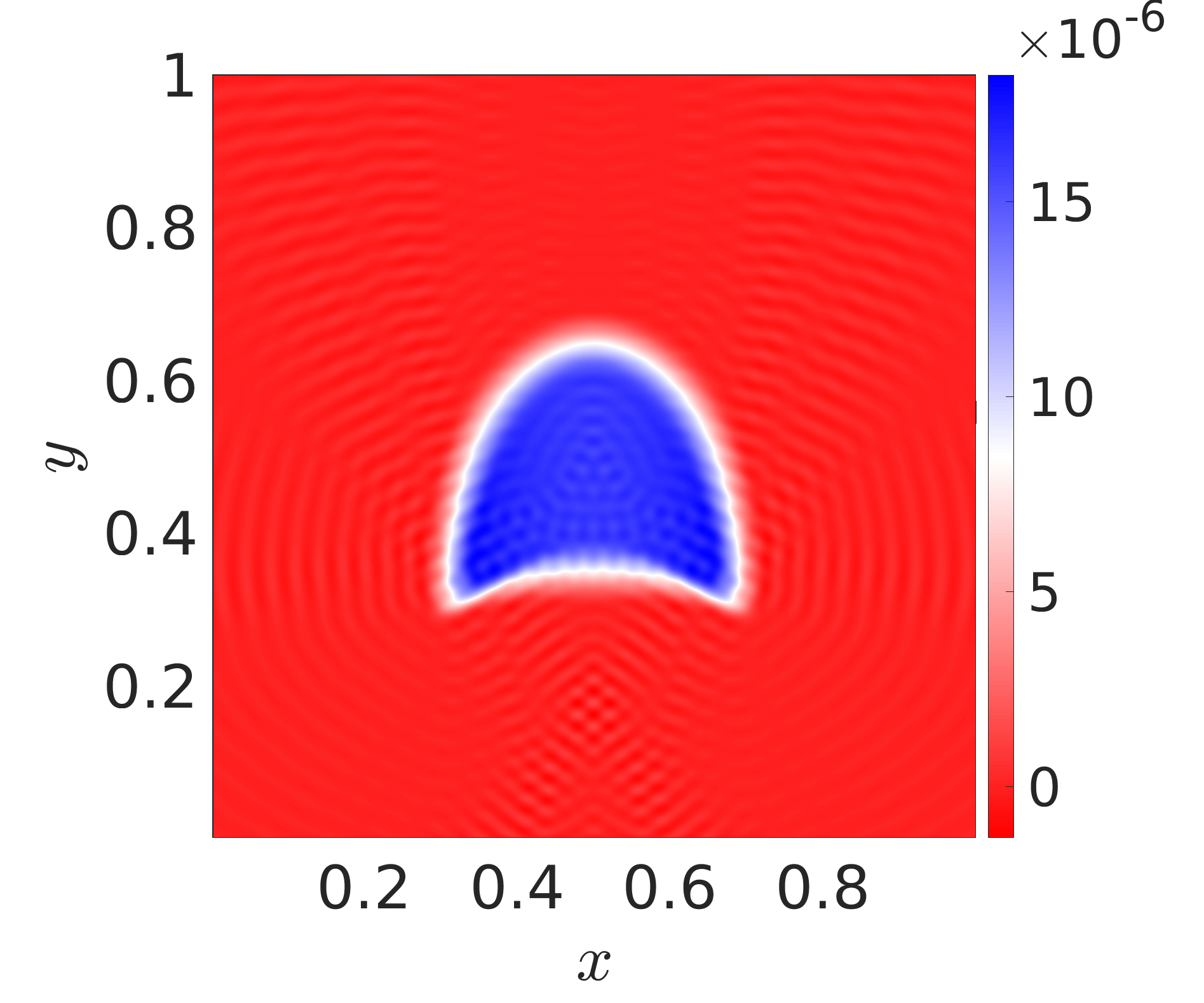}
	\end{subfigure}
	\begin{subfigure}[t]{.29\textwidth}
		\centering
		\includegraphics[width=\linewidth]{\fpath/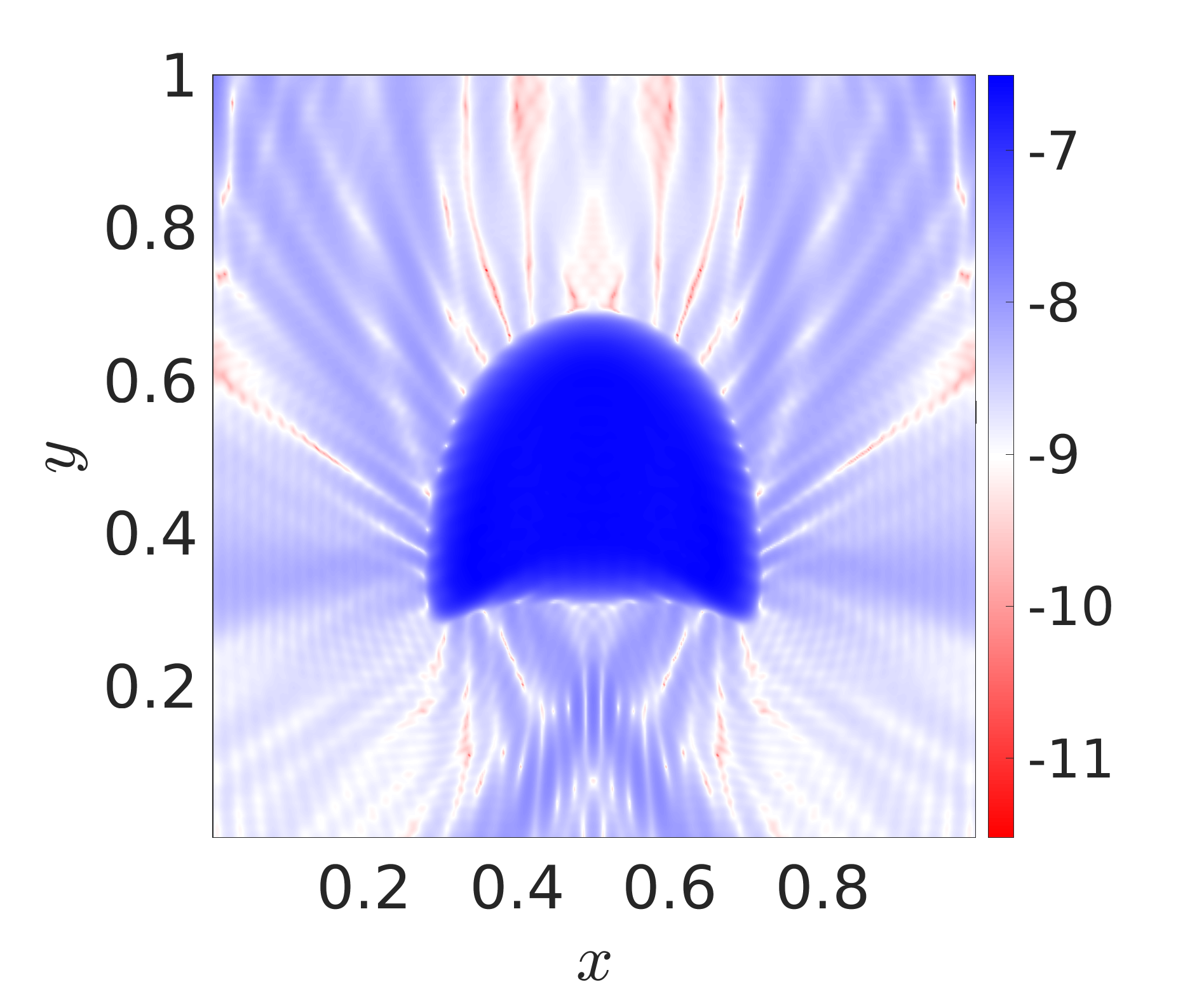}
	\end{subfigure}	
	\begin{subfigure}[t]{.37\textwidth}
		\centering
		\includegraphics[width=\linewidth]{\fpath/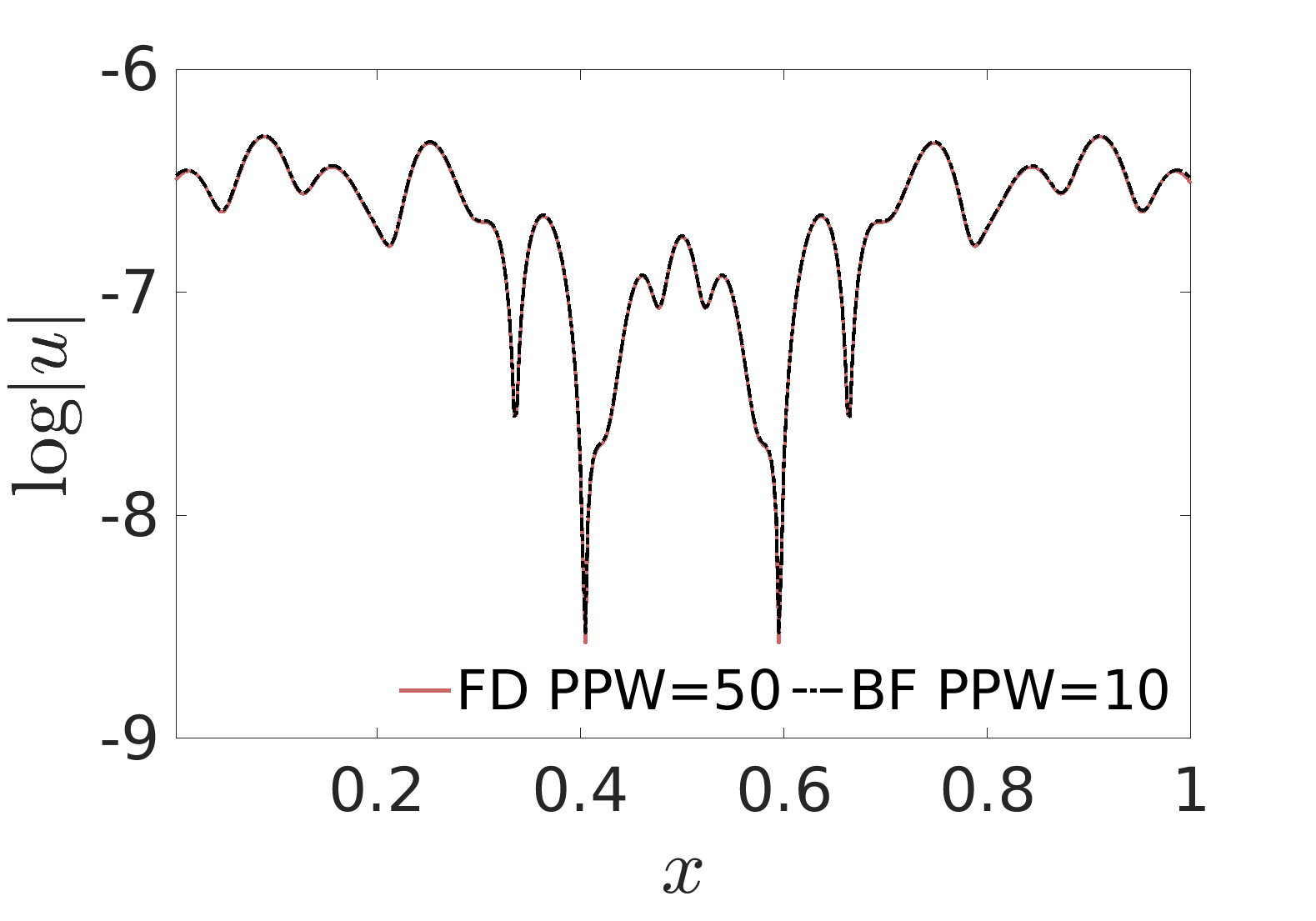}
	\end{subfigure}		

    \begin{subfigure}[t]{.29\textwidth}
		\centering
		\includegraphics[width=\linewidth]{\fpath/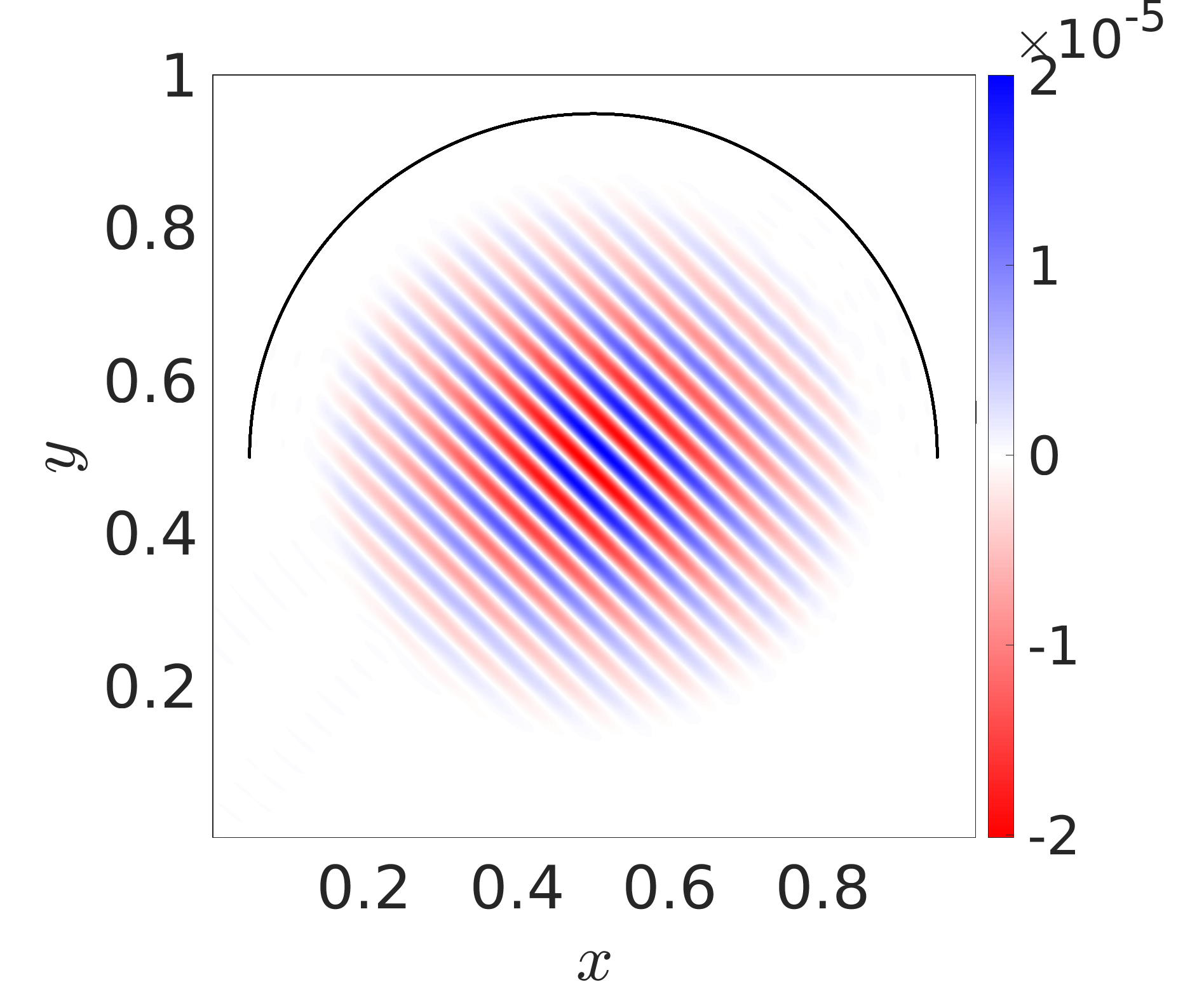}
	\end{subfigure}
	\begin{subfigure}[t]{.29\textwidth}
		\centering
		\includegraphics[width=\linewidth]{\fpath/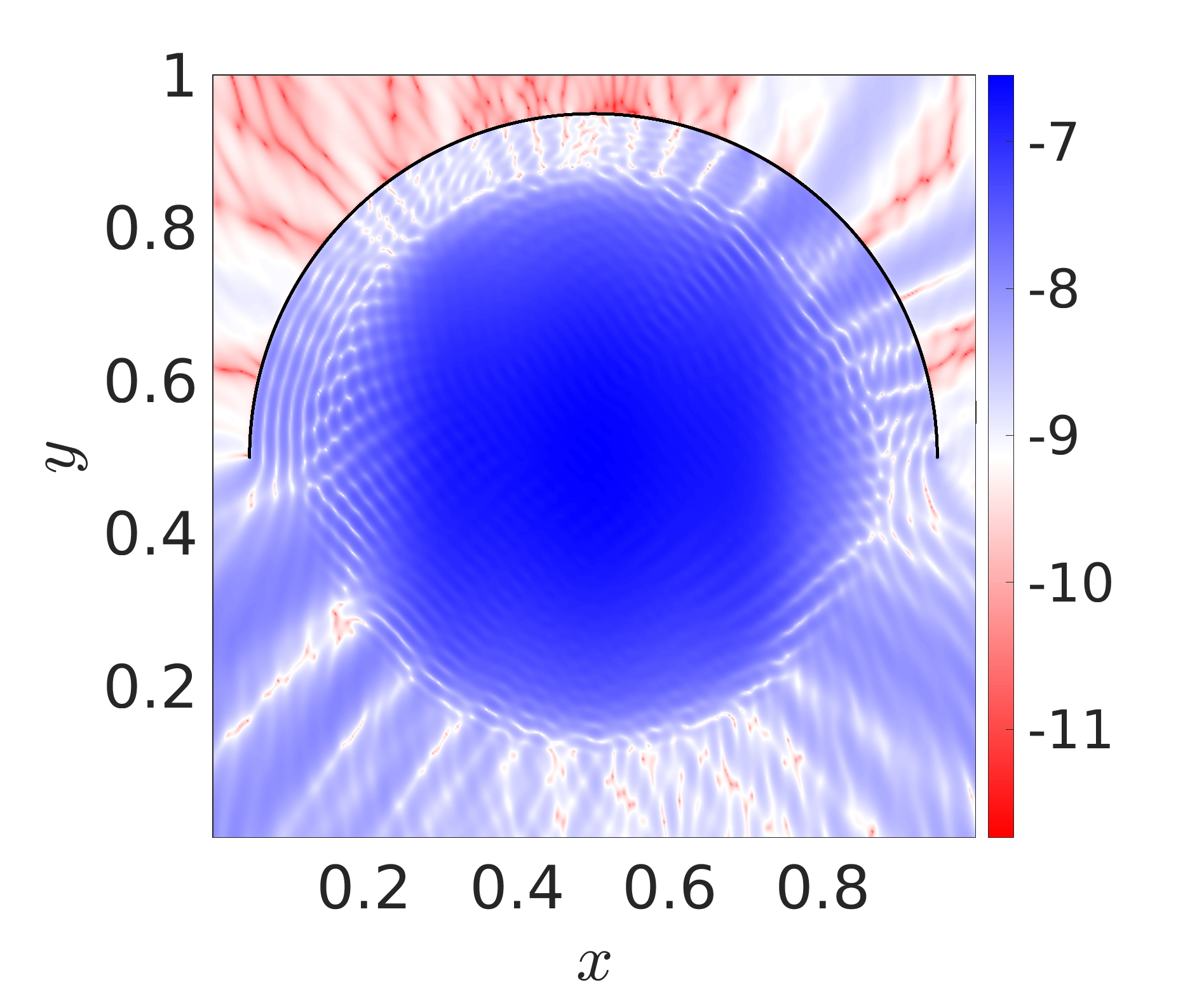}
	\end{subfigure}	
	\begin{subfigure}[t]{.37\textwidth}
		\centering
		\includegraphics[width=\linewidth]{\fpath/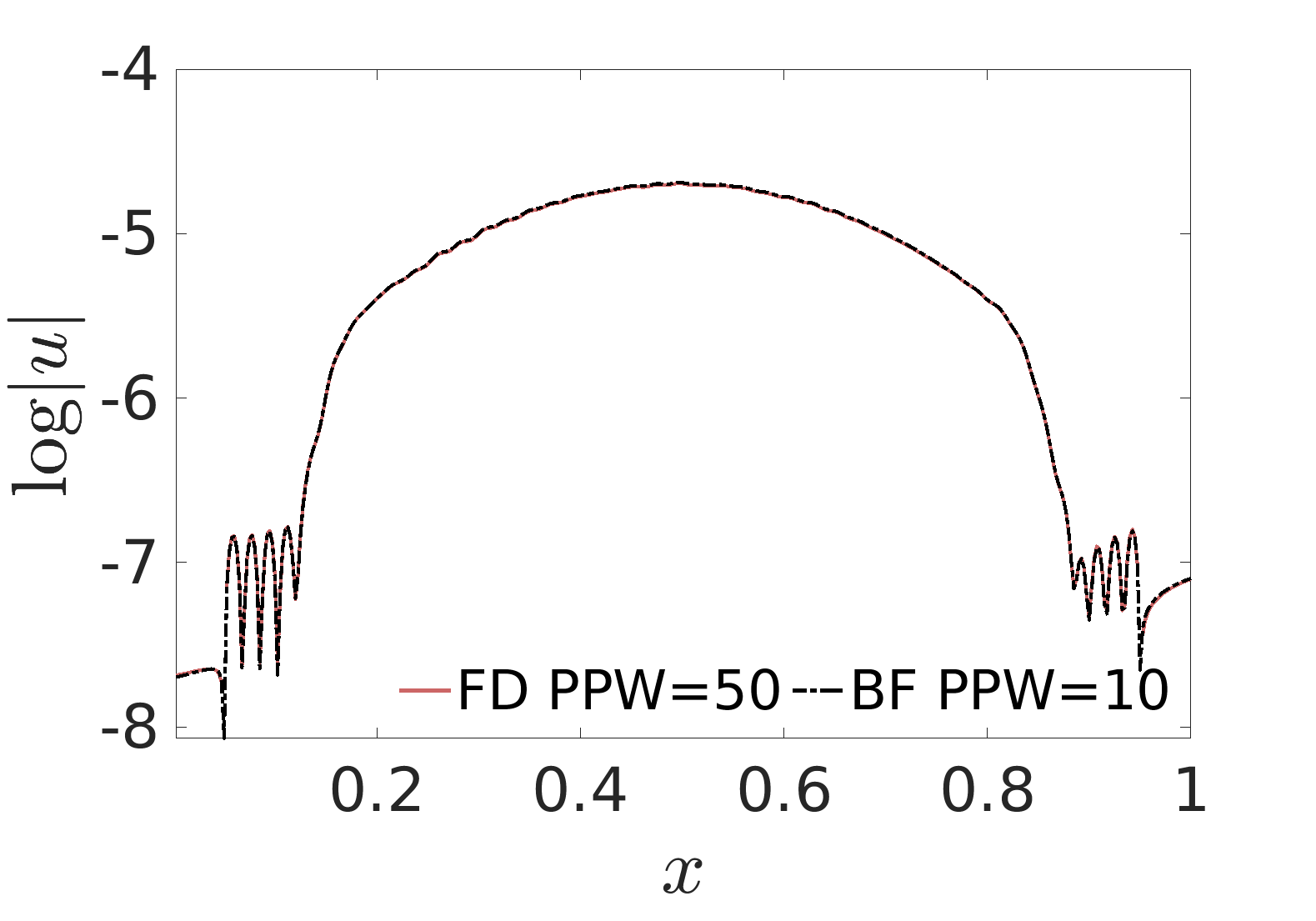}
	\end{subfigure}	
	
	\vspace{-5pt}
	\caption{Waveguide media. $\omega=40\pi$ (40 wavelengths each direction). Left column: the field $\mathrm{Re}(u_{\rm hb})$ (in the linear scale) computed by the proposed scheme. Middle column: difference $|u_{\rm hb}-u_{\rm fd}|$ (in the log scale) between the fields computed by the proposed scheme (PPW=10) and FDFD (PPW=50). Right column: the fields $|u_{\rm hb}|,|u_{\rm fd}|$ (in the log scale) drawn along the line $y=y_{\rm post}$. Row 1: point source with $y_{\rm post}=1-10h$ and $h$ corresponding to PPW=10. Row 2: Gaussian packet source  with $y_{\rm post}=1-10h$. Row 3: concave kite-shaped source with $y_{\rm post}=1-10h$. Row 4: Gaussian packet source with $y_{\rm post}=0.5$ and a semi-circle inclusion.}	
	\label{fig:ex4_f20}
\end{figure}

\subsubsection{3D domain without inclusion\label{sec:acc_3d}}
We consider the following examples of homogeneous and inhomogeneous media for $d=3$:
\begin{itemize}
	\item \textbf{Constant media}: the computational domain is $V=[0,0.5]^3$ with $n(\mathbf{r})=n(x,y,z)=2$. The phase function has an analytical formula $\tau(\mathbf{r},\mathbf{r}_0)=n(\mathbf{r}_0)|\mathbf{r}-\mathbf{r}_0|$. The HB coefficients have analytical form $v_0(\mathbf{r},\mathbf{r}_0)=n(\mathbf{r}_0)/(2\pi)$ and $v_1(\mathbf{r},\mathbf{r}_0)=0$. The HB integrator becomes $g(\mathbf{r},\mathbf{r}_0)=\exp(\mathrm{i}\omega n(\mathbf{r}_0) |\mathbf{r}-\mathbf{r}_0|)/(4\pi |\mathbf{r}-\mathbf{r}_0|)$, i.e., the well-known form of the free-space Green's function. We use the Lax-Friedrichs WENO method with $h_0=0.01$ to solve (\ref{eikonalnormal}) and (\ref{equVs}) with point sources, construct the low-rank representation with an order of $n_I=3$ for the Chebyshev interpolation, and compare the results with these exact formulas. 
	\item \textbf{Constant-gradient media}: the computational domain is $V=[0,0.52]^3$ with $n(\mathbf{r})=n(x,y,z)=\frac{1}{0.5-0.8(y-0.26)}=\frac{-1.25}{y-0.885}$. Note that $1/n$ has a nonzero constant partial derivative in $y$. The phase function has an analytical form $\tau(\mathbf{r},\mathbf{r}_0)=\frac{1}{|\mathbf{G_0}|}\mathrm{arccosh}\Big(1+\frac{1}{2}n(\mathbf{r})n(\mathbf{r}_0)|\mathbf{G}_0|^2|\mathbf{r}-\mathbf{r}_0|^2\Big)$ with $\mathbf{G}_0=[0,-0.8,0]$. The HB coefficients have no known analytical expressions. We use the fifth-order LxF WENO method with $h_0=0.01$ to solve (\ref{eikonalnormal}) and (\ref{equVs}) with point sources, and construct the low-rank representation with an order of $n_I=9$ for the Chebyshev interpolation. It is worth mentioning that the exact Green's function \cite{Holford1981Elementary} exists as $g(\mathbf{r},\mathbf{r}_0)=\frac{(|y-c||y_0-c|)^{1/2}}{2\pi RR'}\exp\big(2\mathrm{i}(a^2\omega^2-1/4)^{1/2}\mathrm{arctanh}(R/R')\big)$. Here $a=1.25$, $c=0.885$, $R=|\mathbf{r}-\mathbf{r}_0|$, and $R'=\sqrt{(x-x_0)^2 + (y+y_0-2c)^2 + (z-z_0)^2}$.   	
\end{itemize}	

To compute $K^{v2v}$, we discretize the domain with $h=2\pi/(\omega n_{\max}n_p)$ with $n_p\sim 5$ and $n_{\max}$ being the maximum refractive index over the domain. The HODBF compression of $K^{v2v}$ is computed with tolerance ${\rm tol}=10^{-6}$ in (\ref{eqn:ID}) and oversampling factor $\chi=4$ in (\ref{eq:proxyv2v}). We apply (\ref{eq:m_v2v}) to the following RHSs centered at $\mathbf{r}_c$ (the domain center): 
\begin{itemize}
	\item \textbf{Point source:} $s(\mathbf{r})=1/h^3$ if $\mathbf{r}$ is inside the source cell centered at $\mathbf{r}_c$.  
	\item \textbf{Gaussian wavepacket source:} $s(\mathbf{r})={\rm exp}(-|\mathbf{r}-\mathbf{r}_c|^2/(2\sigma^2)){\rm exp}(i\omega_0(\mathbf{r}\cdot\mathbf{d}))t(|\mathbf{r}-\mathbf{r}_c|,w_1,w_2)$ with $\sigma=0.15$, $w_1=0.05$, $w_2=0.05$, $\omega_0=0.9\omega$, $\mathbf{d}=\frac{1}{\sqrt 3}[1,1,1]$, and $t$ is the tapering function in (\ref{eq:taper}). 		
\end{itemize}

As for the reference FDFD solver, we use the 27-point staggered grid scheme in \cite{operto20073d}. The computational domain is extended in each direction with a perfectly matched layer (PML). The resulting sparse linear system is solved with STRUMPACK \cite{ghysels2017robust,liu2021sparse}.

For the constant medium, we use PMLs of thickness $3\pi/(\omega n_{\max})$ (1.5 wavelengths). We consider $\omega=40\pi$ (20 wavelengths in each direction). The fields computed by the proposed scheme (PPW=5) and FDFD (PPW=10) and their differences are shown in \cref{fig:ex1_f20_3d}. For the point source (Row 1 in \cref{fig:ex1_f20_3d}), the exact solution is also plotted (in dashed green) in the third column. We can see that for both the point source and Gaussian wavepacket source, the results from FDFD with PPW=10 are not satisfactory. The FDFD scheme with PPW=10 results in a sparse system of dimension $240^3=13,824,000$, and PPW=20 will lead to $480^3=110,592,000$. In contrast, the proposed scheme results in a dense, compressed matrix of dimension $N_v=1,030,301$.

For the constant-gradient medium, we use PMLs of thickness $10\pi/(\omega n_{\max})$ (5 wavelengths). We first consider $\omega=32\pi$ (27 wavelengths in each direction). The fields computed by the proposed scheme (PPW=4.56) and FDFD (PPW=9.12) and their differences are shown in \cref{fig:ex2_f16_3d}. For the point source (Row 1 in \cref{fig:ex2_f16_3d}), the exact solution is also plotted (in dashed green) in the third column. The FDFD scheme with PPW=9.12 results in a sparse system of dimension $352^3=43,614,208$. The proposed scheme results in a dense, compressed matrix of dimension $N_v=2,248,091$. \ylrevnew{Next, we consider $\omega=64\pi$ (54 wavelengths in each direction). The fields computed by the proposed scheme (PPW=4.56) and FDFD (PPW=4.56) and their differences are shown in \cref{fig:ex2_f16_3d}. For the point source (Row 3 in \cref{fig:ex2_f16_3d}), the exact solution is also plotted (in dashed green) in the third column. Clearly, the FDFD results are not reliable. The FDFD scheme with PPW=4.56 results in a sparse system of dimension $305^3=28,372,625$ and even higher PPWs for the FDFD scheme lead to large-scale linear systems of dimensions that no existing sparse direct solvers can deal with. In contrast, the proposed scheme results in a dense, compressed matrix of dimension $N_v=17,779,581$.}

\begin{figure}[!htp]
	\centering
	\vspace{-7.5pt}
	\begin{subfigure}[t]{.29\textwidth}
		\centering
		\includegraphics[width=\linewidth]{\tfpath/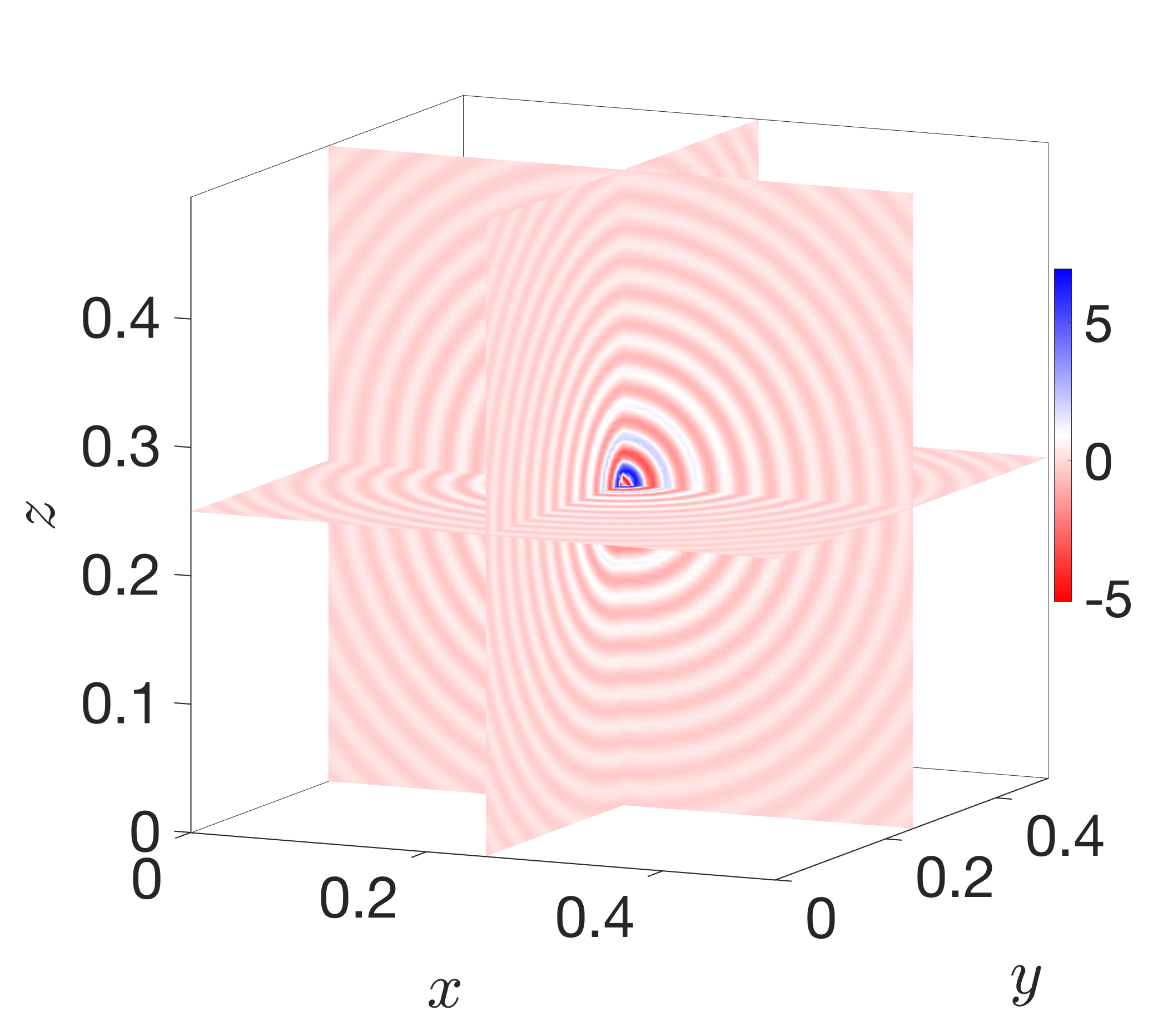}
	\end{subfigure}
	\begin{subfigure}[t]{.29\textwidth}
		\centering
		\includegraphics[width=\linewidth]{\tfpath/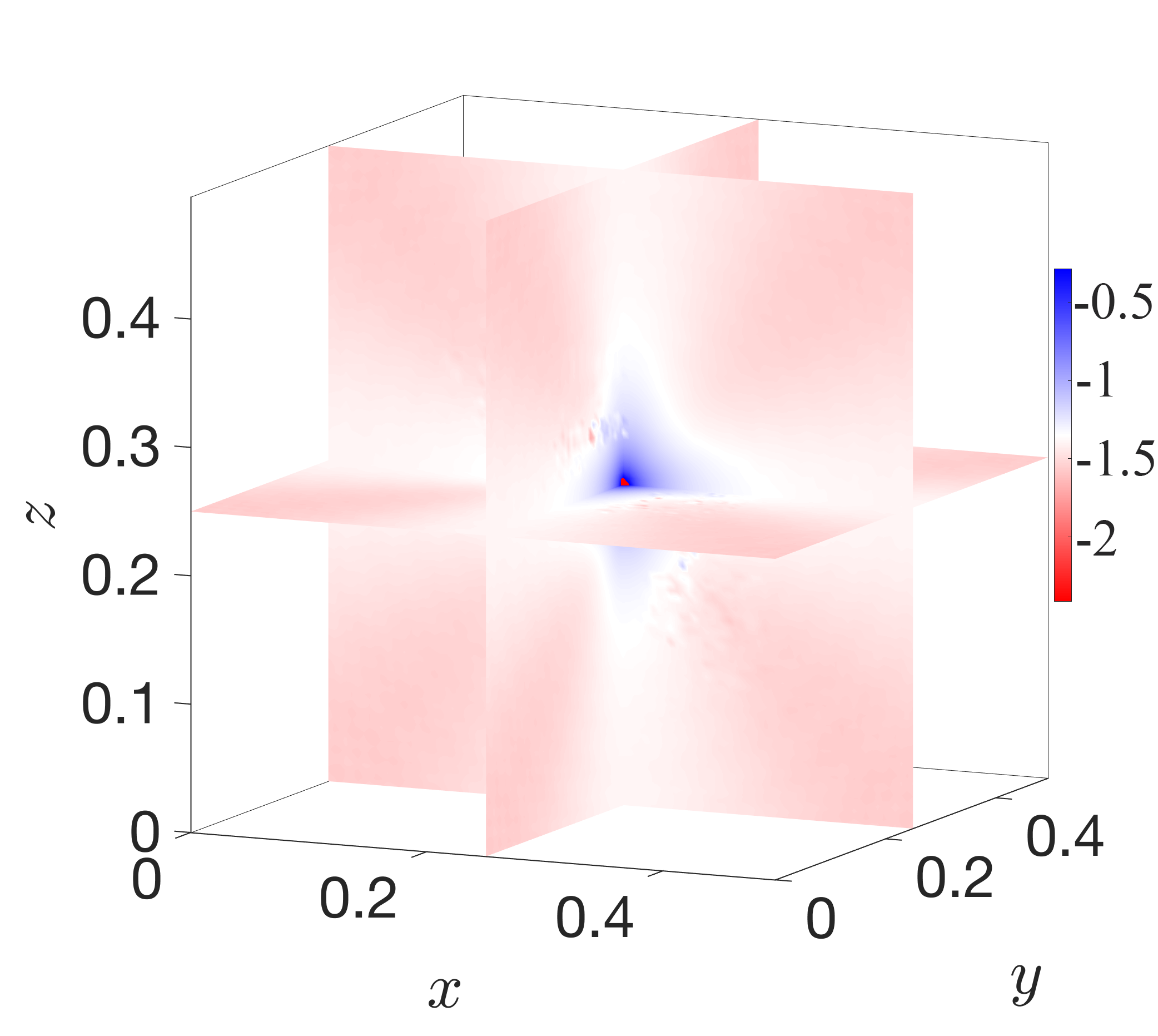}
	\end{subfigure}	
	\begin{subfigure}[t]{.37\textwidth}
		\centering
		\includegraphics[width=\linewidth]{\tfpath/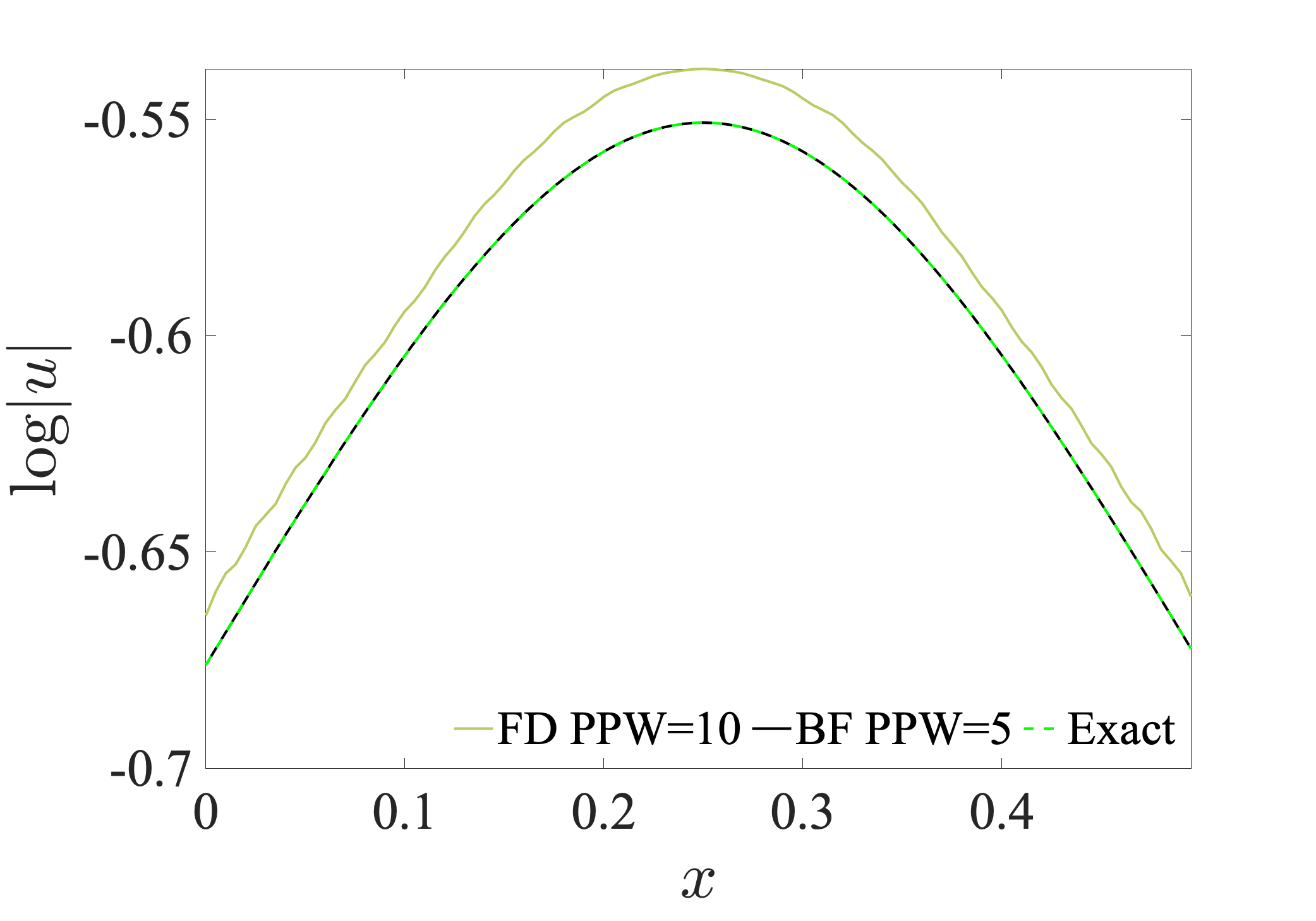}
	\end{subfigure}	
	\begin{subfigure}[t]{.29\textwidth}
		\centering
		\includegraphics[width=\linewidth]{\tfpath/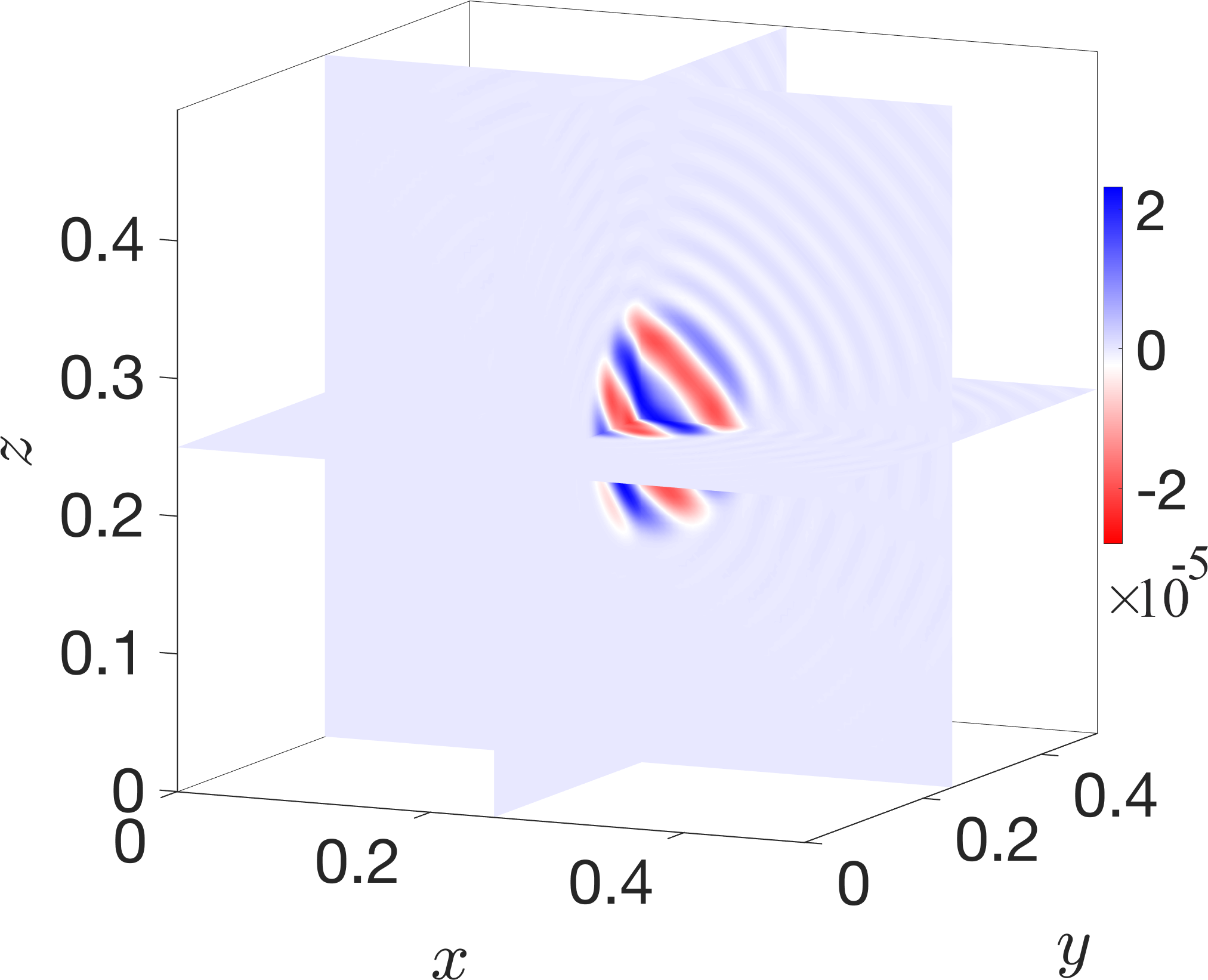}
	\end{subfigure}
	\begin{subfigure}[t]{.29\textwidth}
		\centering
		\includegraphics[width=\linewidth]{\tfpath/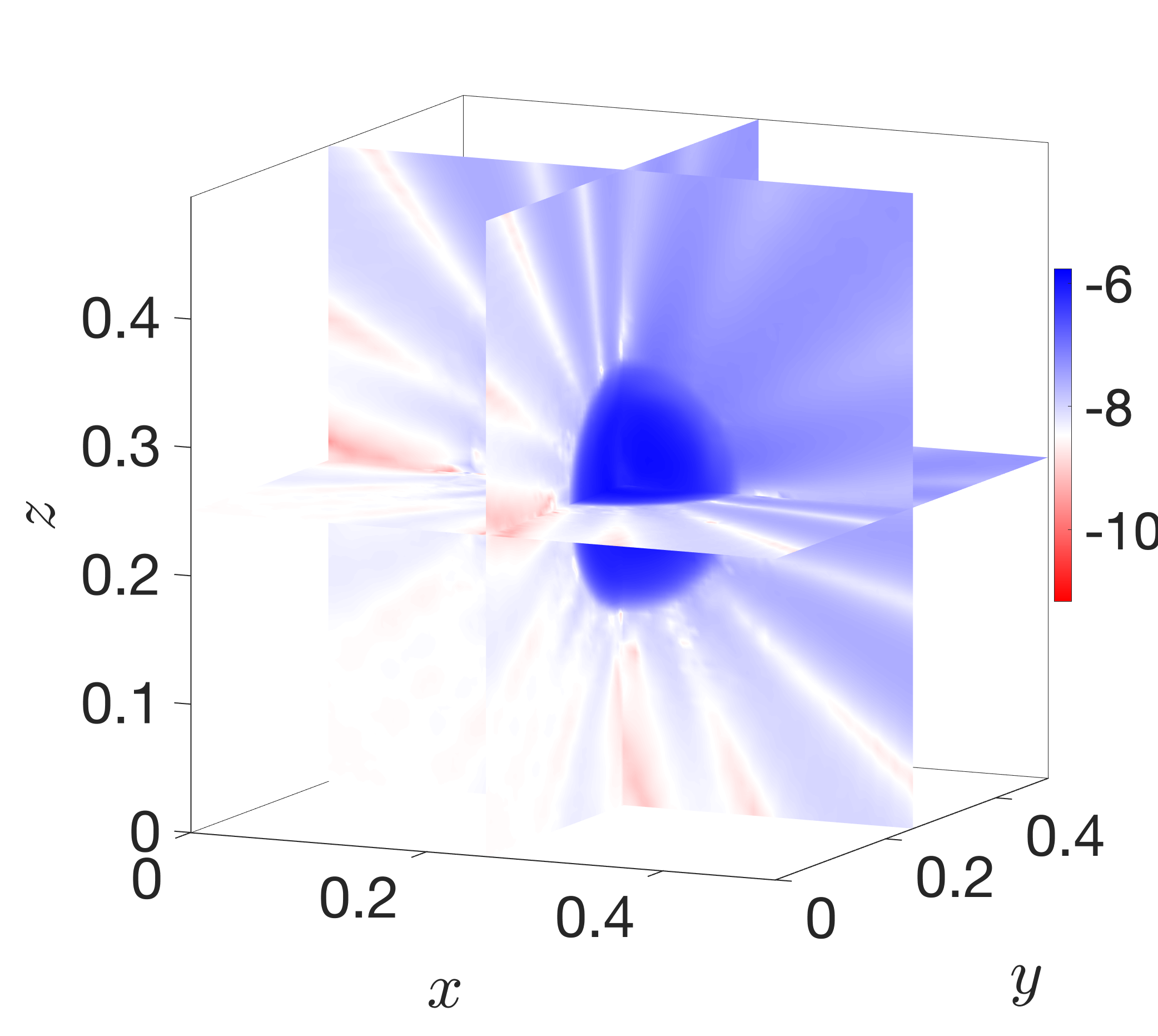}
	\end{subfigure}	
	\begin{subfigure}[t]{.37\textwidth}
		\centering
		\includegraphics[width=\linewidth]{\tfpath/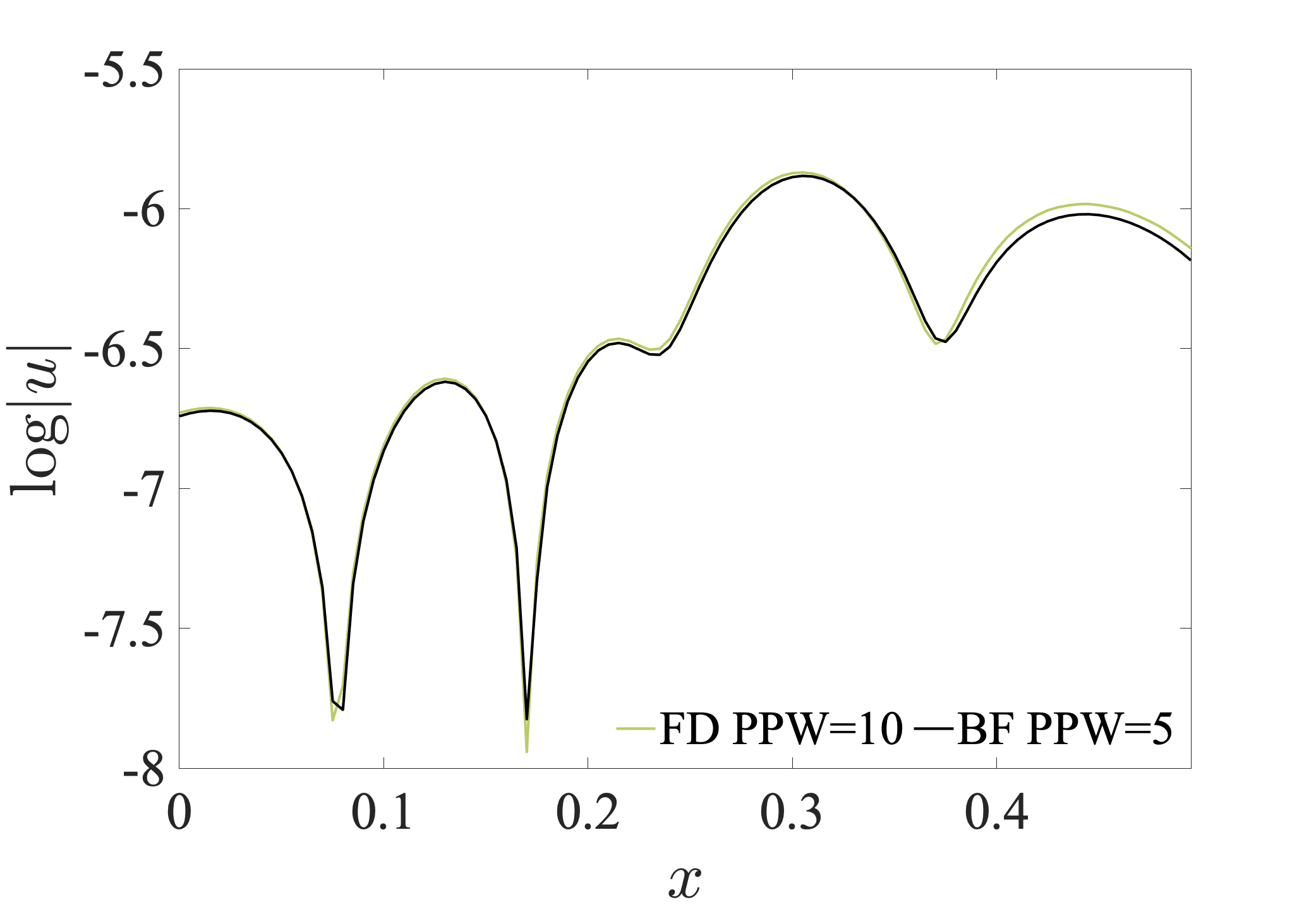}
	\end{subfigure}	
	
	\vspace{-5pt}
	\caption{Constant media in $d=3$. $\omega=40\pi$ (20 wavelengths each direction). Left column: the field $\mathrm{Re}(u_{\rm hb})$ (in linear scale) computed by the proposed scheme. Middle column: difference $|u_{\rm hb}-u_{\rm fd}|$ (in log scale) between the fields computed by the proposed scheme (PPW=5) and FDFD (PPW=10). Right column: the field $|u_{\rm hb}|,|u_{\rm fd}|,|u_{true}|$ (in log scale) drawn along the line $y=0.5-10h$  and $z=0.5-10h$ with $h$ corresponding to PPW=5. Row 1: point source. Row 2: Gaussian packet source.}
	\label{fig:ex1_f20_3d}
\end{figure}

\begin{figure}[!htp]
	\centering
	\vspace{-7.5pt}
	\begin{subfigure}[t]{.29\textwidth}
		\centering
		\includegraphics[width=\linewidth]{\tfpath/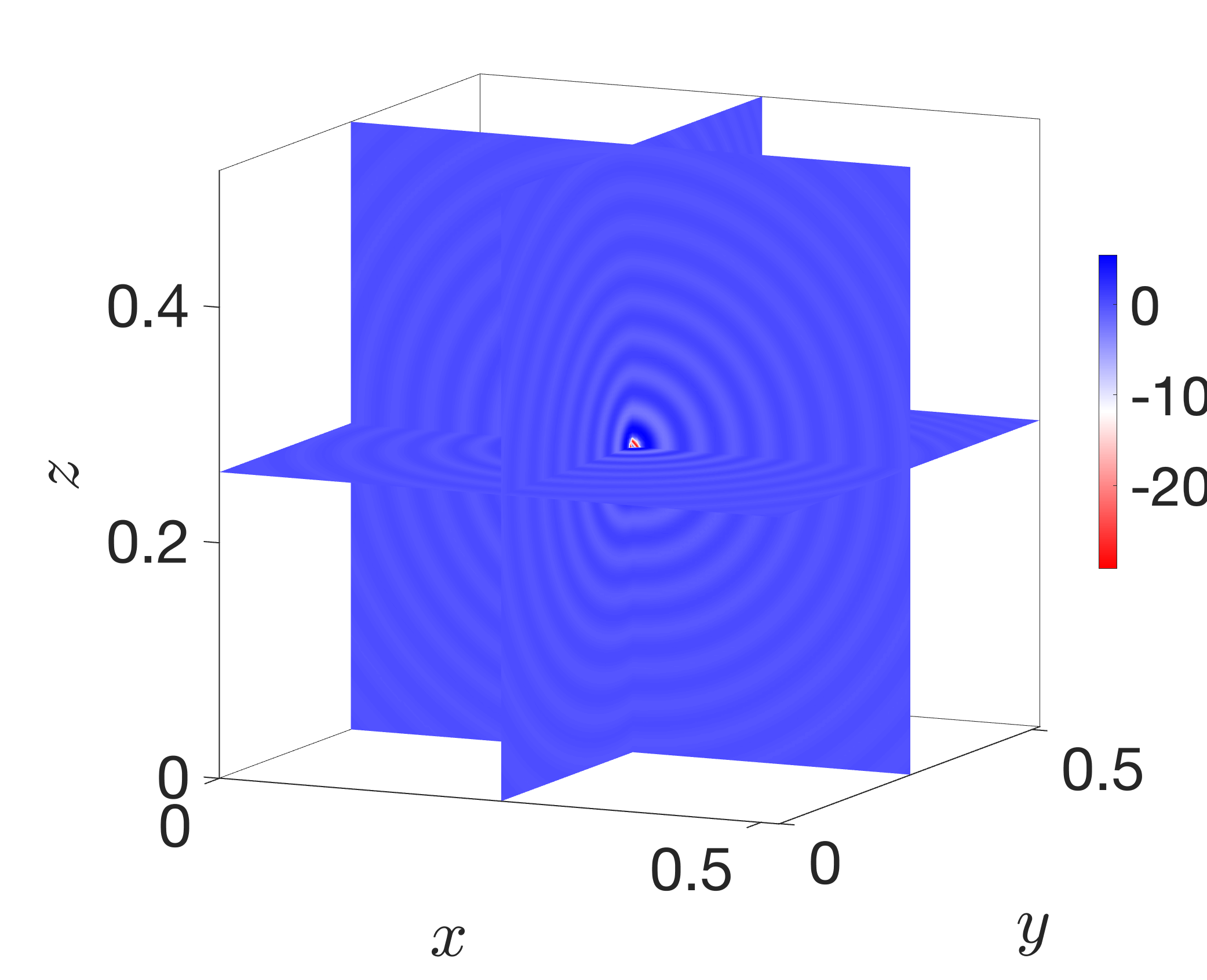}
	\end{subfigure}
	\begin{subfigure}[t]{.29\textwidth}
		\centering
		\includegraphics[width=\linewidth]{\tfpath/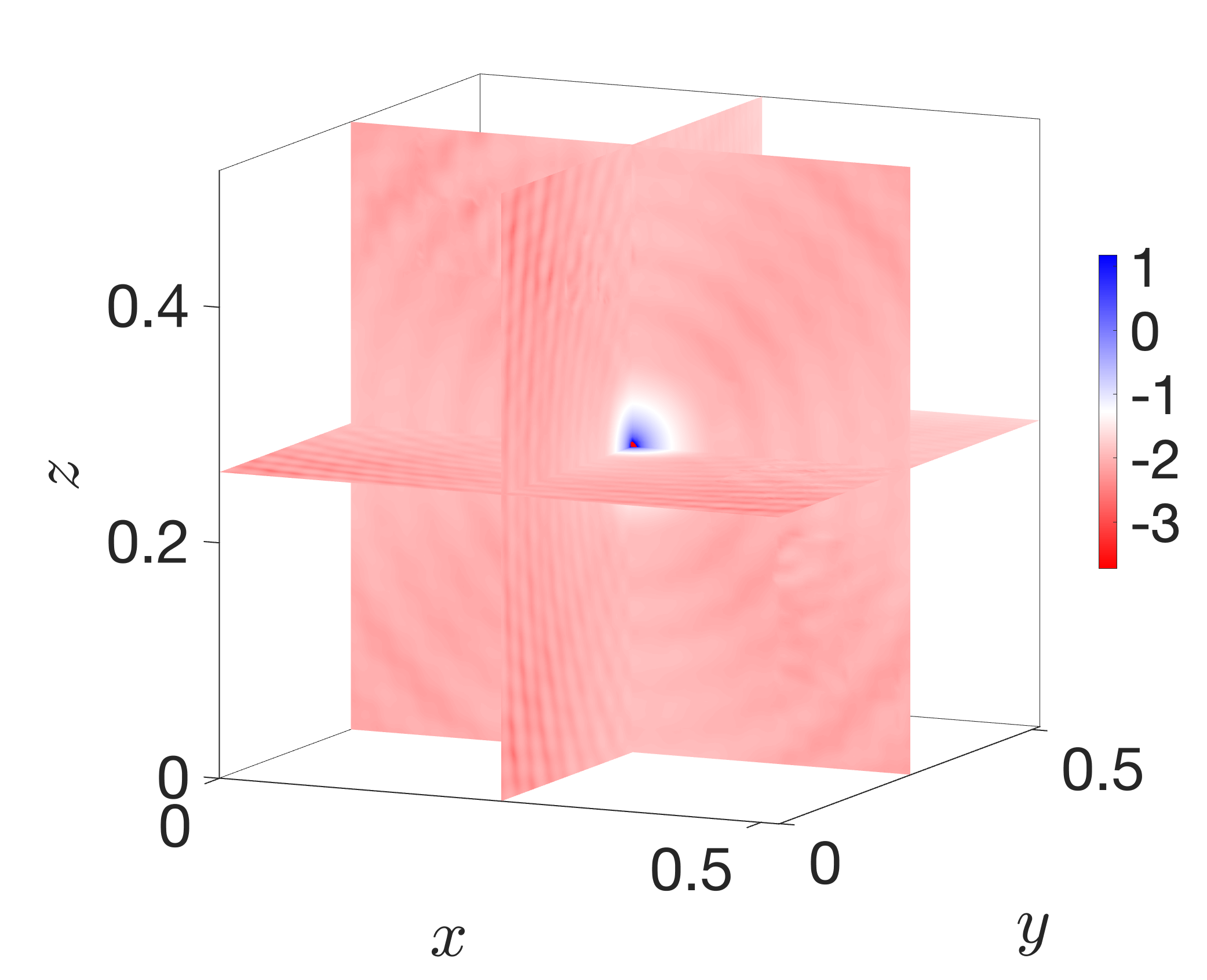}
	\end{subfigure}	
	\begin{subfigure}[t]{.37\textwidth}
		\centering
		\includegraphics[width=\linewidth]{\tfpath/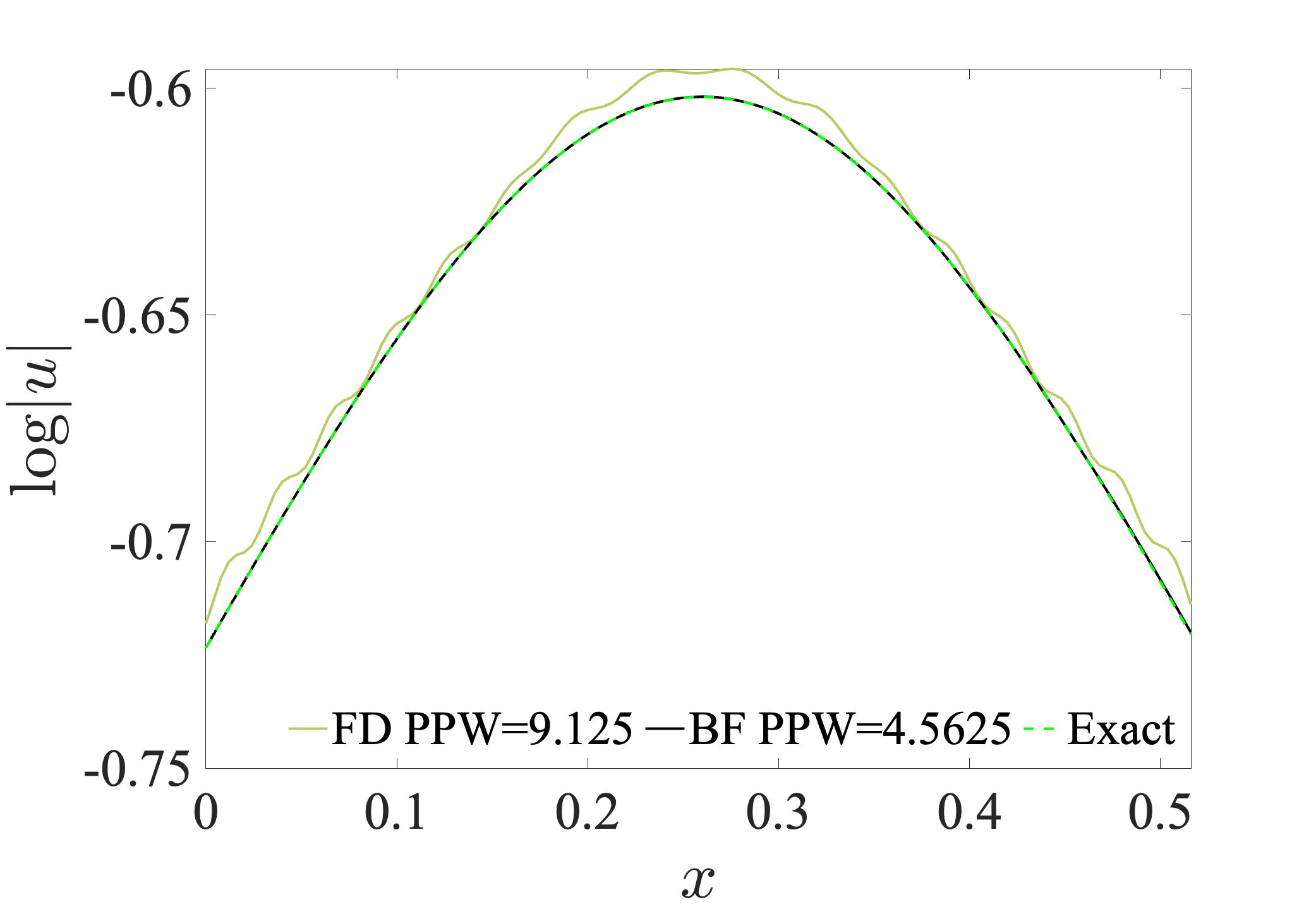}
	\end{subfigure}	
	\begin{subfigure}[t]{.29\textwidth}
		\centering
		\includegraphics[width=\linewidth]{\tfpath/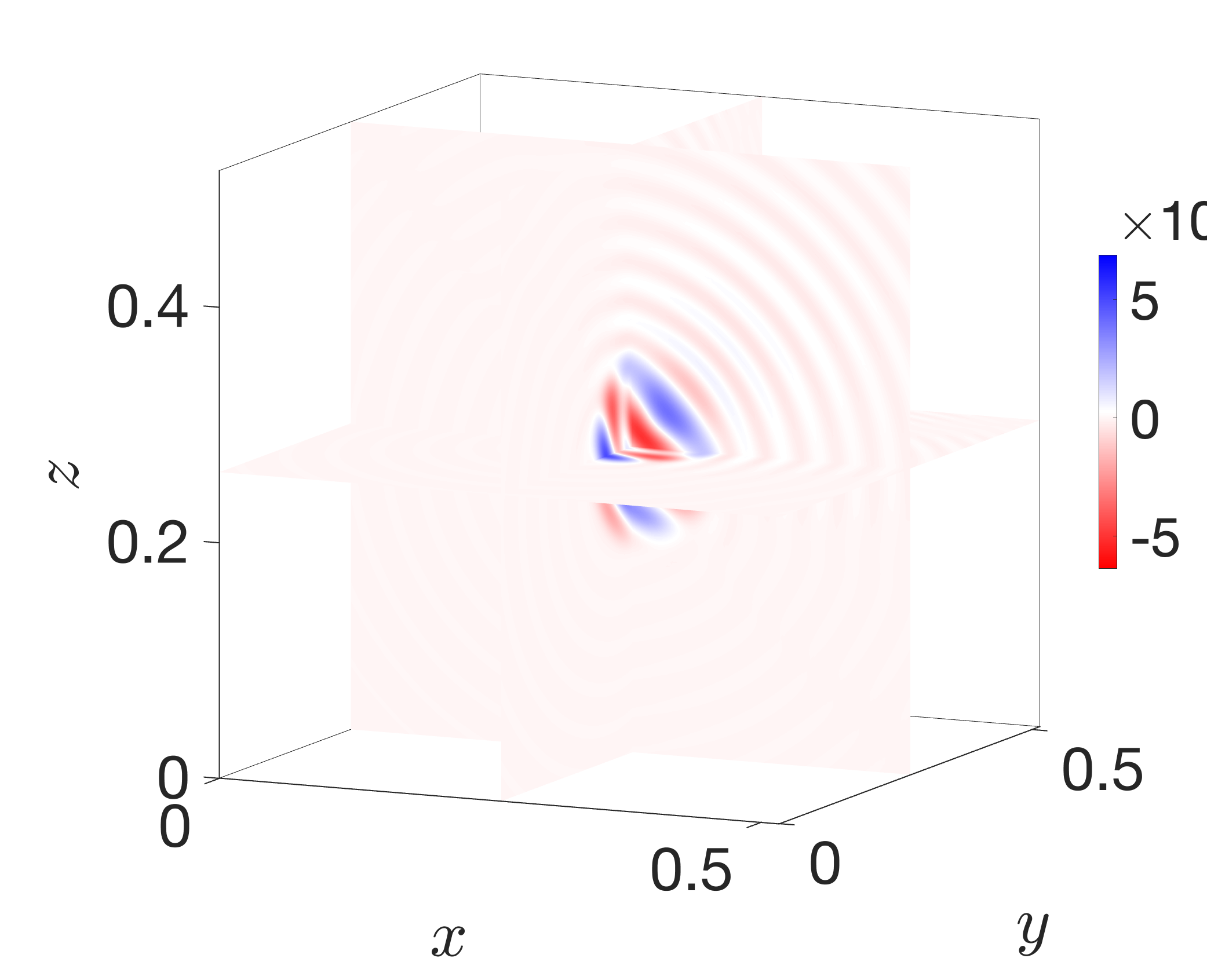}
	\end{subfigure}
	\begin{subfigure}[t]{.29\textwidth}
		\centering
		\includegraphics[width=\linewidth]{\tfpath/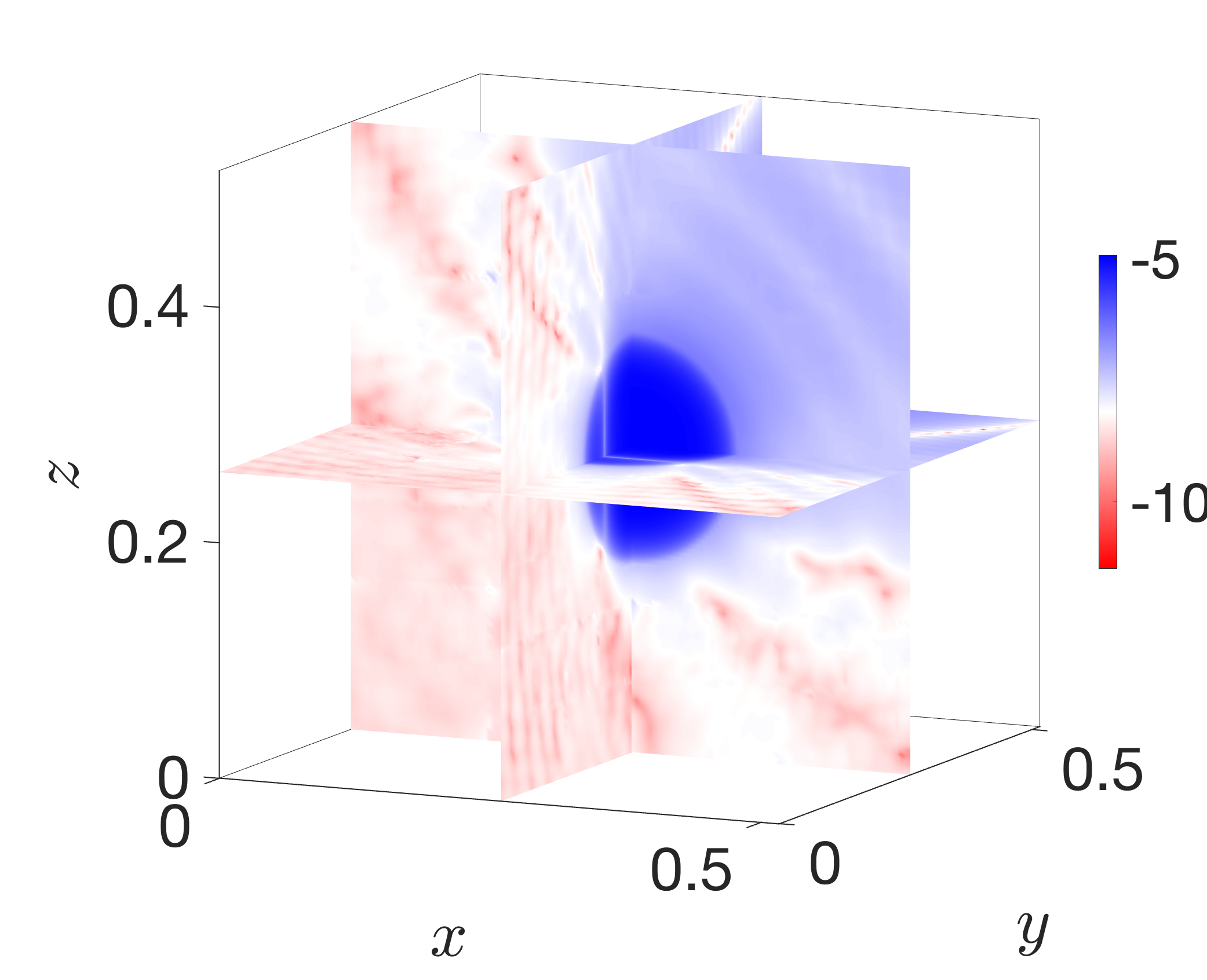}
	\end{subfigure}	
	\begin{subfigure}[t]{.37\textwidth}
		\centering
		\includegraphics[width=\linewidth]{\tfpath/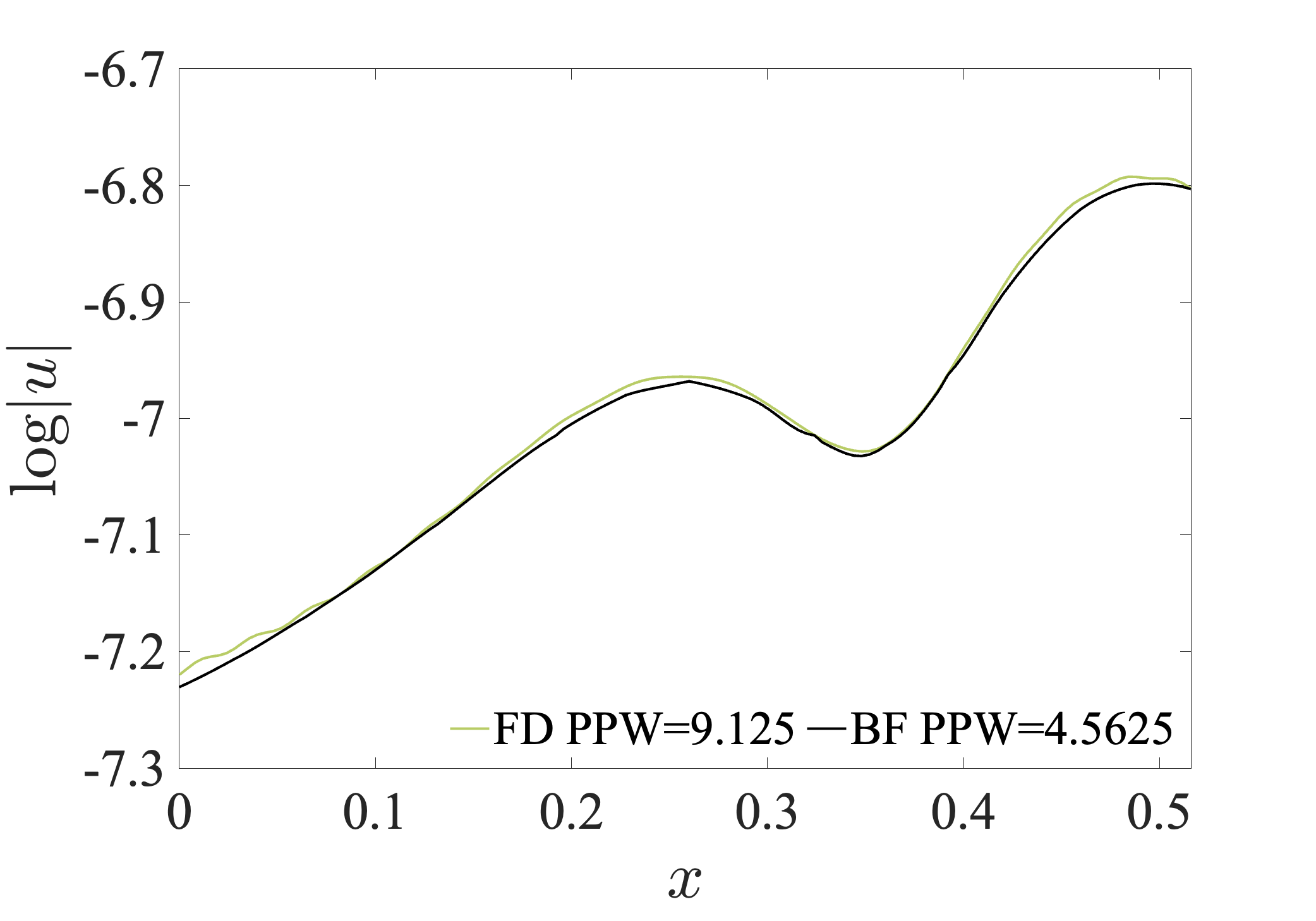}
	\end{subfigure}
	\begin{subfigure}[t]{.29\textwidth}
	\centering
	\includegraphics[width=\linewidth]{\tfpath/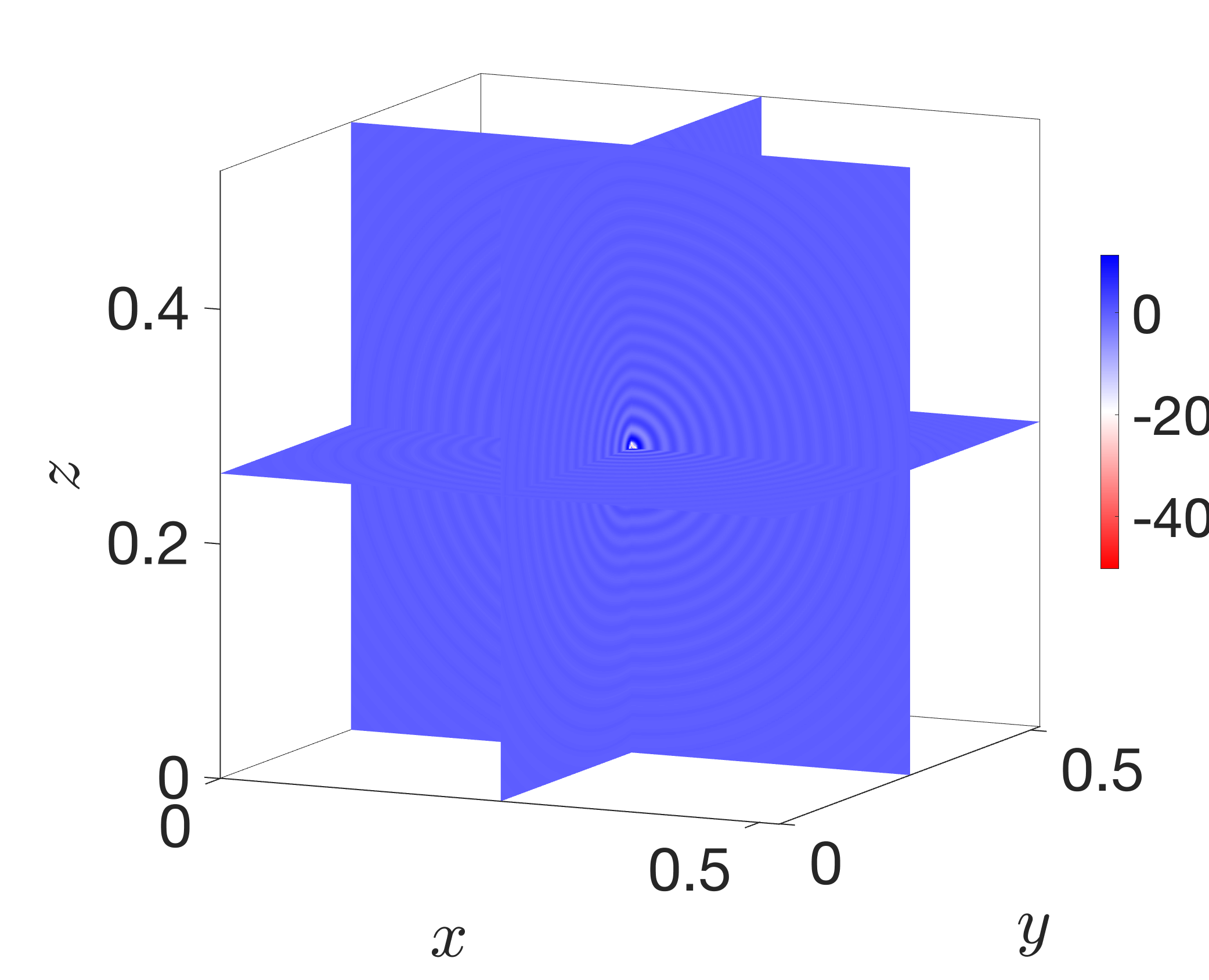}
\end{subfigure}
\begin{subfigure}[t]{.29\textwidth}
	\centering
	\includegraphics[width=\linewidth]{\tfpath/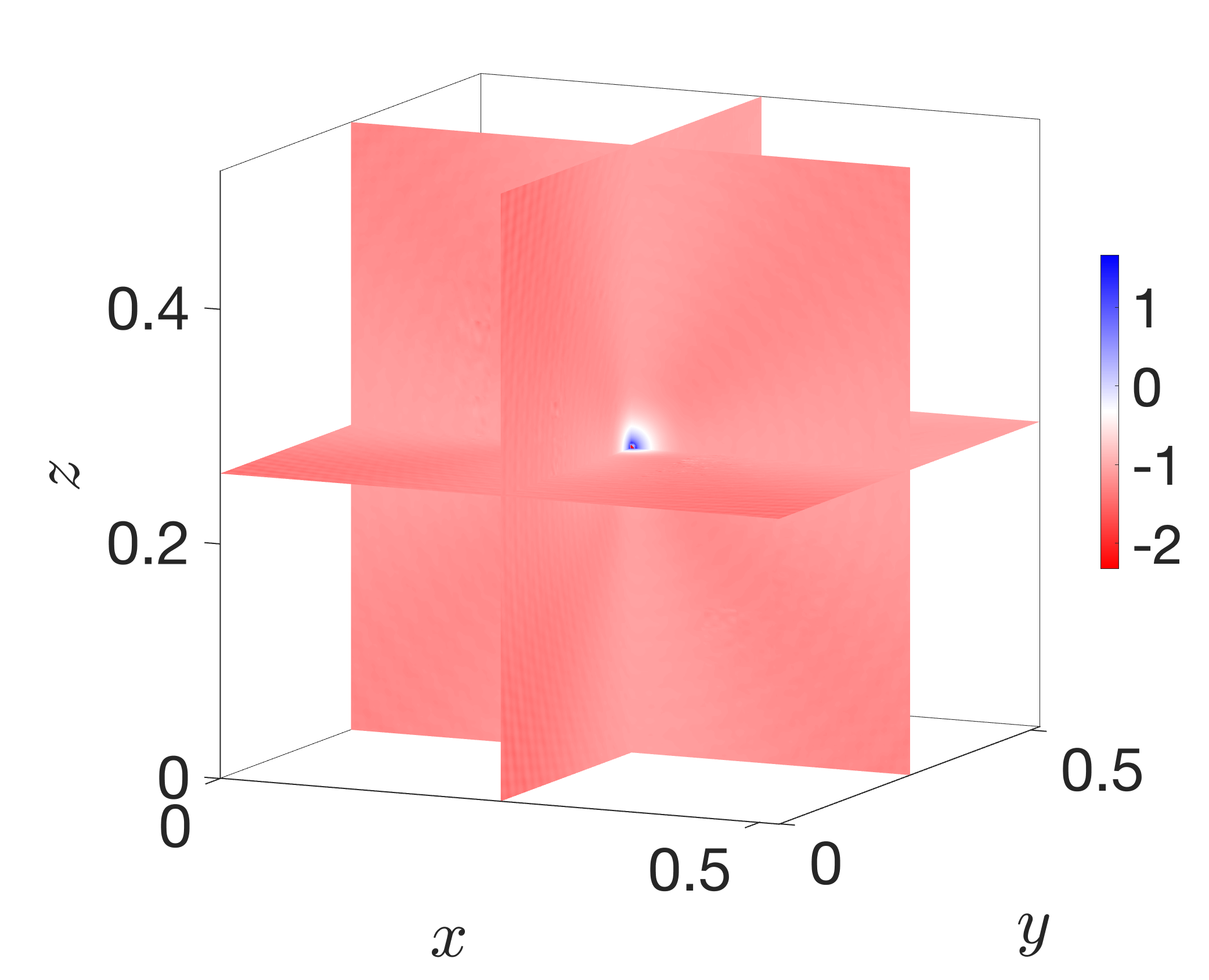}
\end{subfigure}	
\begin{subfigure}[t]{.37\textwidth}
	\centering
	\includegraphics[width=\linewidth]{\tfpath/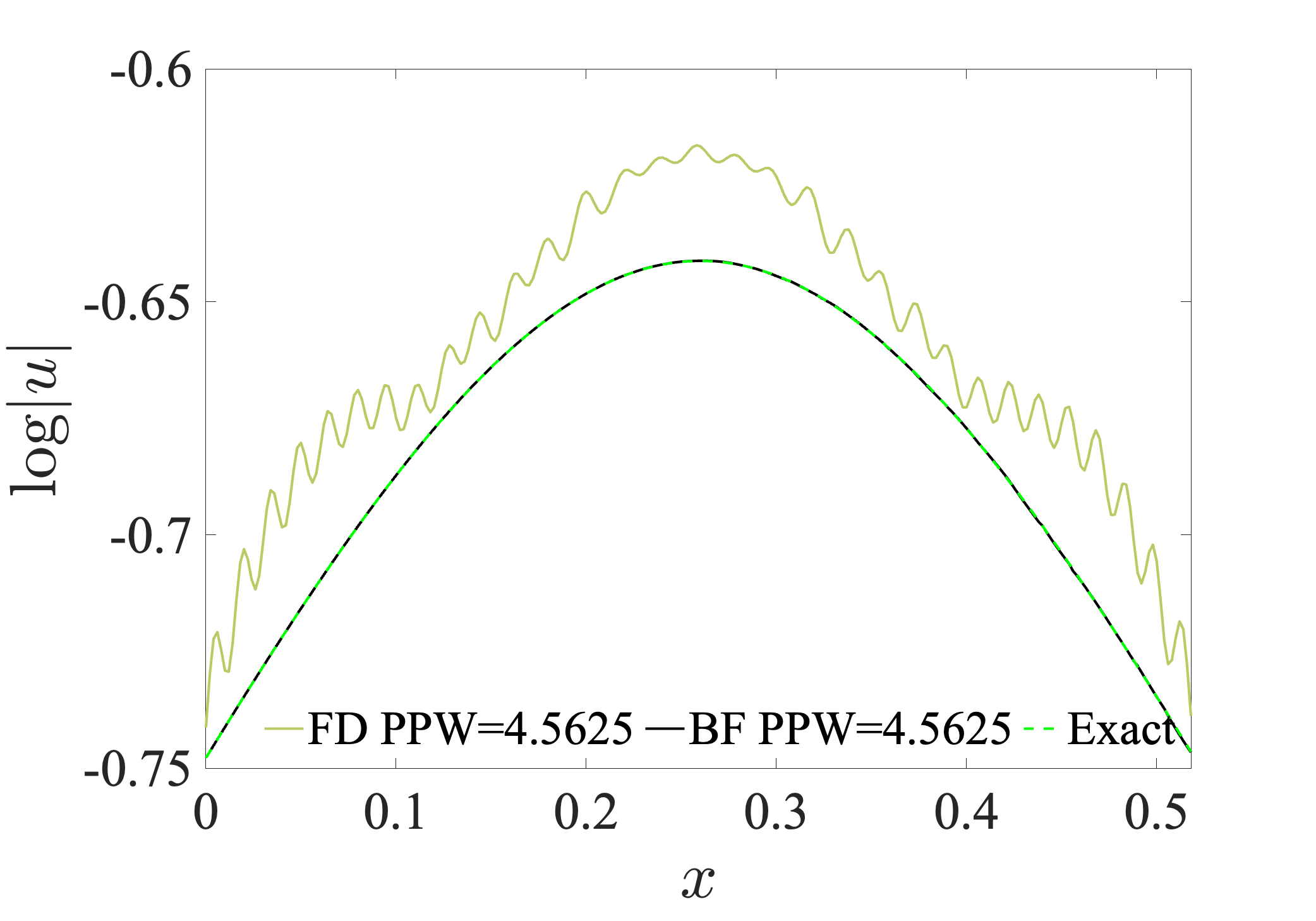}
\end{subfigure}		
	\vspace{-5pt}
	\caption{Constant-gradient media in $d=3$. Left column: the field $\mathrm{Re}(u_{\rm hb})$ (in linear scale) computed by the proposed scheme. Middle column: difference $|u_{\rm hb}-u_{\rm fd}|$ (in log scale) between the fields computed by the proposed scheme and FDFD. Right column: the field $|u_{\rm hb}|,|u_{\rm fd}|,|u_{true}|$ (in log scale) drawn along the line $y=10h$ and $z=10h$ with $h$ corresponding to PPW=4.56. Row 1: $\omega=32\pi$ (27 wavelengths each direction) with point source. Row 2: $\omega=32\pi$ (27 wavelengths each direction) with Gaussian packet source. \ylrevnew{Row 3: $\omega=64\pi$ (54 wavelengths each direction) with point source.}}	
	\label{fig:ex2_f16_3d}
\end{figure}

\subsection{Convergence test}
Next, we validate the convergence of high-order Lax-Friedrichs WENO schemes and the overall error estimates \eqref{totalerror} and \eqref{totalerrorAdded} of the HB ansatz using the 3-D constant-gradient model with a point source excitation, where $d=3$. As mentioned in \cref{sec:acc_3d}, both the phase $\tau$ and the Green's function in such a medium have exact formulas. 

First, errors of the phase computed by the first-, third- and fifth-order Lax-Friedrichs WENO schemes with varying $h_0$ are shown in \cref{fig:convergence3d} (left), which behave as $O(h_0^{\beta})$ with higher convergence order $\beta$ for higher order WENO schemes. When we apply the fifth-order Lax-Friedrichs WENO scheme to compute $\tau$, $v_0$ and $v_1$ are expected to have third-order and first-order accuracy, respectively, since $v_0$ and $v_1$ are computed from $\tau$. Because exact solutions of $v_0$ and $v_1$ are unknown, we will use the third- and first-order accuracy of $\tau$ as the reference accuracy for $v_0$ and $v_1$, respectively.   

Second, overall errors of wavefields (Green's functions) using the one-term ($N=0$) or two-term ($N=1$) HB ansatz \eqref{babich} with the HB ingredients computed by the fifth-order Lax-Friedrichs WENO scheme are shown in \cref{fig:convergence3d} (right), which behave as $O(\omega^{-1})$ and $O(\omega^{-2})$ as estimated by \eqref{totalerror} and \eqref{totalerrorAdded}, respectively, for the one-term and two-term expansions. 

When $N=1$, the HB coefficient $v_1$ only has first-order accuracy $O(h_0)$ which is dominant over the accuracy of $\tau$ and $v_1$, and thus the overall error $E_{\rm total}$ in \eqref{totalerrorAdded} reduces to 
\begin{equation}
\label{totalerrorN1}
   E_{\rm total}= O\left(({1}/{\omega})^{2}\right) + \ylrev{O(h_0^3) + O(\omega h_0^5) + O\left(\frac{h_0}{\omega}\right)},
\end{equation}
where the butterfly compression is not used and thus its error does not appear in the above. When \ylrev{the $O(({1}/{\omega})^{2})$ term is larger than the other terms combined}, the \ylrev{first error term} dominates so that we can observe the second-order asymptotic convergence in $1/{\omega}$ clearly; however, once $\omega$ is so large that the \ylrev{sum of three terms, $O(h_0^3)+O(\omega h_0^5) + O(\frac{h_0}{\omega})$,} dominates, the overall error nearly saturates since the \ylrev{$O(\omega h_0^5)$} term \ylrev{increases slowly} as $\omega$ does. Such convergence behavior can be seen clearly in \cref{fig:convergence3d}. 

When $N=0$, $v_1$ disappears in the HB expansion and $v_0$ has third-order accuracy $O(h_0^3)$ which is dominant over that of $\tau$; hence, the overall error $E_{\rm total}$ in \eqref{totalerror} reduces to 
\begin{equation}
\label{totalerrorN0}
   E_{\rm total}= O({1}/{\omega}) + O(h_0^3)+\ylrev{O(\omega h_0^5)},
\end{equation}
where the butterfly compression is not used and thus its error does not appear in the above. When \ylrev{the $O({1}/{\omega})$ term is larger than the other terms combined}, then the \ylrev{first error term} dominates so that we can observe the first-order asymptotic convergence in $1/{\omega}$ clearly for a much broader band of $\omega$. Such convergence behavior can be seen clearly in \cref{fig:convergence3d}. 

\begin{figure}[!tp]
	\centering
	\vspace{-7.5pt}
	\begin{subfigure}[t]{.49\textwidth}
		\centering
		\includegraphics[width=\linewidth]{\tfpath/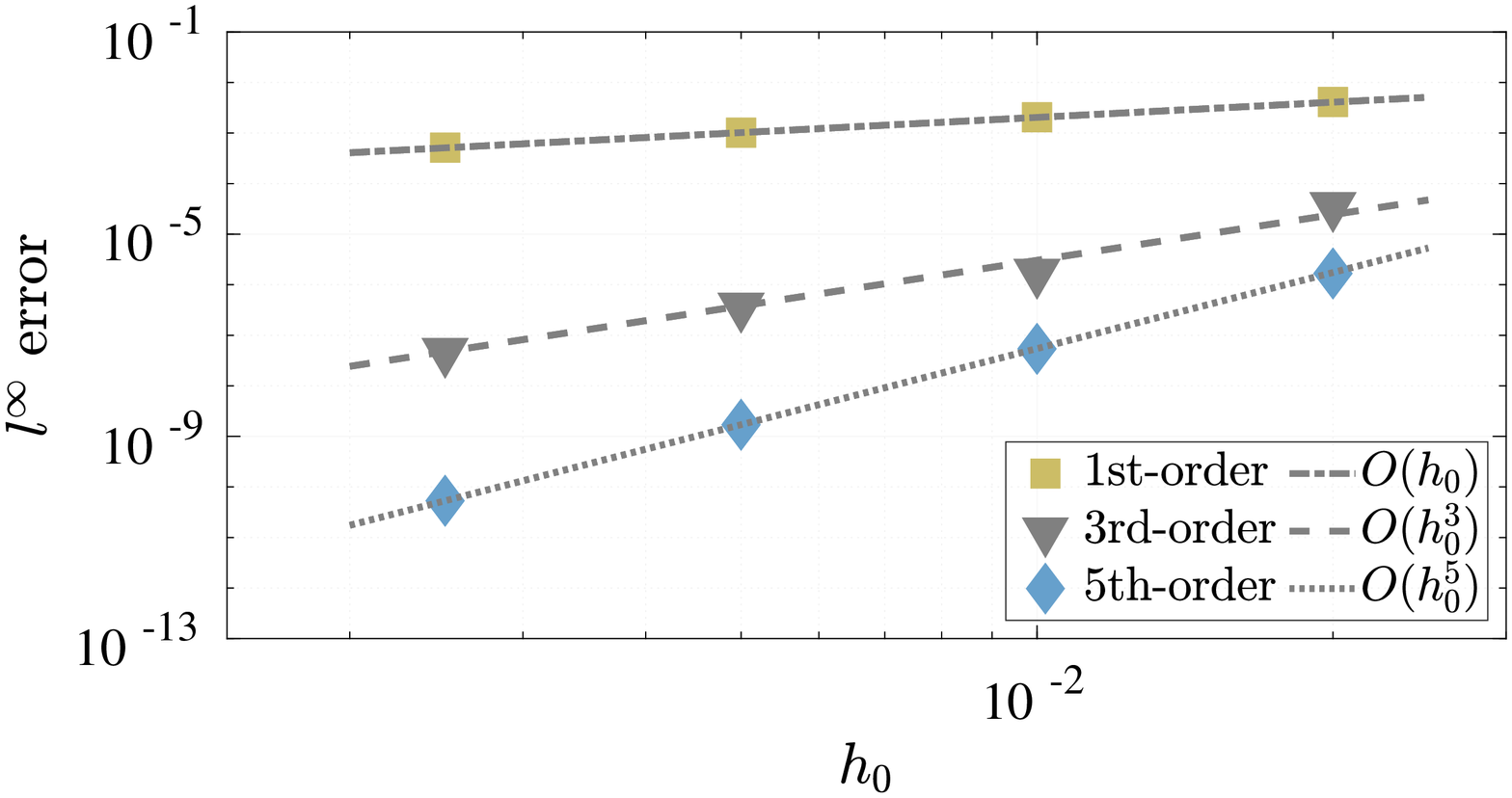}
	\end{subfigure}
	\begin{subfigure}[t]{.49\textwidth}
		\centering
		\includegraphics[width=\linewidth]{\tfpath/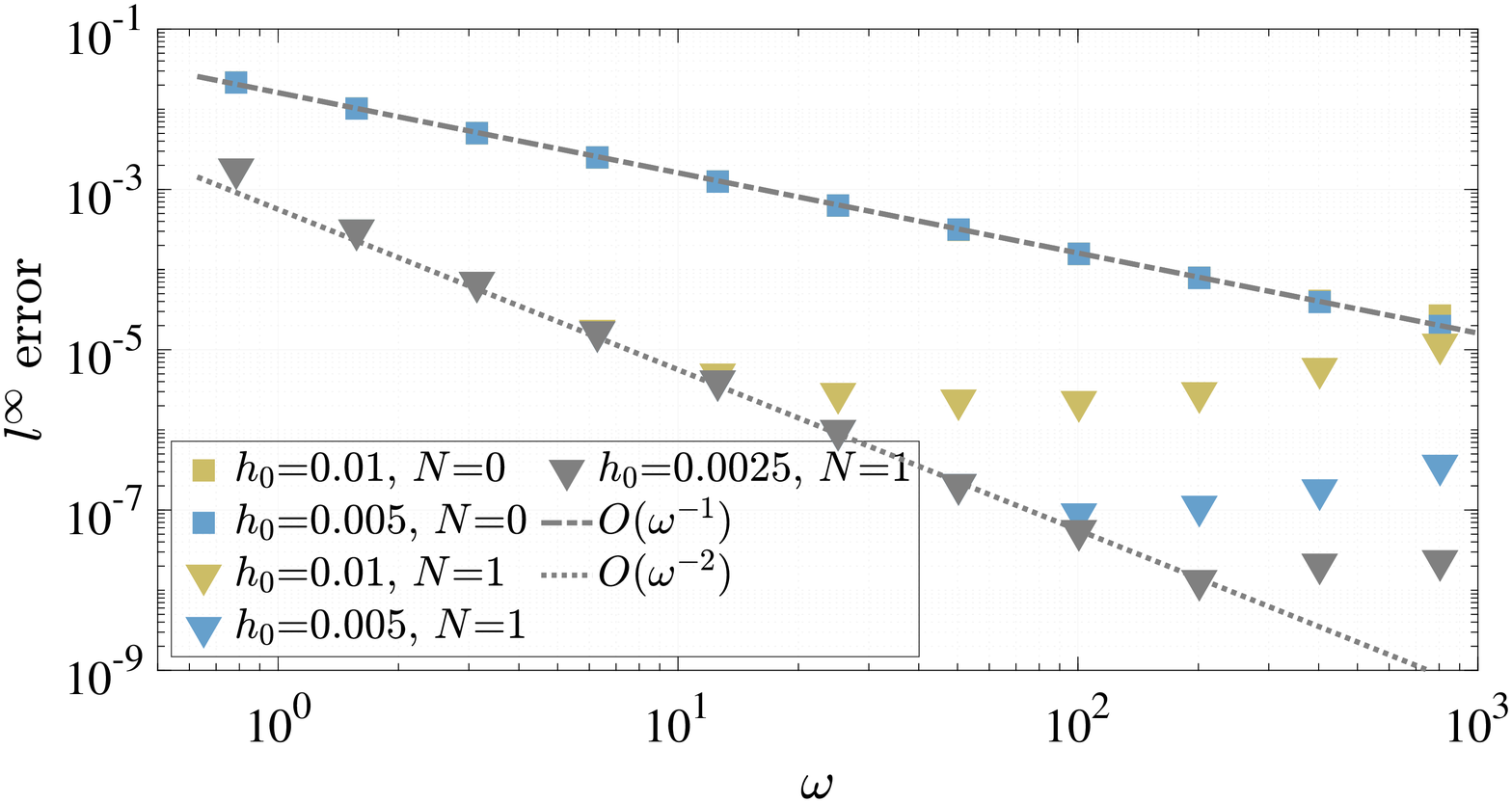}
	\end{subfigure}		
	\vspace{-5pt}
	\caption{Convergence test for a constant-gradient model in $d=3$ with a point source. Left:  Errors (w.r.t. the exact solution) of the phase function computed by the first-, third-, and fifth-order Lax-Friedrichs WENO schemes with varying $h_0$. Right: Errors (w.r.t. the exact solution; see \eqref{eq:inferror}) of the wavefield using the one-term ($N=0$) or two-term ($N=1$) HB ansatz \eqref{babich}, where the phase function is computed by the fifth-order Lax-Friedrichs WENO scheme.}
	\label{fig:convergence3d}
\end{figure}

\subsection{Complexity validation}
In this subsection, we validate the CPU and memory complexities of the proposed scheme claimed in subsections \ref{sec:v2v}, \ref{sec:s2v}, and \ref{sec:s2s} using the constant and constant-gradient medium for $d=2$ and $d=3$.   
\subsubsection{2D domains}
For the constant medium, we consider the computational domain $[0,2]^2$ with an open square inclusion of side length $0.8m$, as described in \cref{sec:acc_2d}. The domain and the inclusion are discretized with PPW $n_p=10$ and $n_p=500$, respectively. We vary the frequency and \ylrev{cell count} from $\omega=10\pi$ and $N_v=40,401$ to $\omega=320\pi$ and $N_v=40,972,801$, respectively. Note that $\omega=320\pi$ corresponds to $640$ wavelengths per direction. Each simulation uses 64 Cori Haswell nodes. The CPU time and memory requirement for computing $\mat{K}^{v2v}$, $\mat{K}^{s2v}$ and $\mat{K}^{s2s}$ (and its inverse) are plotted in \cref{fig:cc2d} (top). Note that the value of $N_s$ for each sample of $N_v$ is not shown here. For the memory requirement, all discretized integral operators scale at most as $O(N_v\log^2N_v)$ as expected. For the computational time, $\mat{K}^{v2v}$ and $\mat{K}^{s2s}$ scale at most as $O(N_v\log^2N_v)$, and $\mat{K}^{s2v}$ scales as $O(N_v^{7/6})$, which can be further improved by additional matrix partitioning or analytical interpolation-based compression. We note that the time for $\mat{K}^{s2v}$ is about 8 times faster than $\mat{K}^{v2v}$ when $\omega=320\pi$. 

For the constant-gradient medium, we consider the computational domain $[0,1]^2$ with an open square inclusion of side length $0.8$, as described in \cref{sec:acc_2d}. The domain and the inclusion are discretized with PPW $n_p=10$ and $n_p=500$, respectively. We vary the frequency and \ylrev{cell count} from $\omega=12.5\pi$ and $N_v=63,001$ to $\omega=200\pi$ and $N_v=16,008,001$, respectively. Note that $\omega=200\pi$ corresponds to $400$ wavelengths per direction. Each simulation uses 64 Cori Haswell nodes. The CPU time and memory requirement for computing $\mat{K}^{v2v}$, $\mat{K}^{s2v}$ and $\mat{K}^{s2s}$ (and its inverse) are plotted in \cref{fig:cc2d} (bottom). Just like the constant medium, the memory requirement and CPU time mostly scale as at most $O(N_v\log^2N_v)$. 

\subsubsection{3D domains}	
For the constant medium, we consider the computational domain $[0,0.5]^3$ as described in \cref{sec:acc_3d}. The domain is discretized with PPW $n_p=5$. We vary the frequency and \ylrev{cell count} from $\omega=10\pi$ and $N_v=26^3=17576$ to $\omega=80\pi$ and $N_v=201^3=8,120,601$, respectively. Note that $\omega=80\pi$ corresponds to $40$ wavelengths per direction. Each simulation uses 64 Cori Haswell nodes. The CPU time and memory requirement for computing $\mat{K}^{v2v}$ are plotted in \cref{fig:cc3d} (top). As estimated in \cref{sec:v2v}, the memory requirement scales as $O(N_v\log^2N_v)$, and the CPU time scales as $O(N_v^{4/3})$. From \cref{sec:v2v}, the matrix entry computation requires $O(N_v\log^2N_v)$ and the IDs require $O(N_v^{4/3})$. Both theoretical curves are plotted in \cref{fig:cc3d} (top). This sub-optimal CPU complexity for the IDs can be further improved via analytical interpolation schemes.  

For the constant-gradient medium, we consider the domain $[0,0.52]^3$ as described in \cref{sec:acc_3d}. The domain is discretized with PPW $n_p=4.56$. We vary the frequency and \ylrev{cell count} from $\omega=8\pi$ and $N_v=36^3=46656$ to $\omega=64\pi$ and $N_v=262^3=17,984,728$, respectively. Note that $\omega=64\pi$ corresponds to $54$ wavelengths per direction. 
Each simulation uses 64 Cori Haswell nodes. The CPU time and memory requirement for computing $\mat{K}^{v2v}$ are plotted in \cref{fig:cc3d} (top). The conclusion is very similar to the case of constant media.

\begin{figure}[!htp]
	\centering
	\vspace{-7.5pt}
	\begin{subfigure}[t]{.49\textwidth}
		\centering
		\includegraphics[width=\linewidth]{\fpath/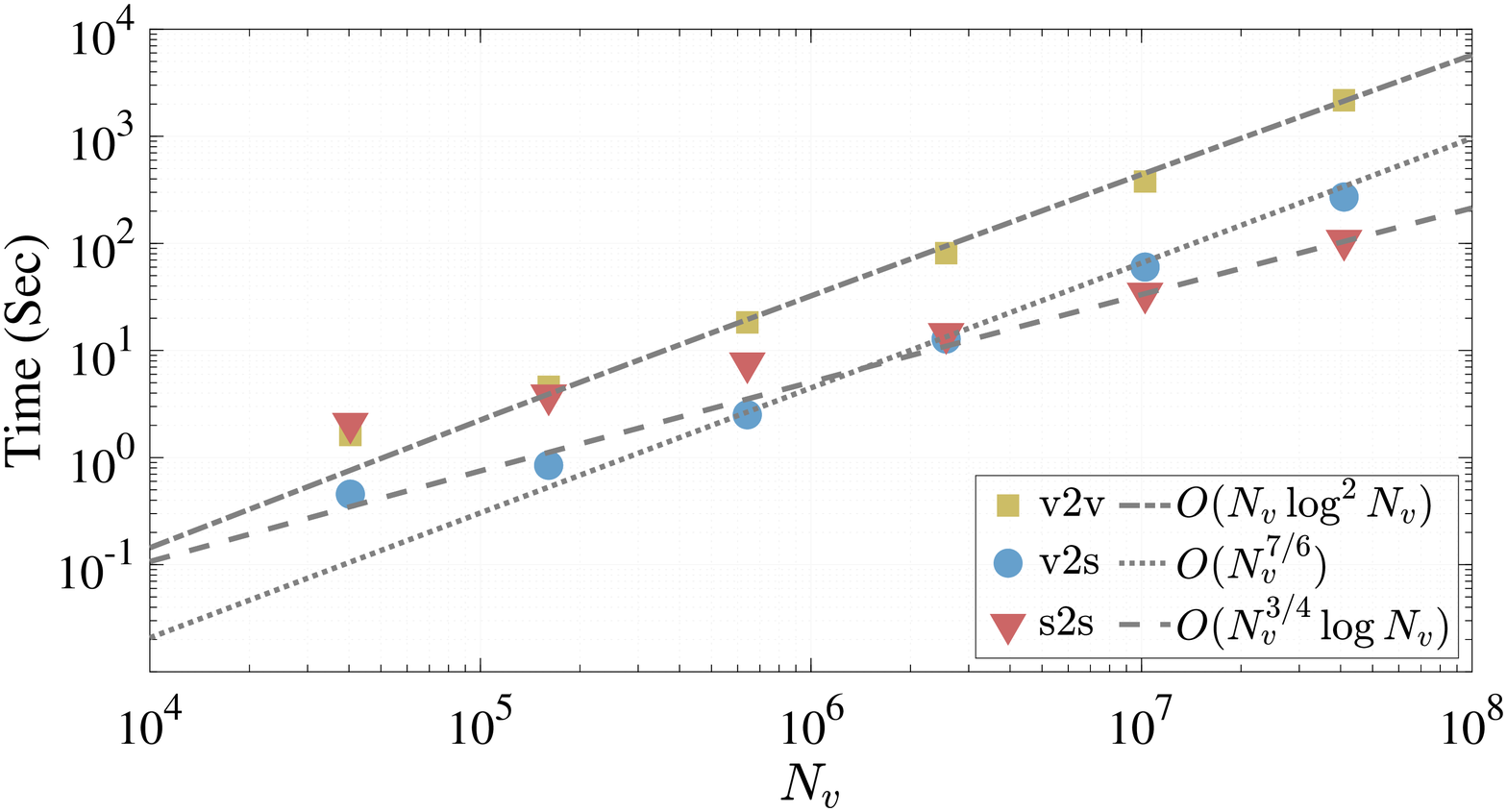}
	\end{subfigure}
	\begin{subfigure}[t]{.49\textwidth}
		\centering
		\includegraphics[width=\linewidth]{\fpath/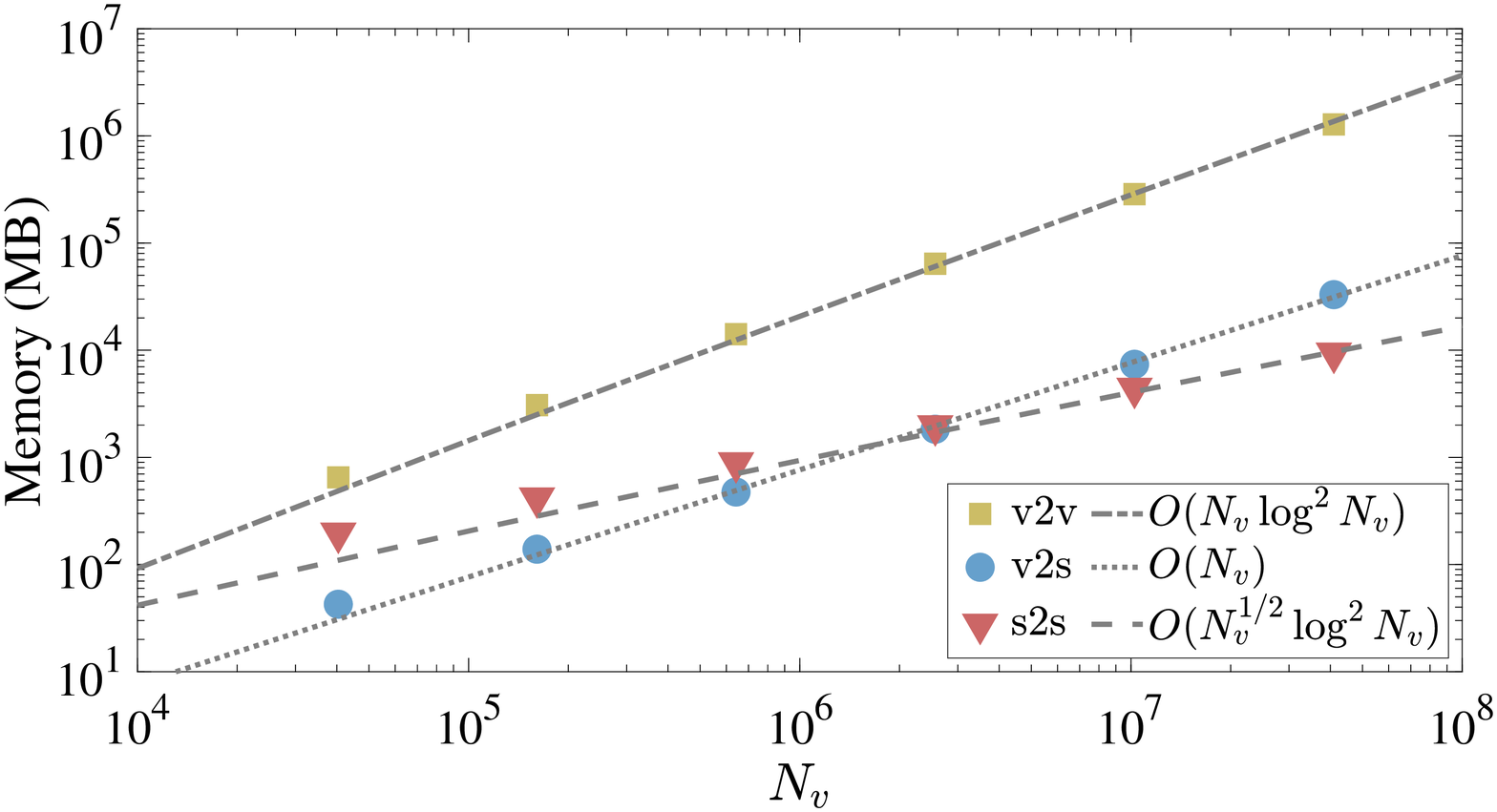}
	\end{subfigure}	
	\begin{subfigure}[t]{.49\textwidth}
	\centering
	\includegraphics[width=\linewidth]{\fpath/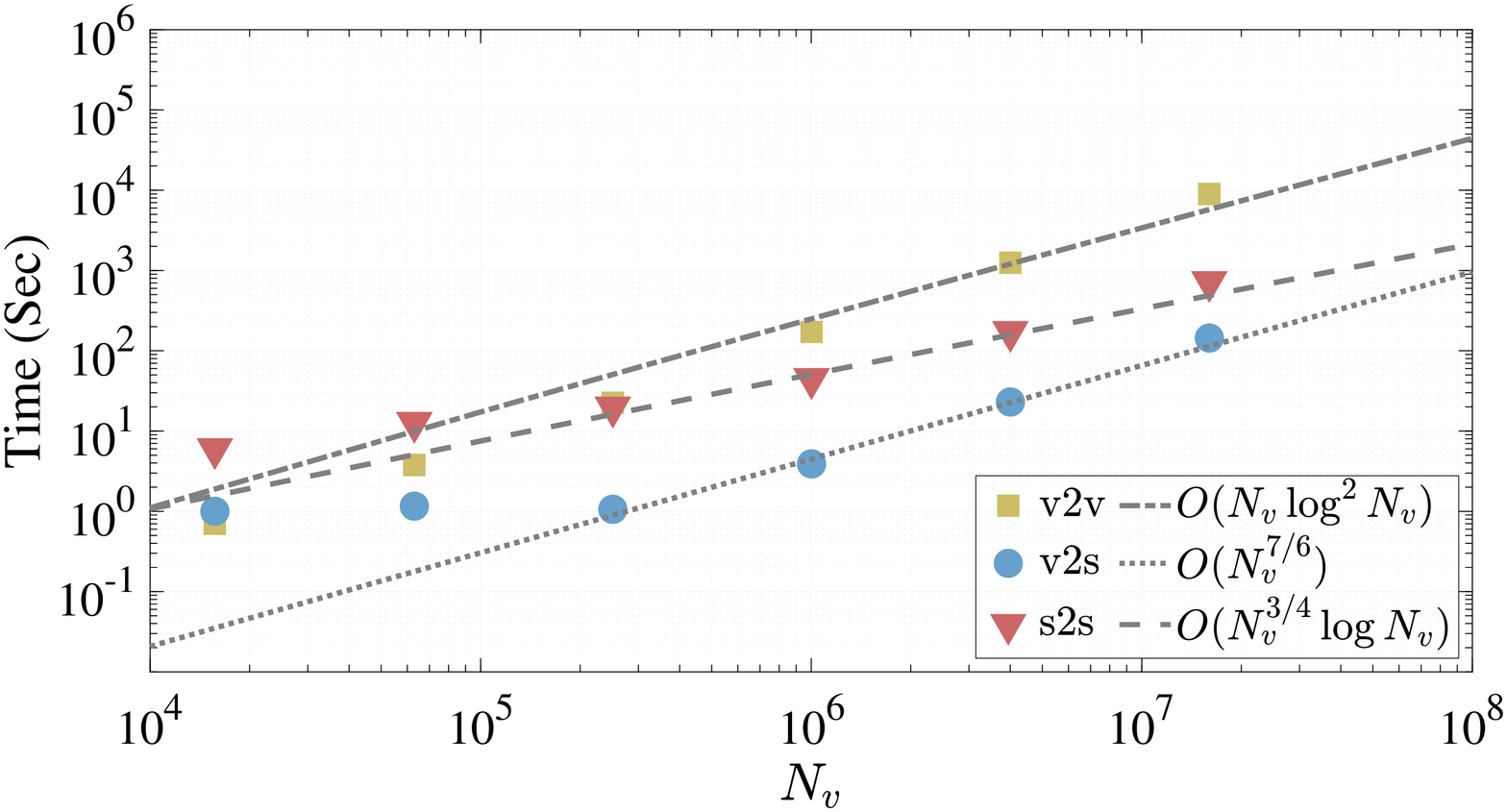}
\end{subfigure}
\begin{subfigure}[t]{.49\textwidth}
	\centering
	\includegraphics[width=\linewidth]{\fpath/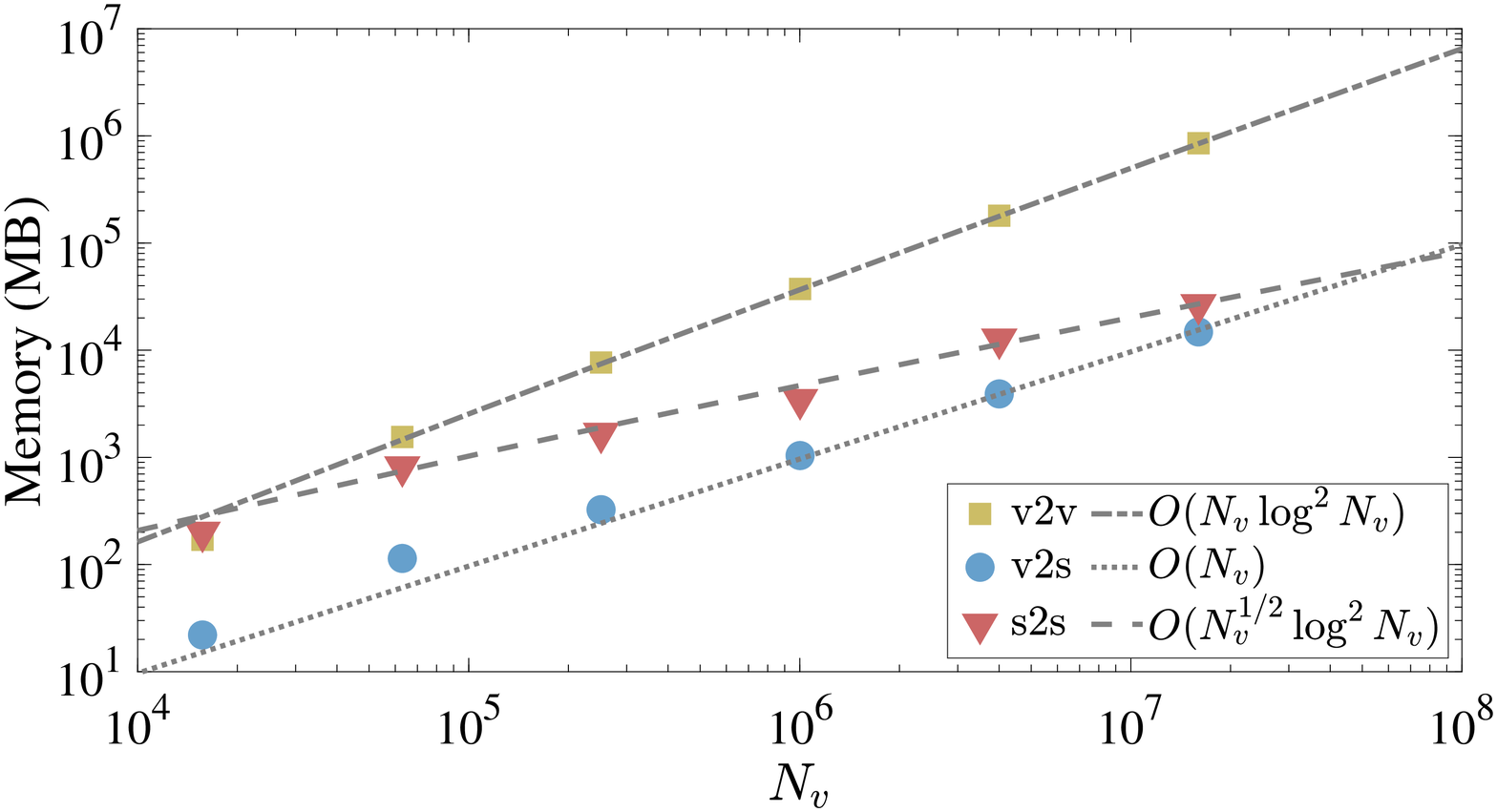}
\end{subfigure}	
	
	\vspace{-5pt}
	\caption{CPU time (left) and storage units (right) for computing $\mat{K}^{v2v}$, $\mat{K}^{s2v}$ and $\mat{K}^{s2s}$ (and its inverse) with problem dimension $d=2$. (Top): constant media with an open square inclusion. The largest data point corresponds to $640$ wavelengths per direction. \ylrev{(Bottom)}: constant-gradient media with an open square inclusion. The largest data point corresponds to $400$ wavelengths per direction.}
	\label{fig:cc2d}
   \end{figure}

\begin{figure}[!htp]
	\centering
	\vspace{-7.5pt}
	\begin{subfigure}[t]{.49\textwidth}
		\centering
		\includegraphics[width=\linewidth]{\tfpath/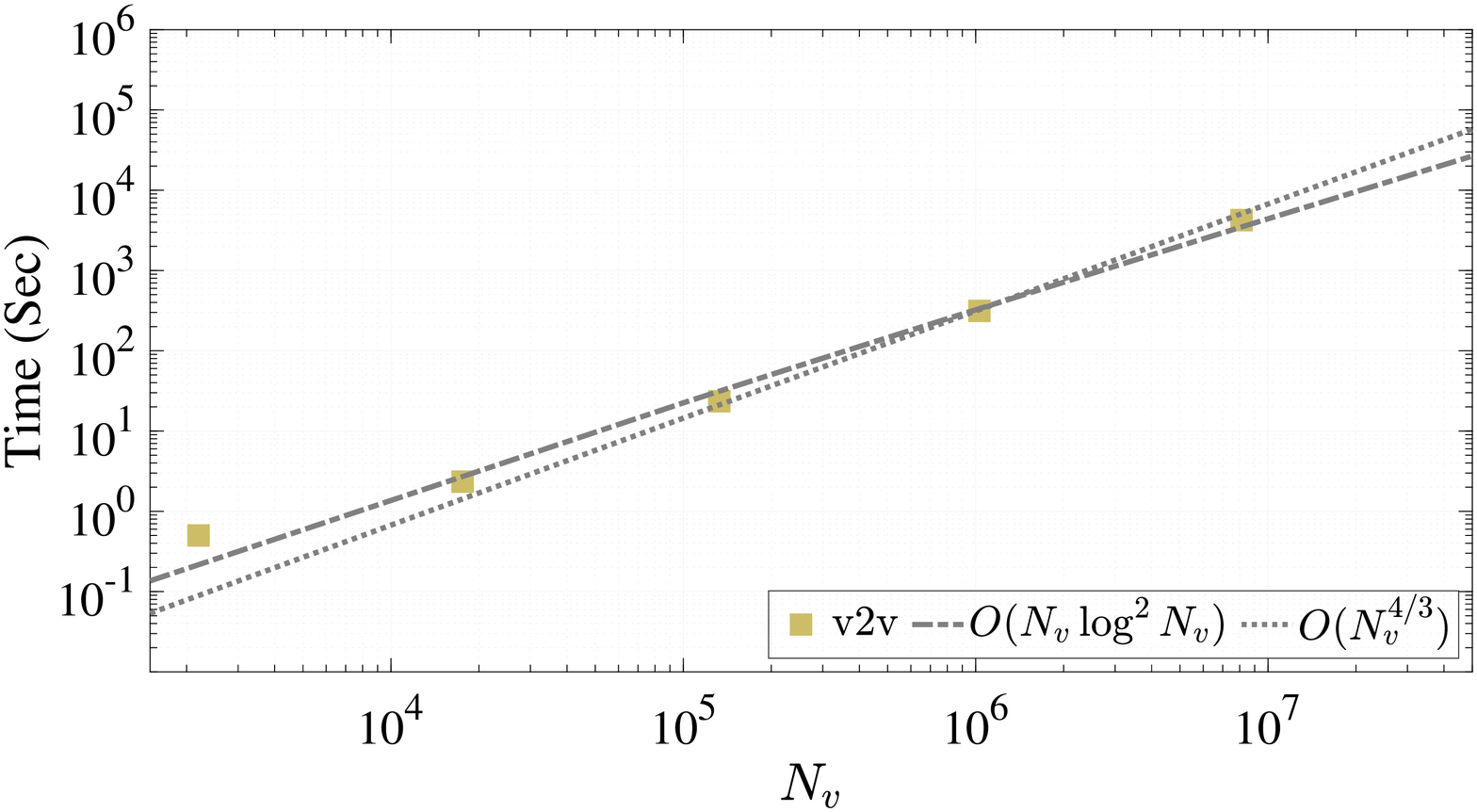}
	\end{subfigure}
	\begin{subfigure}[t]{.49\textwidth}
		\centering
		\includegraphics[width=\linewidth]{\tfpath/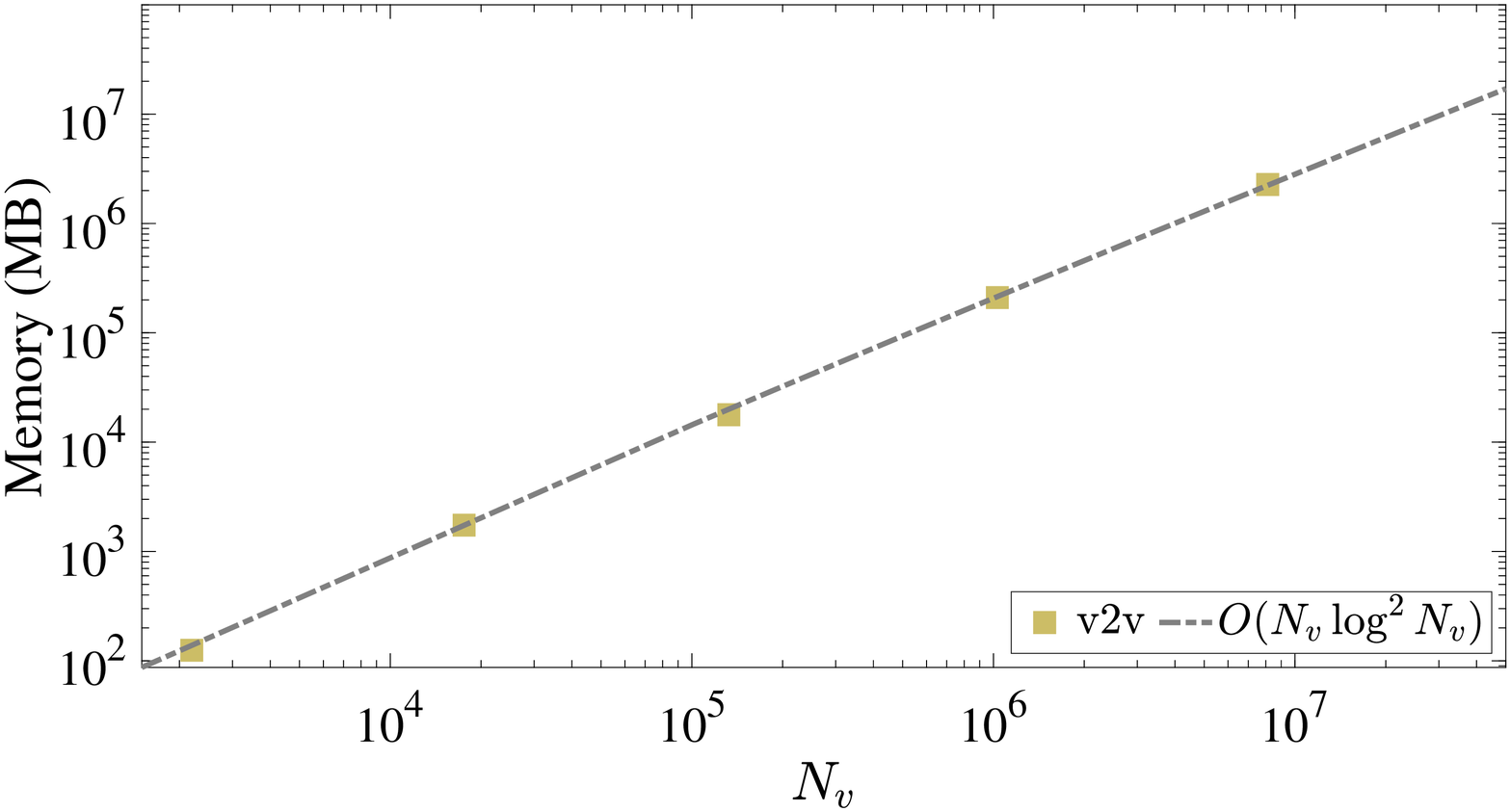}
	\end{subfigure}	
	\begin{subfigure}[t]{.49\textwidth}
	\centering
	\includegraphics[width=\linewidth]{\tfpath/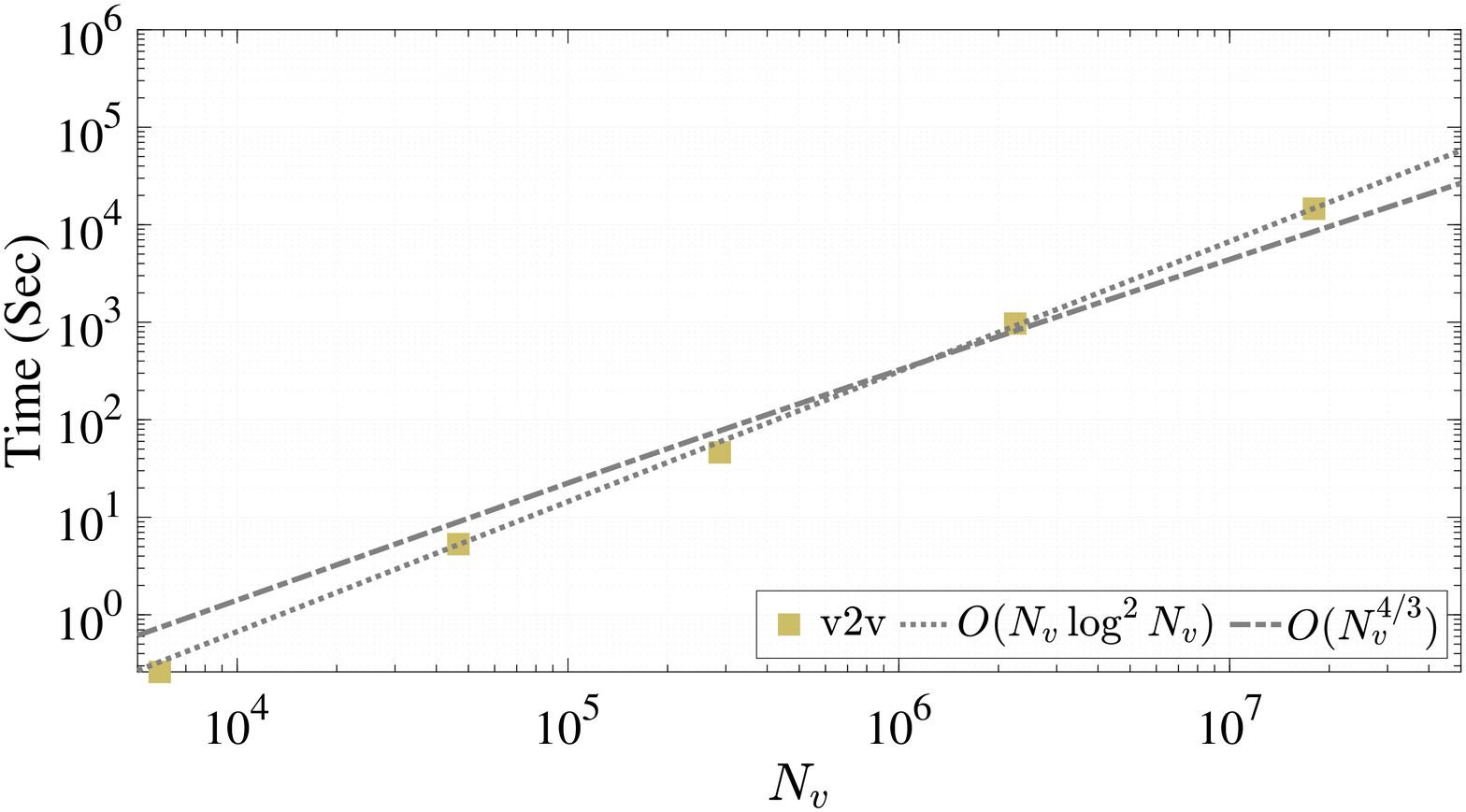}
\end{subfigure}
\begin{subfigure}[t]{.49\textwidth}
	\centering
	\includegraphics[width=\linewidth]{\tfpath/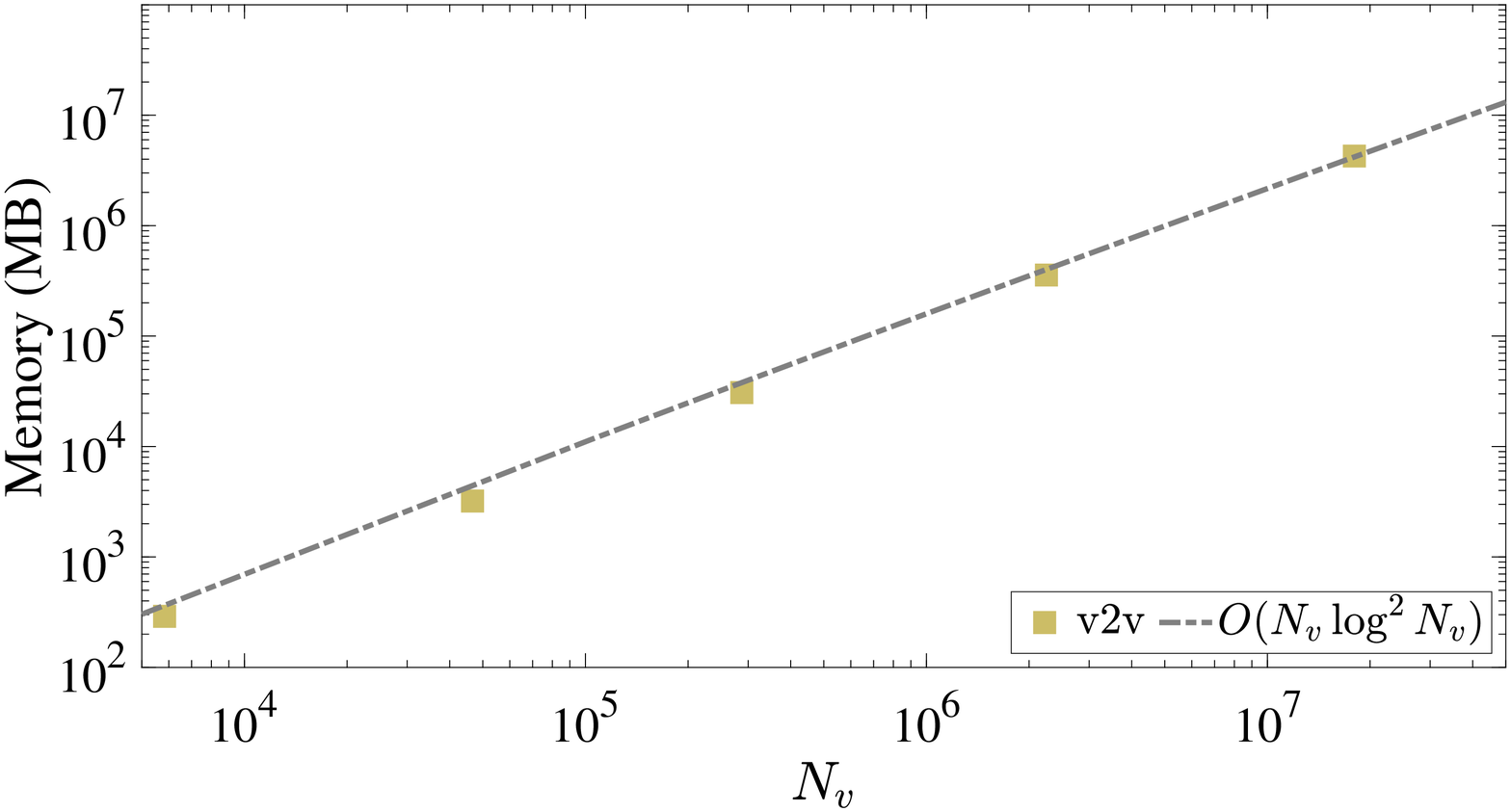}
\end{subfigure}	
	
	\vspace{-5pt}
	\caption{CPU time (left) and storage units (right) for computing $\mat{K}^{v2v}$ with problem dimension $d=3$. \ylrev{(Top)}: constant media. The largest data point corresponds to $40$ wavelengths per direction. \ylrev{(Bottom)}: constant-gradient media. The largest data point corresponds to $54$ wavelengths per direction.}
	\label{fig:cc3d}
\end{figure}

\section{Conclusion\label{sec:conclude}} 
We present a fast and accurate scheme based on the Hadamard-Babich integrator for solving high-frequency Helmholtz equations in smooth, inhomogeneous media with arbitrary sources. The scheme low-rank compresses the phase and HB coefficients in the HB integrator with Chebyshev interpolation by solving their governing equations via Lax-Friedrichs WENO schemes with point sources located at the Chebyshev nodes. Once compressed, the phase and HB coefficients are used in the butterfly and HODBF compression of the resulting HB integrator discretized using $N_v$ \ylrev{cells}. Construction and application of the HB integrator require $O(N_v\log^2N_v)$ CPU time and storage units. The scheme can also handle scattering problems involving sound-hard inclusion in the computational domain. In addition, the new scheme requires a much smaller number of discretization points per wavelength compared to finite-difference solvers. As a result, the scheme can model wave propagation for inhomogeneous media in so-far the largest 2D and 3D domains in terms of wavelength volume \ylrev{on a state-the-art supercomputer at Lawrence Berkeley National Laboratory}. Future research direction includes extension of the proposed scheme to non-smooth media or those media permitting presence of caustics, as well as to Maxwell's equations.

	\section*{Acknowledgements}
	This research was supported in part by the Exascale Computing Project (17-SC-20-SC), a collaborative effort of the U.S.~Department of Energy Office of Science and the National Nuclear Security Administration, and in part by the U.S.~Department of Energy, Office of Science, Office of Advanced Scientific Computing Research, Scientific Discovery through Advanced Computing (SciDAC) program through the FASTMath Institute under Contract No.~DE-AC02-05CH11231 at Lawrence Berkeley National Laboratory. This research used resources of the National Energy Research Scientific Computing Center (NERSC), a U.S.~Department of Energy Office of Science User Facility operated under Contract No.~DE-AC02-05CH11231. Qian is partially supported by NSF (grants \#2012046 and \#2152011). Qian is grateful to Lexing Ying for his insightful comments and suggestions on this project. We are also grateful to anonymous reviewers for constructive comments and suggestions.  
	
\appendix


\section{Computation of Self-interaction Terms}\label{sec:self}
\subsection{2-D self-interaction terms}
By the H-B ansatz \eqref{babich}, we need to integrate the leading order term, which reduces to integrating   
\begin{eqnarray}
f_0(\omega,\tau) &=& i\frac{\sqrt{\pi}}{2}H_0^{(1)}(\omega\tau(\mathbf{r}, \mathbf{r}_0)) \nonumber\\
&=& i\frac{\sqrt{\pi}}{2}H_0^{(1)}(\omega n(\mathbf{r})|\mathbf{r}- \mathbf{r}_0|)
\end{eqnarray}
over the cell $c_j$ of size $h$ with respect to $\mathbf{r}_0$, where $\mathbf{r}$ is the center of $c_j$. Here the leading HB coefficient is constant and is suppressed for now. 

\begin{lemma}
The integral 
\[I(\mathbf{r}) = \int_{c_j}H_0^{(1)} (\omega n(\mathbf{r})|\mathbf{r}- \mathbf{r}_0|)d\mathbf{r}_0\]
can be reduced to evaluating 
\begin{eqnarray}
I(\mathbf{r}) &=&\frac{1}{(n_0\omega)^2}\left[8\int_{0}^{\frac{\pi}{4}}{\frac{hn_0\omega}{2\cos\theta}}H_1^{(1)}\left({\frac{hn_0\omega}{2\cos\theta}}\right)d\theta+{4i} \right],
\label{HankelIntegral2dlemma}
\end{eqnarray}
where $n_0=n(\mathbf{r})$, and we need to use the Bessel function to evaluate the above integral.
\end{lemma}
\begin{proof}
The integral $I$ can be reduced to the integration over a cell of size $h$ centered at the origin. We further partition this cell into eight equal triangles. By using the geodesic polar coordinates centered at the origin, the integral $I$ can be reduced to evaluating the following integral over one triangle, 
\begin{eqnarray} 
I(\mathbf{r}) &=& 8\int_{0}^{\frac{\pi}{4}}d\theta\int_{0}^{\frac{h}{2\cos\theta}}H_0^{(1)}(\omega n_0 r)r dr.  
\end{eqnarray}

Using the following relation from \cite{abrste65}, formula 9.1.30, 
\[\frac{d}{dt}\left(t\;H_1^{(1)}(t)\right)=t\;H_0^{(1)}(t),\] 
we consider the integral 
\begin{eqnarray} 
\int_{0}^{\frac{h}{2\cos\theta}}H_0^{(1)}(\omega n_0 r)\;r\; dr 
&=&\frac{1}{(n_0\omega)^2}\int_{0}^{\frac{hn_0\omega}{2\cos\theta}} H_0^{(1)}(t)\;t\;dt \nonumber \\
&=&\frac{1}{(n_0\omega)^2}\int_{0}^{\frac{hn_0\omega}{2\cos\theta}}\frac{d}{dt}(tH_1^{(1)}(t)) dt \nonumber \\
&=&\frac{1}{(n_0\omega)^2}\left[{\frac{hn_0\omega}{2\cos\theta}}H_1^{(1)}\left({\frac{hn_0\omega}{2\cos\theta}}\right)+\frac{2i}{\pi} \right],  
\end{eqnarray}
where we have used the formula 9.1.9 of \cite{abrste65} to obtain the constant term. 

Now the integral $I$ can be further reduced to 
\begin{eqnarray}
I(\mathbf{r}) &=&\frac{1}{(n_0\omega)^2}\left[8\int_{0}^{\frac{\pi}{4}}{\frac{hn_0\omega}{2\cos\theta}}H_1^{(1)}\left({\frac{hn_0\omega}{2\cos\theta}}\right)d\theta+{4i} \right],
\label{HankelIntegral2d}
\end{eqnarray}
where we need to use the Bessel function to evaluate the above integral. This yields the formula \eqref{HankelIntegral2dlemma}. 
\end{proof}

Therefore, we approximate the self term as  
\begin{eqnarray}
\int_{c_j}g(\mathbf{r}, \mathbf{r}_0)d\mathbf{r}_0 &\approx& \int_{c_j} \left[v_0(\mathbf{r},\mathbf{r}_0)\frac{i\sqrt{\pi}}{2} H_0^{(1)}(\omega n(\mathbf{r})|\mathbf{r}- \mathbf{r}_0|)\right]d\mathbf{r}_0\nonumber\\ 
&\approx& v_0(\mathbf{r},\mathbf{r})\frac{i\sqrt{\pi}}{2} I(\mathbf{r})= \frac{1}{2\sqrt{\pi}}\frac{i\sqrt{\pi}}{2} I(\mathbf{r}) = \frac{i}{4} I(\mathbf{r}),
\end{eqnarray} 
where $I$ is defined in \eqref{HankelIntegral2dlemma}.    

\subsection{3-D self-interaction terms}
Near the source in the 3-D case, from formulas \eqref{babich} and \eqref{termf} we have by keeping the leading-order term, 
\begin{equation}
\label{babich1} g(\mathbf{r},\mathbf{r}_0) = n_0\frac{e^{i\omega\tau}}{4\pi\tau},
\end{equation}
where $n_0$ is the slowness at the source, and we refer to \cite{qiayualiuluobur16} for the reduction process. 

Formula \eqref{babich1} represents the 3-D Green's function near the source. What we need is its integration over the 3-D cube of side $h$. We first consider a special case. 
\subsubsection{Integration of Hankel in 3-D: a special case}
\label{specialcase}
The basic idea of the calculation is to write the integral of \eqref{babich1} over a cube of side $h$ with source point at the center, taken as the origin. The faces of the cube are the planes $x=\pm\frac{h}{2}$, $y=\pm\frac{h}{2}$, and $z=\pm\frac{h}{2}$. 
 
Since we have 
\begin{eqnarray}
h_0^{(1)}(r) &=& j_0(r)+i y_0(r) \quad\quad [\cite{abrste65}, 10.1.1] \nonumber\\
&=&\frac{\sin r - i\cos r}{r} \quad\quad [\cite{abrste65}, 10.1.11, 10.1.12] \nonumber\\
&=& -i\;\frac{e^{i\;r}}{r}, 
\label{e1}
\end{eqnarray}
which is a scaled form of \eqref{babich1}, we start with the integral of $h_0^{(1)}(r)$, where $r$ is the spherical polar radius.

We need to integrate this function over the cube of side $h$, so we need 
\begin{eqnarray}
Q = \int_{S_2}\int_0^{r_1(\mathbf{k})} h_0^{(1)}(r) r^2dr\;d\mathbf{k},
\label{e2} 
\end{eqnarray}
where $\mathbf{k}$ is the unit vector direction of $\mathbf{x}=(x,y,z)$ and $r_1(\mathbf{k})$ is the value of $r$ where the ray $\mathbf{k}$ pierces through the surface of the cube of ``radius'' $\frac{h}{2}$, and $S_2$ is the unit spherical surface in $R^3$. 

We will integrate this over the tetrahedron, 
\[0<z<\frac{h}{2},\quad 0<x<z,\quad 0<y<x.\]
So now let us consider the parameterization of $\mathbf{k}$. Let the ray in direction $\mathbf{k}$ cut the plane $z=1$ in $(\xi,\eta,1)$ or $(\rho\cos\phi,\rho\sin\phi,1)$, where $\rho$ is the cylindrical polar radius and $\phi$ is the azimuthal angle, varying from $0$ to $\frac{\pi}{4}$. We will integrate over the region $0<z<\frac{h}{2}, 0<x<y$, drawn as quite a small tetrahedron near the origin $\mathbf{0}$. The whole cube of side $h$ contains 48 of these small tetrahedra. 

Equation \eqref{e2} may be rewritten 
\begin{eqnarray}
Q &=& -48 i \int_{S_2^{'}} \int_{0}^{r_1(\mathbf{k})}e^{i\;r}rdr\;d\mathbf{k},
\label{e3}
\end{eqnarray}
where $d\mathbf{k}$ is the surface element of the unit sphere on which $\mathbf{k}$ lies, and $S_2^{'}$ is a part of the unit spherical surface to be made precise below. 

We parameterize $\mathbf{k}$ in the first instance by $(\xi,\eta)$, which are $x$ and $y$ coordinates in the plane $z=1$. 
Thus, we have 
\begin{eqnarray}
\mathbf{k} = \frac{(\xi,\eta,1)}{\sqrt{1+\xi^2+\eta^2}} = \frac{(\rho\cos\phi,\rho\sin\phi,1)}{\sqrt{1+\rho^2}}.
\label{e4}
\end{eqnarray}

Let us radically project the element $d\xi d\eta=\rho\;d\rho\;d\phi$ onto the unit sphere. Thus, 
\begin{equation}
d\mathbf{k} = \frac{1}{1+\rho^2} \cos\psi\;d\xi\;d\eta.
\end{equation}
Here $\psi$ is the angle between the normal to the plane $z=1$ and $\mathbf{k}$, i.e., 
\begin{equation}
\cos\psi = (0,0,1)^T\cdot (\xi,\eta,1)^T\frac{1}{\sqrt{1+\rho^2}}=\frac{1}{\sqrt{1+\rho^2}}.
\label{e6}
\end{equation}
Thus, 
\begin{equation}
d\mathbf{k} = \frac{d\xi\;d\eta}{(1+\rho^2)^{\frac{3}{2}}}= \frac{\rho\;d\rho\;d\phi}{(1+\rho^2)^{\frac{3}{2}}}. 
\label{e7}
\end{equation}
From \eqref{e2} and \eqref{e7}, we get 
\begin{eqnarray}
Q &=& -48i\int_{0}^{\frac{\pi}{4}}d\phi \int_{0}^{r_1(\rho)} r e^{ir} dr \frac{\rho\;d\rho}{(1+\rho^2)^{\frac{3}{2}}}, 
\label{e8}\\
r_1(\rho) &=& \frac{h}{2}(1+\rho^2)^{\frac{1}{2}}. 
\label{e9}
\end{eqnarray}
The $r$ integral can be done easily using integration by parts: 
\begin{eqnarray}
\int_{0}^{r_1} r e^{ir}dr &=& [\frac{1}{i}re^{ir}]_0^{r_1}-\int_{0}^{r_1}\frac{1}{i}e^{ir}dr \nonumber\\
                                    &=& \frac{1}{i}r_1e^{ir_1}+[e^{ir}]_{0}^{r_1} \nonumber\\
                                    &=& -i r_1e^{ir_1}+e^{ir_1} - 1\nonumber\\
                                    &=& e^{ir_1}(1-ir_1)-1. \label{e10}
\end{eqnarray}

So, from \eqref{e3}, 
\begin{eqnarray}
Q= -48 i \int_{0}^{\pi/4}d\phi \int_0^{\frac{1}{\cos\phi}}(e^{ir_1}(1-ir_1)-1)\frac{\rho\;d\rho}{(1+\rho^2)^{\frac{3}{2}}}.
\label{e11}
\end{eqnarray}
Set
\begin{eqnarray}
\rho^{'} = \rho\cos\phi,\quad\quad d\rho = \frac{1}{\cos\phi}d\rho^{'}. 
\label{e11a}
\end{eqnarray}
Then 
\begin{eqnarray}
Q = -48i \int_{0}^{\pi/4}d\phi\int_{0}^{1}d\rho'[e^{ir_1(\rho)}(1-ir_1(\rho))-1]\frac{\rho}{(1+\rho^2)^{\frac{3}{2}}}\frac{1}{\cos\phi}. 
\label{e12}
\end{eqnarray}
The function $f(\rho',\phi)$ is given by 
\begin{eqnarray}
f(\rho',\phi) = [e^{ir_1(\rho)}(1-ir_1(\rho))-1]\frac{\rho}{(1+\rho^2)^{\frac{3}{2}}}\frac{1}{\cos\phi},
\label{e13}
\end{eqnarray}
where 
\begin{eqnarray}
\rho &=& \frac{\rho'}{\cos\phi}, \nonumber\\
r_1(\rho)&=& \frac{h}{2}(1+\rho^2)^{\frac{1}{2}}
\label{e14}
\end{eqnarray}
by equations \eqref{e9} and \eqref{e11a}. 

\subsubsection{Integration of Hankel in 3-D: generic case}
Formula \eqref{babich1} represents the 3-D Green's function near the source. What we need is its integration over the 3-D cube of side $h$. To do that, we just need to carry out a coordinate transformation to transfer the integral of $h_0^{(1)}$ derived in \cref{specialcase} to our current setting. 

We have near the source
\begin{eqnarray}
\label{green3d} g(\mathbf{r},\mathbf{r}_0) &=& n_0\frac{e^{i\omega\tau}}{4\pi\tau} \nonumber\\
&=&\frac{i n_0\omega}{4\pi} (-i)\frac{e^{i\omega\tau}}{\omega\tau} \nonumber\\
&=&\frac{i n_0\omega}{4\pi} h_0^{(1)}({\omega\tau}) \nonumber\\
&=&\frac{i n_0\omega}{4\pi} h_0^{(1)}({\omega n_0 |\mathbf{r}-\mathbf{r}_0|}) \nonumber\\
&=&\frac{i n_0\omega}{4\pi} h_0^{(1)}({\omega n_0 r}), 
\end{eqnarray}
where $r=|\mathbf{r}-\mathbf{r}_0|$. 

The integration of the above Green's function in the cell centered at the source will be 
\begin{eqnarray}
I &=&\frac{in_0\omega}{4\pi}\int_{S_2}\int_0^{r_1(\mathbf{k})} h_0^{(1)}(\omega n_0 r) r^2dr\;d\mathbf{k} \nonumber \\
&=&\frac{in_0\omega}{4\pi}\int_{S_2}\int_0^{\omega n_0 r_1(\mathbf{k})} h_0^{(1)}(t) \frac{t^2}{\omega^2n_0^2}\frac{dt}{\omega n_0}\;d\mathbf{k} \nonumber \\
&=&\frac{i}{4\pi\omega^2 n_0^2}\int_{S_2}\int_0^{\omega n_0 r_1(\mathbf{k})} h_0^{(1)}(t)\;t^2\;dt\;d\mathbf{k} \nonumber \\
&=&\frac{i}{4\pi\omega^2 n_0^2}Q_s, 
\end{eqnarray}
where $Q_s$ is the scaled integral of $Q$ as defined in \eqref{e12} and hence is defined by the following integration, 
\begin{eqnarray}
Q_s &=& -48i \int_{0}^{\pi/4}d\phi\int_{0}^{1}d\rho'\left[e^{ir_1(\rho)}(1-ir_1(\rho))-1\right]\frac{\rho}{(1+\rho^2)^{\frac{3}{2}}}\frac{1}{\cos\phi} \nonumber \\
&=& -48i \int_{0}^{\pi/4}d\phi\int_{0}^{1}d\rho' f(\rho',\phi). 
\label{e312}
\end{eqnarray}
The function $f(\rho',\phi)$ is given by 
\begin{eqnarray}
f(\rho',\phi) = \left[e^{ir_1(\rho)}(1-ir_1(\rho))-1\right]\frac{\rho}{(1+\rho^2)^{\frac{3}{2}}}\frac{1}{\cos\phi},
\label{e313}
\end{eqnarray}
where 
\begin{eqnarray}
\rho= \frac{\rho'}{\cos\phi} \;\mbox{ and }\; r_1(\rho)=\omega n_0 \frac{h}{2}(1+\rho^2)^{\frac{1}{2}}
\label{e314}
\end{eqnarray}
by equation \eqref{e14}. Here $r_1(\rho)$ is scaled by the factor $\omega n_0$.  
	
\ylrev{	
\section{Numerical accuracy of the truncated 3-D H-B ansatz}\label{sec:go2hb}	
Since the H-B ansatz \eqref{hb} is based on Hankel functions, we may directly analyze the effect of accuracy of the H-B ingredients, 
such as the phase and H-B coefficients, on wave solution. However, since, {\it away from the point source}, the H-B ansatz is equivalent to the traditional geometrical optics \cite{qiayualiuluobur16,luqiabur16}, to avoid some technical details we will consider the following truncated geometrical-optics ansatz (valid for 3-D wave motion) as a proxy for the truncated H-B ansatz away from the point source:  
\begin{eqnarray}
g_{\rm GO}(\mathbf{r},\mathbf{r}_0) = \sum_{s=0}^N \frac{A_s(\mathbf{r},\mathbf{r}_0)}{(i\omega)^s}	e^{i\omega\tau(\mathbf{r},\mathbf{r}_0)},
\end{eqnarray}	
where $\mathbf{r}_0$ is the source, $\mathbf{r}$ is the observation point, $N$ is an integer, $\tau$ is the phase satisfying the eikonal equation, and $\{A_s\}_{s=0}^N$ are amplitudes satisfying transport equations \cite{avikel63,bab65,qiayualiuluobur16,luqiabur16}.  

Since these $A_s$ functions are directly linked to the H-B coefficients $v_s$ \cite{bab65,qiayualiuluobur16,luqiabur16}, we can assume that these $A_s$ functions are computed to the same orders of accuracy as $v_s$ away from the point source.  Starting from this assumption, we briefly analyze the error between 
$g_{\rm GO}(\mathbf{r},\mathbf{r}_0)	$ and its numerical solution $g^h_{\rm GO}(\mathbf{r},\mathbf{r}_0)$ for any point $\mathbf{r}$ away from the 
source $\mathbf{r}_0$. 	

When $N=0$, we have 
\begin{eqnarray}
|g_{\rm GO}(\mathbf{r},\mathbf{r}_0)-g_{\rm GO}^h(\mathbf{r},\mathbf{r}_0)|&=& |A_0(\mathbf{r},\mathbf{r}_0)e^{i\omega\tau(\mathbf{r},\mathbf{r}_0)}-A_0^h(\mathbf{r},\mathbf{r}_0)e^{i\omega\tau^h(\mathbf{r},\mathbf{r}_0)}|\nonumber\\
&\leq& O(h_0^3) + O(\omega h_0^5).
\end{eqnarray}	

When $N=1$, we have 
\begin{eqnarray}
|g_{\rm GO}(\mathbf{r},\mathbf{r}_0)-g_{\rm GO}^h(\mathbf{r},\mathbf{r}_0)|&\leq& |A_0(\mathbf{r},\mathbf{r}_0)e^{i\omega\tau(\mathbf{r},\mathbf{r}_0)}-A_0^h(\mathbf{r},\mathbf{r}_0)e^{i\omega\tau^h(\mathbf{r},\mathbf{r}_0)}|\nonumber\\
&& +\frac{1}{\omega} |A_1(\mathbf{r},\mathbf{r}_0)e^{i\omega\tau(\mathbf{r},\mathbf{r}_0)}-A_1^h(\mathbf{r},\mathbf{r}_0)e^{i\omega\tau^h(\mathbf{r},\mathbf{r}_0)}|\nonumber\\
&\leq& O(h_0^3) + O(\omega h_0^5) + O\left(\frac{h_0}{\omega}\right) + O(h_0^5) \nonumber \\
&=&O(h_0^3) + O(\omega h_0^5) + O\left(\frac{h_0}{\omega}\right). 
\end{eqnarray}	

Therefore, {\it away from the point source} we will use the above geometrical-optics estimates as the proxy for the truncated H-B estimates in the total error estimates. On the other hand, {\it near the source but  excluding the source}, such truncated H-B estimates also hold since we have the following two observations: (1) the H-B ansatz is an uniformly asymptotic solution to the point-source Helmholtz equation so that it can be treated as the exact solution of the point-source equation, and (2) the computed H-B ingredients in the truncated H-B expansion are initialized near the point source according to specified orders of accuracy.
}
	
	\bibliographystyle{plain}
	\bibliography{references,RayFEMreferences,myref}

\end{document}